\documentclass[conference, compsocconf,  letterpaper, 10pt, times]{IEEEtran}
\usepackage{amsthm}
\usepackage{amsmath}
\usepackage{amssymb}
\theoremstyle{definition}
\newtheorem{definitionInt}{Definition}[section]
\newtheorem{attackInt}{Attack}
\newtheorem{exampleInt}{Example}[section]
\newtheorem{decisionproblem}{Problem}
\theoremstyle{plain}
\newtheorem{theorem}{Theorem}[section]

\newtheorem{proposition}{Proposition}[section]
\newtheorem{fact}{Fact}[section]

\usepackage{multicol}
\usepackage [table]{xcolor} 
\usepackage{array}
\usepackage{balance}
\usepackage{afterpage}
\usepackage{booktabs}
\usepackage{subcaption}
\usepackage{xfrac}
\usepackage{pgfplots}
\pgfplotsset{compat=newest}
\usepgfplotslibrary{units}
\usepackage{pgfplotstable}
\usepackage[noend,ruled]{algorithm2e}
\SetAlFnt{\small\sffamily}

\SetCommentSty{mycommfont}

\makeatletter
\newcommand{\thickhline}{%
    \noalign {\ifnum 0=`}\fi \hrule height 1.5pt
    \futurelet \reserved@a \@xhline
}
\newcolumntype{"}{@{\hskip\tabcolsep\vrule width 1.5pt\hskip\tabcolsep}}
\makeatother

\usepackage{theoremref} 
\usepackage{multirow} 
\usepackage{cite}
\usepackage{listings}
\usepackage{comment}
\usepackage{todonotes}
\usepackage{caption}
\usepackage{balance}
\usepackage{paralist}
\usepackage{stmaryrd}
\usepackage{rotating}
\usepackage{proof}

\usepackage{cals}

\makeatletter
\newcommand{\xRightarrow}[2][]{\ext@arrow 0359\Rightarrowfill@{#1}{#2}}
\makeatother

\usepackage{float}

\newfloat{derivation}{thp}{lop}
\floatname{derivation}{Derivation}

\usepackage{pgf}
\usepackage{tikz}
\usetikzlibrary{arrows,automata}
\usetikzlibrary{positioning}
\usetikzlibrary{shapes.geometric}
\usetikzlibrary{plotmarks}
\usetikzlibrary{patterns}
\usetikzlibrary{calc}
\usetikzlibrary{shapes.misc}

\usetikzlibrary{pgfplots.groupplots}

\usepackage[normalem]{ulem}

 \newenvironment{example}
 {\begin{exampleInt} } 
 { $\hfill \blacksquare$\end{exampleInt} }

  \newenvironment{definition}
 {\begin{definitionInt} } 
 { $\hfill \square$\end{definitionInt} }
 
   \newenvironment{problem}
 {\begin{decisionproblem} } 
 { $\hfill \square$\end{decisionproblem} }

 \newcommand{\concat}{\cdot}

 \newcommand{\confidentiality}{data confidentiality}

  \newcommand{\toolText}{Angerona}
 \newcommand{\tool}{\textsc{\toolText{}}}
 \newcommand{\problog}{\textsc{ProbLog}}
 
 \newcommand{\atklog}{\textsc{AtkLog}}
 \newcommand{\datalog}{\textsc{Datalog}}

  \newcommand{\acf}{PDP}

  \newcommand{\atom}[2]{{#1}\mathalpha{::}{#2}}
\usepackage{footnote}
\usepackage{pifont}

\usepackage{setspace}

\usepackage{xstring}

\newcommand{\shortVersion}{false} %set to true to generate the short paper version
\newcommand{\techReportAppendix}[1]{%
    \IfEqCase{\shortVersion}{%
        {false}{Appendix \ref{#1}}%
        {true}{\cite{technicalReport}}%
    }
}%

\newcommand{\techReportAppendices}[2]{%
    \IfEqCase{\shortVersion}{%
        {false}{Appendices \ref{#1}--\ref{#2}}%
        {true}{\cite{technicalReport}}%
    }
}%

\newcommand{\onlyTechReport}[1]{%
    \IfEqCase{\shortVersion}{%
        {false}{#1}%
    }
}%

\newcommand{\onlyShortVersion}[1]{%
    \IfEqCase{\shortVersion}{%
        {true}{#1}%
    }
}%

\pgfdeclareimage[height=12pt, width=12pt]{time}{figures/clock.png}
\pgfdeclareimage[height=16pt, width=16pt]{files}{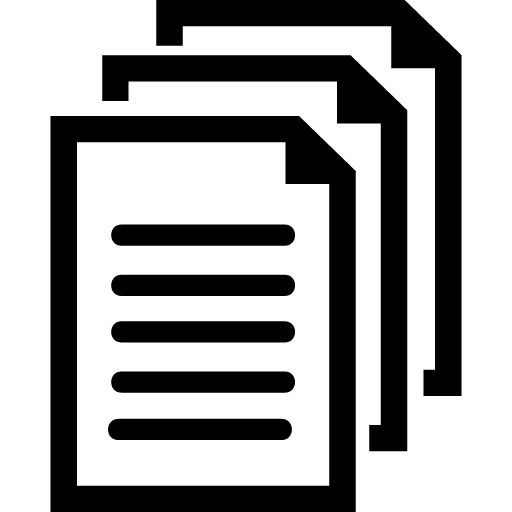}
\pgfdeclareimage[height=16pt, width=16pt]{db}{figures/db.png}
\pgfdeclareimage[height=16pt, width=16pt]{dice}{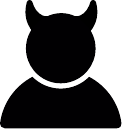}

\newcommand{\para}[1]{\smallskip\noindent{\bf #1.}\hspace{2pt}}%{\smallskip\noindent{\bf #1.}\hspace{2pt}}

\AtBeginDocument{%
  \mathchardef\mathcomma\mathcode`\,
  \mathcode`\,="8000 
}
{\catcode`,=\active
  \gdef,{\mathcomma\discretionary{}{}{}}
}

\title{
Securing  Databases from Probabilistic Inference
 }

\author{\IEEEauthorblockN{Marco Guarnieri}
\IEEEauthorblockA{
	 Institute of Information Security\\
       Department of Computer Science\\
       ETH Zurich, Switzerland\\
       {\small {\tt marco.guarnieri@inf.ethz.ch}}}
\and
\IEEEauthorblockN{Srdjan Marinovic}
\IEEEauthorblockA{The Wireless Registry, Inc.\\
		Washington DC, US\\
       	{\small {\tt srdjan@wirelessregistry.com}}}
\and
\IEEEauthorblockN{David Basin}
\IEEEauthorblockA{Institute of Information Security\\
       Department of Computer Science\\
       ETH Zurich, Switzerland\\
       {\small {\tt basin@inf.ethz.ch}}}
}
\begin{document}

\maketitle

\begin{abstract}
Databases can leak confidential information when users combine query results with probabilistic data dependencies and prior knowledge.
Current research offers mechanisms that  either handle a limited class of dependencies or lack tractable enforcement algorithms. 
We propose a foundation for Database Inference Control based on \problog{}, a probabilistic logic programming language. 
We leverage this foundation to develop \tool{}, a provably secure enforcement mechanism that prevents information leakage in the presence of probabilistic dependencies. 
We then provide a tractable inference algorithm for a practically relevant fragment of \problog{}.
We empirically evaluate \tool{}'s performance showing that it scales to relevant security-critical problems.
\end{abstract}

\section{Introduction}

Protecting the confidentiality of sensitive data stored in data\-bases requires protection from both \textit{direct} and \textit{indirect} access.
The former happens when a user observes query results, and the latter happens when a user infers sensitive information by combining results with  {external} information, such as data dependencies or prior knowledge. 
Controlling indirect access to data is often referred to as \emph{Data\-base Inference Control}~\cite{farkas2002inference} (DBIC).
This topic has attracted considerable attention in recent years, and current research considers different sources of external information, such as the data\-base schema~\cite{chen2007protection,hale1997catalytic,hinke1997protecting,su1991controlling,qian1993detection,su1987data,guarnieri2016strong}, the system's semantics~\cite{guarnieri2016strong}, statistical information~\cite{dobkin1979secure, chin1982auditing, domingo2002inference, adam1989security, denning1980secure}, exceptions~\cite{guarnieri2016strong}, error messages~\cite{Kabra:2006:RIL:1142473.1142489}, user-defined functions~\cite{Kabra:2006:RIL:1142473.1142489},  and data dependencies~\cite{bonatti1995foundations,toland2010inference,brodsky2000secure,yip1998data,morgenstern1987security,morgenstern1988controlling, thuraisingham1987security}.

An important and relevant class of data dependencies are probabilistic dependencies, such as those found in genomics~\cite{humbert2013addressing,lauritzen2003graphical, koller2009probabilistic},  social networks~\cite{he2006inferring}, and location tracking~\cite{mathew2012predicting}. 
Attackers can exploit these dependencies to infer sensitive information with high confidence.
To effectively prevent probabilistic inferences, DBIC mechanisms should
\begin{inparaenum}[(1)]
\item support a large class of probabilistic dependencies, and
\item  have tractable runtime performance.
\end{inparaenum}
The former is needed to express different attacker models. The latter is necessary for mechanisms to scale to real-world databases.

Most existing DBIC mechanisms support only precise data dependencies~\cite{bonatti1995foundations,toland2010inference,brodsky2000secure,thuraisingham1987security,yip1998data} or just limited classes of probabilistic dependencies~\cite{morgenstern1987security,morgenstern1988controlling,katos2011framework,chen2007protection,chen2006database,wiese2010keeping,
hale1997catalytic}.
As a result, they cannot reason about the complex probabilistic dependencies that exist in many realistic settings.
Mardziel et al.'s mechanism~\cite{mardziel2013dynamic} instead supports arbitrary probabilistic dependencies, but no 
complexity bounds have been established and  
their algorithm appears to be intractable.

\para{Contributions}
We develop a tractable and practically useful DBIC mechanism based on probabilistic logic programming. 

First, we develop \atklog{}, a language for formalizing users' beliefs and how they evolve while interacting with the system.
\atklog{} builds on \problog{}~\cite{de2007problog,fierens2015inference, de2015probabilistic}, a state-of-the-art probabilistic extension of \datalog{}, and extends its semantics by building on three key ideas from~\cite{clarkson2005belief,mardziel2013dynamic,kenthapadi2005simulatable}: (1) users' beliefs can be represented as probability distributions, (2) belief revision can be performed by conditioning the probability distribution based on the users' observations, and (3) rejecting queries as insecure may leak information. 
By combining \datalog{} with probabilistic models and belief revision based on users' knowledge,  \atklog{} provides a natural and expressive language to model users' beliefs  and thereby serves as a foundation for DBIC in the presence of probabilistic inferences.

Second, we identify acyclic \problog{} programs, a class of programs where probabilistic inference's data complexity is \textsc{PTime}.
We precisely characterize this class and  develop a dedicated inference engine.
Since \problog{}'s inference is intractable in general, we see acyclic programs as an {essential} building block to effectively using \atklog{} for DBIC.

Finally,  we present  \tool{}\footnote{\toolText{} is the Roman goddess of silence and secrecy, and 
She is the keeper of the city's sacred, and secret, name.}, a novel DBIC mechanism that secures databases against probabilistic inferences.
We prove that \tool{} is secure with respect to any \atklog{}-attacker. 
In contrast to existing mechanisms, \tool{} provides precise tractability and completeness guarantees for a practically relevant class of  attackers.
We empirically show that \tool{} scales to relevant problems of interest.

\para{Structure}
In \S\ref{sect:motivating:example}, we illustrate the security risks associated with probabilistic data dependencies.
In \S\ref{sect:system:model}, we present our system model, which we formalize in \S\ref{sect:formal:model}.
We introduce \atklog{} in \S\ref{sect:language} and in \S\ref{sect:inference} we present our inference engine for acyclic programs.
In  \S\ref{sect:enforcement}, we present  \tool{}.
We discuss related work in  \S\ref{sect:related:work} and  draw conclusions in~\S\ref{sect:conclusion}.
\onlyShortVersion{
An extended version of this paper  with proofs of all results is available at~\cite{technicalReport}, whereas a prototype of our enforcement mechanism is available at~\cite{prototype}.
}
\onlyTechReport{
A prototype of our enforcement mechanism is available at~\cite{prototype}.
}

\section{Motivating Example}\label{sect:motivating:example}\label{sect:motivating:example:medical}

Hospitals and medical research centres store  large quantities of health-related information for purposes ranging from diagnosis to research.
As this information is extremely sensitive, the databases used must be carefully secured~\cite{hipaa,eulaw}.
This task is, however, challenging due to the dependencies between health-related data items.
For instance, information about someone's hereditary diseases or genome can  be inferred from information about her relatives. 
Even seemingly non-sensitive information, such as someone's job or habits, may leak sensitive health-related information such as her predisposition to diseases.
Most of these dependencies can be formalized using probabilistic models developed by medical researchers.

Consider a database storing information about the smoking habits of patients and whether they have been diagnosed with lung cancer. 
The database contains the tables $\mathit{patient}$, $\mathit{smokes}$, $\mathit{cancer}$, $\mathit{father}$, and $\mathit{mother}$.
The first table contains all patients, the second contains all regular smokers, the third contains all diagnosed patients, and the last two associate patients with their parents.
Now consider the following probabilistic model:
\begin{inparaenum}[(a)]
\item every patient has a $5\%$ chance of developing cancer,
\item for each parent with cancer, the likelihood that a child develops cancer increases by $15\%$, and
\item if a patient smokes regularly, his probability of developing cancer increases by $25\%$.
\end{inparaenum}
We intentionally work with a simple model since, despite its simplicity, it illustrates the challenges of securing data with probabilistic dependencies. 
We refer the reader to medical research for more realistic probabilistic models~\cite{pmid7895211,pmid23534801}.

The database is shared between different medical researchers, each conducting a research study on a subset of the patients.
All researchers have access to the $\mathit{patient}$, $\mathit{smokes}$, $\mathit{father}$, and $\mathit{mother}$ tables.
Each researcher, however, has access only to the subset of the $\mathit{cancer}$ table associated with the patients that opted-in to his research study.
We want to protect our database against a malicious  researcher whose goal is to infer the health status of patients not participating in the study.
This is challenging since restricting direct access to the $\mathit{cancer}$ table is insufficient.
Sensitive information may be leaked even by queries involving only authorized data.
For instance, the attacker may know that the patient $\mathit{Carl}$, which has not disclosed his health status, smokes regularly.
From this, he can infer that $\mathit{Carl}$'s probability of  developing lung cancer is, at least, $30\%$.
If, additionally, $\mathit{Carl}$'s parents opted-in to the research study and both have cancer, the attacker can directly infer that the probability of $\mathit{Carl}$ developing lung cancer is $60\%$ by accessing his parents' information.

Security mechanisms that ignore such probabilistic dependencies allow attackers to infer sensitive information.
An alternative is to use standard DBIC mechanisms and encode all dependencies as precise, non-probabilistic, dependencies.
This, however, would result in an unusable system.
Medical researchers, even honest ones, would  be able to access the health-related status only of those patients whose relatives also opted-in to the user study, independently of the amount of leaked information, which may be negligible.
Hence, to secure the database and retain usability, it is essential to reason about the probabilistic dependencies.

\section{System Model}\label{sect:system:model}

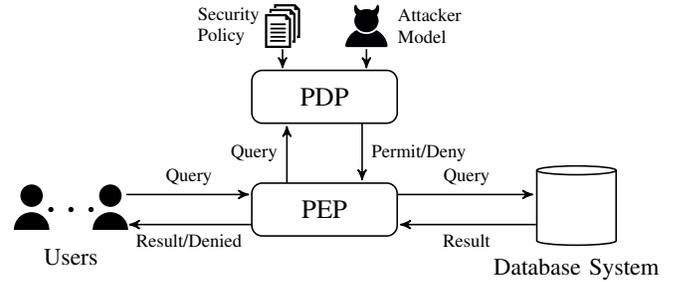
\begin{figure}
\centering

\begin{tikzpicture}[->,>=stealth',shorten >=1pt,auto, semithick]

%%% PDP
\node[fill=none,draw=black, shape = rectangle, rounded corners, inner sep=0pt, outer sep=0pt, minimum height = 20pt, minimum width = 55pt]  (pep) at (0,0) {PEP};

\node[fill=none,draw=black, shape = rectangle, rounded corners, inner sep=0pt, outer sep=0pt, minimum height = 20pt, minimum width = 55pt]  (pdp) at ($(pep) + (0,+1.5)$) {PDP};

%%%% SECURITY POLICY
\node[anchor = north west] (accessControlPolicy) at ($(pdp.north west)+(0,1)$) {\pgfuseimage{files}};
\node[anchor = east,text width=1cm, font=\scriptsize] (accessControlPolicyLabel) at ($(accessControlPolicy)$) {Security Policy};
% 
%
%%%%% PROBABILISTIC MODEL 
%%
\node[anchor = north east] (probabilisticModel) at ($(pdp.north east) + (0,1)$) {\pgfuseimage{dice}}; 
\node[anchor = west,text width=1cm, font=\scriptsize] (probabilisticModelLabel) at ($(probabilisticModel.west) + (.7,0)$) {Attacker Model};
\draw[->, black] ($(probabilisticModel.south) + (0,.1)$) -- ($(probabilisticModel) + (0,-.6)$);
\draw[->, black]  let \p1 = (probabilisticModel), \p2 = (accessControlPolicy), \p3 = ($(probabilisticModel) + (0,-.6)$) in ($(accessControlPolicy.south) + (0,.1)$) -- (\x2,\y3);

%%% LAST USER
\node[shape=circle,fill=none,inner sep=0pt, minimum size=2pt, anchor=south, outer sep=0pt, left = 52.5pt of pep](hiddenNeck1)  {}; 
\node[shape=circle,fill=none,inner sep=0pt, minimum size=2pt, anchor=south, outer sep=0pt, left = 47.5pt of pep](hiddenNeck3)  {}; 
\node[shape=semicircle,fill=black,inner sep=3pt, anchor=south, outer sep=0pt, below =0pt of hiddenNeck1](body1)  {}; 
\node[shape=circle,fill=black,inner sep=3pt, anchor=south, outer sep=0pt, above=0pt of hiddenNeck1](head1)  {};
%%%% DOTS
\node[fill = black, shape = circle,inner sep=0pt,minimum size=2pt,  left = 5pt of hiddenNeck1] (dot1) {};

\node[fill = black, shape = circle,inner sep=0pt,minimum size=2pt,  left = 5pt of dot1] (dot2) {};

\node[fill = black, shape = circle,inner sep=0pt,minimum size=2pt,  left = 5pt of dot2] (dot3) {};
%%% FIRST USER
\node[shape=circle,fill=none,inner sep=0pt, minimum size=2pt, anchor=south, outer sep=0pt, left = 5pt of dot3](hiddenNeck2)  {};
\node[shape=semicircle,fill=black,inner sep=3pt, anchor=south, outer sep=0pt, below = 0pt of hiddenNeck2](body2)  {}; 
\node[shape=circle,fill=black,inner sep=3pt, anchor=south, outer sep=0pt, above=0pt of hiddenNeck2](head2)  {};

%%% USER LABEL
   \node[below = 10pt of dot2] (userLabel) {{\small{Users}}};

%%% DATABASE
\node[shape = cylinder, shape border rotate=90, aspect = 0.5, fill=none,draw=black,  minimum height = 30pt, minimum width = 30pt, right = 52.5pt of pep ]  (db) {};
%%% DATABASE LABEL
   \node[below = 1pt of db] (databaseLabel) {\small{Database System}};

\path ($(hiddenNeck3.east) + (0,.2)$) edge[] node[above] { {\scriptsize Query}} ($(pep.west) + (0,.2)$);   
\path ($(pep.west) + (0,-.2)$) edge[] node[below] { {\scriptsize Result/Denied}}  ($(hiddenNeck3.east) + (0,-.2)$) ;  

\path ($(pep.east) + (0,.2)$) edge[] node[above] { {\scriptsize Query}} ($(db.west) + (0,.2)$);   
\path ($(db.west) + (0,-.2)$) edge[] node[below] { {\scriptsize Result}}  ($(pep.east) + (0,-.2)$);  

\path ($(pep.north) + (-.5,0)$) edge[] node[left] { {\scriptsize Query}} ($(pdp.south) + (-.5,0)$);   
\path ($(pdp.south) + (.5,0)$)  edge[] node[right] { {\scriptsize Permit/Deny}} ($(pep.north) + (+.5,0)$);  

\end{tikzpicture}
\caption{System model.}
\label{fig:systemModel}
\end{figure}

Figure~\ref{fig:systemModel} depicts our system model.
Users interact with two components: a database system and an inference control system, which consists of a Policy Decision Point (PDP) and a Policy Enforcement Point (PEP).
We assume that all communication between users and the components and between the components themselves is over secure channels.

\para{Database System}
The database system manages the system's data.
Its state is a mapping from tables to sets of tuples.

\para{Users}
Each user has a unique account used to retrieve information from the database system by issuing \texttt{SELECT} commands.
Note that these commands do not change the database state. 
This reflects settings where users have only read-access to a database.
Each command is checked by the inference control system and is executed if and only if the command is authorized by the security policy.

\para{Security policy}
The system's security policy consists of a set of \emph{negative permissions} specifying information to be kept secret.
These permissions express bounds on users' beliefs, formalized as probability distributions, about the actual database content. 
Negative permissions are formalized using commands of the form $\mathtt{SECRET}\ q\ \mathtt{FOR}\ u\  \mathtt{THRESHOLD}\ l$, where $q$ is a query, $u$ is a user identifier, and $l$ is a rational number, $0 \leq l \leq 1$. 
This represents the requirement that 
``A user $u$'s belief in the result of $q$ must be less than $l$.''
Namely, the probability assigned by $u$'s belief to $q$'s result must be less than $l$. 
Requirements like ``A user $u$ is not authorized to know the result of $q$'' can be formalized as $\mathtt{SECRET}\ q\ \mathtt{FOR}\ u\  \mathtt{THRESHOLD}\ 1$.
The system also supports commands of the form $\mathtt{SECRET}\ q\ \mathtt{FOR}\ \mathtt{USERS}\ \mathtt{NOT\ IN}\ \{u_1, \ldots, u_n\}\ \mathtt{THRESHOLD}\ l$, which represents the requirement that 
``For all users $u \not\in  \{u_1, \ldots, u_n\}$, $u$'s belief in the result of $q$ must be less than $l$.''

\para{Attacker}
An attacker is a system user with an assigned user account, and each user is a potential attacker.
An attacker's goal is to violate the security policy, that is, to read or infer information about one of the \texttt{SECRET}s with a probability of at least the given threshold.
An attacker can interact with the system and observe its behaviour in response to his commands.
Furthermore, he can reason about this information and infer information by exploiting domain-specific relationships between  data items.
We assume that  attackers know the database schema as well as any integrity constraints on it.

\para{Attacker Model}
An attacker model represents each user's
 initial beliefs about the actual database state and 
 how he updates his beliefs by  interacting with the system and observing its behaviour in response to his commands.
These beliefs may reflect the attacker's knowledge of domain-specific relationships between the data items or prior knowledge.

\para{Inference Control System}
The inference control system protects the confidentiality of database data. 
It consists of a PEP and a PDP, configured with a security policy $P$ and an attacker model $\mathit{ATK}$.
For each user, the inference control system keeps track of the user's beliefs according to $\mathit{ATK}$.

The system intercepts all commands issued by the users. 
When a user $u$ issues a command $c$, the inference control system decides whether $u$ is authorized to execute $c$.
If $c$ complies with the policy, i.e., 
the users' beliefs still satisfy $P$ even after executing $c$, then  the system forwards the command to the database, which executes $c$ and returns its result to $u$. 
Otherwise, it raises a \emph{security exception} and rejects~$c$.

\section{Formal Model}\label{sect:formal:model}

\subsection{Database Model}\label{sect:database:model}

We introduce here background and notation for data\-bases and queries.
Our formalization follows~\cite{abiteboul1995foundations}.

Let ${\cal R}$ be a  countably infinite set representing identifiers of relation schemas.
A \emph{data\-base schema} $D$ is a pair $\langle \Sigma, \mathbf{dom} \rangle$, where $\Sigma$ is a first-order signature and $\mathbf{dom}$ is a fixed domain.
For simplicity, we consider just a single domain.
Extensions to the many-sorted case are straightforward~\cite{abiteboul1995foundations}.
The signature $\Sigma$ consists of a set of \emph{relation schemas} $R \in {\cal R}$, each schema with arity $|R|$, and one constant symbol for each constant in $\mathbf{dom}$.
We interpret constants by themselves in the semantics. 

A \emph{state} $s$ of $D$ is a finite $\Sigma$-structure with domain $\mathbf{dom}$ that interprets each relation schema $R$ by an $|R|$-ary relation over $\mathbf{dom}$.
We denote  by $\Omega_{D}$ the set of all states.
Given a schema $R \in D$, $s(R)$ denotes the tuples that belong to (the interpretation of) $R$ in the state $s$.
We assume that the domain $\mathbf{dom}$ is finite, as is standard for many application areas combining databases and probabilistic reasoning~\cite{suciu2011probabilistic,koller2009probabilistic,getoor2007introduction, de2015probabilistic}. 
In this case, the set of all  states $\Omega_D$ is finite.

A \emph{query} $q$ over a schema $D$ is of the form $\{\overline{x}\,|\,\phi \}$, where $\overline{x}$ is a sequence of variables, $\phi$ is a relational calculus formula over $D$, and $\phi$'s free variables are those in $\overline{x}$.
A \emph{boolean query} is a query $\{\,|\, \phi\}$, also written  as $\phi$, where $\phi$ is a sentence. 
The result of executing a query $q$ on a state $s$, denoted by $[q]^{s}$, is a boolean value in $\{\top,\bot\}$, if $q$ is a boolean query, or a set of tuples otherwise.
Furthermore, given a sentence $\phi$, $\llbracket \phi \rrbracket$ denotes the set  $\{ s \in \Omega_D \mid [\phi]^s = \top \}$.
We denote by $\mathit{RC}$ (respectively $\mathit{RC}_{\mathit{bool}}$) the set of all  relational calculus queries (respectively sentences).
We consider only \emph{domain-independent queries} and we employ the standard relational calculus semantics~\cite{abiteboul1995foundations}.

An \emph{integrity constraint over $D$} is a  relational calculus sentence $\gamma$ over $D$.
Given a state $s$, we say that \emph{$s$ satisfies the constraint $\gamma$} iff $[\gamma]^{s} = \top$.
Given a set of constraints $\Gamma$, $\Omega_{D}^{\Gamma}$ denotes the set of all states satisfying the constraints in $\Gamma$, i.e., $\Omega_{D}^{\Gamma} = \{s \in \Omega_{D}\,|\, \bigwedge_{\gamma \in  \Gamma} [\gamma]^{s} = \top\}$.

\begin{example}\label{example:formal:model:database:schema}
The database associated with the example in \S\ref{sect:motivating:example:medical} consists of five relational schemas $\mathit{patient}$, $\mathit{smokes}$, $\mathit{cancer}$, $\mathit{father}$, and $\mathit{mother}$, where the first three schemas have arity 1 and the last two have arity 2.
We assume that there are only three patients $\mathtt{Alice}$, $\mathtt{Bob}$, and $\mathtt{Carl}$, so the domain $\mathbf{dom}$ is $\{\mathtt{Alice}, \mathtt{Bob}, \mathtt{Carl}\}$.
The integrity constraints are as follows:
\begin{compactitem}
\item $\mathtt{Alice}$, $\mathtt{Bob}$, and $\mathtt{Carl}$ are patients.
\[
	\mathit{patient}(\mathtt{Alice}) \wedge \mathit{patient}(\mathtt{Bob}) \wedge \mathit{patient}(\mathtt{Carl})
\]
\item $\mathtt{Alice}$ and $\mathtt{Bob}$ are $\mathtt{Carl}$'s parents.
\begin{align*}
	& \forall x,y.\, \left(\mathit{father}(x,y) \leftrightarrow \left( x = \mathtt{Bob} \wedge y = \mathtt{Carl} \right) \right)\wedge \\
	& \forall x,y.\, \left(\mathit{mother}(x,y) \leftrightarrow \left( x = \mathtt{Alice} \wedge y = \mathtt{Carl} \right)\right)
\end{align*}
\item $\mathtt{Alice}$ does not smoke, whereas  $\mathtt{Bob}$ and $\mathtt{Carl}$ do.
\[
	\neg\mathit{smokes}(\mathtt{Alice}) \wedge \mathit{smokes}(\mathtt{Bob}) \wedge \mathit{smokes}(\mathtt{Carl})
\]
\end{compactitem}
Given these constraints, there are just $8$ possible database states in $\Omega_D^\Gamma$, which differ only in their $\mathit{cancer}$ relation.
The content of the $\mathit{cancer}$ relation is a subset of $\{\mathtt{Alice}, \mathtt{Bob}, \mathtt{Carl}\}$, whereas the content of the other tables is shown in Figure~\ref{figure:example:database:content:template}. % 
We denote each possible world as $s_{C}$, where the set $C \subseteq \{\mathtt{Alice}, \mathtt{Bob}, \mathtt{Carl}\}$ denotes the users having cancer.
\end{example}

\begin{figure}

\centering
{\small
\begin{tikzpicture}[->,>=stealth',shorten >=1pt,auto, semithick]
 \tikzstyle{every state}=[anchor=north west,rectangle, rounded corners,fill=none,draw=none,text=black,minimum height=1em,
           inner sep=1pt, ultra thin]

\node[state] (patient) at (0,0) { {
	\begin{tabular}{ c } 
		\textbf{patient} \\
		\hline  Alice \\
		\hline  Bob \\
		\hline  Carl \\
	\end{tabular}
}};

\node[state] (smokes) at ($(patient.north east) + (.5,0)$) { {
	\begin{tabular}{ c } 
		\textbf{smokes} \\
		\hline  Bob \\
		\hline Carl \\
	\end{tabular}
}};

\node[state] (father) at ($(smokes.north east) + (.5,0.25)$) { {
		\begin{tabular}{ c  c } 
			\multicolumn{2}{c}{\textbf{father}} \\
			\hline  \multicolumn{1}{c|}{Bob} & Carl \\
	\\
			\multicolumn{2}{c}{\textbf{mother}} \\
			\hline  \multicolumn{1}{c|}{Alice} & Carl \\
		\end{tabular}
}};

\end{tikzpicture}
}
\caption{The template for all database states, where the content of the \textit{cancer} table is left unspecified.
}\label{figure:example:database:content:template}
\end{figure}

\subsection{Security Policies}\label{sect:security:policies}

Existing access control models for data\-bases are inadequate to formalize security requirements capturing probabilistic dependencies.
For example, SQL cannot express statements like 
``A user $u$'s belief that $\phi$ holds must be less than $l$.''
We present a simple framework, inspired by knowl\-edge-based policies~\cite{mardziel2013dynamic}, for expressing such requirements.

A \emph{$D$-secret} is a tuple $\langle U, \phi, l \rangle$, where $U$ is either a finite set of users in ${\cal U}$ or a co-finite set of users, i.e., $U =  {\cal U} \setminus U'$ for some finite $U' \subset {\cal U}$, $\phi$ is  a relational calculus sentence over $D$, and  $l$ is rational number $0 \leq l \leq 1$ specifying the uncertainty threshold.
Abusing notation, when $U$ consists of a single user $u$, we write $u$ instead of $\{u\}$.
Informally, $\langle U, \phi, l \rangle$ represents that 
for each user $u \in U$, $u$'s belief that  $\phi$ holds in the actual database state must be less than $l$. 
Therefore, a command of the form $\mathtt{SECRET}\ q\ \mathtt{FOR}\ u\  \mathtt{THRESHOLD}\ l$ can be represented as $\langle u, q, l\rangle$, whereas a command $\mathtt{SECRET}\ q\ \mathtt{FOR}\ \mathtt{USERS}\ \mathtt{NOT\ IN}\ \{u_1, \ldots, u_n\}\ \mathtt{THRESHOLD}\ l$ can be represented as $\langle {\cal U} \setminus \{u_1, \ldots, u_n\}, q, l \rangle$.
Finally,  a \emph{$D$-security policy}  is a finite set of $D$-secrets.
Given a $D$-security policy $P$, we denote by $\mathit{secrets}(P,u)$ the set of $D$-secrets associated with the user $u$, i.e., $\mathit{secrets}(P,u) = \{ \langle u, \phi, l \rangle \mid \langle U, \phi, l \rangle \in P \wedge u \in U\}$.
Note that the function $\mathit{secrets}$ is computable since the set $U$ is always either finite or co-finite.

Our framework also allows  the specification of lower bounds.
Requirements of the form 
``A user $u$'s belief that $\phi$ holds must be greater than $l$''
 can be formalized as 
 $\langle u, \neg\phi, 1-l\rangle$ (since the probability of $\neg \phi$ is $1 - { P}(\phi)$, where ${ P}( \phi)$ is $\phi$'s probability).
Security policies can be extended to support secrets over non-boolean queries.
A secret $\langle u, \{ \overline{x} \mid \phi(\overline{x}) \},l \rangle$ can be seen as a shorthand for the set $\{ \langle u, \phi[\overline{x} \mapsto \overline{t}], l \rangle \mid \overline{t} \in \bigcup_{s \in \Omega_D^\Gamma} [\{ \overline{x} \mid \phi(\overline{x}) \}]^s\}$, i.e., 
$u$'s belief in any tuple $\overline{t}$ being in the query's result must be less than $l$.

\begin{example}\label{example:formal:model:access:control:policy}
Let $\mathit{Mallory}$ denote the malicious researcher from \S\ref{sect:motivating:example:medical} and $D$ be the schema from Example~\ref{example:formal:model:database:schema}.
Consider the requirement from \S\ref{sect:motivating:example:medical}:
$\mathit{Mallory}$'s belief in a patient having cancer must be less than $50\%$.
This can be formalized as
$\langle \mathit{Mallory}, \mathit{cancer}(\mathtt{Alice}), \sfrac{1}{2} \rangle$, $\langle \mathit{Mallory}, \mathit{cancer}(\mathtt{Bob}), \sfrac{1}{2} \rangle$, and $\langle \mathit{Mallory}, \mathit{cancer}(\mathtt{Carl}), \sfrac{1}{2} \rangle$,
or equivalently as 
$\langle \mathit{Mallory}, \{ p \!\mid\! \mathit{cancer}(p)\}, \sfrac{1}{2} \rangle$.
In contrast, the requirement 
``For all users $u$ that are not $\mathit{Carl}$, $u$'s belief in $\mathit{Carl}$ having cancer must be less than $50\%$''
 can be formalized as 
$\langle {\cal U} \setminus \{\mathit{Carl}\}, \mathit{cancer}(\mathtt{Carl}), \sfrac{1}{2} \rangle$, where $\mathit{Carl}$ denotes the user identifier associated with Carl.
\end{example}

\subsection{Formalized System Model}\label{sect:formal:model:system:model}
\onlyTechReport{
We now formalize our system model.
We first define a system configuration, which describes the database schema and the integrity constraints.
Afterwards, we define the system's state.
Finally, we define a system run, which represents a possible interaction of users with the system.
}

A \emph{system configuration} is a tuple $\langle D, \Gamma\rangle$, where $D$ is a database schema and $\Gamma$ is a set of $D$-integrity constraints.
Let $C = \langle D, \Gamma\rangle$ be a system configuration.
A \emph{$C$-system state} is a tuple $\langle \mathit{db}, U, P\rangle$, where $\mathit{db} \in \Omega_D^\Gamma$ is a database state, $U \subset {\cal U}$ is a finite set of users, and $P$ is a $D$-security policy.
A \emph{$C$-query} is a pair $\langle u, \phi \rangle$ where $u \in {\cal U}$ is a user and $\phi$ is a relational calculus sentence over $D$.\footnote{Without loss of generality, we focus only on boolean queries \cite{abiteboul1995foundations}.
We can support non-boolean queries  as follows.
Given a database state $\mathit{s}$ and a query $q:=\{ \overline{x}\ |\ \phi\}$, if the inference control mechanism authorizes the boolean query $\bigwedge_{\overline{t} \in [q]^{\mathit{s}}} \phi[\overline{x} \mapsto \overline{t}] \wedge (\forall \overline{x}.\, \phi \rightarrow \bigvee_{\overline{t} \in [q]^{\mathit{s}}} \overline{x} = \overline{t})$,   then we return  $q$'s result, and otherwise we reject  $q$ as unauthorized.}
We denote by $\Omega_C$ the set of all system states and by ${\cal Q}_{C}$ the set of all queries.

A \emph{$C$-event} is a triple $\langle q, a, \mathit{res}\rangle$, where $q$ is a $C$-query in ${\cal Q}_C$, $a \in \{\top, \bot\}$ is a security decision, where $\top$ stands for ``authorized query'' and $\bot$ stands for ``unauthorized query'', and $\mathit{res} \in \{\top,\bot, \dagger\}$ is the query's result, where $\top$ and $\bot$ represent the usual boolean values and $\dagger$ represents that the query was not  executed as  access was denied.
Given a $C$-event $e = \langle q, a, \mathit{res}\rangle$, we denote by $q(e)$ (respectively $a(e)$ and $\mathit{res}(e)$) the query $q$ (respectively the decision $a$ and the result $\mathit{res}$).
A \emph{$C$-history} is a finite sequence of $C$-events.
We denote by ${\cal H}_{C}$ the set of all possible $C$-histories.
Moreover, given a sequence $h$, $|{h}|$ denotes its length, ${h}(i)$ its $i$-th element, and ${h}^i$ the sequence containing the first $i$ elements of ${h}$.
We also denote by $h^0$ the empty sequence $\epsilon$, and  $\cdot$ denotes the concatenation operator.

We now formalize Policy Decision Points.
A \emph{$C$-\acf{}} is a function $f: \Omega_{C} \times {\cal Q}_{C} \times {\cal H}_{C} \to \{\top,\bot\}$ taking as input a system state, a query, and a history and returning the security decision, accept ($\top$) or deny ($\bot$). % stands for ``accept'' and $\bot$ stands for ``deny''.

Let $C$ be a system configuration, $s = \langle \mathit{db}, U, P\rangle$ be a $C$-state, and $f$ be a $C$-\acf{}.
A $C$-history $h$ is \emph{compatible with $s$ and $f$} iff for each $1 \leq i \leq |{h}|$,
(1) $f(s, q(h(i)), {h}^{i-1}) = a({h}(i))$, 
(2) if $a({h}(i)) = \bot$, then $\mathit{res}({h}(i)) = \dagger$, and
(3) if $a({h}(i)) = \top$, then $\mathit{res}({h}(i)) = [\phi]^{\mathit{db}}$, where $q({h}(i)) = \langle u,\phi\rangle$.
In other words,  $h$ is compatible with $s$ and $f$ iff it was generated by the \acf{} $f$ starting in state $s$.

A \emph{$(C,f)$-run} is a pair $\langle s, {h} \rangle$, where $s$ is a system state in $\Omega_{C}$ and ${h}$ is a history in ${\cal H}_{C}$  compatible with $s$ and $f$.
Since all queries are \texttt{SELECT} queries, the system state does not change along the run.
Hence, our runs consist of a state and a history instead of e.g., an alternating sequence of states and actions (as is standard for runs).
We denote by $\mathit{runs}(C,f)$ the set of all $(C,f)$-runs.
Furthermore, given a run $r = \langle \langle \mathit{db}, U, P\rangle, {h} \rangle$, 
we denote by $r^i$ the run $ \langle \langle \mathit{db}, U, P\rangle, {h}^i \rangle$, and we use dot notation to access to $r$'s components.
For instance, $r.\mathit{db}$ denotes the database state $\mathit{db}$ and $r.{h}$ denotes the history.

\begin{exampleInt}\label{example:system:model:run}
Consider the run $r = \langle \langle \mathit{db}, U, P\rangle, h \rangle$, where the database state $\mathit{db}$ is the state $s_{\{\mathtt{A},\mathtt{B}, \mathtt{C}\} }$, where $\mathtt{Alice}$, $\mathtt{Bob}$, and $\mathtt{Carl}$ have cancer, the policy $P$ is defined in Example~\ref{example:formal:model:access:control:policy}, the set of users $U$ contains only $\mathit{Mallory}$, and the history ${h}$ is as follows (here we assume that all queries are authorized):
\begin{compactenum}
\item $\mathit{Mallory}$ checks whether \emph{Carl} smokes.
Thus, ${h}(1) = \langle \langle \mathit{Mallory}, \mathit{smokes}(\mathtt{Carl}) \rangle, \top,\top\rangle$.

\item $\mathit{Mallory}$  checks whether \emph{Carl} is \emph{Alice}'s and \emph{Bob}'s son.
Therefore, ${h}(2)$ is $\langle \langle \mathit{Mallory}, \mathit{father}(\mathtt{Bob}, \mathtt{Carl}) \wedge \mathit{mother}(\mathtt{Alice}, \mathtt{Carl}) \rangle, \top, \top\rangle$.

\item  $\mathit{Mallory}$ checks whether \emph{Alice} has cancer. 
Thus, ${h}(3) = \langle \langle \mathit{Mallory}, \mathit{cancer}(\mathtt{Alice}) \rangle, \top, \top \rangle$.

\item $\mathit{Mallory}$ checks whether \emph{Bob} has cancer. 
Thus, ${h}(4) = \langle \langle \mathit{Mallory}, \mathit{cancer}(\mathtt{Bob}) \rangle, \top, \top \rangle$. $\hfill \blacksquare$
\end{compactenum}
\end{exampleInt}

\newcommand{\atk}{\mathit{ATK}}

\subsection{Attacker Model}\label{sect:probabilistic:attacker:model}

To reason about DBIC, it is essential to precisely define
(1) how users interact with the system,
(2) how they reason about the system's behaviour, 
(3) their initial beliefs about the database state, and 
(4) how these beliefs change by observing the system's behaviour.
We formalize this in an attacker model.

Each user has an initial belief about the database state.
Following~\cite{evfimievski2010epistemic,mardziel2013dynamic, clarkson2005belief, clarkson2009quantifying}, we represent a user's beliefs as a probability distribution over all data\-base states.
Furthermore, users observe the system's behaviour and derive information about the database content.
We formalize a user's observations as an equivalence relation over runs, where two runs are equivalent iff the user's observations are the same in both runs, as is standard in information-flow~\cite{askarov2012learning, askarov2007gradual}.
A user's knowledge is the set of all database states that he considers possible given his observations. 
Finally, we use Bayesian conditioning to update a user's beliefs given his knowledge.

Let $C = \langle D, \Gamma \rangle$ be a system configuration and $f$ be a $C$-\acf{}.
A \emph{$C$-probability distribution} is a discrete probability distribution given by a function $P : \Omega_D^\Gamma \to [0,1]$ such that $\sum_{\mathit{db} \in \Omega_D^\Gamma} {P}(\mathit{db}) = 1$.
Given a set $E \subseteq \Omega_D^\Gamma$, $P(E)$ denotes $\sum_{s \in E} P(s)$.
Furthermore, given two sets  $E', E'' \subseteq \Omega_D^\Gamma$ such that $P(E') \neq 0$, $P(E'' \mid E')$ denotes $\sfrac{P(E'' \cap E')}{P(E')}$ as is standard.
We denote by ${\cal P}_C$ the set of all possible $C$-probability distributions.
Abusing notation, we extend probability distributions to formulae: $P(\psi) =  P(\llbracket \psi \rrbracket)$, where $\llbracket \psi \rrbracket = \{ \mathit{db} \in \Omega_D^\Gamma \mid [\psi]^{\mathit{db}} = \top\}$.

We now introduce \textit{indistinstinguishability}, an equivalence relation used in information-flow control~\cite{hedin2011}.
Let $C$ be a system configuration and $f$ be a $C$-\acf{}.
Given a history $h$ and a user $u \in {\cal U}$,  $h|_u$ denotes the history obtained from $h$ by removing all $C$-events from users other than $u$, namely
 $\epsilon|_u = \epsilon$, and 
if $h = \langle \langle u', q\rangle, a, \mathit{res} \rangle  \cdot h'$, then $h|_u = h'|_u$ in case $u \neq u'$, and $h|_u = \langle \langle u, q\rangle, a, \mathit{res} \rangle \cdot h'|_u$ if $u = u'$.
Given two runs $r = \langle \langle db, U, P\rangle, h\rangle$ and $r' = \langle \langle db', U', P'\rangle, h'\rangle$ in $\mathit{runs}(C,f)$ and a user $u \in {\cal U}$, we say that $r$ and $r'$ are \textit{indistinguishable for $u$}, written $r \sim_u r'$, iff $h|_u = h'|_u$.
This means that $r$ and $r'$ are indistinguishable for a user $u$ iff the system's behaviour in response to $u$'s commands is the same in both runs.
Note that $\sim_u$ depends on both $C$ and $f$, which we generally leave implicit.
Given a run $r$, $[r]_{\sim_u}$ is the equivalence class of $r$ with respect to $\sim_u$, i.e,  $[r]_{\sim_u} = \{r' \in \mathit{runs}(C,f) \mid r' \sim_u r\}$, whereas $\llbracket r \rrbracket_{\sim_u}$ is set of all databases associated to the runs in $[r]_{\sim_u}$, i.e., $\llbracket r \rrbracket_{\sim_u} = \{\mathit{db} \mid \exists U,P, h.\ \langle \langle \mathit{db}, U, P  \rangle, h \rangle \in [r]_{\sim_u} \}$.

\begin{definition}
Let $C = \langle D, \Gamma \rangle$ be a configuration and $f$ be a $C$-\acf{}.
A \emph{$(C,f)$-attacker model} is a function $\atk : {\cal U} \to {\cal P}_C$ associating to each user $u \in {\cal U}$ a $C$-probability distribution representing $u$'s initial beliefs. 
Additionally, for all users  $u \in {\cal U}$ and all states $s \in \Omega_D^\Gamma$,  we require that  $\atk(u)(s) > 0$.
The \emph{semantics of} $\atk$ is  $\llbracket \mathit{ATK} \rrbracket(u,r) = \lambda s \in \Omega_D^\Gamma.\, \atk(u)( s \mid \llbracket r \rrbracket_{\sim_u})$, where $u \in {\cal U}$ and $r \in \mathit{runs}(C,f)$.
\end{definition}

The semantics of an attacker model $\mathit{ATK}$ associates to each user $u$ and each run $r$ the probability distribution obtained by updating $u$'s initial beliefs given his knowledge with respect  to the run $r$.
We informally refer to $\llbracket \mathit{ATK} \rrbracket(u,r)(\llbracket \phi\rrbracket)$ as $u$'s beliefs in a sentence $\phi$ (given a run $r$).

\begin{example}\label{example:formal:model:attacker:model}
The attacker model for the example from \S\ref{sect:motivating:example:medical} is as follows.
Let $X_{\texttt{Alice}}$, $X_{\texttt{Bob}}$, and $X_{\texttt{Carl}}$ be three boolean random variables, representing the probability that the corresponding patient has cancer.
They define the following joint probability distribution, which represents a user's initial beliefs about the actual database state:
	${P}(X_{\texttt{Alice}},X_{\texttt{Bob}}, X_{\texttt{Carl}}) = {P}(X_{\texttt{Alice}}) \cdot {P}(X_{\texttt{Bob}}) \cdot {P}(X_{\texttt{Carl}} \mid X_{\texttt{Alice}},X_{\texttt{Bob}}).$
The probability distributions of these variables are given in Figure~\ref{figure:example:distribution:table:example} and they are derived from the probabilistic model in \S\ref{sect:motivating:example:medical}.
We associate each outcome $(x,y,z)$ of $X_{\texttt{Alice}},X_{\texttt{Bob}}, X_{\texttt{Carl}}$ with the corresponding database state $s_{C}$, where $C$ is the set of patients such that the outcome of the corresponding variable is $\top$.
For each user $u \in {\cal U}$, the distribution $P_u$ is defined as $P_u(s_{C}) = P(X_{\texttt{Alice}} = x,X_{\texttt{Bob}} = y, X_{\texttt{Carl}}=z)$, where $x$ (respectively $y$ and $z$) is $\top$ if $\mathtt{Alice}$ (respectively $\mathtt{Bob}$ and $\mathtt{Carl}$) is in $C$ and $\bot$ otherwise.
Figure~\ref{figure:example:distribution:possible:worlds} shows the probabilities associated with each state in $\Omega_D^\Gamma$, i.e., a user's initial beliefs.
Finally, the attacker model is $\atk = \lambda u \in {\cal U}. P_u$.
\end{example}

\begin{figure}

\centering
{\small 
\begin{tikzpicture}[->,>=stealth',shorten >=1pt,auto, semithick]
 \tikzstyle{every state}=[anchor=north west,rectangle, rounded corners,fill=none,draw=none,text=black,minimum height=1em,
           inner sep=1pt, ultra thin]

\node[state] (s1) at (0,0) { {
	\begin{tabular}{ c | c } 
		& $X_{\mathtt{Alice}}$ \\
		\hline  $\top$ & $\sfrac{1}{20}$ \\
		  $\bot$ & $\sfrac{19}{20}$ \\
	\end{tabular}
}};

\node[state] (s2) at ($(s1.south west) + (0,-.2)$) { {
	\begin{tabular}{ c | c } 
		& $X_{\mathtt{Bob}}$ \\
		\hline  $\top$ & $\sfrac{6}{20}$ \\
		  $\bot$ & $\sfrac{14}{20}$ \\
	\end{tabular}
}};

\node[state] (s3) at ($(s1.north east) + (.5,0)$) { {
\begin{tabular}{c c | c c }
& & \multicolumn{2}{c}{$X_{\mathtt{Carl}}$} \\
$X_{\mathtt{Alice}}$ & $X_{\mathtt{Bob}}$ & $\top$ & $\bot$  \\
\hline
$\top$ & $\top$ & $\sfrac{12}{20}$ & $\sfrac{8}{20}$ \\
$\top$ & $\bot$ & $\sfrac{9}{20}$ & $\sfrac{11}{20}$ \\
$\bot$ & $\top$ & $\sfrac{9}{20}$ & $\sfrac{11}{20}$ \\
$\bot$ & $\bot$ & $\sfrac{6}{20}$ & $\sfrac{14}{20}$ \\
\end{tabular}
}};

\end{tikzpicture}
}

\caption{Probability distribution for the random variables $X_{\mathtt{Alice}}$, $X_{\mathtt{Bob}}$, and $X_{\mathtt{Carl}}$ from Example~\ref{example:formal:model:attacker:model}.}
\label{figure:example:distribution:table:example}
\end{figure}

\begin{figure}%[!hbtp]

\centering
{\small
\begin{tikzpicture}[->,>=stealth',shorten >=1pt,auto, semithick]
 \tikzstyle{every state}=[anchor=north west,rectangle, rounded corners,fill=none,draw=none,text=black,minimum height=1em,
           inner sep=1pt, ultra thin]

\node[state] (s1) at (0,0) { {
		\begin{tabular}{c | c }
	State & Probability  \\\hline
	$s_{\emptyset}$  & $0.4655$ 	\\%$
	$s_{\{\mathtt{A}\}}$ & $0.01925$	\\
	$s_{\{\mathtt{B}\}}$ & $0.15675$	\\
	$s_{\{\mathtt{C}\}}$ & $0.1995$		
	\end{tabular}
}};

\node[state] (s2) at ($(s1.north east) + (.5,0)$) { {
		\begin{tabular}{c | c }
	State & Probability  \\\hline
	$s_{\{\mathtt{A},\mathtt{B}\}}$  & $0.006$ 	\\%$
	$s_{\{\mathtt{A},\mathtt{C}\}}$ & $0.01575$	\\
	$s_{\{\mathtt{B},\mathtt{C}\}}$ & $0.12825$	\\
	$s_{\{\mathtt{A},\mathtt{B}, \mathtt{C}\}}$ & $0.009$		
	\end{tabular}
}};
\end{tikzpicture}
}
\caption{Probability distribution over all database states.
Each state is denoted as $s_C$, where $C$ is the content of the $\mathit{cancer}$ table. Here we denote the patients' names with their initials. %, e.g., $\mathtt{A}$ instead of $\mathtt{Alice}$.
}
\label{figure:example:distribution:possible:worlds}
\end{figure}

\subsection{Confidentiality}\label{sect:confidentiality}

We first define the notion of a secrecy-preserving run for a secret $\langle u, \phi, l \rangle$ and an attacker model $\mathit{ATK}$.
Informally, a run $r$ is secrecy-preserving for $\langle u, \phi, l \rangle$ iff whenever an attacker's belief in the secret $\phi$ is below the threshold $l$, then there is no way for the attacker to increase his belief in $\phi$ above the threshold.
Our notion of secrecy-preserving runs is inspired by existing security notions for query auditing~\cite{evfimievski2010epistemic}.

\begin{definition}
Let $C = \langle D, \Gamma\rangle$ be a configuration, $f$ be a $C$-\acf{}, and $\mathit{ATK}$ be a $(C,f)$-attacker model.    
A run~$r$~is \emph{secrecy-preserving} for a secret $\langle u, \phi, l \rangle$ and $\mathit{ATK}$ iff
for all $0 \leq i <\!|r|$, 
 $\llbracket \mathit{ATK} \rrbracket (u,r^i)( \phi ) < l$ implies $\llbracket \mathit{ATK} \rrbracket (u,r^{i+1})( \phi ) < l$.
\end{definition}

We now formalize our confidentiality notion.
A \acf{} provides \confidentiality{} for an attacker model $\mathit{ATK}$ iff all runs are secrecy-preserving for $\mathit{ATK}$.
Note that our security notion can be seen  as a probabilistic generalization of opacity~\cite{schoepe2015understanding} for the database setting. 
Our notion is also inspired by the semantics of knowledge-based policies~\cite{mardziel2013dynamic}.

\begin{definition}
Let $C = \langle D, \Gamma\rangle$ be a system configuration, $f$ be a $C$-\acf{}, and $\mathit{ATK}$ be a $(C,f)$-attacker model.    
We say that the \acf{} \emph{$f$ provides \confidentiality{} with respect to $C$ and $\mathit{ATK}$} iff for all runs $r = \langle \langle \mathit{db}, U, P\rangle, h \rangle$ in  $\mathit{runs}(C,f)$,
for all users $u \in U$,
for all secrets $s \in \mathit{secrets}(P,u)$,
$r$ is secrecy-preserving for $s$ and $\mathit{ATK}$.
\end{definition}

A \acf{} providing confidentiality ensures that if an  attacker's initial belief in a secret $\phi$ is below the corresponding threshold, then there is no way for the attacker to increase his belief in $\phi$ above the threshold by interacting with the system.
This guarantee does not however apply to \emph{trivial non-secrets}, i.e., those secrets an attacker knows with a probability at least the threshold even before interacting with the system.
No \acf{} can prevent their disclosure since the disclosure does not depend on the attacker's interaction with the database.

\begin{example}\label{example:confidentiality}
Let $r$ be the run given in Example~\ref{example:system:model:run}, $\mathit{ATK}$ be the attacker model in Example~\ref{example:formal:model:attacker:model}, and $u$ be the user \emph{Mallory}.
In the following, $\phi_1$, $\phi_2$, and $\phi_3$ denote $\mathit{cancer}(\mathtt{Carl})$, $\mathit{cancer}(\mathtt{Bob})$, and $\mathit{cancer}(\mathtt{Alice})$ respectively, i.e., the three secrets from Example~\ref{example:formal:model:access:control:policy}.
Furthermore, we assume that the policy contains an additional secret $\langle \mathit{Mallory}, \phi_4, \sfrac{1}{2} \rangle$, where $\phi_4 := \neg\mathit{cancer}(\mathtt{Alice})$.

Figure~\ref{figure:confidentiality:example} illustrates 
\emph{Mallory}'s beliefs about $\phi_1, \ldots, \phi_4$ and
whether the run is secrecy-preserving for the secrets $\phi_1, \ldots, \phi_4$.
The probabilities in the tables can be obtained by combining the states in $\llbracket r^i \rrbracket_{\sim_u}$, for $0 \leq i \leq 4$, and $\llbracket \phi_j \rrbracket$, for $1 \leq j \leq 4$, with the probabilities from Figure~\ref{figure:example:distribution:possible:worlds}.
As shown in Figure~\ref{figure:confidentiality:example}, the run is not secrecy-preserving for the secrets $\phi_1$ and $\phi_2$ as it completely discloses that $\mathit{Alice}$ and $\mathit{Bob}$ have cancer, in the third and fourth steps respectively.
Secrecy-preservation is also violated for the secret $\phi_1$, even though $r$ does not directly disclose any information about $\mathit{Carl}$'s health status.
Indeed, in the last step of the run, \emph{Mallory}'s belief in $\phi_1$ is $0.6$, which is higher than the threshold $\sfrac{1}{2}$, even though his belief in $\phi_1$ before learning that $\mathit{Bob}$ had cancer was below the threshold.
Note that $\phi_4$ is a trivial non-secret: even before interacting with the system, \emph{Mallory}'s belief in $\phi_4$ is $0.95$. 
\end{example}

\subsection{Discussion}

Our approach assumes that the attacker's capabilities are well-defined.
While this, in general, is a strong assumption, there are many domains where such information is known.
There are, however, domains where this information is lacking.
In these cases, security engineers must  
\begin{inparaenum}[(1)]
	\item determine the appropriate beliefs capturing the desired attacker models, and
	\item formalize them.
\end{inparaenum}
The latter can be done, for instance, using \atklog{} (see \S\ref{sect:language}).
Note however that precisely eliciting the attackers' capabilities is still an open problem in DBIC.\looseness=-1

\begin{figure*}
\centering
	{\small 
	\begin{tabular}{ c | c  c  c c | c c c c | c c c c }
 \multirow{2}{*}{$i$} &
 \multicolumn{4}{c|}{$\llbracket \mathit{ATK} \rrbracket(u,r^i)(\llbracket \phi \rrbracket )$} & 
 \multicolumn{4}{c|}{$\llbracket \mathit{ATK} \rrbracket(u,r^{i+1})(\llbracket \phi \rrbracket)$} & 
 %\multicolumn{4}{c}{Confidentiality} \\
 \multicolumn{4}{c}{Secrecy} \\

 & $\phi_1$ & $\phi_2$ & $\phi_3$ & $\phi_4$ &  $\phi_1$ & $\phi_2$ & $\phi_3$ & $\phi_4$ & $\phi_1$ & $\phi_2$ & $\phi_3$ & $\phi_4$ \\
 \hline\hline
 
 $0$ & $0.3525$  & $0.3$ & $0.05$ & $0.95$ & $0.3525$  & $0.3$ & $0.05$ & $0.95$ & $\checkmark$ & $\checkmark$  & $\checkmark$ & $*$ \\ \hline
 $1$ & $0.3525$  & $0.3$ & $0.05$ & $0.95$ & $0.3525$  & $0.3$ & $0.05$ & $0.95$ & $\checkmark$ & $\checkmark$  & $\checkmark$ & $*$\\ \hline
 $2$ & $0.3525$  & $0.3$ & $0.05$ & $0.95$ & $0.495$  & $0.3$ & $1$ & $0$ & $\checkmark$ & $\checkmark$  & ${\cal X}$ & $*$\\ \hline
 $3$ & $0.495$  & $0.3$ & $1$ & $0$ & $0.6$  & $1$ & $1$ & $0$  & ${\cal X}$ & ${\cal X}$ & ${\cal X}$ & $*$\\ \hline
 $4$ & $0.6$  & $1$ & $1$ & $0$  & -- & -- & -- & -- & -- & -- & -- & --\\ 
	\end{tabular}
	  }

  \caption{
Evolution of \emph{Mallory}'s beliefs in the secrets $\phi_1, \ldots, \phi_4$ for the run $r$ and the attacker model $\mathit{ATK}$ from  Example~\ref{example:confidentiality}. 
In the table, ${\cal X}$ and $\checkmark$ denote that secrecy-preservation is violated and satisfied respectively, whereas  $*$ denotes trivial secrets.
}\label{figure:confidentiality:example}
\end{figure*}

\section{\atklog{}}\label{sect:language}

\subsection{Probabilistic Logic Programming}\label{sect:language:problog}

\problog{}~\cite{de2007problog,fierens2015inference, de2015probabilistic} is a probabilistic logic programming language with associated  tool support.
An exact inference engine for \problog{} is available at~\cite{problog}.

Conventional logic programs are constructed from terms, atoms, literals,  and rules. 
In the following, we consider only function-free logic programs, also called \datalog{} programs.
In this setting,  terms are either variable identifiers or constants.\looseness=-1

Let $\Sigma$ be a first-order signature, $\mathbf{dom}$ be a finite domain, and $\mathit{Var}$ be a countably infinite set of variable identifiers.
A \textit{$(\Sigma,\mathbf{dom})$-atom} $R(v_1, \ldots, v_n)$ consists of a predicate symbol $R \in \Sigma$ and arguments $v_1, \ldots, v_n$ such that $n$ is the arity of $R$, and each $v_i$, for $1 \leq i \leq n$, is either a variable identifier in $\mathit{Var}$ or a constant in $\mathbf{dom}$.
We denote by ${\cal A}_{\Sigma,\mathbf{dom}}$ the set $\{ R(v_1, \ldots, v_{|R|}) \mid  R \in \Sigma \wedge v_1, \ldots, v_{|R|} \in \mathbf{dom} \cup \mathit{Var} \}$ of all $(\Sigma,\mathbf{dom})$-atoms.
A \emph{$(\Sigma,\mathbf{dom})$-literal} $l$ is either a $(\Sigma,\mathbf{dom})$-atom  $a$ or its negation $\neg a$, where $a \in {\cal A}_{\Sigma,\mathbf{dom}}$.
We denote by ${\cal L}_{\Sigma, \mathbf{dom}}$ the set ${\cal A}_{\Sigma, \mathbf{dom}} \cup \{ \neg a \mid a \in {\cal A}_{\Sigma, \mathbf{dom}}\}$ of all $(\Sigma,\mathbf{dom})$-literals.
Given a literal $l$, $\mathit{vars}(l)$ denotes the set of its variables, $\mathit{args}(l)$ the list of its arguments, and $\mathit{pred}(l)$ the predicate symbol used in $l$.
As is standard, we say that a literal $l$ is \emph{positive} if it is an atom in ${\cal A}_{\Sigma, \mathbf{dom}}$ and \emph{negative} if it is the negation of an atom.
Furthermore, we say that a literal $l$ is \emph{ground} iff $\mathit{vars}(l) = \emptyset$.

A \textit{$(\Sigma,\mathbf{dom})$-rule} is of the form $h \leftarrow l_1, \ldots, l_n, e_1, \ldots, e_m$, where 
$h \in {\cal A}_{\Sigma,\mathbf{dom}}$ is a $(\Sigma,\mathbf{dom})$-atom,
$l_1, \ldots, l_n \in {\cal L}_{\Sigma,\mathbf{dom}}$ are $(\Sigma,\mathbf{dom})$-literals, and
$e_1, \ldots, e_m$ are equality and inequality constraints over the variables in $h,l_1, \ldots, l_m$.\footnote{Without loss of generality, we assume that equality constraints involving a variable $v$ and a constant $c$ are of the form $v = c$.}
Given a rule $r$, we denote by $\mathit{head}(r)$ the atom $h$, by $\mathit{body}(r)$ the literals $l_1, \ldots, l_n$, by $\mathit{cstr}(r)$ the constraints $e_1, \ldots, e_m$, and by $\mathit{body}(r,i)$ the $i$-th literal in $r$'s body, i.e., $\mathit{body}(r,i) = l_i$.
Furthermore, we denote by $\mathit{body}^+(r)$ (respectively $\mathit{body}^-(r)$) all positive (respectively negative) literals in $\mathit{body}(r)$.
As is standard, we assume that the free variables in a rule's head are a subset of the free variables of the positive literals in the rule's body, i.e., $\mathit{vars}(\mathit{head}(r)) \subseteq \bigcup_{l \in \mathit{body}^+(r)} \mathit{vars}(l) \cup \bigcup_{(x = c) \in \mathit{cstr}(r)\wedge c \in \mathbf{dom}} \{x\}$.
Finally, a \textit{$(\Sigma,\mathbf{dom})$-logic program} is a set of $(\Sigma,\mathbf{dom})$-ground atoms and $(\Sigma,\mathbf{dom})$-rules.
We consider only programs $p$ that do not contain negative cycles in the rules as is standard for stratified \datalog{}~\cite{abiteboul1995foundations}.\looseness=-1
To reason about probabilities, \problog{} extends logic programming with probabilistic atoms.
A \textit{$(\Sigma,\mathbf{dom})$-prob\-a\-bi\-lis\-tic atom} is a $(\Sigma,\mathbf{dom})$-atom $a$ annotated with a value $0 \leq v \leq 1$, denoted $\atom{v}{a}$.
\problog{} supports both probabilistic ground atoms and rules having probabilistic atoms in their heads.
\problog{} also supports \emph{annotated disjunctions} $\atom{v_1}{a_1}; \ldots; \atom{v_n}{a_n}$, where $a_1, \ldots, a_n$ are ground atoms and $\left( \sum_{1 \leq i\leq n} v_i \right) \leq 1$, which denote that $a_1, \ldots, a_n$ are mutually exclusive probabilistic events happening with probabilities $v_1, \ldots, v_n$.
Annotated disjunctions can either be used as ground atoms or as heads in rules.
Both annotated disjunctions and probabilistic rules are just syntactic sugar and can be expressed using ground probabilistic atoms and standard rules~\cite{de2007problog,fierens2015inference, de2015probabilistic}; see Appendix~\ref{app:problog:introduction}.

A $(\Sigma,\mathbf{dom})$-\problog{} program $p$ defines a probability distribution over all possible $(\Sigma,\mathbf{dom})$-structures, denoted $\llbracket p \rrbracket$.
Note that we consider only function-free \problog{} programs.
Hence, in our setting, \problog{} is a probabilistic extension of \textsc{Datalog}. 
Appendix~\ref{app:problog:introduction} contains a formal account of \problog{}'s semantics.

\para{Medical Data}
We formalize the probability distribution from Example~\ref{example:formal:model:attacker:model} as a \problog{} program.
We reuse the database schema and the domain from Example~\ref{example:formal:model:database:schema} as the first-order signature and the domain for the \problog{} program.
First, we  encode the template shown in Figure~\ref{figure:example:database:content:template} using ground atoms:
$\mathit{patient}(\mathtt{Alice})$,
$\mathit{patient}(\mathtt{Bob})$,
$\mathit{patient}(\mathtt{Carl})$,
$\mathit{smokes}(\mathtt{Bob})$, 
$\mathit{smokes}(\mathtt{Carl})$,
$\mathit{father}(\mathtt{Bob},\mathtt{Carl})$, and
$\mathit{mother}(\mathtt{Alice},\mathtt{Carl})$.
Second, we encode the probability distribution associated with the possible values of the $\mathit{cancer}$ table using the following \problog{} rules, which have probabilistic atoms in their heads:
\begingroup
\allowdisplaybreaks
\begin{align*}
\atom{\sfrac{1}{20}}{\mathit{cancer}(x)} &\leftarrow \mathit{patient}(x)\\
\atom{\sfrac{5}{19}}{\mathit{cancer}(x)} &\leftarrow smokes(x)\\
\atom{\sfrac{3}{14}}{\mathit{cancer}(y)}  &\leftarrow \mathit{father}(x,y), \mathit{cancer}(x), \\*
				& \qquad \mathit{mother}(z,y), \neg \mathit{cancer}(z)\\
\atom{\sfrac{3}{14}}{\mathit{cancer}(y)}  &\leftarrow \mathit{father}(x,y), \neg \mathit{cancer}(x), \\*
				& \qquad \mathit{mother}(z,y), \mathit{cancer}(z)\\
\atom{\sfrac{3}{7}}{\mathit{cancer}(y)}  &\leftarrow \mathit{father}(x,y), \mathit{cancer}(x), \\*
				& \qquad \mathit{mother}(z,y), \mathit{cancer}(z)
\end{align*}
\endgroup
The coefficients in the above example are derived from \S\ref{sect:motivating:example}. 
For instance, the probability that a smoking patient $x$ whose parents are not not in the $\mathit{cancer}$ relation has cancer is $30\%$.
The coefficient in the first rule is $\sfrac{1}{20}$ since each patient has a $5\%$ probability of having cancer.
The coefficient in the second rule is $\sfrac{5}{19}$, which is $\left( \sfrac{6}{20} - \sfrac{1}{20}\right) \cdot \left(1-\sfrac{1}{20}\right)^{-1}$, i.e., the probability that $\mathit{cancer}(x)$ is derived from the second rule given that it has not been derived from the first rule.
This ensures that the overall probability of deriving $\mathit{cancer}(x)$ is $\sfrac{6}{20}$, i.e., $30\%$.
The coefficients for the last two rules are derived analogously.

Informally, a probabilistic ground atom $\atom{\sfrac{1}{2}}{\mathit{cancer}(\mathtt{Bob})}$ expresses that $\mathit{cancer}(\mathtt{Bob})$ holds with a probability $\sfrac{1}{2}$.
Similarly, the rule $\atom{\sfrac{1}{20}}{cancer(x)} \leftarrow \mathit{patient}(x)$ states that, for any $x$ such that $\mathit{patient}(x)$ holds, then $\mathit{cancer}(x)$ can be derived with probability $\sfrac{1}{20}$.
This program yields the probability distribution shown in Figure~\ref{figure:example:distribution:possible:worlds}.

\subsection{\atklog{}'s Foundations}

We first introduce belief programs, which formalize an  attacker's initial beliefs.
Afterwards, we formalize \atklog{}.

\para{Belief Programs}
A belief program formalizes an attacker's beliefs as a probability distribution over the database states.

A database schema $D' = \langle \Sigma', \mathbf{dom}\rangle$ extends a schema $D = \langle \Sigma, \mathbf{dom}\rangle$ iff $\Sigma'$ contains all relation schemas in $\Sigma$. % and  $\mathbf{dom} = \mathbf{dom}'$.
The extension $D'$ may extend  $\Sigma$ with additional predicate symbols necessary to encode probabilistic dependencies.
Given an extension $D'$, a $D'$-state $s'$ \emph{agrees} with a $D$-state $s$ iff $s'(R) = s(R)$ for all $R$ in $D$.
Given a $D$-state $s$, we denote by $\mathit{EXT}(s,D,D')$ the set of all $D'$-states that agree with $s$.

A $(\Sigma', \mathbf{dom})$-\problog{} program $p$, where $D' = \langle \Sigma', \mathbf{dom}\rangle$ extends $D$, is a \emph{belief program over $D$}.
The \emph{$D$-semantics of $p$} is 
$\llbracket p\rrbracket_D = \lambda s \in \Omega_D.\, \sum_{s' \in \mathit{EXT}(s, D, D')} \llbracket p \rrbracket(s')$. 
Given a system configuration $C = \langle D, \Gamma\rangle$, a belief program $p$ over $D$ \textit{complies with $C$} iff $\llbracket p\rrbracket_D $ is a $C$-probability distribution.
With a slight abuse of notation, we lift the semantics of belief programs to sentences: $\llbracket p\rrbracket_D = \lambda \phi \in {\cal RC}_{\mathit{bool}}.\, \sum_{s' \in \{ s \in \Omega_D \mid [\phi]^{s} = \top\}} \llbracket p \rrbracket_D(s')$.

\para{\atklog{}}
An \atklog{} model specifies the initial beliefs of all users in ${\cal U}$ using belief programs.

Let $D$ be a database schema and $C = \langle D, \Gamma \rangle$ be a system configuration.
A \emph{$C$-\atklog{} model} $\mathit{ATK}$ is a function associating to each user $u \in U$, where $U \subset {\cal U}$ is a finite set of users, a belief program $p_u$ and to all users $u \in {\cal U}\setminus U$ a belief program $p_0$,  such that for  all users $u \in {\cal U}$, $\llbracket  \mathit{ATK}  (u)\rrbracket_D $ complies with $C$ and  for all database states $s \in \Omega_D^\Gamma$, $\llbracket  \mathit{ATK}(u) \rrbracket_D (s) >0$, i.e., all database states satisfying the integrity constraints are possible.
Informally, a $C$-\atklog{} model associates a distinct belief program to each user in $U$, and it associates to each user in ${\cal U} \setminus U$ the same belief program $p_0$.

Given a $C$-\acf{} $f$, a $C$-\atklog{} model $\mathit{ATK}$ defines the $(C,f)$-attacker model $\lambda  u \in {\cal U}. \llbracket \mathit{ATK}(u) \rrbracket_D$ that associates to each user $u \in {\cal U}$ the probability distribution defined by the belief program $\mathit{ATK}(u)$.
The semantics of this $(C,f)$-attacker model is:  
$\lambda u \in {\cal U}. \lambda r \in \mathit{runs}(C,f). \lambda s \in \Omega_D^\Gamma.\, \llbracket \mathit{ATK}(u) \rrbracket_D ( s \mid \llbracket  r \rrbracket_{\sim_u})$.
Informally, given a $C$-\atklog{} model $\mathit{ATK}$,  a $C$-\acf{} $f$, and a user $u$, $u$'s belief in a database state $s$, given a run $r$, is obtained by  conditioning the probability distribution defined by the belief program $\mathit{ATK}(u)$ given the set of database states corresponding to all runs $r' \sim_u r$.

\newcommand{\rg}{\mathit{rg}}
\newcommand{\bn}{\mathit{bn}}
\newcommand{\graph}{\mathit{graph}}
\newcommand{\tup}[1]{\langle #1 \rangle}

\section{Tractable Inference for \problog{} programs}\label{sect:inference}

Probabilistic inference in \problog{} is intractable in general.
Its data complexity, i.e., the complexity of  inference  when only the programs' probabilistic ground atoms are part of the input and the rules are considered fixed and not part of the input, is $\#P$-hard; see~\techReportAppendix{app:atklog:lite:extended}.
This limits the practical applicability of \problog{} (and \atklog{}) for DBIC.
To address this, we define acyclic \problog{} programs,  a class of programs where  the data complexity of inference is \textsc{PTime}.

Given a \problog{} program $p$, our inference algorithm consists of three steps:
(1) we compute all of $p$'s derivations,
(2) we compile these derivations into a Bayesian Network (BN) $\mathit{bn}$, and
(3) we perform the inference over $\mathit{bn}$.
To ensure tractability, we leverage two key insights.
First, we exploit guarded negation~\cite{barany2012queries} to develop a sound over-approximation, called the relaxed grounding, of all derivations of a program that is independent of the presence (or absence) of the probabilistic atoms.
This ensures that whenever a ground atom can be derived from a program (for a possible assignment to the probabilistic atoms), the atom is also part of this program's relaxed grounding.
This avoids grounding $p$ for each possible assignment to the probabilistic atoms.
Second, we introduce syntactic constraints (acyclicity) that ensure that $\mathit{bn}$ is a forest of poly-trees.
This ensures tractability since inference for poly-tree BNs can be performed in polynomial time in the network's size~\cite{koller2009probabilistic}.\looseness=-1

We also precisely characterize the expressiveness of acyclic \problog{} programs.
In this respect, we prove that acyclic programs are as expressive as forests of poly-tree BNs, one of the few classes of BNs with tractable inference.

As mentioned in \S\ref{sect:language}, probabilistic rules and annotated disjunctions are just syntactic sugar.
Hence, in the following we consider \problog{} programs consisting just of probabilistic ground atoms and non-probabilistic rules.
Note also that we treat ground atoms as rules with an empty body.

\subsection{Preliminaries}\label{sect:inference:preliminaries}

\para{Negation-guarded Programs}
A rule $r$ is \textit{negation-guarded}~\cite{barany2012queries} iff all the variables occurring in negative literals also occur in positive literals, namely for all negative literals $l$ in  $\mathit{body}^-(r)$, $\mathit{vars}(l) \subseteq \bigcup_{l' \in \mathit{body}^+(r)} \mathit{vars}(l')$.
To illustrate, the rule $C(x) \leftarrow A(x), \neg B(x)$ is negation-guarded, whereas $C(x) \leftarrow A(x), \neg B(x,y)$ is not since the variable $y$ does not occur in any positive literal.
We say that a program $p$ is \emph{negation-guarded} if all rules $r\in p$ are.

\para{Relaxed Grounding}
The relaxed grounding of a program $p$ is obtained by considering all probabilistic atoms as certain and by grounding all positive literals.
For all negation-guarded programs, the relaxed grounding of $p$ is a sound over-approximation of  all possible derivations in $p$.
Given a program $p$ and a rule $r\in p$, $\rg(p)$ denotes $p$'s relaxed grounding and $\rg(p,r)$ denotes the set of $r$'s ground instances.
We formalize relaxed groundings in~\techReportAppendix{app:atklog:lite:extended}.

\begin{example}\label{example:acyclic:1}
Let $p$ be the program consisting of the facts
$\atom{\sfrac{1}{2}}{A(1)}$, $A(2)$, $A(3)$, $D(1)$, $E(2)$, $F(1)$, $O(1,2)$, and $\atom{\sfrac{2}{3}}{O(2,3)}$, 
and the rules $r_a= B(x) \leftarrow A(x),D(x)$, $r_b=B(x) \leftarrow A(x),E(x)$,  and $r_c= B(y) \leftarrow B(x), \neg F(x), O(x,y)$. 
The relaxed grounding of $p$ consists of the initial facts together with $B(1)$, $B(2)$, and $B(3)$, whereas $\rg(p,r_c)$ consists of 
$B(2) \leftarrow B(1), \neg F(1), O(1,2)$ and $B(3) \leftarrow B(2), \neg F(2), O(2,3)$.
\end{example}

\para{Dependency and Ground Graphs}
The \textit{dependency graph} of a program  $p$, denoted $\graph(p)$, is the directed labelled graph having as nodes all the predicate symbols in $p$ and having an edge $a \xrightarrow{r,i} b$ iff there is a rule $r$ such that $a$ occurs in $i$-th literal in $r$'s body and $b$ occurs in $r$'s head.
Figure~\ref{figure:dependency:graph} depicts the dependency graph from Example~\ref{example:acyclic:1}.
The \textit{ground graph} of a program  $p$ is the graph obtained from its relaxed grounding.
Hence, there is an edge $a \xrightarrow{r,\mathit{gr},i} b$ from the ground atom $a$ to the ground atom $b$ iff there is a rule $r$ and a ground rule $\mathit{gr} \in \rg(p,r)$  such that  $\mathit{body}(\mathit{gr},i) \in \{a,\neg a\}$ and $\mathit{head}(\mathit{gr}) = b$. 
Figure~\ref{figure:ground:graph} depicts the ground graph from Example~\ref{example:acyclic:1}.
Note that there are no incoming or outgoing edges from $A(3)$ because the node is not involved in any derivation.

\begin{figure}
\centering
\begin{tikzpicture}
	\node[shape=circle, draw=black, minimum width=2.2em, anchor=center] (A) at (0,0) {$A$};
	\node[shape=circle, draw=black, minimum width=2.2em, anchor=center] (B) at ($(A)+(0,-2)$) {$B$};

	\node[shape=circle, draw=black, minimum width=2.2em, anchor=center] (F) at ($(B)+(-2,-1)$) {$F$};
	\node[shape=circle, draw=black, minimum width=2.2em, anchor=center] (D) at ($(A)+(2,-1)$) {$D$};
	\node[shape=circle, draw=black, minimum width=2.2em, anchor=center] (O) at ($(B)+(+2,-1)$) {$O$};
	\node[shape=circle, draw=black, minimum width=2.2em, anchor=center] (E) at ($(A)+(-2,-1)$) {$E$};
	\path[->]
		(A) edge[bend left] node[right]  {{\small $r_a,1$}} (B)
		(A) edge[bend right] node[left] {{\small $r_b,1$}} (B)

		(B) edge[loop below] node {\small $r_c,1$} (B)
		(F) edge node[pos=0.5, below, sloped] {{\small $r_c,2$}} (B)

	    (O) edge node[pos=0.5, below, sloped] {{\small $r_c,3$}} (B)
		(E) edge node[pos=0.5, below, sloped] {{\small $r_b,2$}} (B)
		(D) edge node[pos=0.5, below, sloped] {{\small $r_a,2$}} (B);
\end{tikzpicture}

\caption{Dependency graph for the program in Example~\ref{example:acyclic:1}.}\label{figure:dependency:graph}
\end{figure}
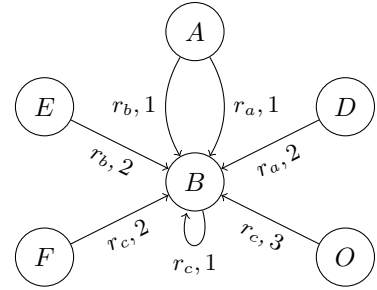

\begin{figure}
	\centering
	\begin{tikzpicture}
		\node[shape=rounded rectangle, draw=black, minimum width=2.2em, anchor=center] (A1) at (0,0) 			{{\small $A(1)$}};
		\node[shape=rounded rectangle, draw=black, minimum width=2.2em, anchor=center] (D1) at ($(A)+(4,0)$) 	{{\small $D(1)$}};
		\node[shape=rounded rectangle, draw=black, minimum width=2.2em, anchor=center] (B1) at ($(A)+(2,-1)$) 	{{\small $B(1)$}};
		\node[shape=rounded rectangle, draw=black, minimum width=2.2em, anchor=center] (A2) at ($(A)+(0,-2)$) 	{{\small $A(2)$}};
		\node[shape=rounded rectangle, draw=black, minimum width=2.2em, anchor=center] (E2) at ($(A)+(4,-2)$) 	{{\small $E(2)$}};
		\node[shape=rounded rectangle, draw=black, minimum width=2.2em, anchor=center] (F1) at ($(A)+(-1,-3)$)	{{\small $F(1)$}};
		\node[shape=rounded rectangle, draw=black, minimum width=2.2em, anchor=center] (O12) at ($(A)+(5,-3)$) 	{{\small $O(1,2)$}};
		\node[shape=rounded rectangle, draw=black, minimum width=2.2em, anchor=center] (B2) at ($(A)+(2,-3)$) 	{{\small $B(2)$}};
		\node[shape=rounded rectangle, draw=black, minimum width=2.2em, anchor=center] (F2) at ($(A)+(0,-4)$) 	{{\small $F(2)$}};
		\node[shape=rounded rectangle, draw=black, minimum width=2.2em, anchor=center] (A3) at ($(A)+(-2,-4)$) 	{{\small $A(3)$}};
		\node[shape=rounded rectangle, draw=black, minimum width=2.2em, anchor=center] (O23) at ($(A)+(4,-4)$) 	{{\small $O(2,3)$}};
		\node[shape=rounded rectangle, draw=black, minimum width=2.2em, anchor=center] (B3) at ($(A)+(2,-5)$) 	{{\small $B(3)$}};
		
		\path[->]
			(A1) 	edge node[sloped,anchor=north,auto=false] {{\tiny $r_a, r_a', 1$}} 	(B1)
			(D1) 	edge node[sloped,anchor=north,auto=false] {{\tiny $r_a, r_a', 2$}} 	(B1)
			(F1) 	edge node[sloped,anchor=north,auto=false] {{\tiny $r_c, r_c^1, 2$}}	(B2)
			(A2)	edge node[sloped,anchor=north,auto=false] {{\tiny $r_b, r_b', 1$}} 	(B2)
			(B1)	edge node[sloped,anchor=south,auto=false] {{\tiny $r_c, r_c^2, 1$}}	(B2)
			(E2)	edge node[sloped,anchor=north,auto=false] {{\tiny $r_b, r_b', 2$}}	(B2)
			(O12)	edge node[sloped,anchor=north,auto=false] {{\tiny $r_c, r_c^1, 3$}}	(B2)
			(F2)	edge node[sloped,anchor=north,auto=false] {{\tiny $r_c, r_c^2, 2$}}	(B3)
			(B2)	edge node[sloped,anchor=south,auto=false] {{\tiny $r_c, r_c^2, 1$}}	(B3)
			(O23)	edge node[sloped,anchor=north,auto=false] {{\tiny $r_c, r_c^2, 3$}}	(B3);
	\end{tikzpicture}
	\caption{Ground graph for the program in Example~\ref{example:acyclic:1}.
	The ground rules $r_a'$, $r_b'$, $r_c^1$, and $r_c^2$ are as follows:
	$r_a' = B(1) \leftarrow A(1),D(1)$,
	$r_b' = B(2) \leftarrow A(2), E(2)$,
	$r_c^1 = B(2) \leftarrow B(1), \neg F(1), O(1,2)$, and
	$r_c^2 = B(3) \leftarrow B(2), \neg F(2), O(2,3)$.}\label{figure:ground:graph}
\end{figure}
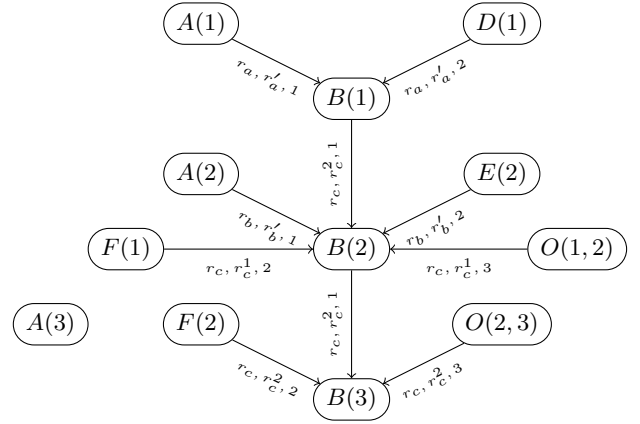

\para{Propagation Maps}
We use propagation maps to track how information flows inside rules.
Given a rule $r$ and a literal $l\in\mathit{body}(r)$, the \textit{$(r,l)$-vertical map} is the mapping $\mu$ from $\{1, \ldots, |l|\}$ to $\{1, \ldots, |\mathit{head}(r)|\}$ such that $\mu(i) = j$ iff $\mathit{args}(l)(i) = \mathit{args}(\mathit{head}(r))(j)$ and $\mathit{args}(l)(i) \in \mathit{Var}$. 
Given a rule $r$ and literals $l$ and $l'$ in $r$'s body, the \textit{$(r,l,l')$-horizontal map} is the mapping $\mu$ from $\{1, \ldots, |l|\}$ to $\{1, \ldots, |l'|\}$ such that $\mu(i) = j$ iff $\mathit{args}(l)(i) = \mathit{args}(l')(j)$ and $\mathit{args}(l)(i) \in \mathit{Var}$.

We say that a path links to a literal $l$ if information flows along the rules to $l$.
This can be formalized by posing constraints on the mapping obtained by combining horizontal and vertical maps along the path. 
Formally, given a literal $l$ and a mapping  $\nu : \mathbb{N} \to \mathbb{N}$, 
 a directed path $\mathit{pr}_1 \xrightarrow{r_1,i_1} \ldots \xrightarrow{r_{n-1}, i_{n-1}} \mathit{pr}_n$ \emph{$\nu$-downward links to $l$} iff
there is a $0 \leq j < n-1$ such that the function $\mu := \mu' \circ \mu_{j} \circ \ldots \circ \mu_1$ 
satisfies $\mu(k) = \nu(k)$ for all $k$ for which $\nu(k)$ is defined,
%is defined for all $k \in K$, 
where for $1 \leq h \leq j$, $\mu_h$ is the $(r_h,\mathit{body}(r_h,i_h))$-vertical map, %connecting $\mathit{body}(r_h,i_h)$ and $r_h$, 
and $\mu'$ is the horizontal map connecting $\mathit{body}(r_{j+1},i_{j+1})$ with $l$.
Similarly, a directed path $\mathit{pr}_1 \xrightarrow{r_1,i_1} \ldots \xrightarrow{r_{n-1}, i_{n-1}} \mathit{pr}_n$ \emph{$\nu$-upward links to $l$} iff
there is a $1 \leq j \leq n-1$ such that the function $\mu :=\mu'^{-1} \circ \mu_{j+1}^{-1} \circ \ldots \circ \mu_{n-1}^{-1}$ 
satisfies $\mu(k) = \nu(k)$ for all $k$ for which $\nu(k)$ is defined,
 where $\mu_h$ is the $(r_h,\mathit{body}(r_h,i_h))$-vertical map, for $j < h \leq n-1$, 
and $\mu'$ is the $(r_{j}, l)$-vertical map. 
A path \textit{$P$ links to a predicate symbol $a$} iff there is an atom $a(\overline{x})$ such that $P$ links to $a(\overline{x})$.

\begin{example}\label{example:acyclic:5}
The horizontal map connecting $A(x)$ and $D(x)$ in $r_a$, i.e., the $(r_a,A(x),D(x))$-horizontal map, is  $\{1 \mapsto 1\}$.
The horizontal map connecting $A(x)$ and $E(x)$ in $r_b$ is $\{1 \mapsto 1\}$ as well.
Hence, the path $A \xrightarrow{r_a,1} B$ downward links to $D$ and the path  $A \xrightarrow{r_b,1} B$ downward links to $E$ for the mapping $\{1 \mapsto 1\}$. 
Furthermore, the path $B \xrightarrow{r_c,1} B$ downward links to $O$ for $\{1 \mapsto 1\}$ since the $(r_c,B(x),O(x,y))$-horizontal map is $\{1 \mapsto 1\}$.
Finally, the path $B \xrightarrow{r_c,1} B$ upward links to $O$ for $\{2 \mapsto 1\}$ since the $(r_c,O(x,y))$-vertical map is $\{2 \mapsto 1\}$.
\end{example}

\subsection{Acyclic \problog{} programs}\label{sect:inference:acyclic}

A sufficient condition for tractable inference is that $p$'s ground graph is a forest of poly-trees.
This requires that $p$'s ground graph neither contains  directed nor undirected cycles, or, equivalently, the undirected version of $p$'s ground graph is acyclic.
To illustrate, the ground graph in Figure~\ref{figure:ground:graph} is a poly-tree.
The key insight here is that a cycle among ground atoms is caused by a (directed or undirected) cycle among $p$'s predicate symbols.
In a nutshell, acyclicity requires that all possible cycles in $\mathit{graph}(p)$ are guarded.
This ensures that cycles in $\mathit{graph}(p)$ do not lead to cycles in the ground graph.
Additionally, acyclicity requires that programs are negation-guarded.
This ensures that the relaxed grounding and the ground graph are well-defined.
In the following, let $p$ be a $(\Sigma,\mathbf{dom})$-\problog{} program.

\newcommand{\order}[1]{\mathit{ORD}(#1)}
\newcommand{\disjoint}[2]{\mathit{DIS}(#1,#2)}
\newcommand{\unique}[2]{\mathit{UNQ}(#1,#2)}
\para{Annotations}
Annotations represent properties of the relations induced by the program $p$, and they are syntactically derived by analysing $p$'s ground atoms and rules.

Let $a,a' \in \Sigma$ be two predicate symbols such that $|a| = |a'|$.
A \textit{disjointness annotation} $\disjoint{a}{a'}$ represents that the relations induced by $a$ and $a'$ (given $p$'s relaxed grounding) are disjoint.
We say that $\disjoint{a}{a'}$ can be derived from $p$ iff no rules in $p$ contain $a$ or $a'$ in their heads, and there is no $\overline{v} \in \mathbf{dom}^{|a|}$ where both $a(\overline{v})$ and $a'(\overline{v})$ appear as (possibly probabilistic) ground atoms in $p$.
Hence, the relations induced by $a$ and $a'$ are disjoint.

Let $n \in \mathbb{N}$ and $A \subseteq \Sigma$ be a set of predicate symbols such that  $|a| = 2n$ for all $a \in A$.
An \emph{ordering annotation} $\order{A}$ represents that the transitive closure of the union of the relations induced by predicates in $A$ given $p$'s relaxed grounding is a strict partial order over $\mathbf{dom}^{n}$.
The annotation $\order{A}$ can be derived from the program $p$ iff there is no rule $r \in p$ that contains any of the predicates in $A$ in its head and 
the transitive closure of $\bigcup_{a \in A} \{ ((v_1, \ldots, v_{n}),(v_{n +1}, \ldots, v_{2n})) \mid \exists k.\ \atom{k}{a(v_1, \ldots, v_{2n})} \in p \}$ is a strict partial order over  $\mathbf{dom}^{n}$.
Hence, the closure of the relation $\bigcup_{a \in A}  \{ ((v_1, \ldots, v_{n}),(v_{n +1}, \ldots, v_{2n})) \mid a(v_1, \ldots, v_{2n}) \in \rg(p) \}$ induced by the relaxed grounding is a strict partial order.

Let $a \in \Sigma$ be a predicate symbol and  $K \subseteq \{1, \ldots, |a|\}$.
A \emph{uniqueness annotation} $\unique{a}{K}$ represents  that the attributes in $K$ are a primary key for the relation induced by $a$ given the relaxed grounding.
We say that $\unique{a}{K}$ can be derived from a program $p$ iff no rule contains $a$ in its head and 
for all $\overline{v}, \overline{v}' \in \mathbf{dom}^{|a|}$, 
if
 (1) $\overline{v}(i) = \overline{v}'(i)$ for all $i \in K$, and
 (2) there are $k$ and $k'$ such $\atom{k}{a(\overline{v})} \in p$ and $\atom{k'}{a(\overline{v}')} \in p$,   
then $\overline{v}  = \overline{v}'$.
This ensures that whenever $a(\overline{v}), a(\overline{v}') \in \rg(p)$ and $\overline{v}(i) = \overline{v}'(i)$ for all $i \in K$, then $\overline{v}  = \overline{v}'$.

A \textit{$\Sigma$-template} ${\cal T}$ is a set of annotations. 
In~\techReportAppendix{app:atklog:lite:extended}, we relax our syntactic rules for deriving annotations.

\begin{example}\label{example:acyclic:2}
We can derive $\disjoint{D}{E}$ from the program in Example~\ref{example:acyclic:1} since no rule generates facts for $D$ and $E$ and the relations defined by the ground atoms are $\{1\}$ and $\{2\}$.
We can also derive $\order{\{O\}}$ since the relation defined by $O$'s ground atoms is $\{(1,2),(2,3)\}$, whose transitive closure is a strict partial order. 
Finally, we can derive $\unique{O}{\{1\}}$, $\unique{O}{\{2\}}$, and $\unique{O}{\{1,2\}}$ since both arguments of $O$ uniquely identify the tuples in the relation induced by $O$.\looseness=-1\end{example}\looseness=-1

\para{Unsafe structures}
An unsafe structure models a part of the dependency graph that may introduce cycles in the ground graph.
We define directed and undirected unsafe structures which may respectively introduce directed and undirected cycles in the ground graph.

A \emph{directed unsafe structure} in $\mathit{graph}(p)$ is a directed cycle $C$ in $\mathit{graph}(p)$.
We say that a directed unsafe structure $C$ \emph{covers} a directed cycle $C'$ iff $C$ is equivalent to $C'$.

An \emph{undirected unsafe structure} in $\mathit{graph}(p)$ is quadruple $\langle D_1, D_2, D_3, U \rangle$ such that
(1) $D_1$, $D_2$, and $D_3$ are directed paths whereas $U$ is an undirected path, 
(2) $D_1$ and $D_2$ start from the same node,
(3) $D_2$ and $D_3$ end in the same node, and
(4) $D_1 \cdot U \cdot D_3 \cdot D_2$ is an undirected cycle in $\graph(p)$.
We say that an unsafe structure  $\tup{D_1, D_2, D_3, U}$ \emph{covers} an undirected cycle $U'$ in $\graph(p)$ iff $D_1 \cdot U \cdot D_3 \cdot D_2$ is equivalent to $U'$.

\begin{example}\label{example:acyclic:3}
The cycle introduced by the rule $r_c$ is captured by the directed unsafe structure $B \xrightarrow{r_c,1} B$, while the cycle introduced by $r_a$ and $r_b$ is captured by the structure $\tup{A \xrightarrow{r_a,1} B, A \xrightarrow{r_b,1} B, \epsilon, \epsilon}$, where $\epsilon$ denotes the empty path.
\end{example}

\newcommand{\treeroot}{\mathit{root}}

\para{Connected Rules}
A connected rule $r$ ensures that a grounding of $r$ is fully determined either by the assignment to the head's variables or to the variables of any literal in $r$'s body.
Formally, a strongly connected rule $r$ guarantees that for any two groundings $\mathit{gr}',\mathit{gr}''$ of $r$, if $\mathit{head}(\mathit{gr}') = \mathit{head}(\mathit{gr}'')$, then $\mathit{gr}' = \mathit{gr}''$.
In contrast, a weakly connected rule $r$ guarantees that for any two groundings $\mathit{gr}',\mathit{gr}''$ of $r$, if $\mathit{body}(\mathit{gr}',i) = \mathit{body}(\mathit{gr}'',i)$ for some $i$, then $\mathit{gr}' = \mathit{gr}''$.
This is done by exploiting uniqueness annotations and the rule's structure.\looseness=-1

Before formalizing connected rules, we introduce join trees.
A join tree represents how multiple predicate symbols in a rule share variables.
In the following, let $r$ be a rule and ${\cal T}$ be a template.
A \emph{join tree for  a rule $r$} is a rooted labelled tree $(N,E, \treeroot, \lambda)$, where
$N\subseteq \mathit{body}(r)$,
$E$ is a set of edges (i.e., unordered pairs over $N^2$),
$\treeroot \in N$ is the tree's root, and
$\lambda$ is the labelling function.
Moreover, we require that for all $n,n'\in N$, 
	if $n \neq n'$ and $(n,n') \in E$, then $\lambda(n,n') = \mathit{vars}(n) \cap \mathit{vars}(n')$ and $\lambda(n,n') \neq \emptyset$.
A join tree $(N,E, \treeroot, \lambda)$ \emph{covers} a literal $l$ iff $l \in N$.
Given a join tree $J=(N,E, \treeroot, \lambda)$ and a node $n \in N$, the \emph{support of $n$}, denoted $\mathit{support}(n)$, is the set $\mathit{vars}(\mathit{head}(r)) \cup \{ x \mid (x = c) \in \mathit{cstr}(r) \wedge c \in \mathbf{dom} \} \cup \{ \mathit{vars}(n') \mid n' \in \mathit{anc}(J,n) \}$, where $\mathit{anc}(J,n)$ is the set of $n$'s ancestors in $J$, i.e., the set of all nodes (different from $n$) on the path from $\treeroot$ to $n$. 
A join tree $J=(N,E, \treeroot, \lambda)$ is \emph{${\cal T}$-strongly connected} iff
for all positive literals $l \in N$,  there is a set 
$K \subseteq \{ i \mid \overline{x} = \mathit{args}(l) \wedge \overline{x}(i) \in \mathit{support}(l) \}$ 
such that $\unique{\mathit{pred}(l)}{K} \in {\cal T}$ and for all negative literals $l \in N$, $\mathit{vars}(l) \subseteq \mathit{support}(l)$.
In contrast, a join tree $(N,E, \treeroot, \lambda)$ is \emph{${\cal T}$-weakly connected} iff
for all $(a(\overline{x}), a'(\overline{x}')) \in E$, there are $K \subseteq \{ i \mid \overline{x}(i) \in L \}$ and $K' \subseteq \{ i \mid \overline{x}'(i) \in L \}$  such that $\unique{a}{K}, \unique{a'}{K'} \in {\cal T}$, where $L= \lambda(a(\overline{x}), a'(\overline{x}'))$.\looseness=-1

We now formalize strongly and weakly connected rules.
A rule $r$ is \emph{${\cal T}$-strongly connected} iff there exist ${\cal T}$-strongly connected join trees $J_1, \ldots, J_n$ that cover all literals in $r$'s body. 
This guarantees that for any two groundings $\mathit{gr}',\mathit{gr}''$ of $r$, if $\mathit{head}(\mathit{gr}') = \mathit{head}(\mathit{gr}'')$, then $\mathit{gr}' = \mathit{gr}''$.

Given a rule $r$, a set of literals $L$, and a template ${\cal T}$, a literal $l \in \mathit{body}(r)$ is \textit{$(r,{\cal T},L)$-strictly guarded} iff 
(1) $\mathit{vars}(l) \subseteq \bigcup_{l' \in  L\cap\mathit{body}^+(r)} \mathit{vars}(l') \cup \{x \mid (x = c) \in \mathit{cstr}(r) \wedge c \in \mathbf{dom}\}$, and
(2) there is a positive literal $a(\overline{x}) \in L$ and an annotation $\unique{a}{K} \in {\cal T}$ such that $\{ \overline{x}(i) \mid i \in K \} \subseteq \mathit{vars}(l)$.
A rule $r$ is \emph{weakly connected for ${\cal T}$} iff 
there exists a ${\cal T}$-weakly connected join tree $J = (N,E,\treeroot,\lambda)$ such that $N \subseteq \mathit{body}^+(r)$, and all literals in $\mathit{body}(r) \setminus N$ are $(r,{\cal T}, N)$-strictly guarded.
This guarantees that for any two groundings $\mathit{gr}',\mathit{gr}''$ of $r$, if $\mathit{body}(\mathit{gr}',i) = \mathit{body}(\mathit{gr}'',i)$ for some $i$, then $\mathit{gr}' = \mathit{gr}''$.\looseness=-1

\begin{example}\label{example:acyclic:4}
Let ${\cal T}$ be the template from Example~\ref{example:acyclic:2}.
The rule $r_c := B(y) \leftarrow B(x), \neg F(x), O(x,y)$ is ${\cal T}$-strongly connected.
Indeed, the join tree having $O(x,y)$ as root and $B(x)$ and $\neg F(x)$ as leaves is such that (1) there is a uniqueness annotation $\unique{O}{\{2\}}$ in ${\cal T}$ such that the second variable in $O(x,y)$ is included in those of $r_c$'s head, (2) the variables in $B(x)$ and $\neg F(x)$ are a subset of those of their ancestors, and (3) the tree covers all literals in $r_c$'s body.
The rule is also ${\cal T}$-weakly connected: the join tree consisting only of $O(x,y)$ is ${\cal T}$-weakly connected and the literals $B(x)$ and $\neg F(x)$ are strictly guarded.
Note that the rules $r_a$ and $r_b$ are trivially both strongly and weakly connected.
\end{example}

\para{Guarded undirected structures}
Guarded undirected structures ensure that undirected cycles in the dependency graph do not correspond to undirected cycles in the ground graph by exploiting disjointness annotations.
Formally, an undirected unsafe structure $\langle D_1, D_2, D_3, U \rangle$ is \textit{guarded by a template ${\cal T}$} iff  either $(D_1,D_2)$ is ${\cal T}$-head-guarded  or $(D_2,D_3)$ is ${\cal T}$-tail-guarded.\looseness=-1

A pair of non-empty paths $(P_1, P_2)$ sharing the same initial node $a$ is \emph{${\cal T}$-head guarded} iff 
\begin{inparaenum}[(1)]
\item if $P_1 = P_2$, all rules in $P_1$ are weakly connected for ${\cal T}$, or
\item if $P_1 \neq P_2$, there is an annotation $\disjoint{\mathit{pr}}{\mathit{pr}'} \in {\cal T}$, a set $K \subseteq \{1, \ldots, |a|\}$, and a bijection $\nu : K \to \{1, \ldots, |\mathit{pr}|\}$ such that $P_1$ $\nu$-downward links to $\mathit{pr}$ and $P_2$ $\nu$-downward links to $\mathit{pr}'$.
\end{inparaenum}
Given two ground paths $P_1'$ and $P_2'$ corresponding to $P_1$ and $P_2$, the first condition ensures that $P_1' = P_2'$ whereas the second  ensures that $P_1'$ or $P_2'$ are not in the ground graph.

Similarly, a pair of non-empty paths $(P_1, P_2)$ sharing the same final node $a$ is \emph{${\cal T}$-tail guarded} iff 
\begin{inparaenum}[(1)]
\item if $P_1 = P_2$, all rules in $P_1$ are strongly connected for ${\cal T}$, or
\item if $P_1 \neq P_2$, there is an annotation $\disjoint{\mathit{pr}}{\mathit{pr}'} \in {\cal T}$, a set $K \subseteq \{1, \ldots, |a|\}$, and a bijection $\nu : K \to \{1, \ldots, |\mathit{pr}|\}$, such that $P_1$ $\nu$-upward links to $\mathit{pr}$ and $P_2$ $\nu$-upward links to $\mathit{pr}'$.
\end{inparaenum}

\begin{example}\label{example:acyclic:6}
The only non-trivially guarded undirected cycle in the graph from Figure~\ref{figure:dependency:graph} is the one represented by the undirected unsafe structure $\tup{A \xrightarrow{r_a,1} B, A \xrightarrow{r_b,1} B, \epsilon, \epsilon}$.
The structure is guarded since the paths $A \xrightarrow{r_a,1} B$ and $A \xrightarrow{r_b,1} B$ are head guarded by $\disjoint{D}{E}$. 
Indeed, for the same ground atom $A(v)$, for some $v \in \{1,2,3\}$, only one of $r_a$ and $r_b$ can be applied since $D$ and $E$ are disjoint.
\end{example}

\para{Guarded directed structures}
Guarded directed structures exploit ordering annotations to ensure that directed cycles in the dependency graph do not correspond to directed cycles among ground atoms.
A directed unsafe structure $\mathit{pr}_1 \xrightarrow{r_1, i_1} \ldots \xrightarrow{r_n,i_n} \mathit{pr}_1$ is \textit{guarded by a template ${\cal T}$} iff 
there is an annotation $\order{O} \in {\cal T}$,
integers $1 \leq y_1 < y_2 < \ldots < y_e = n$, literals $o_1(\overline{x}_1), \ldots, o_e(\overline{x}_e)$ (where $o_1, \ldots, o_e \in O$),
a non-empty set $K \subseteq \{1, \ldots, |\mathit{pr}_1|\}$, and a bijection $\nu : K \to \{1, \ldots, \sfrac{|\mathit{o}|}{2}\}$ such that 
for each $0 \leq k < e$, 
(1) $\mathit{pr}_{y_k} \xrightarrow{r_{y_k}, i_{y_k}} \ldots \xrightarrow{r_{{y_{k+1}}-1}, i_{{y_{k+1}}-1}} pr_{y_{k+1}}$ $\nu$-downward connects to $o_{k+1}(\overline{x}_{k+1})$, and 
(2) $\mathit{pr}_{y_{k+1}-1} \xrightarrow{r_{{y_{k+1}}-1}, i_{{y_{k+1}}-1}} pr_{y_{k+1}}$ $\nu'$-upward connects to $o_{k+1}(\overline{x}_{k+1})$, 
where 
$\nu'(i) = \nu(x) + \sfrac{|\mathit{o}_1|}{2}$ for all $1 \leq i \leq \sfrac{|\mathit{o}_1|}{2}$, and
$y_0 = 1$.

\begin{example}\label{example:acyclic:6}
The directed unsafe structure $B \xrightarrow{r_c,1} B$ is guarded by $\order{\{O\}}$ in the template from Example~\ref{example:acyclic:2}.
Indeed, the strict partial order induced by $O$ breaks the cycle among ground atoms belonging to $B$.
In particular, the path $B \xrightarrow{r_c,1} B$ both downward links and upward links to $O(x,y)$; see Example~\ref{example:acyclic:5}.
\end{example}

\para{Acyclic Programs}
Let $p$ be a negation-guarded program and ${\cal T}$ be the template containing all annotations that can be derived from $p$.
We say that $p$ is \emph{acyclic} iff 
\begin{inparaenum}[(a)]
	\item for all undirected cycles $U$ in $\mathit{graph}(p)$ that are not directed cycles, there is a ${\cal T}$-guarded undirected unsafe structure that covers $U$, and  
	\item for all directed cycles $C$ in $\mathit{graph}(p)$, there is a ${\cal T}$-guarded directed unsafe structure that covers $C$.  
\end{inparaenum}
This ensures the absence of cycles in the ground graph.

\begin{proposition}
Let $p$ be a \problog{} program.
If $p$ is acyclic, then the ground graph of $p$ is a forest of poly-trees. 
\end{proposition}

\begin{example}
The program $p$ from Example~\ref{example:acyclic:1} is acyclic. 
This is reflected in the ground graph in Figure~\ref{figure:ground:graph}.
The program $q = p \cup \{E(1)\}$, however, is not acyclic:
we cannot derive $\disjoint{D}{E}$ from $q$ and the undirected unsafe structure $\tup{A \xrightarrow{r_a,1} B, A \xrightarrow{r_b,1} B, \epsilon,\epsilon}$ is not guarded.
As expected, $q$'s ground graph contains an undirected cycle between $A(1)$ and $B(1)$, as shown in Figure~\ref{figure:ground:graph:cycle}.
\end{example}

\begin{figure}
	\centering
	\begin{tikzpicture}
		\node[shape=rounded rectangle, draw=black, minimum width=2.2em, anchor=center] (A1) at (0,0) 			{{\small $A(1)$}};
		\node[shape=rounded rectangle, draw=black, minimum width=2.2em, anchor=center] (D1) at ($(A)+(4,0)$) 	{{\small $D(1)$}};
		\node[shape=rounded rectangle, draw=black, dashed, minimum width=2.2em, anchor=center] (E1) at ($(A)+(5,-1)$) 	{{\small $E(1)$}};
		\node[shape=rounded rectangle, draw=black, minimum width=2.2em, anchor=center] (B1) at ($(A)+(2,-1)$) 	{{\small $B(1)$}};
		\node[shape=rounded rectangle, draw=black, minimum width=2.2em, anchor=center] (A2) at ($(A)+(0,-2)$) 	{{\small $A(2)$}};
		\node[shape=rounded rectangle, draw=black, minimum width=2.2em, anchor=center] (E2) at ($(A)+(4,-2)$) 	{{\small $E(2)$}};
		\node[shape=rounded rectangle, draw=black, minimum width=2.2em, anchor=center] (F1) at ($(A)+(-1,-3)$)	{{\small $F(1)$}};
		\node[shape=rounded rectangle, draw=black, minimum width=2.2em, anchor=center] (O12) at ($(A)+(5,-3)$) 	{{\small $O(1,2)$}};
		\node[shape=rounded rectangle, draw=black, minimum width=2.2em, anchor=center] (B2) at ($(A)+(2,-3)$) 	{{\small $B(2)$}};
		\node[shape=rounded rectangle, draw=black, minimum width=2.2em, anchor=center] (F2) at ($(A)+(0,-4)$) 	{{\small $F(2)$}};
		\node[shape=rounded rectangle, draw=black, minimum width=2.2em, anchor=center] (A3) at ($(A)+(-2,-4)$) 	{{\small $A(3)$}};
		\node[shape=rounded rectangle, draw=black, minimum width=2.2em, anchor=center] (O23) at ($(A)+(4,-4)$) 	{{\small $O(2,3)$}};
		\node[shape=rounded rectangle, draw=black, minimum width=2.2em, anchor=center] (B3) at ($(A)+(2,-5)$) 	{{\small $B(3)$}};
		
		\path[->]
			(A1) 	edge 						node[sloped,anchor=north,auto=false] {{\tiny $r_a, r_a', 1$}} 	(B1)
			(A1) 	edge[dashed, bend left]  	node[sloped,anchor=south,auto=false]{{\tiny $r_b, r', 1$}} 		(B1)
			(D1) 	edge 						node[sloped,anchor=north,auto=false] {{\tiny $r_a, r_a', 2$}} 	(B1)
			(E1) 	edge[dashed] 				node[sloped,anchor=north,auto=false] {{\tiny $r_b, r', 2$}} 	(B1)
			(F1) 	edge 						node[sloped,anchor=north,auto=false] {{\tiny $r_c, r_c^1, 2$}}	(B2)
			(A2)	edge 						node[sloped,anchor=north,auto=false] {{\tiny $r_b, r_b', 1$}} 	(B2)
			(B1)	edge 						node[sloped,anchor=south,auto=false] {{\tiny $r_c, r_c^2, 1$}}	(B2)
			(E2)	edge 						node[sloped,anchor=north,auto=false] {{\tiny $r_b, r_b', 2$}}	(B2)
			(O12)	edge 						node[sloped,anchor=north,auto=false] {{\tiny $r_c, r_c^1, 3$}}	(B2)
			(F2)	edge 						node[sloped,anchor=north,auto=false] {{\tiny $r_c, r_c^2, 2$}}	(B3)
			(B2)	edge 						node[sloped,anchor=south,auto=false] {{\tiny $r_c, r_c^2, 1$}}	(B3)
			(O23)	edge 						node[sloped,anchor=north,auto=false] {{\tiny $r_c, r_c^2, 3$}}	(B3);
	\end{tikzpicture}
	\caption{Ground graph for the program in Example~\ref{example:acyclic:1} extended with the atom $E(1)$.
	The additional edges and nodes are represented using dashed lines.
	The ground rules $r_a'$, $r_b'$, $r_c^1$, and $r_c^2$ are as in Figure~\ref{figure:ground:graph}, and $r' = B(1) \leftarrow A(1),E(1)$.
}\label{figure:ground:graph:cycle}
\end{figure}
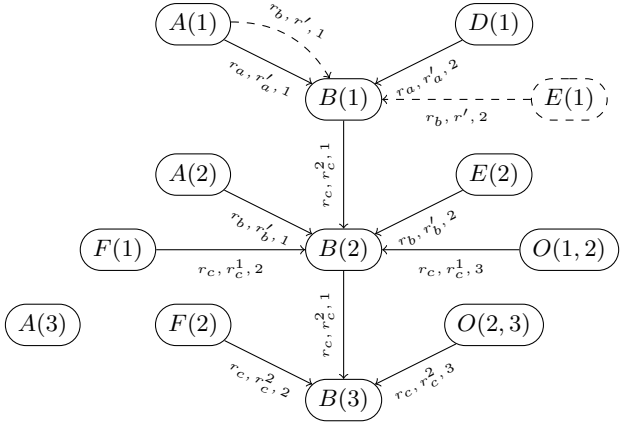

\para{Expressiveness}
Acyclicity trades off the programs expressible in \problog{} for a tractable inference procedure.
Acyclic programs, nevertheless, can encode many relevant probabilistic models.\looseness=-1 

\begin{proposition}\thlabel{theorem:polytree:BN:to:acyclic}
Any forest of poly-tree BNs can be represented as an acyclic \problog{} program. 
\end{proposition}

To clarify this proposition's scope, observe that poly-tree BNs are one of the few classes of BNs with tractable inference procedures.
From Proposition~\ref{theorem:polytree:BN:to:acyclic}, it follows that a large class of probabilistic models with tractable inference can be represented as acyclic programs.
This supports our thesis that our syntactic constraints are not overly restrictive.
In~\techReportAppendix{app:atklog:lite:extended}, we relax acyclicity to support 
 a limited form of annotated disjunctions and
 rules sharing a part of their bodies, which are needed to encode the example from \S\ref{sect:language} and for Proposition~\ref{theorem:polytree:BN:to:acyclic}.
We also provide all the proofs.

\begin{figure}
	\centering
	\begin{tikzpicture}
		\node (A1) at (0,0) {{\small $X[B(2)]$}};
		\node (A2) at (2,0) {{\small $X[F(2)]$}};
		\node (A3) at (4,0)	{{\small $X[O(2,3)]$}};

		\node (B1) at (2,-.75) {{\small $X[r_c,r_c^2,B(3) ]$}};
		% \node (B2) at (2,-1.5) {{\small $X[r_c,B(3) ]$}};
% 		\node (B3) at (2,-2.25) {{\small $X[B(3) ]$}};
		
		\path [->]
			(A1) edge (B1)
			(A2) edge (B1)
			(A3) edge (B1);
			% (B1) edge (B2)
% 			(B2) edge (B3);

	\end{tikzpicture}

	\begin{tabular}{l}
		{{\bf\small Selected CPTs.}}\\\hline
		\rule{0pt}{12.5ex}
	{\small 
	\begin{tabular}{c c c | c c}
		& & & \multicolumn{2}{c}{$X[r_c,r_c^2,B(3) ]$} \\
		$X[B(2)]$ & $X[F(2)]$ & $X[O(2,3)]$ & $\top$ & $\bot$ \\\hline
		$\top$  &  $\top$  &  $\top$   &  0  &  1 \\
		$\top$  &  $\top$  &  $\bot$   &  0  &  1 \\
		$\top$  &  $\bot$  &  $\top$   &  1  &  0 \\
		$\top$  &  $\bot$  &  $\bot$   &  0  &  1 \\
		$\bot$  &  $\top$  &  $\top$   &  0  &  1 \\
		$\bot$  &  $\top$  &  $\bot$   &  0  &  1 \\
		$\bot$  &  $\bot$  &  $\top$   &  0  &  1 \\
		$\bot$  &  $\bot$  &  $\bot$   &  0  &  1 \\
	\end{tabular}
	}
	\end{tabular}

	\caption{Portion of the resulting BN  for the atoms $B(2)$, $F(2)$, and $O(2,3)$, the rule $r_c = B(y) \leftarrow B(x), \neg F(x), O(x,y)$, and the ground rule $r_c^2 =  B(3) \leftarrow B(2), \neg F(2), O(2,3)$, together with the CPT encoding $r_c^2$'s semantics.}\label{figure:bayesian:network}
\end{figure}

\subsection{Inference Engine}\label{sect:inference:encoding}

Our inference algorithm for acyclic \problog{} programs consists of three steps:
\begin{inparaenum}[(1)]
	\item we compute the relaxed grounding of $p$ (cf.~\S\ref{sect:inference:preliminaries}),
	\item we compile the relaxed grounding into a Bayesian Network (BN), and
	\item we perform the inference using standard algorithms for poly-tree Bayesian Networks~\cite{koller2009probabilistic}.
\end{inparaenum}

\para{Encoding as BNs}
We compile the relaxed grounding $\rg(p)$ into the BN $\bn(p)$.
The boolean random variables in $\bn(p)$ are as follows: 
\begin{inparaenum}[(a)]
\item for each atom $a$ in $\rg(p)$ and ground literal $a$ or $\neg a$ occurring in any $\mathit{gr} \in \bigcup_{r \in p} \rg(p,r)$, there is a random variable $X[a]$,
\item for each rule $r \in p$ and each ground atom $a$ such that there is $\mathit{gr} \in \rg(p,r)$ satisfying $a= \mathit{head}(\mathit{gr})$, there is a random variable $X[r,a]$, and
\item for each rule $r \in p$, each ground atom $a$, and each ground rule $\mathit{gr} \in \rg(p,r)$ such that $a= \mathit{head}(\mathit{gr})$, there is a random variable $X[r,\mathit{gr},a]$.
\end{inparaenum}

The edges in $\bn(p)$ are as follows:
\begin{inparaenum}[(a)]
\item for each ground atom $a$, rule $r$, and ground rule $\mathit{gr}$, there is an edge from $X[r,\mathit{gr},a]$ to $X[r,a]$ and an edge from $X[r,a]$ to $X[a]$, and 
\item for each ground atoms $a$ and $b$, rule $r$, and ground rule $\mathit{gr}$, there is an edge from $X[b]$ to $X[r,\mathit{gr},a]$ if $b$ occurs in $\mathit{gr}$'s body.\looseness=-1
\end{inparaenum}

Finally, the Conditional Probability Tables (CPTs) of the variables in $\bn(p)$ are as follows.
The CPT of variables of the form $X[a]$ and $X[r,a]$ is just the OR of the values of their parents, i.e., the value is $\top$ with probability $1$ iff at least one of the parents has value $\top$.
For variables of the form $X[r,\mathit{gr},a]$ such that $\mathit{body}(r) \neq \emptyset$, then the variable's CPT encode the semantics of the rule $r$, i.e., the value of $X[r,\mathit{gr},a]$ is $\top$ with probability $1$ iff all positive literals have value $\top$ and all negative literals have value $\bot$.
In contrast, for variables of the form $X[r,\mathit{gr},a]$ such that $\mathit{body}(r) = \emptyset$, then the variable has value $\top$ with probability $v$ and $\bot$ with probability $1-v$, where $r$ is of the form $\atom{v}{a}$ (if $r = a$ then $v = 1$).

To ensure that the size of the CPT of variables of the form $X[r,a]$ is independent of the size of the relaxed grounding, instead of directly connecting variables of the form $X[r,\mathit{gr},a]$ with $X[r,a]$, we construct a binary tree of auxiliary variables where the leaves are all variables of the form $X[r,\mathit{gr},a]$ and the root is the variable $X[r,a]$.
Figure~\ref{figure:bayesian:network} depicts a portion of the BN for the program in Example~\ref{example:acyclic:1}.

\para{Complexity}
We now introduce the main result of this section.

\begin{theorem}\thlabel{theorem:inference:complexity:ptime}
The data complexity of inference for acyclic \problog{} programs is \textsc{Ptime}. 
\end{theorem} 

This follows from 
\begin{inparaenum}[(1)]
\item the relaxed grounding and the encoding can be computed in \textsc{Ptime} in terms of data complexity,
\item the encoding ensures that, for acyclic programs, the resulting Bayesian Network is a forest of poly-trees, and
\item inference algorithms for poly-tree BNs~\cite{koller2009probabilistic} run in polynomial time in the BN's size.	
\end{inparaenum}
In  \techReportAppendix{app:atklog:lite:extended}, we extend our encoding to handle additional features such as annotated disjunctions, and we prove its correctness and complexity.

\section{\tool{}}\label{sect:enforcement}

\tool{} is a DBIC mechanism that provably secures databases against probabilistic inferences. 
\tool{} is parametrized by an \atklog{} model representing the attacker's capabilities and it leverages \problog{}'s inference capabilities.

\subsection{Checking Query Security}

\begin{algorithm}[tp]
%\DecMargin{10em}
\DontPrintSemicolon
\KwIn{A system state $s = \langle \mathit{db}, U, P \rangle$, a history $h$, an action $\langle u,q\rangle$, a system configuration $C$, and a $C$-\atklog{} model $\mathit{ATK}$.}
\KwOut{The security decision in $\{\top,\bot\}$.}

\SetKwProg{Fn}{function}{}{}
\SetKwIF{If}{ElseIf}{Else}{if}{}{else if}{else}{endif}
\Begin{
  \For{$\langle u, \psi, l\rangle \in \mathit{secrets}(P,u)$}
  	{
    	  \If{$\mathit{secure}(C, \mathit{ATK}, h, \langle u, \psi, l\rangle)$}
  	  {
		% \If{$\mathit{pox}(C, \mathit{ATK}, h, \langle u, q \rangle, \top)$}{
% 			$h' :=  h \cdot \langle \langle u,q\rangle, \top, \top\rangle $\;
% 			\If{$\neg \mathit{secure}(C, \mathit{ATK}, h', \langle u, \psi, l\rangle)$}
%   	  		{
%   	  			\Return{$\bot$}\;
%   	  		}
% 		} 
		\If{$\mathit{pox}(C, \mathit{ATK}, h, \langle u, q \rangle)$}{
			$h' :=  h \cdot \langle \langle u,q\rangle, \top, \top\rangle $\;
			\If{$\neg \mathit{secure}(C, \mathit{ATK}, h', \langle u, \psi, l\rangle)$}
  	  		{
  	  			\Return{$\bot$}\;
  	  		}
		}  	 
  	  	% \If{$\mathit{pox}(C, \mathit{ATK}, h, \langle u,  q \rangle, \bot)$}{
% 			$h' :=  h \cdot \langle \langle u,q\rangle, \top, \bot \rangle$\;
% 			\If{$\neg \mathit{secure}(C, \mathit{ATK}, h', \langle u, \psi, l\rangle)$}
%   	  		{
%   	  			\Return{$\bot$}\;
%   	  		}
% 		}
  	  	\If{$\mathit{pox}(C, \mathit{ATK}, h, \langle u,  \neg q \rangle)$}{
			$h' :=  h \cdot \langle \langle u,q\rangle, \top, \bot \rangle$\;
			\If{$\neg \mathit{secure}(C, \mathit{ATK}, h', \langle u, \psi, l\rangle)$}
  	  		{
  	  			\Return{$\bot$}\;
  	  		}
		}
  	  }
  	}
  \Return{$\top$}
}
\;
\Fn{$\mathit{secure}(\langle D, \Gamma \rangle, \mathit{ATK},   h, \langle u, \psi, l\rangle)$}{
		$p := \mathit{ATK}(u)$\;
		\For{$\phi \in \mathit{knowledge}(h,u)$}
		{
			$p := p \cup \mathit{PL}(\phi) \cup \{\mathit{evidence}(\mathit{head}(\phi), \mathtt{true})\}$\;
		}
		$p := p \cup \mathit{PL}(\psi)$\;
		\Return{$\llbracket p\rrbracket_D(\mathit{head}(\psi)) < l$}	
} 
\;
% \Fn{$\mathit{pox}(\langle D, \Gamma \rangle, \mathit{ATK},   h, \langle u,\psi \rangle, v)$}{
% 	$p := \mathit{ATK}(u)$\;
% 	\For{$\phi \in \mathit{knowledge}(h,u)$}
% 	{
% 		$p := p \cup \mathit{PL}(\phi) \cup \{\mathit{evidence}(\mathit{head}(\phi), \mathtt{true})\}$\;
% 	}
% 	\If{$v = \top$}{
% 		$p := p \cup \mathit{PL}(\psi)$\;
% 		\Return{$\llbracket p\rrbracket_D(\mathit{head}(\psi)) > 0$}\;
% 	}
% 	\Else
% 	{
% 		$p := p \cup \mathit{PL}(\neg \psi)$\;
% 		\Return{$\llbracket p\rrbracket_D(\mathit{head}(\neg \psi)) > 0$}	\;
% 	}
% }
\Fn{$\mathit{pox}(\langle D, \Gamma \rangle, \mathit{ATK},   h, \langle u,\psi \rangle)$}{
	$p := \mathit{ATK}(u)$\;
	\For{$\phi \in \mathit{knowledge}(h,u)$}
	{
		$p := p \cup \mathit{PL}(\phi) \cup \{\mathit{evidence}(\mathit{head}(\phi), \mathtt{true})\}$\;
	}
	$p := p \cup \mathit{PL}(\psi)$\;
	\Return{$\llbracket p\rrbracket_D(\mathit{head}(\psi)) > 0$}\;
} 
\caption{\tool{} Enforcement Algorithm.}
\label{figure:algorithm}
\end{algorithm}

Algorithm~\ref{figure:algorithm} presents \tool{}.
It takes as input a system state $s = \langle \mathit{db},U,P\rangle$, a history $h$, the current query $q$ issued by the user $u$, a system configuration $C$, and an \atklog{} model $\mathit{ATK}$ formalizing the users' beliefs.
\tool{} checks whether disclosing the result of the current query $q$ may violate any secrets in $\mathit{secrets}(P,u)$.
If this is the case, the algorithm concludes that $q$'s execution would be insecure and returns $\bot$.
Otherwise, it returns $\top$ and authorizes  $q$'s execution.
Note that once we fix a configuration $C$ and an \atklog{} model $\mathit{ATK}$, \tool{} is a $C$-\acf{} as defined in \S\ref{sect:formal:model:system:model}. 

To check whether a query $q$ may violate a secret $\langle u, \psi, l \rangle \in \mathit{secrets}(P,u)$, \tool{} first checks whether the secret has  been already violated.
If this is not the case, \tool{} checks whether disclosing $q$ violates any secret. 
This requires checking that $u$'s belief about the secret $\psi$ stays below the threshold independently of the result of the query $q$; hence, we must ensure that $u$'s belief is below the threshold both in case the query $q$ holds in the actual database and in case $q$ does not hold (this ensures that the access control decision itself does not leak information).
\tool{}, therefore, first checks whether there exists at least one possible database state where $q$ is satisfied given $h$, using the  procedure $\mathit{pox}$.
If this is the case, the algorithm extends the current history $h$ with the new event recording that the query $q$ is authorized and its result is $\top$ and it checks whether $u$'s belief about $\psi$ is still below the corresponding threshold once $q$'s result is disclosed, using the $\mathit{secure}$ procedure.
Afterwards, \tool{} checks whether there exists at least a possible database state where $q$ is not satisfied  given $h$, it extends the current history $h$ with another event representing that the query $q$ does not hold, and it checks again whether disclosing that $q$ does not hold in the current database state violates the secret.
Note that checking whether there is a database state where $q$ is (or is not) satisfied is essential to ensure that the conditioning that happens in the $\mathit{secure}$ procedure is well-defined, i.e., the set of states we condition on has non-zero probability.

\tool{} uses the $\mathit{secure}$ subroutine to determine whether a secret's confidentiality is violated.
This subroutine takes as input a system configuration, an \atklog{} model $\mathit{ATK}$, a history $h$, and a secret $\langle u, \psi, l \rangle$.
It first computes the set $\mathit{knowledge}(h,u)$ containing all the authorized queries in the $u$-projection of $h$, i.e., $\mathit{knowledge}(h,u) = \{\phi \mid \exists i.\, h|_u(i) = \langle \langle u, \phi \rangle, \top, \top\rangle \} \cup \{\neg \phi \mid \exists i.\, h|_u(i) = \langle \langle u, \phi \rangle, \top, \bot\rangle \}$.
Afterwards, it generates a \problog{} program $p$ by extending $\mathit{ATK}(u)$ with additional rules.
In more detail, it translates each relational calculus sentence $\phi \in \mathit{knowledge}(h,u)$ to an equivalent set of \problog{} rules $\mathit{PL}(\phi)$.
The translation $\mathit{PL}(\phi)$ is standard~\cite{abiteboul1995foundations}. 
For example, given a query $\phi = (A(1) \wedge B(2)) \vee \neg C(3)$, the translation $\mathit{PL}(\phi)$ consists of the rules $\{(h_1 \leftarrow A(1)), (h_2 \leftarrow B(2)), (h_3 \leftarrow \neg C(3)), (h_4 \leftarrow h_1,h_2), (h_5 \leftarrow h_4), (h_5 \leftarrow h_3)\}$, where $h_1, \ldots, h_5$ are fresh predicate symbols.
We denote by $\mathit{head}(\phi)$ the unique predicate symbol associated with the sentence $\phi$ by the translation $\mathit{PL}(\phi)$.
In our example, $\mathit{head}(\phi)$ is the fresh predicate symbol $h_5$.
The algorithm then conditions the initial probability distribution $\mathit{ATK}(u)$ based on the sentences in $\mathit{knowledge}(h,u)$.
This is done using \emph{evidence statements}, which are special \problog{} statements of the form $\mathit{evidence}(a,v)$, where $a$ is a ground atom and $v$ is either $\mathtt{true}$ or $\mathtt{false}$; see Appendix~\ref{app:problog:introduction}.
For each sentence $\phi \in \mathit{knowledge}(h,u)$, the program $p$ contains a statement $\mathit{evidence}(\mathit{head}(\phi),\mathtt{true})$. 
Finally, the algorithm translates $\psi$ to a set of logic programming rules and checks whether $\psi$'s probability is below the  threshold $l$.

The $\mathit{pox}$ subroutine takes as input a system configuration, an \atklog{} model $\mathit{ATK}$, a history $h$, and a query $\langle u,\psi\rangle$.
%
% It determines whether there is a database $\mathit{db}'$ that complies with the history $h$ and such that $[\psi]^{\mathit{db}'} = \top$.
It determines whether there is a database $\mathit{db}'$ that satisfies $\psi$ and complies with the history $h|_u$.
Internally, the routine again constructs  a \problog{} program starting from $\mathit{ATK}$, $\mathit{knowledge}(h,u)$, and $\psi$.
Afterwards, it uses the program to check whether the probability of $\psi$ given $h|_u$ is greater than $0$.\looseness=-1

Given a run $\langle s,h\rangle$ and a user $u$, the $\mathit{secure}$ and $\mathit{pox}$ subroutines condition $u$'s initial beliefs based on the sentences in  $\mathit{knowledge}(h,u)$, instead of using the set  $\llbracket r \rrbracket_{\sim_u}$ as in the \atklog{} semantics.
The key insight is that, as we prove in \techReportAppendix{app:enforcement:mechanism:proofs}, the set of possible database states defined by the sentences in $\mathit{knowledge}(h,u)$ is equivalent to $\llbracket r \rrbracket_{\sim_u}$, which contains all database states derivable from the runs $r' \sim_u r$. %that are low-indistinguishable from $r$ for the user $u$.
This allows us to use \problog{} to implement \atklog{}'s semantics without explicitly computing $\llbracket r \rrbracket_{\sim_u}$.

 \begin{example}
Let $\mathit{ATK}$ be the attacker model in Example~\ref{example:formal:model:attacker:model}, $u$ be the user \emph{Mallory}, the database state  be $s_{\{\mathtt{A},\mathtt{B}, \mathtt{C}\} }$, where $\mathtt{Alice}$, $\mathtt{Bob}$, and $\mathtt{Carl}$ have cancer, and the policy $P$ be the one from Example~\ref{example:formal:model:access:control:policy}.
Furthermore, let $q_1, \ldots, q_4$ be the queries issued by \emph{Mallory} in Example~\ref{example:system:model:run}.
\tool{} permits the execution of the first two queries since they do not violate the policy.
In contrast, it denies the execution of the last two queries as they leak sensitive information.
\end{example}

\para{Confidentiality}
As we prove in~\techReportAppendix{app:enforcement:mechanism:proofs}, \tool{} provides the desired security guarantees for any \atklog{}-attacker. 
Namely, it authorizes only those queries whose disclosure does not increase an attacker's beliefs in the secrets  above the corresponding thresholds.
\tool{} also provides precise completeness guarantees: it authorizes all secrecy-preserving queries.
Informally, a query $\langle u,q\rangle$ is \textit{secrecy-preserving} given a run $r$ and a secret $\langle u, \psi, l \rangle$ iff disclosing the result of $\langle u,q\rangle$ in any run $r' \sim_u r$ does not violate the secret.

\begin{theorem}\thlabel{theorem:angerona:confidentiality}
Let a system configuration $C$ and a $C$-\atklog{} model $\mathit{ATK}$ be given, and let \tool{} be the $C$-\acf{} $f$.
$\tool{}$ provides \confidentiality{} with respect to $C$ and $\lambda  u \in {\cal U}.\ \llbracket \mathit{ATK}(u) \rrbracket_D$,
 and it authorizes all secrecy-preserving queries.
\end{theorem}

\para{Complexity}
\tool{}'s complexity is dominated by the complexity of inference.
We focus our analysis only on data complexity, i.e., the complexity when only the ground atoms in the \problog{} programs are part of the input while everything else is fixed.
A \emph{literal query} is a query consisting either of a ground atom $a(\overline{c})$ or its negation $\neg a(\overline{c})$.
We call  an \atklog{} model \emph{acyclic} if all belief programs in it are acyclic.
Furthermore, a \emph{literal secret} is a secret $\tup{U, \phi, l}$ such that $\phi$ is a literal query.
We prove in \techReportAppendix{app:enforcement:mechanism:proofs}~that for acyclic \atklog{} models, literal queries, and literal secrets, the \problog{} programs produced by the $\mathit{secure}$ and $\mathit{pox}$ subroutines are acyclic.
We can therefore use our dedicated inference engine from \S\ref{sect:inference} to reason about them.
Hence,  \tool{} can be used to protect databases in \textsc{Ptime} in terms of data complexity.
Literal queries are expressive enough to formulate  queries about the database content such as ``does $\mathit{Alice}$ have cancer?''. 

\begin{theorem}\thlabel{theorem:angerona:complexity:ptime}
For all acyclic \atklog{} attackers, 
for all literal queries $q$,
for all runs $r$ whose histories contain only literal queries and contain only secrets expressed using literal queries,
\tool{}'s data complexity is \textsc{Ptime}. 
\end{theorem} 

\para{Discussion}
Our tractability guarantees apply only to acyclic \atklog{} models, literal queries, and literal secrets.
Nevertheless, \tool{} can still handle relevant problems of interest.
As stated in \S\ref{sect:inference}, acyclic models are as expressive as poly-tree Bayesian Networks, one of the few classes of Bayesian Networks with tractable inference.
Hence, for many probabilistic models that cannot be represented as acyclic \atklog{} models, exact probabilistic inference is intractable.

Literal queries are expressive enough to state simple facts about the database content.
More complex (non-literal) queries can be simulated using (possibly large) sequences of literal queries.
Similarly, policies with non-literal secrets can be implemented as sets of literal secrets, and the Boole--Fr\'{e}chet inequalities~\cite{hailperin1984probability} can be used to derive the desired thresholds.
In both cases, however,  our completeness guarantees hold only for the resulting literal queries, not for the original ones.

Finally, whenever our tractability constraints are violated, \tool{} can still be used  by directly using \problog{}'s inference capabilities.
In this case, one would retain the security and completeness guarantees (Theorem~\ref{theorem:angerona:confidentiality}) but lose the tractability guarantees (Theorem~\ref{theorem:angerona:complexity:ptime}).\looseness=-1

\subsection{Implementation and Empirical Evaluation}

\begin{figure}[t]
%\begin{tabular}{c c}
   % \begin{subfigure}{0.48\textwidth}
\centering
\scalebox{0.8}{
    \begin{tikzpicture}
	\tikzstyle{every node}=[font=\small]    
    
      \begin{axis}[
          width=\linewidth, % Scale the plot to \linewidth
          grid=major, 
          grid style={dashed,gray!30},
          xlabel=Number of patients, % Set the labels
          ylabel={Time [$\mathit{ms}$]},
          scaled y ticks = true,
          scaled x ticks = false,
		  legend style={at={(0,1)},anchor=north west},
          x tick label style={at={(axis description cs:0.5,-0.1)},anchor=north, 
       /pgf/number format/.cd,
            fixed zerofill = false,
            precision=2,
            fixed,
            1000 sep={},
        /tikz/.cd,
        %ymax=3
    }
            ]
        
 		\addplot[mark = o, /tikz/solid, every mark/.append style={solid} ] table[x index={0},y  expr=((\thisrowno{1})) ,col sep=comma] {logs/sizeXtime_example_0.csv};
  		%\addplot[mark = triangle, /tikz/densely dotted, every mark/.append style={solid} ] table[x index={0},y  expr=((\thisrowno{2})/1000),col sep=comma] {logs/sizeXtime_example_0.csv};
  		%\addplot[mark = x, /tikz/densely dotted, every mark/.append style={solid} ] table[x index={0},y  expr=((\thisrowno{3})/1000),col sep=comma] {logs/sizeXtime_example_0.csv};
        %\legend{Acyclic inference, Acyclic inference and grounding, \problog{}}
      \end{axis}
    \end{tikzpicture}
   }

    \caption{\tool{} execution time in seconds.}\label{figure:results}
\end{figure}
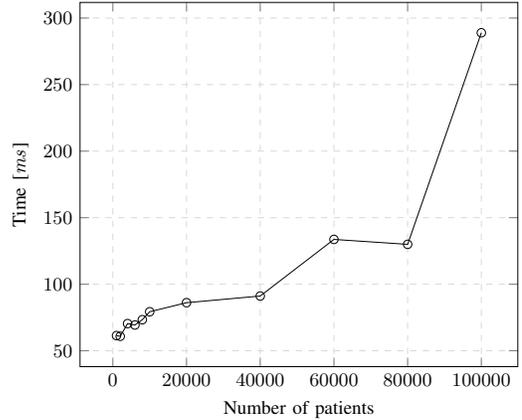

To evaluate the feasibility of our approach in practice, we implemented a prototype of \tool{}, available at~\cite{prototype}.
The prototype implements 
our dedicated inference algorithm for acyclic \problog{} programs (\S\ref{sect:inference}), which computes the relaxed grounding of the input program $p$, constructs the  Bayesian Networks $\mathit{BN}$, and performs the inference over $\mathit{BN}$ using belief propagation~\cite{koller2009probabilistic}.
For inference over $\mathit{BN}$, we rely on the GRRM library~\cite{mallet}.
Observe that evidence statements in \problog{} are encoded by fixing the values of the corresponding random variables in the BN.
Note also that computing the relaxed grounding of $p$ takes polynomial time in terms of data complexity, where the exponent is determined by $p$'s rules.
A key optimization is to pre-compute the relaxed grounding and construct $\mathit{BN}$ off-line.
%
%This can be done off-line, which
This avoids grounding $p$ and constructing the (same) Bayesian Network for each query.
In our experiments we measure this time separately.

We use our prototype to study \tool{}'s efficiency and scalability.
We run our experiments on a PC with an Intel i7 processor and 32GB of RAM.
For our experiments, we consider the database schema from \S\ref{sect:formal:model}.
For the belief programs, we use the \problog{} program given in \S\ref{sect:language}, which can be encoded as an acyclic program when the parent-child relation is a poly-tree. 
We evaluate \tool{}'s efficiency and scalability in terms of the number of ground atoms in the belief programs. 
We generate synthetic belief programs containing 1,000 to 100,000 patients and  for each of these instances, we generate 100 random queries of the form $R(\overline{t})$, where $R$ is a predicate symbol and $\overline{t}$ is a tuple.
For each instance and sequence of queries, we check the security of each  query with our prototype, against a policy containing 100 randomly generated  secrets specified as literal queries. %

Figure~\ref{figure:results} reports the execution times for our case study.
Once the BN is generated,  \tool{} takes under {300} milliseconds, even for our larger examples, to check a query's security.
During the initialization phase of our dedicated inference engine, we ground the original \problog{} program and translate it into a BN.
Most of the time is spent in the grounding process, whose data complexity is polynomial, where the polynomial's degree is determined by the number of free variables in the belief program.
Our prototype uses a naive bottom-up grounding technique, and for our larger examples the initialization times is less than 2.5 minutes. % for the medical data example.
We remark, however, that the initialization is performed just once per belief program.
Furthermore, it can be done offline and its performance can be greatly improved by naive parallelization.

\section{Related Work}\label{sect:related:work}

\para{Database Inference Control}
Existing DBIC approaches protect databases  (either at design time~\cite{morgenstern1987security, morgenstern1988controlling} or at runtime~\cite{chen2006database, chen2007protection,hale1997catalytic}) only against restricted classes of probabilistic dependencies, e.g., those arising from functional and multi-valued dependencies.
\atklog{}, instead, supports arbitrary probabilistic dependencies, and even our acyclic fragment can express probabilistic dependencies that are not supported by~\cite{morgenstern1987security, morgenstern1988controlling, chen2006database, chen2007protection,hale1997catalytic}. 
Weise~\cite{wiese2010keeping} proposes a DBIC framework, based on possibilistic logic, that formalizes secrets as sentences and expresses policies by associating bounds to secrets.
Possibility theory differs from probability theory, which results in subtle differences.
For instance, there is no widely accepted definition of conditioning for possibility distributions, cf.~\cite{Bouchon-Meunier2002}. 
Thus, the probabilistic model from \S\ref{sect:motivating:example} cannot be encoded in Weise's framework~\cite{wiese2010keeping}.

Statistical databases store information associated to different individuals and support queries that return statistical information about an entire population~\cite{chin1982auditing}.
DBIC solutions for statistical databases~\cite{dobkin1979secure, chin1982auditing, domingo2002inference, adam1989security, denning1980secure} prevent leakages of information about single individuals while allowing the execution of statistical queries.
These approaches rely on various techniques, such as perturbating the original data, synthetically generating data, or restricting the data on which the queries are executed. 
Instead,  we protect specific secrets in the presence of probabilistic data dependencies and  we return the original query result, without modifications, if it is secure.

\para{Differential Privacy}
Differential Privacy~\cite{dwork2006differential,dwork2014algorithmic} is widely used for  privacy-preserving data analysis.
Systems such as  ProPer~\cite{ebadi2015differential} or PINQ~\cite{mcsherry2009privacy}  provide users with automated ways to perform differentially private computations.
A differentially private computation guarantees that the presence (or absence) of an individual's data in the input data set affects the probability of the computation's result only in limited way, i.e., by at most a factor $e^\epsilon$ where $\epsilon$ is a parameter controlling the privacy-utility trade-off.
While differential privacy does not make any assumption about the attacker's beliefs, we assume that the attacker's belief is known and we guarantee that for all secrets in the policy, no user can increase his beliefs, as specified in the attacker model, over the corresponding thresholds by interacting with the system.

\para{Information Flow Control}
Quantified Information Flow~\cite{lowe2002quantifying,clark2007static, alvim2010probabilistic,Kopf:2007:IMA:1315245.1315282} aims at quantifying the amount of information leaked by a program. 
Instead of measuring the amount of leaked information, we focus on restricting the information that an attacker may obtain about a set of given secrets.

Non-interference has been extended to consider probabilities~\cite{sabelfeld2000probabilistic, volpano1999probabilistic, aldini2001probabilistic} for concurrent programs.
Our security notion, instead, allows those leaks that do not increase an attacker's beliefs in a secret above the threshold, and it can be seen as a probabilistic extension of \emph{opacity}~\cite{schoepe2015understanding}, which allows any leak except leaking whether the secret holds.

Mardziel et al.~\cite{mardziel2013dynamic} present a general DBIC architecture, where users' beliefs are expressed as probabilistic programs,  security requirements as threshold on these beliefs, and the beliefs are updated in response to the system's behaviour.
Our work directly builds on top of this architecture.
However, instead of using an imperative probabilistic language, we formalize beliefs using probabilistic logic programming, which provides a natural and expressive language for formalizing dependencies arising in the database setting, e.g., functional and multi-valued dependencies, as well as common probabilistic models, like Bayesian Networks.

Mardziel et al.~\cite{mardziel2013dynamic} also propose a DBIC mechanism based on abstract interpretation.
They do not provide any precise complexity bound for their mechanism.
Their algorithm's complexity class, however, appears to be intractable, since they use a probabilistic extension of the polyhedra abstract domain, whose asymptotic complexity is exponential in the number of program variables~\cite{singh2017fast}.
In contrast, \tool{} exploits our inference engine for acyclic programs to secure databases against a practically relevant class of probabilistic inferences, and it provides precise tractability and completeness guarantees.\looseness=-1

We now compare (unrestricted) \atklog{} with the imperative probabilistic language used in~\cite{mardziel2013dynamic}.
\atklog{} allows one to concisely encode probabilistic  properties specifying relations between tuples in the database.
For instance, a property like ``the probability of $A(x)$ is $\sfrac{1}{2}^n$, where $n$ is the number of tuples $(x,y)$ in $B$'' can be encoded as $\atom{\sfrac{1}{2}}{A(x)}\leftarrow B(x,y)$.
Encoding this property as an imperative program is more complex; it requires a \textbf{for} statement to iterate over all variables representing tuples in $B$ and an \textbf{if} statement to filter the tuples.
In contrast to~\cite{mardziel2013dynamic}, \atklog{} provides limited support for numerical constraints (as we support only finite domains).
Mardziel et al.~\cite{mardziel2013dynamic} formalize queries as imperative probabilistic programs.
They can, therefore, also model probabilistic queries or the use of randomization to limit disclosure.
While all these features are supported by \atklog{}, our goal is to protect databases from attackers that use standard query languages like SQL. 
Hence, we formalize queries using relational calculus and  ignore probabilistic queries.
Similarly to~\cite{mardziel2013dynamic}, our approach can be extended to handle some uncertainty on the attackers' capabilities.
In particular, we can associate to each user a finite number of possible beliefs, instead of a single one.
However, like~\cite{mardziel2013dynamic}, we cannot handle infinitely many alternative beliefs.

\para{Probabilistic Programming}
Probabilistic programming is an active area of research~\cite{gordon2014probabilistic}.
Here, we position \problog{} with respect to expressiveness and inference.
Similarly to~\cite{mardziel2013dynamic,getoor2007introduction}, \problog{} can express only discrete probability distributions, and it is less expressive than languages supporting continuous distributions~\cite{shan2017exact,gehr2016psi, gordon2014tabular}. 
Current exact inference algorithms for probabilistic programs are based on program analysis techniques, such as symbolic execution~\cite{gehr2016psi, shan2017exact} or abstract interpretation~\cite{mardziel2013dynamic}.
In this respect, we present syntactic criteria that \textit{ensure} tractable inference for \problog{}.
Sampson et al.~\cite{sampson2014expressing} symbolically execute probabilistic programs and translate them to BNs to verify probabilistic assertions.
In contrast, we translate \problog{} programs to BNs to perform exact inference and our translation is tailored to work together with our acyclicity constraints to allow tractable inference.

\section{Conclusion}\label{sect:conclusion}

Effectively securing databases that store data with probabilistic dependencies requires an expressive language to capture the dependencies and a tractable enforcement mechanism.
To address these requirements, we developed \atklog{}, a formal language providing an expressive and concise way to represent attackers' beliefs while interacting with the system. 
We leveraged this to design \tool{}, a provably secure DBIC mechanism that prevents the leakage of sensitive information in the presence of probabilistic dependencies.
\tool{} is based on a dedicated inference engine for a fragment of \problog{} where exact inference is tractable.

We see these results as providing a foundation for building practical protection mechanisms, which include probabilistic dependencies, as part of real-world database systems.
As future work, we plan to extend our framework to dynamic settings where the database and the policy change. 
We also intend to explore different fragments of \problog{} and relational calculus for which exact inference is practical

% \smallskip
{
% \noindent
\small
\para{Acknowledgments}
We thank Ognjen Maric, Dmitriy Traytel, Der-Yeuan Yu, and  the anonymous reviewers for their comments.
}

\bibliographystyle{IEEEtranS}
\bibliography{IEEEabrv,bib/bib}

\onlyShortVersion{
\appendices

\onlyTechReport{
\clearpage
}
\onlyShortVersion{
 \balance
}
\section{\problog{}}\label{app:problog:introduction}

Here, we provide a formal account of \problog{}, which follows~\cite{de2007problog,fierens2015inference, de2015probabilistic}.
In addition to \problog{}'s semantics, we present here again some (revised) aspects of \problog{}'s syntax, which we introduced in \S\ref{sect:language}.
As mentioned in \S\ref{sect:language},  we restrict ourselves to the function-free fragment of \problog{}.

\para{Syntax}
As introduced in \S\ref{sect:language}, a \emph{$(\Sigma, \mathbf{dom})$-probabilistic atom} is an atom $a \in {\cal A}_{\Sigma,\mathbf{dom}}$ annotated with a value $0 \leq v \leq 1$, denoted $\atom{v}{a}$.
If $v = 1$, then we write $a(\overline{c})$ instead of $\atom{1}{a(\overline{c})}$. % and we treat the probabilistic atom just as a normal atom in ${\cal A}_{\Sigma,\mathbf{dom}}$. 
A \emph{$(\Sigma, \mathbf{dom})$-\problog{} program} is a finite set of ground probabilistic $(\Sigma, \mathbf{dom})$-atoms and  $(\Sigma, \mathbf{dom})$-rules. 
Note that ground atoms $a \in {\cal A}_{\Sigma,\mathbf{dom}}$ are represented as $\atom{1}{a}$.
Observe also that rules do not involve probabilistic atoms (as formalized in \S\ref{sect:language}).
This is without loss of generality: as we show below probabilistic rules and annotated disjunctions can be represented using only probabilistic atoms and non-probabilistic rules.
We denote by $\mathit{prob}(p)$ the set of all probabilistic ground atoms $\atom{v}{a}$ in $p$, i.e., $\mathit{prob}(p) := \{ \atom{v}{a} \in p \mid 0\leq v \leq 1 \wedge a \in {\cal A}_{\Sigma,\mathbf{dom}} \}$, and by $\mathit{rules}(p)$ the non-probabilistic rules in $p$, i.e., $\mathit{rules}(p) :=  p \setminus \mathit{prob}(p)$.
As already stated in \S\ref{sect:language}, we consider only programs $p$ that do not contain negative cycles in the rules.
Finally, we say that a \problog{} program $p$ is a \emph{logic program} iff $v = 1$ for all $\atom{v}{a} \in \mathit{prob}(p)$, i.e., $p$ does not contain probabilistic atoms.

\para{Semantics}
Given a $(\Sigma, \mathbf{dom})$-\problog{} program $p$, a \emph{$p$-grounded instance} is a \problog{} program $A \cup R$, where the set of ground atoms $A$ is a subset of $\{ a \mid \atom{v}{a} \in \mathit{prob}(p) \}$ and $R = \mathit{rules}(p)$.
Informally, a grounded instance of $p$ is one of the logic programs that can be obtained by selecting some of the probabilistic atoms in $p$ and keeping all rules in $p$.
A \emph{$p$-probabilistic assignment} is a total function associating to each probabilistic atom $\atom{v}{a}$ in $\mathit{prob}(p)$ a value in $\{\top,\bot\}$.
We denote by $\mathit{A}(p)$ the set of all $p$-probabilistic assignments.
The \emph{probability} of a $p$-probabilistic assignment $f$ is $\mathit{prob}(f) = \Pi_{\atom{v}{a} \in \mathit{prob}(p)} \left( \Pi_{f(\atom{v}{a}) = \top} v \cdot \Pi_{f(\atom{v}{a}) = \bot} (1-v) \right)$.
Given a $p$-probabilistic assignment $f$,  $\mathit{instance}(p,  f)$ denotes the $p$-grounded instance $\{a \mid \exists v.\, f(\atom{v}{a}) = \top \} \cup \mathit{rules}(p)$. 
Finally, given a $p$-grounded instance $p'$, $\mathit{WFM}(p')$ denotes the well-founded model associated with the logic program $p'$, as defined by standard logic programming semantics~\cite{abiteboul1995foundations}.

The semantics of a $(\Sigma, \mathbf{dom})$-\problog{} program $p$ is defined as a probability distribution over all possible $p$-grounded instances.
Note that \problog{}'s semantics relies on the \emph{closed world assumption}, namely every fact that is not in a given model is considered false.
The semantics of $p$, denoted by $\llbracket p \rrbracket$, is as follows:
$\llbracket p \rrbracket(p') = \sum_{f \in F_{p,p'}} \mathit{prob}(f)$, where $ F_{p,p'} = \{ f \in \mathit{A}(p) \mid p'= \mathit{instance}(p,f) \}$.
We remark that a $(\Sigma, \mathbf{dom})$-\problog{} program $p$ implicitly defines a probability distribution over $(\Sigma, \mathbf{dom})$-structures.
Indeed, the probability of a given $(\Sigma,\mathbf{dom})$-structure $s$ is the sum of the probabilities of all $p$-grounded instances whose well-founded model  is $s$.
With a slight abuse of notation, we extend the semantics of $p$ to $(\Sigma,\mathbf{dom})$-structures and ground atoms as follows:
$\llbracket p \rrbracket(s) = \sum_{f \in {\cal M}(p,s)}\mathit{prob}(f)$, where  $s$ is a $(\Sigma, \mathbf{dom})$-structure and ${\cal M}(p,s)$ is the set of all  assignments $f$ such that $\mathit{WFM}(\mathit{instance}(p,f)) = s$. %, and
Finally, $p$'s semantics can be lifted to sentences as follows: $\llbracket p \rrbracket(\phi) = \Sigma_{s \in \llbracket \phi \rrbracket}\llbracket p \rrbracket(s)$, where $\llbracket \phi \rrbracket = \{s \in \Omega_D^\Gamma \mid [\phi]^s = \top \}$.

\para{Evidence}
\problog{} supports expressing \emph{evidence} inside programs~\cite{de2015probabilistic}.
To express evidence, i.e., to condition a distribution on some event, we use statements of the form $\mathit{evidence}(a, v)$, where $a$ is a ground atom and $v \in \{\mathtt{true},\mathtt{false}\}$.
Let $p$ be a $(\Sigma, \mathbf{dom})$-\problog{} program $p$ with evidence $\mathit{evidence}(a_1, v_1), \ldots, \mathit{evidence}(a_n, v_n)$, and $p'$ be the program without the evidence statements.
Furthermore, let $\mathit{POX}(p')$ be the set of all $(\Sigma, \mathbf{dom})$-structures $s$ complying with the evidence, i.e., the set of all states $s$ such that $a_i$ holds in $s$ iff $v_i = \mathtt{true}$.
Then, $\llbracket p \rrbracket(s)$, for a $(\Sigma, \mathbf{dom})$-structure $\mathit{s} \in \mathit{POX}(p)$, is $\llbracket p' \rrbracket(s) \cdot \left( \sum_{s' \in \mathit{POX}(p)} \llbracket p' \rrbracket(s') \right)^{-1}$.

\para{Syntactic Sugar}
Following~\cite{de2007problog,fierens2015inference, de2015probabilistic}, we extend \problog{} programs with two additional constructs: annotated disjunctions and probabilistic rules.
As shown in~\cite{de2015probabilistic}, these constructs are just syntactic sugar.
A \emph{probabilistic rule} is a \problog{} rule where the head is a probabilistic atom.
The probabilistic rule $\atom{v}{h} \leftarrow l_1, \ldots, l_n$ can be encoded using the additional probabilistic atoms $\atom{v}{sw(\_)}$ and the rule $h \leftarrow l_1, \ldots, l_n, sw(\overline{x})$, where $\mathit{sw}$ is a fresh predicate symbol, $\overline{x}$ is the tuple containing the variables in $\mathit{vars}(h) \cup \bigcup_{1 \leq i \leq n} \mathit{vars}(l_i)$, and $\atom{v}{sw(\_)}$ is a shorthand representing the fact that there is a probabilistic atom $\atom{v}{sw(\overline{t})}$ for each tuple $\overline{t} \in \mathbf{dom}^{|\overline{x}|}$.

An \emph{annotated disjunction} $\atom{v_1}{a_1}; \ldots; \atom{v_n}{a_n}$, where $a_1, \ldots, a_n$ are ground atoms and $\left( \sum_{1 \leq i\leq n} v_i \right) \leq 1$, denotes that $a_1, \ldots, a_n$ are mutually exclusive probabilistic events happening with probabilities $v_1, \ldots, v_n$.
It can be encoded \!as:
\begin{align*}
&\atom{p_1}{sw_1(\_)}\\
%&\atom{p_2}{sw_2(\_)}\\
&\quad \vdots\\
&\atom{p_n}{sw_n(\_)}\\
&a_1(\overline{t}_1) \leftarrow sw_1(\overline{t}_1)\\
%&a_2(\overline{t}_2) \leftarrow \neg sw_1(\overline{t}_1), sw_2(\overline{t}_2)\\
&\quad \vdots\\
&a_n(\overline{t}_n) \leftarrow \neg sw_1(\overline{t}_1), \ldots, \neg sw_{n-1}(\overline{t}_{n-1}), sw_n(\overline{t}_n),
\end{align*}
where each $p_i$, for $1 \leq i \leq n$, is $v_i \cdot \left(1-\sum_{1 \leq j < i} v_j\right)^{-1}$.
Probabilistic rules can be easily extended to support annotated disjunctions in their heads.

\begin{example}
Let $\Sigma$ be a first-order signature with two predicate symbols $V$ and $W$, both with arity $1$, $\mathbf{dom}$ be the domain $\{ a,b \}$, and 
$p$ be the program consisting of the facts $\atom{\sfrac{1}{4}}{T(a)}$ and  $\atom{\sfrac{1}{2}}{T(b)}$, the annotate disjunction $\atom{\sfrac{1}{4}}{W(a)};	\atom{\sfrac{1}{2}}{W(b)}$, and the rule $\atom{\sfrac{1}{2}}{T(x)} \leftarrow W(x)$.

The probability associated to each $(\Sigma,\mathbf{dom})$-structure by the program $p$ is shown in the following table.
\begin{center}
	{\small
\begin{tabular}{c c | c c c c }
    					&				&	\multicolumn{4}{c}{$W$}	\\
    					&				&	$\emptyset$		&	$\{a\}$			&	$\{b\}$			&	$\{a,b\}$	\\ \hline
\multirow{4}{*}{$T$}	&	$\emptyset$	&	$\sfrac{3}{32}$	&	$\sfrac{3}{64}$	&	$\sfrac{3}{32}$	&	$0$	\\
					&	$\{a\}$		&	$\sfrac{1}{32}$	&	$\sfrac{5}{64}$	&	$\sfrac{1}{32}$	&	$0$	\\
					&	$\{b\}$		&	$\sfrac{3}{32}$	&	$\sfrac{3}{64}$	&	$\sfrac{9}{32}$	&	$0$	\\
					&	$\{a,b\}$	&	$\sfrac{1}{32}$	&	$\sfrac{5}{64}$	&	$\sfrac{3}{32}$	&	$0$    
\end{tabular}
}
\end{center}

The empty structure has probability $\sfrac{3}{32}$. 
The only grounded instance whose well-formed model is the empty database is the instance $i_1$ that does not contain grounded atoms.
Its probability is $\sfrac{3}{32}$ because the probability that $T(a)$ is not in $i_1$ is $\sfrac{3}{4}$, the probability that $T(b)$ is not in $i_1$ is $\sfrac{1}{2}$, and the probability that neither $W(a)$ nor $W(b)$ are in $i_1$ is $\sfrac{1}{4}$ and all these events are independent.

The probability of some structures is determined by more than one grounded instance.
For example, the probability of the structure $s$ where $s(T) = \{a,b \}$ and $s(W)= \{ a\}$ is $\sfrac{5}{64}$.
There are two grounded instances $i_2$ and $i_3$ whose well-founded model is $s$.
The instance $i_2$ has probability $\sfrac{1}{16}$ and it consists of the atoms $T(b),W(a), \mathit{sw}(a)$ and the rule $T(x)  \leftarrow W(x),\mathit{sw}(x)$, whereas the instance $i_3$ has probability $\sfrac{1}{64}$ and it consists of  the atoms $T(a),T(b),W(a)$ and the rule $T(x)  \leftarrow W(x),\mathit{sw}(x)$.
Note that before computing the ground instances, we translated probabilistic rules and annotated disjunctions into standard \problog{} rules.
\end{example}

}

\onlyTechReport{
\newpage
\appendices

\section*{Appendices}
Here, we provide additional details about numerous aspects of the paper, such as acyclic \problog{} programs and \tool{}.
We also provide complete proofs of all our results.

\para{Structure}
Appendix~\ref{app:problog:introduction} reviews \problog{}'s syntax and semantics.
In Appendix~\ref{app:atklog:lite:extended}, we formalize acyclicity for \problog{} programs, we present our encoding from acyclic programs to poly-tree Bayesian Networks, we prove the correctness of this encoding, and we give complexity results for inference for acyclic \problog{} programs.
Finally, Appendix~\ref{app:enforcement:mechanism:proofs} contains additional details about \tool{} as well as all the security and complexity proofs.
\clearpage

\section{Acyclic \problog{} programs}\label{app:atklog:lite:extended}
The belief programs in \atklog{} are based on \problog{}, which is a very expressive probabilistic language and it can be used to encode a wide range of probabilistic models.
As a result, performing exact inference in \problog{} is intractable in general, i.e., it is $\#P$-hard in terms of data complexity.

We identify acyclic \problog{} programs, a restricted class of \problog{} programs where exact inference can be done efficiently, in linear time in terms of data complexity.

We first introduce acyclic programs.
Afterwards, we introduce relaxed acyclic programs, a generalization of acyclic programs, and we provide an encoding of relaxed acyclic programs as a  Bayesian Network.
Observe that here we state again (sometimes with some generalizations) some of the concepts we introduced in the main paper.
Note also that in the main paper we do not distinguish between acyclic programs and relaxed acyclic programs.

\subsection{Acyclic Programs}

Here, we define acyclic \problog{} programs.
Intuitively, an acyclic program is a \problog{} program whose probability distribution can be represented using a poly-tree Bayesian Network.
While inference for general Bayesian Networks is intractable, i.e., $\mathit{NP}$-hard, inference for poly-tree Bayesian Networks can be done in linear time in the size of the corresponding Bayesian Network.

\subsubsection{Acyclic \problog{} Programs}
In the following, let $\Sigma$ be a first-order signature, $\mathbf{dom}$ be a finite domain, and $p$ be a $(\Sigma, \mathbf{dom})$-program.
Without loss of generality, in the following we assume that (1) inside rules (not ground atoms), constants are used just inside equality and inequality atoms, i.e., instead of $B(x,7) \leftarrow D(x,y)$ we consider only the equivalent rule $B(x,k) \leftarrow D(x, y), k=7$, and (2) there are no distinct $v_1$ and $v_2$ such that $\atom{v_1}{a(\overline{c})} \in p$ and $\atom{v_2}{a(\overline{c})} \in p$ (this is without loss of generality since multiple occurrences of the same ground atom with different probabilities can be represented as a single ground atom an adjusted probability). 
Observe that in the following we re-define, extend, or make more precise the definitions from \S\ref{sect:inference}.

\para{Dependency Graph}
The \emph{$p$-dependency graph}, denoted $\mathit{graph}(p)$, is the directed labelled graph $(N,E)$ where $N = \Sigma$ and $E = \{ \mathit{pred}(\mathit{body}(r,i)) \xrightarrow{r,i} \mathit{pred}(\mathit{head}(r)) \mid r \in p \wedge i \in \mathbb{N} \}$.

\para{Directed Paths and Cycles}
A \emph{directed path} in $\mathit{graph}(p)$ is a path  $\mathit{pr}_1 \xrightarrow{r_1, i_1} \ldots  \xrightarrow{r_{n-1}, i_{n-1}} \mathit{pr}_n$ in $\mathit{graph}(p)$ such that for all $1 \leq j < n$, $p_j = p_{j+1}$.
Two directed paths $\mathit{pr}_1 \xrightarrow{r_1, i_1} \ldots  \xrightarrow{r_{n-1}, i_{n-1}} \mathit{pr}_n$ and $\mathit{pr}_1' \xrightarrow{r_1', i_1'} \ldots  \xrightarrow{r_{m-1}', i_{m-1}} \mathit{pr}_m'$ are \emph{head-connected} iff  $\mathit{pr}_1 = \mathit{pr}_1'$. 
Similarly, two directed paths $\mathit{pr}_1' \xrightarrow{r_1, i_1} \ldots  \xrightarrow{r_{n-1}, i_{n-1}} \mathit{pr}_n$ and $\mathit{pr}_1 \xrightarrow{r_1', i_1'} \ldots  \xrightarrow{r_{m-1}', i_{m-1}} \mathit{pr}_m'$ are \emph{tail-connected} iff  $\mathit{pr}_n = \mathit{pr}_m'$.
Note that an \emph{empty path} consists of just one node in $N$, i.e., one predicate symbol.
A \emph{directed cycle} in $\mathit{graph}(p)$ is a directed path $\mathit{pr}_1 \xrightarrow{r_1, i_1} \ldots \xrightarrow{r_{n-1}, i_{n-1}} \mathit{pr}_n$ in $\graph(p)$ such that $\mathit{pr}_n = \mathit{pr}_1$. 
Given a directed path $P = \mathit{pr}_1 \xrightarrow{r_1, i_1} \ldots  \xrightarrow{r_{n-1}, i_{n-1}} \mathit{pr}_n$ in $\graph(p)$, we denote by $\mathit{start}(P)$ the predicate symbol $\mathit{pr}_1$ and by $\mathit{end}(P)$ the predicate symbol $\mathit{pr}_n$.
We say that a directed cycle $C$ is \emph{simple} iff (1) each edge occurs at most once in $C$, and (2) there is a predicate symbol $\mathit{pr}$ such that $\mathit{pr}$ occurs twice in $C$ and for all $\mathit{pr}' \neq \mathit{pr}$, $\mathit{pr}'$ occurs at most once in $C$.

\para{Undirected Paths and Cycles}
A directed path $P_1$ in $\graph(p)$ is \emph{connected with $P_2$} iff $\mathit{end}(P_1) = \mathit{end}(P_2)$, $\mathit{start}(P_1) = \mathit{start}(P_2)$, or $\mathit{end}(P_1) = \mathit{start}(P_2)$.
An undirected path in $\graph(p)$ is a sequence of directed paths $P_1, \ldots, P_n$ such that, for all $1 \leq i < n$, $P_i$ is connected with $P_{i+1}$.
An undirected cycle in $\graph(p)$ is a sequence of directed paths $P_1, \ldots, P_n$ such that (1) $P_1, \ldots, P_n$ is an undirected path and (2) $P_n$ is connected with $P_1$.
Note that a directed path (respectively cycle) is also an undirected path (respectively cycle).
We say that two (directed or undirected) cycles are equivalent iff they are the same cycle.

\para{Undirected Unsafe Structures}
An \emph{undirected unsafe structure} in $\mathit{graph}(p)$ is quadruple $\langle D_1, D_2, D_3, U \rangle$ such that
(1) $D_1, D_2, D_3$ are directed paths in $\mathit{graph}(p)$,
(2) $U$ is an undirected path in $\mathit{graph}(p)$,
(3) $D_1$ and $D_2$ are non-empty or $D_2$ and $D_3$ are non-empty, 
(4) $D_1$ and $D_2$ are head-connected,
(5) $D_2$ and $D_3$ are tail-connected, and
(6) $D_1, U, D_3, D_2$ is an undirected cycle in $\graph(p)$.
An undirected unsafe structure $\langle D_1, D_2, D_3, U \rangle$ \emph{covers} an undirected cycle $U'$ iff $D_1, U, D_3, D_2$ is equivalent to $U'$.

\para{Directed Unsafe Structures}
A \emph{directed unsafe structure} in $\mathit{graph}(p)$ is a directed cycle $C$ in $\mathit{graph}(p)$.
A directed unsafe structure $C$ \emph{covers} a directed cycle $C'$ iff $C$ is equivalent to $C'$.

\para{Deriving Ordering Annotations}
A \emph{$\Sigma$-ordering annotation} is $\order{A}$, where $A \subseteq \Sigma$ and there is a $k \in \mathbb{N}$ such that for all $a \in A$, $|a| = 2k$. 
Let $\preceq$ be a well-founded partial order over $2^\Sigma$.
We say that the ordering annotation $\order{A}$ can be derived from the program $p$ given $\preceq$, written $p, \preceq \models \order{A}$, iff one of the following conditions hold:
\begin{compactenum}
\item For all $\mathit{pr} \in A$, there is no rule $r$ in $p$ such that $\mathit{pred}(\mathit{head}(r)) = \mathit{pr}$ and $\mathit{body}(r) \neq \emptyset$,
there is no $\mathit{pr}' \in \Sigma$ such that $\mathit{pr}' \prec \mathit{pr}$ and $p, \preceq \models \order{\mathit{pr}'}$, and 
the transitive closure of the relation $R = \bigcup_{\mathit{pr} \in A} \{ (\overline{c}, \overline{v}) \mid \exists r \in p.\, ((\mathit{head}(r) =pr(\overline{c},\overline{v}) \vee \exists v'.\ \mathit{head}(r) = \atom{v'}{pr(\overline{c},\overline{v})}) \wedge \mathit{body}(r) = \emptyset) \wedge |\overline{c}| = |\overline{v}| = \sfrac{|\mathit{pr}|}{2}\}$ is strict partial order. %, and 
\item There is a set $A' \subseteq \Sigma$ such that (1) $A' \preceq A$, (2) $p, \preceq \models \order{A'}$, (3) there is a $\mathit{pr} \in \Sigma$ and a $\mathit{pr}' \in A'$ such that 
(a) $|\mathit{pr}| = |\mathit{pr}'|$,
(b) $A = (A' \setminus \{\mathit{pr}'\}) \cup \{\mathit{pr}\}$, and (c) for all rules $r$ in $p$ such that $\mathit{pred}(\mathit{head}(r)) = pr$, there are sequences of variables $\overline{x}$ and $\overline{y}$ and an $i$, such that $\mathit{head}(r) = \mathit{pr}(\overline{x}, \overline{y})$,  $\mathit{body}(r,i) =  \mathit{pr}'(\overline{x},\overline{y})$, and $|\overline{x}| = |\overline{y}|$. 
\end{compactenum}
A \textit{$\Sigma$-ordering template ${\cal O}$} is a set of ordering annotations. 
We say that \emph{$p$ complies with ${\cal O}$ with respect to $\preceq$} iff for all ordering annotations $a \in {\cal O}$, $p, \preceq \models a$ holds.
We say that \emph{$p$ complies with ${\cal O}$} iff there exists a well-founded partial order $\preceq$ over $\Sigma$ such that $p$ complies with ${\cal O}$ with respect to $\preceq$.
Intuitively, a program $p$ complies with an ordering template ${\cal O}$ iff all sets of predicates in ${\cal O}$ represent strict partial orders in all possible models of $p$.

\para{Deriving Disjointness Annotations}
A \textit{$\Sigma$-disjointness annotation} is $\disjoint{a}{a'}$, where $a, a' \in \Sigma$ and $|a| = |a'|$.
Let $\preceq$ be a well-founded partial order over $\Sigma^2$.
We say that the disjointness annotation $\disjoint{\mathit{pr}}{\mathit{pr}'}$ can be derived from the program $p$ given $\preceq$, written $p, \preceq \models \disjoint{\mathit{pr}}{\mathit{pr}'}$, iff one of the following conditions hold:
\begin{compactenum}
\item There are no rules $r \in p$ such that $\mathit{body}(r) \neq \emptyset$ and $\mathit{pred}(\mathit{head}(r)) = \mathit{pr}$ or $\mathit{pred}(\mathit{head}(r)) = \mathit{pr}'$, and 
there is no $(\mathit{pr}_1,\mathit{pr}_2) \in \Sigma^2$ such that $(\mathit{pr}_1,\mathit{pr}_2) \prec (\mathit{pr},\mathit{pr}')$ and $p,\preceq \models \disjoint{\mathit{pr}_1}{\mathit{pr}_2}$, and 
the sets $\{ \overline{c} \mid \mathit{pr}(\overline{c}) \in p \vee \exists v.\, \atom{v}{\mathit{pr}(\overline{c})} \in p\}$ and $\{ \overline{c} \mid \exists r \in p.\ \mathit{head}(r) = \mathit{pr}'(\overline{c}) \vee \exists v.\, \atom{v}{\mathit{pr}'(\overline{c})} \in p\}$ are disjoint.
\item There are no rules $r \in p$ such that $\mathit{body}(r) \neq \emptyset$ and $\mathit{pred}(\mathit{head}(r)) = \mathit{pr}$, and 
there is $(\mathit{pr},\mathit{pr}'_1)  \in \Sigma^2$ such that (a) $p, \preceq \models \disjoint{\mathit{pr}}{\mathit{pr}'_1}$, (b) $(\mathit{pr},\mathit{pr}'_1)  \prec (\mathit{pr}, \mathit{pr}')$, and (c) for all rules $r \in p$ such that $\mathit{head}(r) = \mathit{pr}'(\overline{x})$ there is $1 \leq i \leq |\mathit{body}(r)|$ such that $\mathit{body}(r,i) = \mathit{pr}'_1(\overline{x})$.
\item There are no rules $r \in p$ such that $\mathit{body}(r) \neq \emptyset$ and $\mathit{pred}(\mathit{head}(r)) = \mathit{pr}'$, 
there is $(\mathit{pr}_1,\mathit{pr}') \in \Sigma^2$ such that (a) $p, \preceq \models \disjoint{\mathit{pr}_1}{\mathit{pr}'}$, (b) $(\mathit{pr}_1,\mathit{pr}') \prec (\mathit{pr}, \mathit{pr}')$, and (c) for all rules $r \in p$ such that $\mathit{head}(r) = \mathit{pr}(\overline{x})$ there is $1 \leq i \leq |\mathit{body}(r)|$ such that $\mathit{body}(r,i) = \mathit{pr}_1(\overline{x})$.
\item There is  $(\mathit{pr}_1, \mathit{pr}_1') \in \Sigma^2$ such that (a) $p, \preceq \models \disjoint{\mathit{pr}_1}{\mathit{pr}_1'}$, (b) $(\mathit{pr}_1, \mathit{pr}_1')  \prec (\mathit{pr}, \mathit{pr}')$, and (c) for all rules $r,r' \in p$ such that $\mathit{head}(r) = \mathit{pr}(\overline{x})$ and $\mathit{head}(r') = \mathit{pr}'(\overline{x}')$, there are $1 \leq i \leq |\mathit{body}(r)|$, $1 \leq i' \leq |\mathit{body}(r')|$, and $(\mathit{pr}_1, \mathit{pr}_1') \in {\cal D}$ such that $\mathit{body}(r,i)  = \mathit{pr}_1(\overline{x})$, $\mathit{body}(r',i') = \mathit{pr}_1'(\overline{x}')$.
\end{compactenum}
A \textit{$\Sigma$-disjointness template ${\cal D}$} is a set of disjointness annotations. 
We say that \emph{$p$ complies with ${\cal D}$} iff there exists a well-founded partial order $\preceq$ over $\Sigma^2$ such that for all disjointness annotations $a \in {\cal D}$, $p, \preceq \models a$ holds.
Intuitively, a program $p$ complies with a disjointness template ${\cal D}$ iff all pairs of predicates $(\mathit{pr},\mathit{pr'})$ involved in a disjointness annotation in ${\cal D}$ are pairwise disjoint in all possible models of $p$.

\para{Deriving Uniqueness Annotations}
A \textit{$\Sigma$-uniqueness annotation} is $\unique{a}{K}$, where $a \in \Sigma$ and $K$ is a non-empty subset of $\{1, \ldots, |a|\}$.
Let $\preceq$ be a well-founded order over $\Sigma$.
We say that the uniqueness annotation $\unique{\mathit{pr}}{K}$ can be derived from $p$ given $\preceq$, written $p, \preceq \models \unique{\mathit{pr}}{K}$, iff one of the following conditions  hold:
\begin{compactenum}
\item There are no rules $r \in p$ such that $\mathit{body}(r) \neq \emptyset$ and $\mathit{pred}(\mathit{head}(r)) = \mathit{pr}$, and there is no $\mathit{pr}'$ such that $\mathit{pr}' \prec \mathit{pr}$, and for all $\overline{t}, \overline{v}$ in $\{ \overline{c} \mid \mathit{pr}(\overline{c}) \in p \vee \exists v.\ \atom{v}{\mathit{pr}(\overline{c})} \in p \}$, if $\overline{t}(i) = \overline{v}(i)$ for all $i \in K$, then $\overline{t} = \overline{v}$.
\item Both the following conditions hold: 
\begin{compactenum}
\item For each rule $r \in p$ such that $\mathit{head}(r) = \mathit{pr}(\overline{x})$ and $\mathit{body}(r) \neq \emptyset$, there is a mapping $\mu'$ associating to each predicate $\mathit{pr}' \prec \mathit{pr}$ a set $K$ such that (a) $p,\preceq \models \unique{\mathit{pr}'}{\mu(\mathit{pr}')}$, and (b) the following hold:
\[ V \subseteq \bigcup_{l \in \mathit{bound}(\mathit{head}(r), K, \mathit{body}^+(r), \mu')} \mathit{vars}(l),\]
where $V  =  \{\overline{x}(i) \mid i \not\in K \wedge \overline{x}(i) \in \mathit{Var}\}$, $u(\overline{y}, K) = \{ \overline{y}(i) \mid i \in K \}$, and
\begin{align*}
& \mathit{bound}(h, K, L, \mu')  = \bigcup_{
\substack{b(\overline{y}) \in L \wedge  b \prec \mathit{pr} \wedge \\ u(\overline{y}, \mu'(b)) \subseteq u(\mathit{args}(h),K) }} \{b(\overline{y})\} \cup\\
& \bigcup_{\substack{b(\overline{y}) \in L \wedge b \prec \mathit{pr} \wedge\\ \exists l' \in \mathit{bound}(h, K, L, \mu'). (u(\overline{y}, \mu'(b)) \subseteq \mathit{vars}(l') )}} \{b(\overline{y})\}.
\end{align*}
\item For all rules $r_1,r_2 \in p$, if $r_1 \neq r_2$, $\mathit{pred}(\mathit{head}(r_1)) \\ = \mathit{pred}(\mathit{head}(r_2)) = \mathit{pr}$, and  $K \neq \{1,\ldots, |\mathit{pr}|\}$, then there is a value $i \in K$ such that 
$\mathit{args}(\mathit{head}(r_1)) \in \mathbf{dom}$,
$\mathit{args}(\mathit{head}(r_2)) \in \mathbf{dom}$, and
$\mathit{args}(\mathit{head}(r_1)) \neq \mathit{args}(\mathit{head}(r_1))$.
\end{compactenum}
\end{compactenum} 
A \textit{$\Sigma$-uniqueness template ${\cal U}$} is a set of uniqueness annotations. 
We say that \emph{$p$ complies with ${\cal U}$ with respect to $\preceq$} iff for all uniqueness annotations $a \in {\cal U}$, $p, \preceq \models a$ holds.
We say that \emph{$p$ complies with ${\cal U}$} iff there exists a well-founded partial order $\preceq$ over $\Sigma$ such that $p$ complies with ${\cal U}$ with respect to $\preceq$.
Intuitively, a program complies with a uniqueness template ${\cal U}$ iff for all predicates $\mathit{pr}$ each set $K$ such that $\unique{\mathit{pr}}{K} \in {\cal U}$ represents a primary key for $\mathit{pr}$.

\para{Propagation Maps}
We use propagation maps to track how information flows inside rules.
Given a rule $r$ and a literal $l\in\mathit{body}(r)$, the \textit{$(r,l)$-vertical map} is the mapping $\mu$ from $\{1, \ldots, |l|\}$ to $\{1, \ldots, |\mathit{head}(r)|\}$ such that $\mu(i) = j$ iff $\mathit{args}(l)(i) = \mathit{args}(\mathit{head}(r))(j)$ and $\mathit{args}(l)(i) \in \mathit{Var}$. 
Given a rule $r$ and literals $l, l'$ in $r$'s body, the \textit{$(r,l,l')$-horizontal map} is the mapping $\mu$ from $\{1, \ldots, |l|\}$ to $\{1, \ldots, |l'|\}$ such that $\mu(i) = j$ iff $\mathit{args}(l)(i) = \mathit{args}(l')(j)$ and $\mathit{args}(l)(i) \in \mathit{Var}$. 
We say that a path links to a literal $l$ if information flows along the rules to $l$.
This can be formalized by posing constraints on the mapping obtained by combining horizontal and vertical maps along the path. 
Formally, given a literal $l$ and a mapping  $\nu : \mathbb{N} \to \mathbb{N}$, 
 a directed path $\mathit{pr}_1 \xrightarrow{r_1,i_1} \ldots \xrightarrow{r_{n-1}, i_{n-1}} \mathit{pr}_n$ \emph{$\nu$-downward links to $l$} iff
there is a $0 \leq j < n-1$ such that the function $\mu := \mu' \circ \mu_{j} \circ \ldots \circ \mu_1$ 
satisfies $\mu(k) = \nu(k)$ for all $k$ for which $\nu(k)$ is defined, where for $1 \leq h \leq j$, $\mu_h$ is the vertical map connecting $\mathit{body}(r_h,i_h)$ and $r_h$, 
and $\mu'$ is the horizontal map connecting $\mathit{body}(r_{j+1},i_{j+1})$ with $l$.
Similarly, a directed path $\mathit{pr}_1 \xrightarrow{r_1,i_1} \ldots \xrightarrow{r_{n-1}, i_{n-1}} \mathit{pr}_n$ \emph{$\nu$-upward links to $l$} iff
there is a $1 \leq j \leq n-1$ such that the function $\mu :=\mu'^{-1} \circ \mu_{j+1}^{-1} \circ \ldots \circ \mu_{n-1}^{-1}$ satisfies $\mu(k) = \nu(k)$ for all $k$ for which $\nu(k)$ is defined, where $\mu_h$ is the $(r_h,\mathit{body}(r_h,i_h))$-vertical map, for $j < h \leq n-1$, 
and $\mu'$ is the $(r_{j}, l)$-vertical map. 
A path \textit{$P$ links to a predicate symbol $a$} iff there is an atom $a(\overline{x})$ such that $P$ links to $a(\overline{x})$.

\para{Negation-guarded Programs}
A rule $r$ is \textit{negation-guarded}~\cite{barany2012queries} iff for all negative literals $l$ in $r$, $\mathit{vars}(l) \subseteq \bigcup_{l' \in \mathit{body}^+(r)} \mathit{vars}(l')$.
We say that a program $p$ is \emph{negation-guarded} if all rules $r\in p$ are negation-guarded.

\para{Join Trees}
A join tree represents how multiple predicate symbols in a rule share variables.
A \emph{join tree for  a rule $r$} is a rooted labelled tree $(N,E, \treeroot, \lambda)$, where
$N\subseteq \mathit{body}(r)$,
$E$ is a set of edges (i.e., unordered pairs over $N^2$),
$\treeroot \in N$ is the tree's root, and
$\lambda$ is the labelling function.
Moreover, we require that for all $n,n'\in N$, if $n \neq n'$ and $(n,n') \in E$, then $\lambda(n,n') = \mathit{vars}(n) \cap \mathit{vars}(n')$ and $\lambda(n,n') \neq \emptyset$.
A join tree $(N,E, \treeroot, \lambda)$ \emph{covers} a literal $l$ iff $l \in N$.
Given a join tree $J=(N,E, \treeroot, \lambda)$ and a node $n \in N$, the \emph{support of $n$}, denoted $\mathit{support}(n)$, is the set $\mathit{vars}(\mathit{head}(r)) \cup \{ x \mid (x = c) \in \mathit{cstr}(r) \wedge c \in \mathbf{dom} \} \cup \{ \mathit{vars}(n') \mid n' \in \mathit{anc}(J,n) \}$, where $\mathit{anc}(J,n)$ is the set of $n$'s ancestors in $J$, i.e., the set of all nodes on the path from $\treeroot$ to $n$ (if such a path exists). 

Let ${\cal U}$ be a uniqueness template.
A join tree $J=(N,E, \treeroot, \lambda)$ is \emph{${\cal U}$-strongly connected} iff
for all positive literals $l \in N$,  there is a set $K \subseteq \{ i \mid \overline{x} = \mathit{args}(l) \wedge \overline{x}(i) \in \mathit{support}(l) \}$ 
such that $\unique{\mathit{pred}(l)}{K} \in {\cal U}$ and for all negative literals $l \in N$, $\mathit{vars}(l) \subseteq \mathit{support}(l)$.
In contrast, a join tree $(N,E, \treeroot, \lambda)$ is \emph{${\cal U}$-weakly connected} iff
for all $(a(\overline{x}), a'(\overline{x}')) \in E$, there are $K \subseteq \{ i \mid \overline{x}(i) \in L \}$ and $K' \subseteq \{ i \mid \overline{x}'(i) \in L \}$  such that $\unique{a}{K}, \unique{a'}{K'} \in {\cal U}$, where $L= \lambda(a(\overline{x}), a'(\overline{x}'))$.

\para{Connected rules}
A connected rule $r$ ensure that a grounding of $r$ is fully determined either by the assignment to the head's variables (strongly connected rule) or to the variable of any literal in $r$'s body (weakly connected rule).
This is done by exploiting uniqueness annotations and the rule's structure.
In the following, let ${\cal U}$ be a uniqueness template.

We say that a rule $r$ is \emph{strongly connected for ${\cal U}$} iff there exist join trees $J_1, \ldots, J_n$ such that (a) the trees cover all literals in $\mathit{body}(r)$, and (b) for each $1 \leq i \leq n$,  $J_i$ is strongly connected for ${\cal U}$.
We call $J_1, \ldots, J_n$ a \emph{witness for the strong connectivity of $r$ with respect to ${\cal U}$}.
A strongly connected rule $r$ guarantees that for any two groundings $r',r''$ of $r$, if $\mathit{head}(r') = \mathit{head}(r'')$, then $r' = r''$.

Given a rule $r$, a set of literals $L$, and a uniqueness template ${\cal U}$, we say that a literal $l \in \mathit{body}(r)$ is \textit{$(r,{\cal U},L)$-strictly guarded} iff 
(1) $\mathit{vars}(l) \subseteq \bigcup_{l' \in  L\cap\mathit{body}^+(r)} \mathit{vars}(l') \cup \{x \mid (x = c) \in \mathit{cstr}(r) \wedge c \in \mathbf{dom}\}$, and
(2) there is a positive literal $a(\overline{x}) \in L\cap\mathit{body}^+(r)$ and an annotation $\unique{a}{K} \in {\cal U}$ such that $\{ \overline{x}(i) \mid i \in K \} \subseteq \mathit{vars}(l)$.
We say that a rule $r$ is \emph{weakly connected for ${\cal U}$} iff 
there exists a join tree $J = (N,E, \treeroot,\lambda)$ such that (a) $J$ is weakly connected for ${\cal U}$, (b) $N \subseteq \mathit{body}^+(r)$, and (c) all literals in $\mathit{body}(r) \setminus N$ are $(r,{\cal U}, N)$-strictly guarded.
We call $J$ a \emph{witness for the weak connectivity of $r$ with respect to ${\cal U}$}.
A weakly connected rule $r$ guarantees that for any two groundings $r',r''$ of $r$, if $\mathit{body}(r',i) = \mathit{body}(r'',i)$ for some $i$, then $r' = r''$.

\para{Guarded Cycles}
A set of predicates \emph{$O$ guards a directed cycle $\mathit{pr}_1 \xrightarrow{r_1, i_1} \ldots \xrightarrow{r_{n-1}, i_{n-1}} \mathit{pr}_n \xrightarrow{r_n, i_n} \mathit{pr}_1$}  iff there are integers $1 \leq y_1 < y_2 < \ldots < y_e = n$, literals $o_1(\overline{x}_1), \ldots, o_e(\overline{x}_e)$ (where $o_j \in O$ and $|\overline{x}_j| = |o_j|$),
a  non-empty set $K \subseteq \{1, \ldots, |\mathit{pr}_1|\}$, and a bijection $\nu : K \to \{1, \ldots, \sfrac{|\mathit{o}_1|}{2}\}$ such that 
for each $0 \leq k < e$, 
(1) $\mathit{pr}_{y_k} \xrightarrow{r_{y_k}, i_{y_k}} \ldots \xrightarrow{r_{{y_{k+1}}-1}, i_{{y_{k+1}}-1}} pr_{y_{k+1}}$ $\nu$-downward connects to $o_{k+1}(\overline{x}_{k+1})$, and 
(2) $\mathit{pr}_{y_{k+1}-1} \xrightarrow{r_{{y_{k+1}}-1}, i_{{y_{k+1}}-1}} pr_{y_{k+1}}$ $\nu'$-upward connects to $o_{k+1}(\overline{x}_{k+1})$, 
where 
$\nu'(i) = \nu(x) + \sfrac{|\mathit{o}_1|}{2}$ for all $1 \leq i \leq \sfrac{|\mathit{o}_1|}{2}$, and
$y_0 = 1$. 
A directed cycle $C$ is guarded by an ordering template ${\cal O}$ iff there is an annotation $\order{A} \in {\cal O}$ such that  $A$ guards $C$.

\para{Head-Guarded Paths}
A pair of head-connected non-empty directed paths $(P_1, P_2)$ is \emph{$({\cal D}, {\cal U})$-head guarded} iff one of the following conditions hold:
\begin{compactitem}
%\item $P_1$ is an empty path or $P_2$ is an empty path,

%\item if $P_1 = P_2$ and $P_1$ is non-empty, then all rules in $P_1$ are weakly connected for ${\cal U}$, or
\item if $P_1 = P_2$, then all rules in $P_1$ are weakly connected for ${\cal U}$, or

%\item if $P_1 \neq P_2$ and $P_1$ is non-empty, then there is a pair of predicate symbols $(\mathit{pr},\mathit{pr}') \in {\cal D}$ that head-guards $(P_1,P_2)$.
\item if $P_1 \neq P_2$, then %there is %a %pair of predicate symbols $(\mathit{pr},\mathit{pr}') \in {\cal D}$ that head-guards $(P_1,P_2)$.
there is an annotation $\disjoint{\mathit{pr}}{\mathit{pr}'} \in {\cal D}$, a set $K \subseteq \{1, \ldots, |a|\}$, and a bijection $\nu : K \to \{1, \ldots, |\mathit{pr}|\}$ such that $P_1$ $\nu$-downward links to $\mathit{pr}$ and $P_2$ $\nu$-downward links to $\mathit{pr}'$.
\end{compactitem}
Given two ground paths $P_1'$ and $P_2'$ corresponding to $P_1$ and $P_2$, the first condition ensures that $P_1' = P_2'$ whereas the second  ensures that $P_1'$ or $P_2'$ are not in the ground graph.

\para{Tail-Guarded Paths}
A pair of tail-connected non-empty directed paths $(P_1, P_2)$ are \emph{$({\cal D}, {\cal U})$-tail guarded} iff one of the following conditions hold:
\begin{compactitem}
%\item $P_1$ is an empty path or $P_2$ is an empty path,

%\item if $P_1 = P_2$ and $P_1$ is non-empty, then all rules in $P_1$ are strongly connected for ${\cal U}$, or
\item if $P_1 = P_2$, then all rules in $P_1$ are strongly connected for ${\cal U}$, or

%\item if $P_1 \neq P_2$ and $P_1$ is non-empty, then there is a pair of predicate symbols $(\mathit{pr},\mathit{pr}') \in {\cal D}$ that tail-guards $(P_1,P_2)$.
\item %if $P_1 \neq P_2$, then there is a pair of predicate symbols $(\mathit{pr},\mathit{pr}') \in {\cal D}$ that tail-guards $(P_1,P_2)$.
if $P_1 \neq P_2$, there is an annotation $\disjoint{\mathit{pr}}{\mathit{pr}'} \in {\cal D}$, a set $K \subseteq \{1, \ldots, |a|\}$, and a bijection $\nu : K \to \{1, \ldots, |\mathit{pr}|\}$, such that $P_1$ $\nu$-upward links to $\mathit{pr}$ and $P_2$ $\nu$-upward links to $\mathit{pr}'$.
\end{compactitem}

\para{Guarded Unsafe Structures}
An undirected unsafe structure $\langle D_1, D_2, D_3, U \rangle$ is \textit{$({\cal O}, {\cal D}, {\cal U})$-guarded} iff 
 either $(D_1,D_2)$ is $({\cal D}, {\cal U})$-head-guarded  or $(D_2,D_3)$ is $({\cal D}, {\cal U})$-tail-guarded. % by ${\cal D}$.
In contrast, a directed unsafe structure $C$ is \textit{$({\cal O}, {\cal D}, {\cal U})$-guarded} iff $C$ is guarded by ${\cal O}$.

\para{Acyclic Programs}
A $(\Sigma, \mathbf{dom})$-program $p$ is $({\cal O}, {\cal D}, {\cal U})$-\emph{acyclic}, where ${\cal O}$ is a $\Sigma$-ordering template, ${\cal D}$ is a $\Sigma$-disjointness template, and ${\cal U}$ is  a $\Sigma$-uniqueness template, iff
\begin{compactenum}
\item $p$ complies with ${\cal O}$, ${\cal D}$, and ${\cal U}$,
\item $p$ is a negation-guarded program, and
\item for all undirected cycles $U$ in $\mathit{graph}(p)$ that are not simple directed cycles, 
there is an undirected unsafe structure $S$ in $\mathit{graph}(p)$ that covers $U$, and is   $({\cal O}, {\cal D}, {\cal U})$-guarded
\item for all directed cycles $C$ in $\graph(p)$, there is a directed unsafe structure $S$ in $\mathit{graph}(p)$ that covers $C$, and is   $({\cal O}, {\cal D}, {\cal U})$-guarded
\end{compactenum}
A $(\Sigma, \mathbf{dom})$-program $p$ is \emph{acyclic} iff there are a $\Sigma$-ordering template ${\cal O}$, a $\Sigma$-disjointness template ${\cal D}$, and a $\Sigma$-uniqueness template ${\cal U}$ such that $p$ is $({\cal O}, {\cal D}, {\cal U})$-acyclic.

\subsubsection{Exact Grounding}\label{app:atklog:lite:extended:exact:grounding}
We now introduce the exact grounding of a \problog{} program $p$ given a $p$-probabilistic assignment.
The exact grounding encodes the \problog{} semantics.

Let $p$ be a \problog{} program and $f$ be a $p$-probabilistic assignment.
Furthermore, let $\mu$ be the mapping from predicate symbols in $p$ to $\mathbb{N}$:
$\mu(a) = \mathit{max}_{t \in \mathit{paths}(p,a)}(\mathit{w}(t))$, where $\mathit{paths}(p,a)$ is the set of all directed paths that ends in $a$ in $p$'s dependency graph and the weight of each path is
\[
\mathit{w}(t) = 
\begin{cases}
0 & \text{if}\ t = \epsilon\\
\mathit{w}(t') & \text{if}\ t = a_1 \xrightarrow{r_1,i_1} a_2 \xrightarrow{r_2,i_2} \ldots \xrightarrow{r_n,i_n} a_{n+1}  \\
	& \quad \wedge t' = a_2 \xrightarrow{r_2,i_2} \ldots \xrightarrow{r_n,i_n} a_{n+1}  \\
	& \quad \wedge \mathit{body}(r_1,i_1) \in {\cal A}^+_{\Sigma,\mathbf{dom}}\\
1+\mathit{w}(t') & \text{if}\ t = a_1 \xrightarrow{r_1,i_1} a_2 \xrightarrow{r_2,i_2} \ldots \xrightarrow{r_n,i_n} a_{n+1}  \\
	& \quad \wedge t' = a_2 \xrightarrow{r_2,i_2} \ldots \xrightarrow{r_n,i_n} a_{n+1}  \\
	& \quad \wedge \mathit{body}(r_1,i_1) \not\in {\cal A}^+_{\Sigma,\mathbf{dom}}\\
\end{cases}
\]
We remark that $\mu$ is a valid stratification for the program $p$.
From more details about logic programming with stratified negation we refer the reader to~\cite{abiteboul1995foundations}.

The functions $g_{f}(p,j,i)$, $g_{f}(p,j)$, and $g_{f}(p)$ are defined as follows:
\begin{align*}
g_{f}(p,0,0) =& \{ a(\overline{c}) \mid a(\overline{c}) \in p \wedge \mu(a) = 0 \} \cup \\
& \{ a(\overline{c}) \mid \exists v.\ \atom{v}{a(\overline{c})} \in p \wedge f(\atom{v}{a(\overline{c})}) = \top \\
& \qquad \wedge \mu(a) = 0 \}
\displaybreak[0] \\
g_{f}(p,0,i) =& g_{f}(p,0,i-1) \cup  \\
& \{ \mathit{head}(r)\Theta \mid \Theta \in ASS(r) \wedge  r \in p  \wedge \\
& \qquad \forall l \in \mathit{body}^+(r).\ l\Theta \in \mathit{g}_{f}(p, 0, i-1) \wedge \\
& \qquad \mathit{body}^-(r) = \emptyset \wedge \mu(\mathit{pred}(\mathit{head}(r))) = 0 \}
\displaybreak[0] \\
g_{f}(p,j,0) =& g_{f}(p,j-1) \cup \\
&\{ a(\overline{c}) \mid a(\overline{c}) \in p \wedge \mu(a) = j \} \cup \\
&\{ a(\overline{c}) \mid \exists v.\ \atom{v}{a(\overline{c})} \in p \wedge f(\atom{v}{a(\overline{c})}) = \top \\
& \qquad \wedge \mu(a) = j \}
\displaybreak[0] \\
g_{f}(p,j,i) =& g_{f}(p,j,i-1) \cup \\
& \{ \mathit{head}(r)\Theta \mid \Theta \in ASS(r) \wedge  r \in p \wedge  \\
& \qquad \forall l \in \mathit{body}^+(r).\ l\Theta \in \mathit{g}_{f}(p, j, i-1) \wedge \\
& \qquad \forall l \in \mathit{body}^-(r).\ l\Theta \not\in \mathit{g}_{f}(p, j-1) \wedge \\
& \qquad \wedge \mu(\mathit{pred}(\mathit{head}(r))) = j \}
\displaybreak[0] \\
g_{f}(p,j) =& \bigcup_{i \in \mathbb{N}} g_{f}(p,j,i)
\displaybreak[0] \\
g_{f}(p) =& \bigcup_{j \in \mathbb{N}} g_{f}(p,j)
\end{align*}
Furthermore, the functions $g_{f}(p,r,j,i)$, $g_{f}(p,r,j)$, and $g_{f}(p,r)$ are defined as follows:
\begin{align*}
g_{f}(p,a(\overline{c}),j,i) =& \{ a(\overline{c}) \mid a(\overline{c}) \in p \wedge \mu(a) = j \} 
\displaybreak[0] \\
g_{f}(p,\atom{v}{a(\overline{c})},j,i) =& \{ a(\overline{c}) \mid \atom{v}{a(\overline{c})} \in p \wedge f(\atom{v}{a(\overline{c})}) = \top \wedge  \\
& \qquad \mu(a) = j \}
\displaybreak[0] \\
g_{f}(p,r,0,0) =& \emptyset
\displaybreak[0]\\
g_{f}(p,r,0,i) =& g_{f}(p,r,0,i-1) \cup \\
& \ \{ h\Theta \leftarrow l_1\Theta, \ldots, l_n\Theta \mid \\
& \ r = h \leftarrow l_1, \ldots, l_n \wedge \\
& \ \Theta \in ASS(r) \wedge \\
& \ \forall l \in \mathit{body}^+(r).\, l\Theta \in g_{f}(p,0,i-1) \wedge \\
& \ \mathit{body}^-(r) = \emptyset \wedge \mu(\mathit{pred}(h)) = 0 \}
\displaybreak[0]\\
g_{f}(p,r,j,0) =& g_{f}(p,r,j-1)
\displaybreak[0]\\
g_{f}(p,r,j,i) =& g_{f}(p,r,j,i-1) \cup \\
& \ \{ h\Theta \leftarrow l_1\Theta, \ldots, l_n\Theta \mid \\
& \ r = h \leftarrow l_1, \ldots, l_n \wedge \\
& \ \Theta \in ASS(r) \wedge \\
& \ \forall l \in \mathit{body}^+(r).\, l\Theta \in g_{f}(p,j,i-1) \wedge \\
& \ \forall l \in \mathit{body}^-(r).\, l\Theta \not\in g_{f}(p,j-1) \wedge \\
& \ \mu(\mathit{pred}(h)) = j \}
\displaybreak[0]\\
g_{f}(p,r,j) &= \bigcup_{i \in \mathbb{N}} g_{f}(p,r,j,i)
\displaybreak[0]\\
g_{f}(p,r) &= \bigcup_{j \in \mathbb{N}} g_{f}(p,r,j)
\end{align*}

The fact below follows directly from $g_f$'s definition and the notion of well-founded models for stratified logic programs~\cite{abiteboul1995foundations}. 

\begin{fact}
For any \problog{} program $p$ and any $p$-probabilistic assignment $f$, $g_f(p) = \mathit{WMF}(\mathit{instance}(p,f))$.
\end{fact}

\subsubsection{Relaxed Grounding}
Here we introduce the notion of relaxed grounding. This is a key notion for our acyclicity and correctness proofs and it is also a key component of our compilation from \problog{} programs to poly-tree Bayesian Networks.

Given a \problog{} acyclic program $p$, the \emph{relaxed grounding of $p$} is a set $M$ of ground atoms and rules containing all possible ground atoms and rules that can be derived in any ground instance of $p$.
We remark that a relaxed grounding is an over-approximation of the ground atoms and/or rules that could be derived by the ground instances of $p$. 
Formally, let $\Sigma$ be a first-order signature, $\mathbf{dom}$ be a finite domain, $p$ be a $(\Sigma,\mathbf{dom})$-acyclic \problog{}, and $i \in \mathbb{N}$ be an natural number.
The function $\mathit{ground}(p,i)$ is defined as follows:
$\mathit{ground}(p,0) =  \{ a(\overline{c}) \mid a(\overline{c}) \in p \vee \exists v.\, \atom{v}{a(\overline{c})} \in p\}$, and 
$\mathit{ground}(p,i) = \mathit{ground}(p,i-1) \cup \{ h\Theta \mid \Theta \in \mathit{ASS}(r) \wedge \exists r.\, r \in p \wedge h = \mathit{head}(r)  \wedge \forall l \in \mathit{body}^+(r).\ l\Theta \in \mathit{ground}(p, i-1) \wedge \forall i,j.\, \mathit{consistent}(\mathit{body}(r,i)\Theta, \mathit{body}(r,j)\Theta) \}$, where $\mathit{ASS}(r)$ is the set of all mappings $\mathbf{Var} \to \mathbf{dom}$ that satisfy the equality and inequality constraints in $r$ while $\mathit{consistent}(a(\overline{v}), \neg a(\overline{v}))  = \mathit{consistent}(\neg a(\overline{v}), a(\overline{v})) = \bot$ and $\mathit{consistent}(l,l') = \top$ otherwise.
The function $\mathit{ground}(p,r,i)$, where $r \in p$, is defined as follows:
$\mathit{ground}(p, a(\overline{c}), i) = \{ a(\overline{c})\}$, 
$\mathit{ground}(p, \atom{v}{a(\overline{c})}, i) = \{ a(\overline{c})\}$, 
$\mathit{ground}(p, r,0) = \emptyset$, and 
$\mathit{ground}(p,r,i)$ is the set $\{ h\Theta \leftarrow l_1\Theta, \ldots, l_n\Theta \mid \Theta \in \mathit{ASS}(r) \wedge \forall l_i \in \mathit{body}^+(r).\ l_i\Theta \in  \mathit{ground}(p,i-1) \wedge \forall i,j.\, \mathit{consistent}(l_i\Theta, l_j\Theta) \}$, where $r = h \leftarrow l_1, \ldots, l_n$.
Finally,  the relaxed grounding of $p$, denoted by $\mathit{ground}(p)$, is the set $\mathit{ground}(p) = \bigcup_{i \in \mathbb{N}} \mathit{ground}(p,i)$, whereas $\mathit{ground}(p,r) = \bigcup_{i \in \mathbb{N}} \mathit{ground}(p,r,i)$.

The fact below follows directly from $\mathit{ground}$'s definition and the definition of acyclic \problog{} programs.

\begin{fact}
For any acyclic \problog{} program $p$ and any $p$-probabilistic assignment $f$, $g_f(p) \subseteq \mathit{ground}(p)$ and $g_f(p,r) \subseteq \mathit{ground}(p,r)$.
\end{fact}

\subsubsection{Ground Graphs}

Let $\Sigma$ be a first-order signature, $\mathbf{dom}$ be a finite domain, and $p$ be a $(\Sigma,\mathbf{dom})$-acyclic \problog{} program.
The ground graph of $p$, denoted $\mathit{gg}(p)$, is the directed graph $(N,E)$ defined as follows:
\begin{compactitem}
\item There is one node $a(\overline{c})$ in $N$ per literal $a(\overline{c})$ or $\neg a(\overline{c})$ used in a ground rule $r' \in \mathit{ground}(p,r)$ for some $r \in p$.
For each rule $r \in p$, ground rule $s \in \mathit{ground}(p,r)$, and position $1 \leq j \leq |\mathit{body}(r')|$, there is a node $(r,s,j)$.

\item There is an edge from $b(\overline{v})$ to  $(r,s,j)$ if $\mathit{body}(s,j) = b(\overline{v})$ or $\mathit{body}(s,j) = \neg b(\overline{v})$.
There is an edge from  $(r,s,j)$ to $a(\overline{c})$ if $\mathit{head}(s) = a(\overline{c})$.
\end{compactitem}

\subsubsection{Proofs about Templates}
Here we prove preliminary results about the semantic relationships enforced by the ordering, disjointness, and uniqueness templates.

\begin{proposition}\thlabel{theorem:template:ordering:semantics}
Let $\Sigma$ be a first-order signature, $\mathbf{dom}$ be a finite domain, $p$ be a $(\Sigma,\mathbf{dom})$-\problog{} program, ${\cal O}$ be a $\Sigma$-ordering template, and $\preceq$ be a well-founded partial order over $2^\Sigma$.
If $p$ complies with ${\cal O}$ with respect to $\preceq$ and $\order{A} \in {\cal O}$, then if $\mathit{pr}_1(\overline{a}_1, \overline{a}_2), \mathit{pr}_2(\overline{a}_2, \overline{a}_3),  \ldots, \mathit{pr}_{n-1}(\overline{a}_{n-1},\overline{a}_n)$ are in $\mathit{ground}(p)$ and $\{\mathit{pr}_1, \ldots, \mathit{pr}_{n-1}\} \subseteq A$, then $\overline{a}_1 \neq \overline{a}_n$.
\end{proposition}

\begin{proof}
Let $\Sigma$ be a first-order signature, $\mathbf{dom}$ be a finite domain, $p$ be a $(\Sigma,\mathbf{dom})$-\problog{} program, ${\cal O}$ be a $\Sigma$-ordering template, and $\preceq$ be a well-founded partial order over $2^\Sigma$.
We assume that $p$ complies with ${\cal O}$ with respect to $\preceq$, $\order{A} \in {\cal O}$, $\mathit{pr}_1(\overline{a}_1, \overline{a}_2), \ldots, \mathit{pr}_{n-1}(\overline{a}_{n-1},\overline{a}_n)$ are in $\mathit{ground}(p)$, and $\{\mathit{pr}_1, \ldots, \mathit{pr}_{n-1}\} \subseteq A$.
We now prove, by induction over $\preceq$, that the transitive closure of the unions of the relations induced by $\mathit{pr}_1, \ldots, \mathit{pr}_{n-1}$ over $\mathit{ground}(p)$ is a strict partial order over $\mathbf{dom}^{\sfrac{|\mathit{pr}|}{2}}$.
From this, it follows that  if $\mathit{pr}_1(\overline{a}_1, \overline{a}_2), \mathit{pr}_2(\overline{a}_2, \overline{a}_3), \ldots, \mathit{pr}_{n-1}(\overline{a}_{n-1},\overline{a}_n)$ are in $\mathit{ground}(p)$, then $\overline{a}_1 \neq \overline{a}_n$.

\para{Base Case}
Let $A$ be a set such that there is no $A' \preceq A$ such that $p, \preceq \models \order{A'}$.
From this, $\order{A} \in {\cal O}$, and $p$ complies with ${\cal O}$ with respect to $\preceq$, it follows that 
(1) for all $\mathit{pr} \in A$, there is no rule $r$ in $p$ such that $\mathit{pred}(\mathit{head}(r)) = \mathit{pr}$ and $\mathit{body}(r) \neq \emptyset$, and
(2) the transitive closure of the relation $R = \bigcup_{\mathit{pr} \in A} \{ (\overline{c}, \overline{v}) \mid \exists r \in p.\, ((\mathit{head}(r) =pr(\overline{c},\overline{v}) \vee \exists v'.\ \mathit{head}(r) = \atom{v'}{pr(\overline{c},\overline{v})}) \wedge \mathit{body}(r) = \emptyset) \wedge |\overline{c}| = |\overline{v}| = \sfrac{|\mathit{pr}|}{2}\}$ is strict partial order. 
From this and $\{ (\overline{c}, \overline{v}) \mid \exists \mathit{pr} \in A.\ \mathit{pr}(\overline{c}, \overline{v}) \in \mathit{ground}(p) \wedge |\overline{c}| = |\overline{v}| = \sfrac{|\mathit{pr}|}{2} \} \subseteq R$, it follows that the transitive closure of the union of the relations induced by $\mathit{pr}$ over $\mathit{ground}(p)$ is a strict partial order over $\mathbf{dom}^{\sfrac{|\mathit{pr}|}{2}}$.

\para{Induction Step}
Assume that the relation induced by  any $A' \prec A$   over $\mathit{ground}(p)$ is a strict partial order over $\mathbf{dom}^{\sfrac{|\mathit{pr'}|}{2}}$ if $p,\preceq \models \order{A'}$, where $\mathit{pr}'$ is any predicate symbol in $A'$. 
Without loss of generality, we assume that there is at least one $A' \prec A$ such that $p,\preceq \models \order{A'}$ (otherwise the proof is the same as in the base case).
We now prove that also the relation induced by $A$  over $\mathit{ground}(p)$ is a strict partial order over $\mathbf{dom}^{\sfrac{|\mathit{pr}|}{2}}$.
Let $\mathit{pr}$, $A'$, $\mathit{pr}_1$ be such that $A = (A' \setminus\{\mathit{pr}_1\}) \cup \{\mathit{pr}\}$, $A' \prec A$, $p,\preceq \models \order{A'}$, and $\mathit{pr}_1 \in A'$. 
Note that the previous values always exist.
We now prove that also the relation induced by $A$  over $\mathit{ground}(p)$ is a strict partial order over $\mathbf{dom}^{\sfrac{|\mathit{pr}|}{2}}$.
Since $\mathit{pr} \in A$, $A = (A' \setminus\{\mathit{pr}_1\}) \cup \{\mathit{pr}\}$, $\mathit{pr}_1 \in A'$, and $\order{A} \in {\cal O}$, it follows that  for all rules $r$ in $p$ such that $\mathit{pred}(\mathit{head}(r)) = pr$, there are sequences of variables $\overline{x}$ and $\overline{y}$ and an $i$, $\mathit{head}(r) = \mathit{pr}(\overline{x}, \overline{y})$,  $\mathit{body}(r,i) =  \mathit{pr}_1(\overline{x},\overline{y})$, and $|\overline{x}| = |\overline{y}|$.
From this, it follows that the $\{ (\overline{v}, \overline{w}) \mid \mathit{pr}(\overline{v},\overline{w}) \in \mathit{ground}(p) \} \subseteq \{ (\overline{v}, \overline{w}) \mid \mathit{pr}_1(\overline{v},\overline{w}) \in \mathit{ground}(p) \}$.
From the induction hypothesis, it follows that the transitive closure of $\bigcup_{\mathit{pr}' \in A'} \{ (\overline{v}, \overline{w}) \mid \mathit{pr}'(\overline{v},\overline{w}) \in \mathit{ground}(p) \}$ is a strict partial order over $\mathbf{dom}^{\sfrac{|\mathit{pr}'|}{2}}$.
From this, $\mathit{pr}_1 \in A'$, $A = (A' \setminus\{\mathit{pr}_1\}) \cup \{\mathit{pr}\}$, $\mathit{pr}_1 \in A'$, and $\{ (\overline{v}, \overline{w}) \mid \mathit{pr}(\overline{v},\overline{w}) \in \mathit{ground}(p) \} \subseteq \{ (\overline{v}, \overline{w}) \mid \mathit{pr}_1(\overline{v},\overline{w}) \in \mathit{ground}(p) \}$, it follows that the transitive closure of $\bigcup_{\mathit{pr}' \in A} \{ (\overline{v}, \overline{w}) \mid \mathit{pr}'(\overline{v},\overline{w}) \in \mathit{ground}(p) \}$ is a strict partial order over $\mathbf{dom}^{\sfrac{|\mathit{pr}'|}{2}}$.
This completes the proof of our claim.
\end{proof}

\begin{proposition}\thlabel{theorem:template:disjointness:semantics}
Let $\Sigma$ be a first-order signature, $\mathbf{dom}$ be a finite domain, $p$ be a $(\Sigma,\mathbf{dom})$-\problog{} program, ${\cal D}$ be a $\Sigma$-disjointness template, and $\preceq$ be a well-founded partial order over $\Sigma^2$.
If $p$ complies with ${\cal D}$ with respect to $\preceq$ and $\disjoint{\mathit{pr}}{\mathit{pr}'} \in {\cal D}$, then there is no tuple $\overline{v}$ such that $\mathit{pr}(\overline{v}) \in \mathit{ground}(p)$ and $\mathit{pr}'(\overline{v}) \in \mathit{ground}(p)$.
\end{proposition}

\begin{proof}
Let $\Sigma$ be a first-order signature, $\mathbf{dom}$ be a finite domain, $p$ be a $(\Sigma,\mathbf{dom})$-\problog{} program, ${\cal D}$ be a $\Sigma$-disjointness template, and $\preceq$ be a well-founded partial order over $\Sigma^2$.
Furthermore, we assume that $p$ complies with ${\cal D}$ with respect to $\preceq$ and $\disjoint{\mathit{pr}}{\mathit{pr}'} \in {\cal D}$.
We prove, by induction over $\preceq$, that for any $\disjoint{\mathit{pr}}{\mathit{pr}'} \in {\cal D}$, there is no tuple $\overline{v}$ such that $\mathit{pr}(\overline{v}) \in \mathit{ground}(p)$ and $\mathit{pr}'(\overline{v}) \in \mathit{ground}(p)$.

\para{Base Case}
Let $(\mathit{pr}, \mathit{pr}')$ be a pair in $\Sigma^2$ such that there is no $(\mathit{pr}_1, \mathit{pr}'_1) \prec (\mathit{pr}, \mathit{pr}')$ and $p, \preceq \models \disjoint{\mathit{pr}}{\mathit{pr}'}$.
From this,  $\disjoint{\mathit{pr}}{\mathit{pr}'} \in {\cal D}$, and $p$ complies with ${\cal D}$, it follows that (1) there are no rules $r \in p$ such that $\mathit{body}(r) \neq \emptyset$ and $\mathit{pred}(\mathit{head}(r)) = \mathit{pr}$ or $\mathit{pred}(\mathit{head}(r)) = \mathit{pr}'$, and (2) the sets $A = \{ \overline{c} \mid \mathit{pr}(\overline{c}) \in p \vee \exists v'. \atom{v'}{\mathit{pr}(\overline{c})} \in p\}$ and $A' = \{ \overline{c} \mid \exists r \in p.\ \mathit{head}(r) = \mathit{pr}'(\overline{c})  \vee \exists v'. \atom{v'}{\mathit{pr}'(\overline{c})} \in p\}$ are disjoint.
From this, it follows that there is no tuple $\overline{v}$ such that $\mathit{pr}(\overline{v}) \in \mathit{ground}(p)$ and $\mathit{pr}'(\overline{v}) \in \mathit{ground}(p)$.

\para{Induction Step}
Assume that there is no tuple $\overline{v}$ such that $\mathit{pr}_1(\overline{v}) \in \mathit{ground}(p)$ and $\mathit{pr}'_1(\overline{v}) \in \mathit{ground}(p)$ for any $(\mathit{pr}_1, \mathit{pr}'_1) \prec (\mathit{pr}, \mathit{pr}')$ such that $p, \preceq \models \disjoint{\mathit{pr}_1}{\mathit{pr}'_1}$.
We now prove that there is no tuple $\overline{v}$ such that $\mathit{pr}(\overline{v}) \in \mathit{ground}(p)$ and $\mathit{pr}'(\overline{v}) \in \mathit{ground}(p)$.
There are four cases:
\begin{compactitem}
\item There are no rules $r \in p$ such that $\mathit{body}(r) \neq \emptyset$ and $\mathit{pred}(\mathit{head}(r)) = \mathit{pr}$, and 
there is a $(\mathit{pr},\mathit{pr}'_1) \in \Sigma^2$ such that (a) $\disjoint{\mathit{pr}}{\mathit{pr}'_1} \in {\cal D}$, (b) $(\mathit{pr},\mathit{pr}'_1)  \prec (\mathit{pr}, \mathit{pr}')$, and (c) for all rules $r \in p$ such that $\mathit{head}(r) = \mathit{pr}'(\overline{x})$ there is $1 \leq i \leq |\mathit{body}(r)|$ such that $\mathit{body}(r,i) = \mathit{pr}'_1(\overline{x})$.
From $(\mathit{pr},\mathit{pr}'_1)  \prec (\mathit{pr}, \mathit{pr}')$, $\disjoint{\mathit{pr}}{\mathit{pr}'} \in {\cal D}$, and the induction's hypothesis, it follows that there is no tuple $\overline{v}$ such that $\mathit{pr}(\overline{v}) \in \mathit{ground}(p)$ and $\mathit{pr}'_1(\overline{v}) \in \mathit{ground}(p)$.
Furthermore, the relation associated with $\mathit{pr}'$ is a subset of the relation associated to $\mathit{pr}'_1$.
Therefore, there is no tuple $\overline{v}$ such that $\mathit{pr}(\overline{v}) \in \mathit{ground}(p)$ and $\mathit{pr}'(\overline{v}) \in \mathit{ground}(p)$.

\item There are no rules $r \in p$ such that $\mathit{body}(r) \neq \emptyset$ and $\mathit{pred}(\mathit{head}(r)) = \mathit{pr}'$, and
there is a $(\mathit{pr}_1,\mathit{pr}') \in \Sigma^2$ such that (a) $\disjoint{\mathit{pr}_1}{\mathit{pr}'} \in {\cal D}$, (b) $(\mathit{pr}_1,\mathit{pr}') \prec (\mathit{pr}, \mathit{pr}')$, and (c) for all rules $r \in p$ such that $\mathit{head}(r) = \mathit{pr}(\overline{x})$ there is $1 \leq i \leq |\mathit{body}(r)|$ such that $\mathit{body}(r,i) = \mathit{pr}_1(\overline{x})$. %, $(\mathit{pr}_1,\mathit{pr}') \in {\cal D}$, and $(\mathit{pr}_1,\mathit{pr}') \prec (\mathit{pr}, \mathit{pr}')$.
From $(\mathit{pr}_1,\mathit{pr}')  \prec (\mathit{pr}, \mathit{pr}')$, $\disjoint{\mathit{pr}_1}{\mathit{pr}'} \in {\cal D}$, and the induction's hypothesis, it follows that there is no tuple $\overline{v}$ such that $\mathit{pr}_1(\overline{v}) \in \mathit{ground}(p)$ and $\mathit{pr}'(\overline{v}) \in \mathit{ground}(p)$.
Furthermore, the relation associated with $\mathit{pr}$ is a subset of the relation associated to $\mathit{pr}_1$.
Therefore, there is no tuple $\overline{v}$ such that $\mathit{pr}(\overline{v}) \in \mathit{ground}(p)$ and $\mathit{pr}'(\overline{v}) \in \mathit{ground}(p)$.

\item There is a $(\mathit{pr}_1, \mathit{pr}_1') \in \Sigma^2$ such that (a) $\disjoint{\mathit{pr}_1}{\mathit{pr}_1'} \in {\cal D}$, (b) $(\mathit{pr}_1, \mathit{pr}_1')  \prec (\mathit{pr},  \mathit{pr}')$, and (c) for all rules $r,r' \in p$ such that $\mathit{head}(r) = \mathit{pr}(\overline{x})$ and $\mathit{head}(r') = \mathit{pr}'(\overline{x}')$, there are $1 \leq i \leq |\mathit{body}(r)|$, $1 \leq i' \leq |\mathit{body}(r')|$ such that $\mathit{body}(r,i)  = \mathit{pr}_1(\overline{x})$, $\mathit{body}(r',i') = \mathit{pr}_1'(\overline{x}')$.
From $(\mathit{pr}_1,\mathit{pr}'_1)  \prec (\mathit{pr}, \mathit{pr}')$, $\disjoint{\mathit{pr}_1}{\mathit{pr}_1'} \in {\cal D}$, and the induction's hypothesis, it follows that there is no tuple $\overline{v}$ such that $\mathit{pr}_1(\overline{v}) \in \mathit{ground}(p)$ and $\mathit{pr}'_1(\overline{v}) \in \mathit{ground}(p)$.
Furthermore, the relation associated with $\mathit{pr}$ is a subset of the relation associated with $\mathit{pr}_1$ and the relation associated with $\mathit{pr}'$ is a subset of the relation associated with $\mathit{pr}_1'$.
Therefore, there is no tuple $\overline{v}$ such that $\mathit{pr}(\overline{v}) \in \mathit{ground}(p)$ and $\mathit{pr}'(\overline{v}) \in \mathit{ground}(p)$.
\end{compactitem}
This completes the proof of our claim.
\end{proof}

\begin{proposition}\thlabel{theorem:template:uniqueness:semantics}
Let $\Sigma$ be a first-order signature, $\mathbf{dom}$ be a finite domain, $p$ be a $(\Sigma,\mathbf{dom})$-\problog{} program,  ${\cal U}$ be a $\Sigma$-uniqueness template, and $\preceq$ be a well-founded order over $\Sigma$.
If $p$ complies with ${\cal U}$ with respect to $\preceq$ and $\unique{\mathit{pr}}{K} \in {\cal U}$, then 
for all tuples $\mathit{pr}(\overline{v}), \mathit{pr}(\overline{w}) \in \mathit{ground}(p)$, if $\overline{v}(i) = \overline{w}(i)$ for all $i \in K$, then $\overline{v} = \overline{w}$.
\end{proposition}

\begin{proof}
Let $\Sigma$ be a first-order signature, $\mathbf{dom}$ be a finite domain, $p$ be a $(\Sigma,\mathbf{dom})$-\problog{} program, ${\cal U}$ be a $\Sigma$-uniqueness template, and $\preceq$ be a well-founded partial order over ${\cal U}$.
Furthermore, we assume that $p$ complies with ${\cal U}$ with respect to $\preceq$ and $\unique{\mathit{pr}}{K} \in {\cal U}$.
We prove, by induction over $\preceq$, that for all $\unique{\mathit{pr}}{K} \in {\cal U}$, and all tuples $\mathit{pr}(\overline{v}), \mathit{pr}(\overline{w}) \in \mathit{ground}(p)$, if $\overline{v}(i) = \overline{w}(i)$ for all $i \in K$, then $\overline{v} = \overline{w}$.
Without loss of generality, in the following we assume that $K \neq \{1, \ldots, |\mathit{pr}|\}$.
If this is not the case, then our claim holds trivially. 

\para{Base Case}
Let $\mathit{pr}$ be a predicate such that there is no $\mathit{pr}' \prec \mathit{pr}$.
From this and $\unique{\mathit{pr}}{K} \in {\cal U}$, it follows that there are no rules $r$ in $p$ such that $\mathit{body}(r) \neq \emptyset$ and $\mathit{pred}(\mathit{head}(r)) = \mathit{pr}$, and for all $\overline{w}, \overline{v}$ in $A = \{ \overline{c} \mid \mathit{pr}(\overline{c}) \in p \vee \exists v. \atom{v}{\mathit{pr}(\overline{c})}\in p \}$, if $\overline{v}(i) = \overline{w}(i)$ for all $i \in K$, then $\overline{v} = \overline{w}$.
From this and $\{\overline{c} \mid \mathit{pr}(\overline{c}) \in \mathit{ground}(p)\} \subseteq A$, it follows that for all tuples $\mathit{pr}(\overline{v}), \mathit{pr}(\overline{w}) \in \mathit{ground}(p)$, if $\overline{v}(i) = \overline{w}(i)$ for all $i \in K$, then $\overline{v} = \overline{w}$.

\para{Induction Step}
Assume that the claim holds for any $\mathit{pr}' \prec \mathit{pr}$ and $K \subseteq \{1, \ldots, |\mathit{pr}|\}$ such that $\unique{\mathit{pr}}{K} \in {\cal U}$ (we denote this induction hypothesis as $(\clubsuit)$).
We now prove that the claim holds for $\mathit{pr}$ as well.
Without loss of generality, we assume that there is at least one rule $r$ such that $\mathit{body}(r) \neq \emptyset$ and $\mathit{pred}(\mathit{head}(r)) = \mathit{pr}$ (otherwise the proof is identical to the base case).
Assume, for contradiction's sake, that there are two ground atoms $\mathit{pr}(\overline{v}), \mathit{pr}(\overline{w}) \in \mathit{ground}(p)$ such that $\overline{v}(i) = \overline{w}(i)$ for all $i \in K$ but $\overline{v} \neq \overline{w}$.
There are two cases:
\begin{compactitem}
\item $\mathit{pr}(\overline{v})$ and  $\mathit{pr}(\overline{w})$ are generated by two ground instances of the same rule $r$.
From $\overline{v} \neq \overline{w}$, it follows that there is $j \not\in K$ such that $\overline{v}(j) \neq \overline{w}(j)$.
From this and the fact that both atoms are generated by the same rule $r$, it follows that there is a variable $x$ such that $x \in \{ \overline{x}(i) \mid i \not\in K\}$ and $x = \overline{x}(j)$, where $\overline{x} = \mathit{args}(\mathit{head}(r))$.
From this and $\mathit{pr} \in {\cal U}$, it follows that $x \in \bigcup_{l \in \mathit{bound}(\mathit{head}(r),K,\mathit{body}^+(r), \mu')} \mathit{vars}(l)$.
We claim that the value of $x$ is determined by the values of the variables whose positions are in $K$.
From this and the fact that $\overline{v}$ and $\overline{w}$ agree on all values whose positions are in $K$, it follows that $\overline{v} = \overline{w}$, leading to a contradiction.

We now prove our claim that the value of $x$ is determined by the values of the variables whose positions are in $K$.
The value of the variable $x$ is determined by the grounding of one of the atoms in $\mathit{bound}(\mathit{head}(r),K,\mathit{body}^+(r), \mu')$.
We first slightly modify the definition of $\mathit{bound}$.
We denote by $\mathit{bd}^0(h, K, L, \mu')$ the function $\bigcup_{
\substack{b(\overline{y}) \in L \wedge  b \prec \mathit{pr} \wedge \\ u(\overline{y}, \mu'(b)) \subseteq u(\mathit{args}(h),K) }} \{b(\overline{y})\}$ and by $\mathit{bd}^i(h, K, L, \mu')$, where $i > 0$, the function $\mathit{bd}^{i-1}(h, K, L, \mu') \cup \bigcup_{\substack{b(\overline{y}) \in L \wedge b \prec \mathit{pr} \wedge\\ \exists l' \in \mathit{bound}^{i-1}(h, K, L, \mu'). (u(\overline{y}, \mu'(b)) \subseteq \mathit{vars}(l') )}} \{b(\overline{y})\}$.
It is easy to see that (1) $\mathit{bound}(\mathit{head}(r),K,\mathit{body}^+(r), \mu') = \bigcup_{i \in \mathbb{N}} \mathit{bd}^i(h, K, L, \mu')$, and (2) the fix-point is always reached in a finite amount of steps (bounded by $|\mathit{body}^+(r)|$).
We now prove by induction on $i$ that the groundings of the literals in $\mathit{bd}^i(h, K, L, \mu')$ is always determined by the values of the variables in $u(\mathit{args}(\mathit{head}(r)), K)$.
From this, our claim trivially follows.
For the base case, let $i = 0$.
Then, for any literal $b(\overline{y}) \in \mathit{bd}^0(h, K, L, \mu')$, it follows that (a) $b \prec \mathit{pr}$, and (b) $u(\overline{y}, \mu'(b)) \subseteq u(\mathit{args}(\mathit{head}(r)),K)$.
From this, $p$ complies with ${\cal U}$, and the induction's hypothesis $(\clubsuit)$, it follows that the variables in $\unique{b}{\mu'(b)}$ uniquely determine the grounding of $b(\overline{y})$.
From this and $u(\overline{y}, \mu'(b)) \subseteq u(\mathit{args}(\mathit{head}(r)),K)$, it follows that the variables in $u(\mathit{args}(\mathit{head}(r)),K)$ uniquely determine the grounding of $b(\overline{y})$.
For the induction's step, assume that the claim hold for all $j < i$ (we denote this induction hypothesis as $(\spadesuit)$).
Let $b(\overline{y}) \in \mathit{bd}^i(h, K, L, \mu')$ (without loss of generality, $b(\overline{y}) \not\in \mathit{bd}^{i-1}(h, K, L, \mu')$).
From this, it follows that $b \preceq \mathit{pr}$ and there is a  $l' \in \mathit{bound}^{i-1}(h, K, L, \mu')$ such that $u(\overline{y}, \mu'(b)) \subseteq \mathit{vars}(l')$.
From the induction hypothesis $(\spadesuit)$, it follows that the grounding of $l'$ is directly determined by the grounding of the values of the variables in $u(\mathit{args}(\mathit{head}(r)),K)$.
Furthermore, from $b \preceq \mathit{pr}$, $p, \preceq \models \unique{b}{\mu'(b)}$, and the induction hypothesis $(\clubsuit)$, it follows that the values of the variables in $u(\overline{y}, \mu'(b))$ uniquely determine the grounding of $b(\overline{y})$.
Therefore, the values of the variables in  $u(\mathit{args}(\mathit{head}(r)),K)$ uniquely determine the grounding of $b(\overline{y})$.
This completes the proof of the claim.

\item $\mathit{pr}(\overline{v})$ and  $\mathit{pr}(\overline{w})$ are generated by ground instances of two different rules $r_1$ and $r_2$.
From this, $\unique{\mathit{pr}}{K} \in {\cal U}$, and $K \neq \{1, \ldots, |\mathit{pr}|\}$, it follows that  
it follows that there is a value $i \in K$ such that $\mathit{args}(\mathit{head}(r_1)) \in \mathbf{dom}$,
$\mathit{args}(\mathit{head}(r_2)) \in \mathbf{dom}$, and
$\mathit{args}(\mathit{head}(r_1)) \neq \mathit{args}(\mathit{head}(r_1))$.
Therefore,  there is an $i \in K$ such that $\overline{v}(i) \neq \overline{w}(i)$.
This contradicts the fact that $\overline{v}(i) = \overline{w}(i)$ for all $i \in K$. 
\end{compactitem}

This completes the proof of our claim.
\end{proof}

\subsubsection{Proofs about Propagation Maps}

Here we prove some results about propagation maps, which we afterwards use in proving the main result.

\begin{proposition}\thlabel{theorem:vertical:maps}
Let $\Sigma$ be a first-order signature, $\mathbf{dom}$ be a finite domain, $p$ be a $(\Sigma,\mathbf{dom})$-\problog{} program, $r$ be a rule in $p$, and $l$ be the $i$-th literal in $\mathit{body}(r)$.
Furthermore, let $\mu$ be the $(r,l)$-vertical map.
Given a rule $r' \in \mathit{ground}(p,r)$, then $\overline{b}(j) = \overline{h}(\mu(j))$ for any $j$ such that $\mu(j)$ is defined, where $\overline{h} = \mathit{args}(\mathit{head}(r'))$ and $\overline{b} = \mathit{args}(\mathit{body}(r',i))$.
\end{proposition}

\begin{proof}
Let $\Sigma$ be a first-order signature, $\mathbf{dom}$ be a finite domain, $p$ be a $(\Sigma,\mathbf{dom})$-\problog{} program, $r$ be a rule in $p$, and $l$ be the $i$-th literal in $\mathit{body}(r)$.
Furthermore, let $\mu$ be the $(r,l)$-vertical map and  $r' \in \mathit{ground}(p,r)$.
From the definition of $\mathit{ground}(p,r)$, it follows that there is an assignment $\Theta$ from variables to elements in $\mathbf{dom}$ such that $r' = r \Theta$.
From this, $\overline{b}(j) = \Theta(\mathit{args}(l)(j))$ and  $\overline{h}(\mu(j)) = \Theta(\mathit{args}(\mathit{head}(r))(\mu(j)))$.
Note that from the definition of $(r,l)$-vertical map it follows that $\mathit{args}(l)(j) = \mathit{args}(\mathit{head}(r))(\mu(j))$.
From this, $\overline{b}(j) = \overline{h}(\mu(j))$ whenever $\mu(j)$ is defined.
\end{proof}

\begin{proposition}\thlabel{theorem:horizontal:maps}
Let $\Sigma$ be a first-order signature, $\mathbf{dom}$ be a finite domain, $p$ be a $(\Sigma,\mathbf{dom})$-\problog{} program, $r$ be a rule in $p$, $l$ be the $i$-th literal in $\mathit{body}(r)$, and $l'$ be the $j$-th literal in $\mathit{body}(r)$.
Furthermore, let $\mu$ be the $(r,l,l')$-horizontal map.
Given a rule $r' \in \mathit{ground}(p,r)$, then $\overline{b_1}(k) = \overline{b_2}(\mu(k))$ for any $k$ such that $\mu(k)$ is defined, where $\overline{b_1} = \mathit{args}(\mathit{body}(r',i))$ and $\overline{b_2} = \mathit{args}(\mathit{body}(r',j))$.
\end{proposition}

\begin{proof}
Let $\Sigma$ be a first-order signature, $\mathbf{dom}$ be a finite domain, $p$ be a $(\Sigma,\mathbf{dom})$-\problog{} program, $r$ be a rule in $p$, $l$ be the $i$-th literal in $\mathit{body}(r)$, and $l'$ be the $j$-th literal in $\mathit{body}(r)$.
Furthermore, let $\mu$ be the $(r,l,l')$-horizontal map  and  $r' \in \mathit{ground}(p,r)$.
From the definition of $\mathit{ground}(p,r)$, it follows that there is an assignment $\Theta$ from variables to elements in $\mathbf{dom}$ such that $r' = r \Theta$.
From this, $\overline{b_1}(j) = \Theta(\mathit{args}(l)(j))$ and  $\overline{b_2}(\mu(j)) = \Theta(\mathit{args}(l')(\mu(j)))$.
Note that from the definition of $(r,l,l')$-vertical map it follows that $\mathit{args}(l)(j) = \mathit{args}(l')(\mu(j))$.
From this, $\overline{b_1}(j) = \overline{b_2}(\mu(j))$ whenever $\mu(j)$ is defined.
\end{proof}

\begin{proposition}\thlabel{theorem:downward:connections}
Let $\Sigma$ be a first-order signature, $\mathbf{dom}$ be a finite domain, $p$ be a $(\Sigma,\mathbf{dom})$-\problog{} program,  $P = \mathit{pr}_1 \xrightarrow{r_1,i_1} \ldots \xrightarrow{r_{n-1}, i_{n-1}} \mathit{pr}_n$ be a directed path in $\graph(p)$, and $\nu : \mathbb{N} \to \mathbb{N}$ be a mapping.
Furthermore, let $P' = a_1 \rightarrow (r_1, s_1,i_1) \rightarrow a_2 \rightarrow (r_2,s_2, i_2) \rightarrow \ldots \rightarrow (r_{n-1}, s_{n-1}, i_{n-1}) \rightarrow a_n$ be a directed path in $\mathit{gg}(p)$ corresponding to $P$.
If $P$ $\nu$-downward links to $l$, where $l$ is the $k$-th literal in $r_j$, then $\overline{b_1}(m) = \overline{v_2}(\nu(m))$ whenever $\nu(m)$ is defined, where $\overline{b_1} = \mathit{args}(\mathit{body}(s_1,i_1))$ and $\overline{v_2} = \mathit{args}(\mathit{body}(s_j,k))$.
\end{proposition}

\begin{proof}
Let $\Sigma$ be a first-order signature, $\mathbf{dom}$ be a finite domain, $p$ be a $(\Sigma,\mathbf{dom})$-\problog{} program,  $P = \mathit{pr}_1 \xrightarrow{r_1,i_1} \ldots \xrightarrow{r_{n-1}, i_{n-1}} \mathit{pr}_n$ be a directed path in $\graph(p)$, and $\nu : \mathbb{N} \to \mathbb{N}$ be a mapping.
Furthermore, let $P' = a_1 \rightarrow (r_1, s_1,i_1) \rightarrow a_2 \rightarrow (r_2,s_2, i_2) \rightarrow \ldots \rightarrow (r_{n-1}, s_{n-1}, i_{n-1}) \rightarrow a_n$ be a directed path in $\mathit{gg}(p)$ corresponding to $P$.
Assume that  $P$ $\nu$-downward links to $l$, where $l$ is the $k$-th literal in $r_j$.
From this, it follows that the function $\mu := \mu' \circ \mu_{j} \circ \ldots \circ \mu_1$ satisfies $\mu(m) = \nu(m)$ for all $m$ for which $\nu(k)$ is defined, where for $1 \leq h \leq j$, $\mu_h$ is the vertical map connecting $\mathit{body}(r_h,i_h)$ and $r_h$, 
and $\mu'$ is the horizontal map connecting $\mathit{body}(r_{j+1},i_{j+1})$ with $l$.
By repeatedly applying \thref{theorem:vertical:maps} to the rules in $P'$, we have that $\overline{b_1}(m) = \overline{b_j}(\phi(m))$ whenever $\phi(m)$ is defined, where $\overline{b_1} = \mathit{args}(\mathit{body}(s_1,i_1))$, $\overline{b_j} = \mathit{args}(\mathit{body}(s_j,i_j))$, and $\phi = \mu_{j} \circ \ldots \circ \mu_1$.
Moreover, by applying \thref{theorem:horizontal:maps}, we have that $\overline{b_j}(m) = \overline{v_2}(\mu'(m))$ whenever $\mu'(m)$ is defined, where $\overline{b_j} = \mathit{args}(\mathit{body}(s_j,i_j))$ and $\overline{v_2} = \mathit{args}(\mathit{body}(s_j,k))$.
Therefore, we have that  $\overline{b_1}(m) = \overline{v_2}(\mu(m))$  whenever $\mu'(m)$ is defined (by composing the previous results).
From this and $\mu(m) = \nu(m)$ for all $m$ for which $\nu(m)$ is defined, it follows that $\overline{b_1}(m) = \overline{v_2}(\nu(m))$  whenever $\mu'(m)$ is defined.
\end{proof}

\begin{proposition}\thlabel{theorem:upward:connections}
Let $\Sigma$ be a first-order signature, $\mathbf{dom}$ be a finite domain, $p$ be a $(\Sigma,\mathbf{dom})$-\problog{} program,  $P = \mathit{pr}_1 \xrightarrow{r_1,i_1} \ldots \xrightarrow{r_{n-1}, i_{n-1}} \mathit{pr}_n$ be a directed path in $\graph(p)$, and $\nu : \mathbb{N} \to \mathbb{N}$ be a mapping.
Furthermore, let $P' = a_1 \rightarrow (r_1, s_1,i_1) \rightarrow a_2 \rightarrow (r_2,s_2, i_2) \rightarrow \ldots \rightarrow (r_{n-1}, s_{n-1}, i_{n-1}) \rightarrow a_n$ be a directed path in $\mathit{gg}(p)$ corresponding to $P$.
If $P$ $\nu$-upward links to $l$, where $l$ is the $k$-th literal in $r_j$, then $\overline{b_n}(m) = \overline{v_2}(\nu(m))$ whenever $\nu(m)$ is defined, where $\overline{b_n} = \mathit{args}(\mathit{head}(r_{n-1}))$ and $\overline{v_2} = \mathit{args}(\mathit{body}(s_j,k))$.
\end{proposition}

\begin{proof}
Let $\Sigma$ be a first-order signature, $\mathbf{dom}$ be a finite domain, $p$ be a $(\Sigma,\mathbf{dom})$-\problog{} program,  $P = \mathit{pr}_1 \xrightarrow{r_1,i_1} \ldots \xrightarrow{r_{n-1}, i_{n-1}} \mathit{pr}_n$ be a directed path in $\graph(p)$, and $\nu : \mathbb{N} \to \mathbb{N}$ be a mapping.
Furthermore, let $P' = a_1 \rightarrow (r_1, s_1,i_1) \rightarrow a_2 \rightarrow (r_2,s_2, i_2) \rightarrow \ldots \rightarrow (r_{n-1}, s_{n-1}, i_{n-1}) \rightarrow a_n$ be a directed path in $\mathit{gg}(p)$ corresponding to $P$.
Assume that $P$ $\nu$-upward links to $l$, where $l$ is the $k$-th literal in $r_j$.
From this, it follows that the function $\mu :=\mu'^{-1} \circ \mu_{j+1}^{-1} \circ \ldots \circ \mu_{n-1}^{-1}$ satisfies $\mu(k) = \nu(k)$ for all $k$ for which $\nu(k)$ is defined, where $\mu_h$ is the $(r_h,\mathit{body}(r_h,i_h))$-vertical map, for $j < h \leq n-1$, 
and $\mu'$ is the $(r_{j}, l)$-vertical map. 
By repeatedly applying \thref{theorem:vertical:maps} to the rules in $P'$, we obtain that  $\overline{b_{j+1}}(\phi^{-1}(m)) = \overline{b_n}(m)$ whenever $\phi^{-1}(m)$ is defined, where $\overline{b_{j+1}} = \mathit{args}(\mathit{body}(s_{j+1},i_{j+1}))$, $\overline{b_n} = \mathit{args}(\mathit{head}(s_{n-1}))$, and $\phi^{-1} = \mu_{j+1}^{-1} \circ \ldots \circ \mu_{n-1}^{-1}$.
Furthermore, by applying \thref{theorem:vertical:maps} to the $(r_j,l)$-vertical map, we have that $\overline{v_{2}}({\mu'}^{-1}(m)) = \overline{b_{j+1}}(m)$ whenever ${\mu'}^{-1}(m)$ is defined, where $\overline{b_{j+1}} = \mathit{args}(\mathit{head}(s_j))$ and $\overline{v_2} = \mathit{args}(\mathit{body}(s_j,k))$.
From this and $\mu(m) = \nu(m)$ for all $m$ for which $\nu(m)$ is defined, it follows that $\overline{b_n}(m) = \overline{v_2}(\nu(m))$ whenever $\nu(m)$ is defined.
\end{proof}

\subsubsection{Proofs about Connected Rules}

We now prove that, for strongly connected rules, the grounding of a rule's head uniquely determines the grounding of the rule's body (\thref{theorem:join:tree:1}), whereas for weakly connected rules the grounding of one of the atoms in the body determines the head's grounding (\thref{theorem:join:tree:2}).

\begin{proposition}\thlabel{theorem:join:tree:1}
Let $\Sigma$ be a first-order signature, $\mathbf{dom}$ be a finite domain, $p$ be a $(\Sigma,\mathbf{dom})$-\problog{} program, $r$ be a rule in $p$, and $\mathit{U}$ be a $\Sigma$-uniqueness template.
If $p$ complies with ${\cal U}$ and $r$ is strongly connected for ${\cal U}$, then for all $r_1,r_2 \in \mathit{ground}(p,r)$, if $\mathit{head}(r_1) = \mathit{head}(r_2)$, then $r_1 = r_2$.
\end{proposition}

\begin{proof}
Let $\Sigma$ be a first-order signature, $\mathbf{dom}$ be a finite domain, $p$ be a $(\Sigma,\mathbf{dom})$-\problog{} program, $r$ be a rule in $p$, and $\mathit{U}$ be a $\Sigma$-uniqueness template.
Furthermore, we assume that (1)  $p$ complies with ${\cal U}$, and (2) there are join trees $J_1, \ldots, J_n$ such that (a) the trees cover all literals in $\mathit{body}(r)$, and (b) for each $1 \leq i \leq n$, $J_i$  is strongly connected for ${\cal U}$.
Finally, let $r_1,r_2 \in \mathit{ground}(p,r)$ be two ground rules such that $\mathit{head}(r_1) = \mathit{head}(r_2)$.
Assume, for contradiction's sake, that $r_1 \neq r_2$.
Therefore, there is a position $1 \leq i \leq |\mathit{body}(r)|$ such that $\mathit{body}(r_1,i) \neq \mathit{body}(r_2,i)$.
Let $J = (N, E, \treeroot, \lambda)$ be one of the join trees that cover $l = \mathit{body}(r,i)$.
We claim that the grounding of $\mathit{head}(r)$ determines the grounding of all literals in $J$.
From this, it follows that   $\mathit{body}(r_1,i) = \mathit{body}(r_2,i)$ leading to a contradiction.

We now prove our claim that the grounding of $\mathit{head}(r)$ determines the grounding of all literals in $J$.
Let $V(J,i)$ be the set of all the nodes in $J$ at distance at most $i$ from the root $\treeroot$.
Furthermore, we denote by $\mathit{ground}(r,r',J,i)$ the set $\{ \mathit{body}(r',j) \mid \mathit{body}(r,j) \in V(J,i) \}$.
We prove, by induction on $i$, that for all $i$ and all ground rules $r_1$ and $r_2$ instances of $r$, if $\mathit{head}(r_1) = \mathit{head}(r_2)$, then $\mathit{ground}(r,r_1,J,i) = \mathit{ground}(r,  r_2,J,i)$.
From this, it follows that the grounding of $\mathit{head}(r)$ determines the grounding of all literals in $J$.

\para{Base case}
For $i = 0$, there is a $j$ such that $V(J,0) = \{ \mathit{body}(r,j) \}$.
From this, it follows that $\mathit{ground}(r,r_1,J,0) = \{ \mathit{body}(r_1,j) \}$ and $\mathit{ground}(r, r_2, J,0) = \{ \mathit{body}(r_2,j)\}$.
Furthermore, $\mathit{anc}(J, \mathit{body}(r,j)) = \emptyset$.
There are two cases depending on whether $\mathit{body}(r,j)$ is a positive literal or not:
\begin{compactenum}
\item  If $\mathit{body}(r,j)$ is a positive literal of the form $a(\overline{x})$, then $\mathit{body}(r_1,j) = a(\overline{c}_1)$ and $\mathit{body}(r_2,j) = a(\overline{c}_2)$. 
From the fact that $J$ is strongly connected for ${\cal U}$, it follows that there is a set of variables $K \subseteq \{ i \mid \overline{x}(i) \in \mathit{support}(a(\overline{x}))\}$ such that $\unique{a}{K} \in {\cal U}$.
From $\mathit{support}(a(\overline{x})) = \mathit{vars}(\mathit{head}(r)) \cup \{x \mid (x = c) \in \mathit{cstr}(r) \wedge c \in \mathbf{dom}\}$ and $\mathit{head}(r_1) = \mathit{head}(r_2)$, it follows that the values assigned to the variables associated to the indexes in $K$ are the same in $r_1$ and $r_2$.
From this, $\unique{a}{K} \in {\cal U}$, $p$ complies with ${\cal U}$, $\{a(\overline{c}_1),a(\overline{c}_2)\} \subseteq \mathit{ground}(p)$, and  \thref{theorem:template:uniqueness:semantics}, it follows $\overline{c}_1 = \overline{c}_2$ and $\mathit{ground}(r,r_1,J,0) = \mathit{ground}(r, r_2, J,0)$.

\item  If $\mathit{body}(r,j)$ is a negative literal of the form $\neg a(\overline{x})$, then $\mathit{body}(r_1,j) = \neg a(\overline{c}_1)$ and $\mathit{body}(r_2,j) = \neg a(\overline{c}_2)$. 
From the fact that $J$ is strongly connected for ${\cal U}$, it follows that  $\mathit{vars}(\neg a(\overline{x})) \subseteq \mathit{support}(\neg a(\overline{x}))$.
From this, $\mathit{support}(\neg a(\overline{x})) = \mathit{vars}(\mathit{head}(r)) \cup \{x \mid (x = c) \in \mathit{cstr}(r) \wedge c \in \mathbf{dom}\}$, and $\mathit{head}(r_1) = \mathit{head}(r_2)$, it follows that the values of the variables in $\mathit{support}(\neg a(\overline{x}))$ are the same in $r_1$ and $r_2$.
From this, it follows that $\overline{c}_1 = \overline{c}_2$ and $\mathit{ground}(r,r_1,J,0) = \mathit{ground}(r, r_2, J,0)$.
\end{compactenum}

\para{Induction Step}
Assume that for all $j < i$ and all ground rules $r_1,r_2 \in \mathit{ground}(p,r)$, if $\mathit{head}(r_1) = \mathit{head}(r_2)$, then $\mathit{ground}(r,r_1,J,j)  = \mathit{ground}(r,r_2,J,j)$.
We now show that $\mathit{ground}(r,r_1,J,i) = \mathit{ground}(r,r_2,J,i)$.
Assume, for contradiction's sake, that this is not the case, namely $\mathit{ground}(r,r_1, J,i) \neq \mathit{ground}(r,r_2,J,i)$.
From the definition of $\mathit{ground}(r,r', J,i)$, it follows that $\mathit{ground}(r,r_1,J,i) = \mathit{ground}(r,r_1,J,i-1) \cup \{ \mathit{body}(r_1, j) \mid \mathit{body}(r,j) \in V(J,i) \setminus V(J,i-1)\}$ and $\mathit{ground}(r,r_2,J,i) = \mathit{ground}(r,r_2,J,i-1) \cup \{ \mathit{body}(r_2, j) \mid \mathit{body}(r,j) \in V(J,i) \setminus V(J,i-1)\}$.
From this, the induction's hypothesis, and $\mathit{ground}(r,r_1,J,i) \neq \mathit{ground}(r,  r_2,J,i)$, it follows that $\{ \mathit{body}(r_1, j) \mid \mathit{body}(r,j) \in V(J,i) \setminus V(J,i-1)\} \neq \{ \mathit{body}(r_2, j) \mid \mathit{body}(r,j) \in V(J,i) \setminus V(J,i-1)\}$.
Therefore, there is a $j$ such that $\mathit{body}(r,j)  \in V(J,i) \setminus V(J,i-1)$ and $\mathit{body}(r_1,j)  \neq \mathit{body}(r_2, j)$.
There are two cases, depending on whether $\mathit{body}(r,j)$ is a positive literal:
\begin{compactenum}
\item If $\mathit{body}(r,j) = a(\overline{x})$, then $\mathit{body}(r_1, j) = a(\overline{c}_1)$, $\mathit{body}(r_2, j) =  a(\overline{c}_2)$, and $\overline{c}_1 \neq \overline{c}_2$.
From the fact that $J$ is strongly connected for ${\cal U}$, it follows that there is a set of variables $K \subseteq \{ i \mid \overline{x}(i) \in \mathit{support}(a(\overline{x}))\}$ such that $\unique{a}{K} \in {\cal U}$.
From $\mathit{support}(a(\overline{x})) = \mathit{vars}(\mathit{head}(r)) \cup \{x \mid (x = c) \in \mathit{cstr}(r) \wedge c \in \mathbf{dom}\} \cup \bigcup_{l \in V(J,i-1)} \mathit{vars}(l) $, $\mathit{ground}(r, r_1, J,i-1) = \mathit{ground}(r, r_2,J, i-1)$ (from the induction's hypothesis), and $\mathit{head}(r_1) = \mathit{head}(r_2)$, it follows that the values assigned to the variables associated to the indexes in $K$ are the same in $r_1$ and $r_2$.
From this, $\unique{a}{K} \in {\cal U}$, $p$ complies with ${\cal U}$, $\{a(\overline{c}_1),a(\overline{c}_2)\} \subseteq \mathit{ground}(p)$, and  \thref{theorem:template:uniqueness:semantics}, it follows that $\overline{c}_1 = \overline{c}_2$ leading to a contradiction.

\item If $\mathit{body}(r,j) = \neg a(\overline{x})$, then $\mathit{body}(r_1, j) = \neg a(\overline{c}_1)$, $\mathit{body}(r_2, j) = \neg a(\overline{c}_2)$, and $\overline{c}_1 \neq \overline{c}_2$.
From the fact that $J$ is strongly connected for ${\cal U}$, it follows that $\mathit{vars}(\neg a(\overline{x})) \subseteq  \mathit{support}(a(\overline{x}))$.
From $\mathit{support}(a(\overline{x})) = \mathit{vars}(\mathit{head}(r)) \cup \{x \mid (x = c) \in \mathit{cstr}(r) \wedge c \in \mathbf{dom}\} \cup \bigcup_{l \in V(J,i-1)} \mathit{vars}(l) $, $\mathit{ground}(r,r_1, J,i-1) = \mathit{ground}(r, r_2, J, i-1)$ (from the induction's hypothesis), it follows that the values assigned to the variables  in $\mathit{support}(a(\overline{x}))$ are the same in $r_1$ and $r_2$.
From this,  it follows that $\overline{c}_1 = \overline{c}_2$ leading to a contradiction.
\end{compactenum}
Since both cases lead to a contradiction, $\mathit{ground}(r,r_1,J,i) = \mathit{ground}(r,r_2,J,i)$.
\end{proof}

\begin{proposition}\thlabel{theorem:join:tree:2}
Let $\Sigma$ be a first-order signature, $\mathbf{dom}$ be a finite domain, $p$ be a $(\Sigma,\mathbf{dom})$-\problog{} program, $r$ be a rule in $p$,  and $\mathit{U}$ be a $\Sigma$-uniqueness template.
If $p$ complies with ${\cal U}$ and $r$ is weakly connected for ${\cal U}$, then for all $r_1,r_2 \in \mathit{ground}(p,r)$, if there is an $1 \leq i \leq |\mathit{body}(r)|$ such that $\mathit{body}(r_1,i) = \mathit{body}(r_2,i)$, then $r_1 = r_2$.
\end{proposition}

\begin{proof}
Let $\Sigma$ be a first-order signature, $\mathbf{dom}$ be a finite domain, $p$ be a $(\Sigma,\mathbf{dom})$-\problog{} program, $r$ be a rule in $p$,  and $\mathit{U}$ be a $\Sigma$-uniqueness template.
Furthermore, we assume that $p$ complies with ${\cal U}$ and that $r$ is weakly connected for ${\cal U}$, namely there is a join tree $J = (N,E, \treeroot, \lambda)$ such that (a) $J$ is weakly connected for ${\cal U}$, (b) $N \subseteq \mathit{body}^+(r)$, and (c) all literals in $\mathit{body}(r) \setminus N$ are $(r,{\cal U}, N)$-strictly guarded.
Let $r_1,r_2 \in \mathit{ground}(p,r)$ be two ground rules such that there is an $1 \leq i \leq |\mathit{body}(r)|$ such that $\mathit{body}(r_1,i) = \mathit{body}(r_2,i)$, and let $l$ be the literal $\mathit{body}(r,i)$.
There are two cases:
\begin{compactenum}
\item $\mathit{body}(r,i) \in N$.
Given a node $n$ in a join tree $J$, we denote by $\mathit{adjacent}(n,i)$, where $i \in \mathbb{N}$, the sub-tree obtained by considering only the nodes reachable from $n$ using at most $i$ edges.
We claim that for all nodes in $\mathit{adjacent}(l,j)$, the corresponding ground atoms in $r_1$ and $r_2$ are the same.
From this and the fact that there is a $j$ such that $\mathit{adjacent}(l,j) = J$, it follows that for all $i \in \{ j \mid \mathit{body}(r,j) \in N\}$, $\mathit{body}(r_1,i) = \mathit{body}(r_2,i)$.
From this and the fact that the literals in $\mathit{body}(r) \setminus N$ are $(r,{\cal U}, N)$-strictly guarded, it follows that for all $i \in \{ j \mid \mathit{body}(r,j) \in \mathit{body}(r) \setminus N\}$, $\mathit{body}(r_1,i) = \mathit{body}(r_2,i)$ (since $\mathit{vars}(l) \subseteq  \bigcup_{l' \in N\cap\mathit{body}^+(r)} \mathit{vars}(l') \cup \{ x \mid (x = c) \in \mathit{cstr}(r) \wedge c \in \mathbf{dom} \}$ for any literal $l$ in $\mathit{body}(r) \setminus N$).
Therefore, $\mathit{body}(r_1) = \mathit{body}(r_2)$.

\item $\mathit{body}(r,i) \not\in N$. 
From this and the fact that $r$ is weakly connected, it follows that there a $j$ such that $\mathit{body}(r,j) = a(\overline{x})$, $\mathit{body}(r,j) \in N$, and a $\unique{a}{K} \in {\cal U}$ such that $\{ \overline{x}(i) \mid i \in K \} \subseteq \mathit{vars}(\mathit{body}(r,i))$.
From this, $p$ complies with ${\cal U}$, $\mathit{body}(r_1,i) = \mathit{body}(r_2,i)$, $\{\mathit{body}(r_1,j),\mathit{body}(r_2,j)\} \subseteq \mathit{ground}(p)$, and \thref{theorem:template:uniqueness:semantics}, it follows that  $\mathit{body}(r_1, j) = \mathit{body}(r_2,j)$.
We proved above that if $\mathit{body}(r_1,j) = \mathit{body}(r_2,j)$ and $\mathit{body}(r,j) \in N$, then $\mathit{body}(r_1) = \mathit{body}(r_2)$.
Therefore, $\mathit{body}(r_1) = \mathit{body}(r_2)$.

\end{compactenum}
From $\mathit{vars}(\mathit{head}(r)) \subseteq \bigcup_{l \in \mathit{body}^+(r)} \mathit{vars}(l)$ and $\mathit{body}(r_1) = \mathit{body}(r_2)$, it follows that $\mathit{head}(r_1) = \mathit{head}(r_2)$.
Therefore, $r_1 = r_2$.

We now prove, by induction on $j$, that for all nodes in $\mathit{adjacent}(l,j)$, the corresponding ground atoms in $r_1$ and $r_2$ are the same.

\para{Base Case}
The sub-tree $\mathit{adjacent}(l,0)$ contains only the node $l = \mathit{body}(r,i)$.
Since $\mathit{body}(r_1,i) = \mathit{body}(r_2,i)$, the claim holds for the base case.

\para{Induction step}
Assume now that the claim holds for all $j' < j$.
We now prove that the claim holds also for $j$.
The sub-tree $\mathit{adjacent}(l,j)$ is obtained by extending $\mathit{adjacent}(l,j - 1)$ with either edges of the form $n_1 \xrightarrow{L_1} n_1'$, where $n_1' \in \mathit{adjacent}(l,j - 1)$, or $n_2' \xrightarrow{L_2} n_2$, where $n_2' \in \mathit{adjacent}(l,j - 1)$.
In both cases, from the induction hypothesis, the ground literals corresponding to $n_1'$ and $n_2'$ are $l_1'$ and $l_2'$ and they are the same in $r_1$ and $r_2$.
From the definition of weakly connected join tree, there are variables $K_1 \subseteq L_1$ and $K_2 \subseteq L_2$ such that $\unique{\mathit{pred}(n_1)}{K_1} \in {\cal U}$ and $\unique{\mathit{pred}(n_2)}{K_2} \in {\cal U}$.
From this, $p$ complies with ${\cal U}$, the fact that the value of $K_1$ and $K_2$ are fixed by $l_1'$ and $l_2'$, and \thref{theorem:template:uniqueness:semantics}, it follows that the ground atoms corresponding to $n_1$ and $n_2$ are the same in $r_1$ and $r_2$ (because the values in $K_1$ and $K_2$ determines all values in $n_1$ and $n_2$).
\end{proof}

\subsubsection{Acyclicity Proof}

We are now ready to prove our first key result, namely that the ground graph associated to an acyclic \problog{} program is a forest of \emph{poly-trees}, i.e., its undirected version does not contain simple cycles (which are cycles without repetitions of edges and vertices other than the starting and ending vertices).

\begin{proposition}\thlabel{theorem:graph:acyclic}
Let $\Sigma$ be a first-order signature, $\mathbf{dom}$ be a finite domain, and $p$ be a $(\Sigma,\mathbf{dom})$-acyclic \problog{} program.
The graph $\mathit{gg}(p)$ is a forest of poly-trees, i.e., its undirected version does not contain simple cycles.
\end{proposition}

\begin{proof}
Let $\Sigma$ be a first-order signature, $\mathbf{dom}$ be a finite domain, and $p$ be a $(\Sigma,\mathbf{dom})$-acyclic \problog{} program.
We prove that the undirected version of $\mathit{gg}(p)$ is acyclic. 
Assume, for contradiction's sake, that this is not the case, namely there is a simple cycle $C:= n_1 \to n_2 \to \ldots \to n_k \to n_1$ in the undirected version of $\mathit{gg}(p)$.
There are two cases:
(1) there is a directed simple cycle in $\mathit{gg}(p)$ that directly corresponds to $C$, or
(2) there are $n$ directed paths $P_1, \ldots, P_n$ in $\mathit{gg}(p)$ that induce a simple cycle $C$ in the undirected version of $\mathit{gg}(p)$ (and $P_1, \ldots, P_n$ does not correspond to any directed simple cycle in $\mathit{gg}{p}$).
From the first case, it follows that there is a directed cycle in $\mathit{graph}(p)$, whereas from the second case, it follows that there is an undirected cycle in $\mathit{graph}(p)$.

\para{Directed Cycle}
Assume that  $C:=n_1 \to n_2 \to \ldots \to n_k \to n_1$ directly corresponds to a directed cycle in $\mathit{gg}(p)$.
Without loss of generality, we assume that $n_1$ is a node of the form $a(\overline{c})$.
Therefore, the cycle has form $a_1(\overline{c}_1) \to (r_1,s_1,i_1) \to  a_2(\overline{c}_2) \to \ldots \to a_{n}(\overline{c}_{n}) \to (r_n,s_n,i_n) \to a_{n+1}(\overline{c}_{n+1})$, where $n = \sfrac{k}{2}$ and $a_{n+1}(\overline{c}_{n+1}) = a_1(\overline{c}_1)$.
From this, it follows that there are rules $r_1,\ldots, r_n$ and ground rules $s_1, \ldots, s_n$ such that:
\begin{inparaenum}[(1)]
\item there is a directed cycle $a_1 \xrightarrow{r_1,i_1} a_2 \xrightarrow{r_2,i_2} \ldots \xrightarrow{r_{n-1}, i_{n-1}} a_n \xrightarrow{r_n,i_n} a_1$ in $p$'s dependency graph $\mathit{graph}(p)$, and
\item  for all $1 \leq i \leq n$, $s_i \in \mathit{ground}(p,r_i)$ and $\mathit{head}(s_i) = a_{i+1}(c_{i+1})$.
\end{inparaenum}
From this, it follows that $p$'s dependency graph $\mathit{graph}(p)$ contains a directed cycle $C' = a_1 \xrightarrow{r_1,i_1} a_2 \xrightarrow{r_2,i_2} \ldots \xrightarrow{r_{n-1}, i_{n-1}}  a_n \xrightarrow{r_n,i_n} a_1$.
Note that $C'$ may contain loops (i.e., it is not simple).
Since $p$ is acyclic and $C'$ is a directed cycle in $\graph(p)$, it follows that there is an ordering template ${\cal O}$,  a disjointness template ${\cal D}$, and a uniqueness template ${\cal U}$ such that (1) $p$ complies with ${\cal O}$, ${\cal D}$, and ${\cal U}$ and (2) there is a directed unsafe structure $S$ that covers $C'$ and is $({\cal O}, {\cal D}, {\cal U})$-guarded.
Without loss of generality, we assume that $S = C'$.
From $S$ covers $C'$ and $S$ is $({\cal O}, {\cal D}, {\cal U})$-guarded, it follows that there is an ordering annotation $\order{O} \in {\cal O}$ such that there are integers $1 \leq y_1 < y_2 < \ldots < y_e = n$, literals $o_1(\overline{x}_1), \ldots, o_e(\overline{x}_e)$ (where $o_j \in O$ and $|\overline{x}_j| = |o_j|$),
a  non-empty set $K \subseteq \{1, \ldots, |\mathit{pr}_1|\}$, and a bijection $\nu : K \to \{1, \ldots, \sfrac{|\mathit{o}_1|}{2}\}$ such that 
for each $0 \leq k < e$, 
(1) $\mathit{pr}_{y_k} \xrightarrow{r_{y_k}, i_{y_k}} \ldots \xrightarrow{r_{{y_{k+1}}-1}, i_{{y_{k+1}}-1}} pr_{y_{k+1}}$ $\nu$-downward connects to $o_{k+1}(\overline{x}_{k+1})$, and 
(2) $\mathit{pr}_{y_{k+1}-1} \xrightarrow{r_{{y_{k+1}}-1}, i_{{y_{k+1}}-1}} pr_{y_{k+1}}$ $\nu'$-upward connects to $o_{k+1}(\overline{x}_{k+1})$, 
where 
$\nu'(i) = \nu(x) + \sfrac{|\mathit{o}_1|}{2}$ for all $1 \leq i \leq \sfrac{|\mathit{o}_1|}{2}$, and
$y_0 = 1$.
By applying \thref{theorem:downward:connections} and \thref{theorem:upward:connections} to the paths $\mathit{pr}_{y_k} \xrightarrow{r_{y_k}, i_{y_k}} \ldots \xrightarrow{r_{{y_{k+1}}-1}, i_{{y_{k+1}}-1}} pr_{y_{k+1}}$ $\nu$-downward connects to $o_{k+1}(\overline{x}_{k+1})$ for each $0 \leq k < e$, it follows that (1) ground atoms $o_1(\overline{b}_1, \overline{b}_2),  \ldots, o_e(\overline{b}_{|K|}, \overline{b}_{|K|+1})$ are in $\mathit{ground}(p)$, and (2) for all $1 \leq w \leq \sfrac{|\mathit{pr}|}{2}$, both $\overline{b}_1(w) = \mathit{args}(\mathit{body}(s_1,i_1))(\nu(w))$ and $\overline{b}_{|K|+1}(w) = \mathit{args}(\mathit{head}(s_n))(\nu'(w))$ hold.
From this, $\order{O} \in {\cal O}$, $p$ complies with ${\cal O}$, and \thref{theorem:template:ordering:semantics}, it follows that $\overline{b}_1 \neq \overline{b}_{|K|+1}$.
From this  $\overline{b}_1(w) = \mathit{args}(\mathit{body}(s_1,i_1))(\nu(w))$ and $\overline{b}_{|K|+1}(w) = \mathit{args}(\mathit{head}(s_n))(\nu'(w))$  for all $1 \leq w \leq \sfrac{|\mathit{pr}|}{2}$, it follows that $\mathit{body}(s_1,i_1) \neq \mathit{head}(s_n)$.
This contradicts $\mathit{head}(s_n) = a(\overline{c}_1)$ and $\mathit{body}(s_1,i_1) = a(\overline{c}_1)$ (which directly follows from the existence of the cycle $C$).

\para{Undirected Cycle}
Assume that $C$ does not directly correspond to any directed simple cycle in $\mathit{gg}(p)$.
From this, it follows that there are $n$ directed paths $P_1, \ldots, P_n$ in $\mathit{gg}(p)$ such that $P_1,\ldots, P_n$ correspond to the simple cycle $C$ in the undirected version of $\mathit{gg}(p)$.
From this, it follows that $P_1, \ldots, P_n$ form an undirected cycle in $\mathit{gg}(p)$.
Therefore:
\begin{compactenum}
\item For $1 \leq j \leq n$, there is a directed path $D_j := a_1 \xrightarrow{r_1,i_1} a_2 \xrightarrow{r_2,i_2} \ldots \xrightarrow{r_{n_j-1}, i_{n_j-1}} a_{n_j}$ in $p$'s dependency graph $\mathit{graph}(p)$, where 
$P_j$ is $a_1(\overline{c}_1) \rightarrow (r_1,s_1, i_1) \rightarrow a_2(\overline{c}_2) \rightarrow \ldots \rightarrow a_{n_j-1}(\overline{c}_{n_j-1}) \rightarrow (r_{n_j-1}, s_{n_j -1}, i_{n_j-1}) \rightarrow a_{n_j}(\overline{c}_{n_j})$. 
\item  For all $1 \leq j \leq n$, $1 \leq h < n_j$, $s_h \in \mathit{ground}(p,r_h)$ and $\mathit{head}(s_h) = a_{h+1}(\overline{c}_{h+1})$, where 
 $P_j := a_1(\overline{c}_1) \rightarrow (r_1,s_1, i_1) \rightarrow a_2(\overline{c}_2) \rightarrow \ldots \rightarrow a_{n_j-1}(\overline{c}_{n_j-1}) \rightarrow (r_{n_j-1}, s_{n_j -1}, i_{n_j-1}) \rightarrow a_{n_j}(\overline{c}_{n_j})$. 
\item $D_1, \ldots, D_n$ form an undirected cycle in $\graph(p)$.
\item $D_1, \ldots, D_n$ is not a directed simple cycle.
Indeed, if $D_1, \ldots, D_n$ is a directed simple cycle, then $C$ would contain loops (contradicting our assumption that $C$ is a simple cycle).
\end{compactenum}
Since $p$ is acyclic and $D_1, \ldots, D_n$ form an undirected cycle in $\graph(p)$ that is not a directed simple cycle, it follows that there is an ordering template ${\cal O}$,  a disjointness template ${\cal D}$, and a uniqueness template ${\cal U}$ such that (1) $p$ complies with ${\cal O}$, ${\cal D}$, and ${\cal U}$ and (2) there is an unsafe structure $S$ that covers $D_1, \ldots, D_n$ and is $({\cal O}, {\cal D}, {\cal U})$-guarded.
We assume that there are values $i,a,b,c$ such that 
\begin{compactenum}
\item $S = \langle D_a, D_b, D_c, U \rangle$, 
\item  $a = i$, $b = (i+1)\%n$, and $c = (i+2)\%n$ (i.e., $P_a$, $P_b$, and $P_c$ are adjacent in the cycle), and
\item $U$ is the undirected path containing all $D_i$'s that are different from $D_a$, $D_b$, and $D_c$.
\end{compactenum}
Note that the previous assumption is without loss of generality: we can always pick the directed paths in $\mathtt{gg}(p)$ inducing $C$ in such a way that they match the guarded structure $S$.

Let $P_a$ be $a_1(\overline{a}_1) \rightarrow (r_1,s_1,i_1) \rightarrow a_2(\overline{a}_2) \rightarrow (r_2,s_2,i_2) \rightarrow \ldots \rightarrow (r_{n_a-1},s_{n_a-1}, i_{n_a-1}) \rightarrow a_{n_a}(\overline{a}_{n_a})$, $P_b$ be $b_1(\overline{b}_1) \rightarrow (r_1',s_1',i_1') \rightarrow b_2(\overline{b}_2) \rightarrow (r_2',s_2',i_2') \rightarrow \ldots \rightarrow (r_{n_b-1}',s_{n_b-1}', i_{n_b-1}') \rightarrow b_{n_b}(\overline{b}_{n_b})$, and $P_c$ be $c_1(\overline{c}_{1}) \rightarrow (r_1'',s_1'',i_1'') \rightarrow c_2(\overline{c}_{2}) \rightarrow (r_2'',s_2'',i_2'') \rightarrow \ldots \rightarrow (r_{n_c-1}'',s_{n_c-1}'', i_{n_c-1}'') \rightarrow c_{n_c}(\overline{c}_{n_c})$. 
From $S= \langle D_a, D_b, D_c, U \rangle$, it follows that $a_1(\overline{a}_1) = b_1(\overline{b}_1)$ and $b_{n_b}(\overline{b}_{n_b}) = c_{n_c}(\overline{c}_{n_c})$.
From $S = \langle D_a, D_b, D_c, U \rangle$, $S$ covers $D_1, \ldots, D_n$, and $S$ is  $({\cal O}, {\cal D}, {\cal U})$-guarded, there are two cases:
\begin{compactenum}
\item $(D_a, D_b)$ is $({\cal D}, {\cal U})$-head guarded.
Therefore, $D_a$ and $D_b$ are non-empty.
There are two cases:
\begin{compactenum}

\item $D_a = D_b$. 
From this, it follows that (1) $P_a$ and $P_b$ are directed paths in $\mathit{gg}(p)$, (2)  $a_1(\overline{a}_1) = b_1(\overline{b}_1)$, (3) $n_a = n_b$, and  (4) $D_a$ and $D_b$ are head-connected.
From this and $(D_a,D_b)$ is head-guarded, it follows that the rules in $D_a$ and $D_b$ are weakly connected for ${\cal U}$.
We claim that $s_h = s_h'$ for any $1 \leq h \leq n_a$.
From this, it follows that $P_a = P_b$. 
This contradicts the fact that $P_1, \ldots, P_n$ induces a simple cycle in the undirected version of $\mathit{gg}(p)$ (because the cycle is not simple).

We now prove, by induction on $d$, that $s_d = s_d'$ and $a_d(\overline{a}_d) = b_d(\overline{b}_d)$.

\para{Base Case}
Assume that $d = 1$.
From $r_1 = r_1'$, $r_1$ is weakly connected for ${\cal U}$,  $\mathit{body}(s_1,i_1) = \mathit{body}(s_1',i_1') = a_1(\overline{c}_1)$, $s_d \in \mathit{ground}(p,r_d)$, $s_d' \in \mathit{ground}(p,r_d')$ , and \thref{theorem:join:tree:2}, it follows that $s_1 = s_1'$. 
From this, $a_1(\overline{a}_1) = \mathit{head}(s_1)$, and $b_1(\overline{b}_1) = \mathit{head}(s_1')$, it follows that $a_1(\overline{a}_1) = b_1(\overline{b}_1)$.

\para{Induction Step}
Assume that the claim holds for all $d' < d$.
We now prove that $s_d = s_d'$ and $a_d(\overline{a}_d) = b_d(\overline{b}_d)$ hold as well.
From the induction's hypothesis, it follows that $\mathit{head}(s_{d-1}) = \mathit{head}(s_{d-1}')$.
From this and $\mathit{gg}$'s definition, it follows that $\mathit{body}(s_d, i_d) = \mathit{body}(s_d', i_d')$.
From this, $r_d = r_d'$, $i_d = i_d'$, $r_d$ is weakly connected for ${\cal U}$, $s_d \in \mathit{ground}(p,r_d)$,  $s_d' \in \mathit{ground}(p,s_d)$, and \thref{theorem:join:tree:2}, it follows that $s_d = s_d'$. 
From this, $a_d(\overline{a}_d) = \mathit{head}(s_d)$, and $b_d(\overline{b}_d) = \mathit{head}(s_d')$, it follows that $a_d(\overline{a}_d) = b_d(\overline{b}_d)$.

\item $D_a \neq D_b$. 
From this and $(D_a,D_b)$ are head-guarded, it follows that there is an annotation $\disjoint{\mathit{pr}}{\mathit{pr}'} \in {\cal D}$ a set $K \subseteq \{1, \ldots, |a|\}$, and a bijection $\nu : K \to \{1, \ldots, |\mathit{pr}|\}$ such that $D_a$ $\nu$-downward links to $\mathit{pr}$ and $D_a$ $\nu$-downward links to $\mathit{pr}'$.
From $D_a$ $\nu$-downward links to $\mathit{pr}$, $P_a$ is a ground instance of $D_a$, and \thref{theorem:downward:connections}, it follows that there is a positive literal $\mathit{pr}(\overline{v})$ in the body of one of the ground rules such that $a_1(\overline{a}_1)(k) = \overline{v}(\nu(K))$ for any $k \in K$.
From this and the definition of $\mathit{ground}$, it follows that $\mathit{pr}(\overline{v}) \in \mathit{ground}(p)$.
From $D_b$ $\nu$-downward links to $\mathit{pr}'$, $P_b$ is a ground instance of $D_b$, and \thref{theorem:downward:connections}, it follows that there is a positive literal $\mathit{pr}'(\overline{v}')$ in the body of one of the ground rules such that $b_1(\overline{b}_1)(k) = \overline{v}'(\nu(K))$ for any $k \in K$.
From this and the definition of $\mathit{ground}$, it follows that $\mathit{pr}'(\overline{v}') \in \mathit{ground}(p)$.
Finally, from the fact that $P_a$ and $P_b$ are head-connected, it follows that $a_1(\overline{a}_1) = b_1(\overline{b}_1)$.
From this, $a_1(\overline{a}_1)(k) = \overline{v}(\nu(K))$ for any $k \in K$, and $b_1(\overline{b}_1)(k) = \overline{v}'(\nu(K))$ for any $k \in K$, it follows that $\overline{v} = \overline{v}'$.
From this,  $\mathit{pr}(\overline{v}) \in \mathit{ground}(p)$ and $\mathit{pr}'(\overline{v}) \in \mathit{ground}(p)$.
This contradicts \thref{theorem:template:disjointness:semantics}, $\disjoint{\mathit{pr}}{\mathit{pr}'} \in {\cal D}$, and $p$ complies with ${\cal D}$.

\end{compactenum}

\item $(D_b, D_c)$ is $({\cal D}, {\cal U})$-tail guarded.
Therefore, $D_b$ and $D_c$ are non-empty.
There are two cases:
\begin{compactenum}
\item $D_b = D_c$. % and $D_b$ is non-empty.
From this, it follows that (1) $P_c$ and $P_b$ are directed paths in $\mathit{gg}(p)$, (2)  $c_{n_c}(\overline{c}_{n_c}) = b_{n_b}(\overline{b}_{n_b})$, (3) $n_b = n_c$, and (4) $D_b$ and $D_c$ are tail-connected.
From this and $(D_b,D_c)$ is tail-guarded,  it follows that (1) $D_c$ and $D_b$ are tail-connected, and (2) the rules in $D_c$ and $D_b$ are strongly connected for ${\cal U}$.
We claim that $s_h'' = s_h'$ for any $1 \leq h \leq n_c$.
From this, it follows that $P_c = P_b$. 
This contradicts the fact that $P_1, \ldots, P_n$ induces a simple cycle in the undirected version of $\mathit{gg}(p)$ (because the cycle is not simple).

We now prove,  by induction on $d$, that $s_{n_b-d}'' = s_{n_b-d}'$ and $\mathit{head}(s_{n_b-d}'') = \mathit{head}(s_{n_b-d}')$.

\para{Base Case}
Assume $d = 0$.
From $r_{n_b}'' = r_{n_b}'$, $r_{n_b}'$ is strongly connected for ${\cal U}$, $\mathit{head}(s_{n_b}'') = \mathit{head}(s_{n_b}') = b_{n_b}(\overline{b}_{n_b})$, $s_{n_b}'' \in \mathit{ground}(p,r_d'')$, $s_d' \in \mathit{ground}(p,r_d')$ , and \thref{theorem:join:tree:1}, it follows that $s_1 = s_1'$. 
From this, $a_1(\overline{a}_1) = \mathit{head}(s_1)$, and $b_1(\overline{b}_1) = \mathit{head}(s_1')$, it follows that $b_{n_b}(\overline{b}_{n_b}) = c_{n_b}(\overline{c}_{n_b})$.

\para{Induction Step}
Assume that the claim holds for all $d' < d$.
We now prove that $s_{n_b-d}'' = s_{n_b-d}'$ and $c_{n_b - d}(\overline{c}_{n_b - d}) = b_{n_b - d}(\overline{b}_{n_b - d})$ hold as well.
From the induction's hypothesis, it follows that $s_{n_b-d+1}'' = s_{n_b-d+1}'$.
From this and $\mathit{gg}$'s definition, it follows that $\mathit{head}(s_{n_b - d}'', i_{n_b - d}'') = \mathit{head}(s_{n_b - d}', i_{n_b - d}')$.
From this, $r_{n_b - d}'' = r_{n_b - d}'$, $i_{n_b - d}'' = i_{n_b - d}'$, $r_{n_b - d}''$ is strongly connected for ${\cal U}$, $s_{n_b - d}'' \in \mathit{ground}(p,r_{n_b - d}'')$,  $s_{n_b - d}' \in \mathit{ground}(p,s_{n_b - d})$, and \thref{theorem:join:tree:1}, it follows that $s_{n_b - d}d'' = s_{n_b - d}'$.
From this, $c_{n_b - d}(\overline{c}_{n_b - d}) = \mathit{head}(s_{n_b - d}'')$, and $b_{n_b - d}(\overline{b}_{n_b - d}) = \mathit{head}(s_{n_b - d}')$, it follows that $c_{n_b - d}(\overline{c}_{n_b - d}) = b_{n_b - d}(\overline{b}_{n_b - d})$.

\item $D_b \neq D_c$. 
From this and $(D_b, D_c)$ is tail-guarded, it follows that there is an annotation $\disjoint{\mathit{pr}}{\mathit{pr}'} \in {\cal D}$, a set $K \subseteq \{1, \ldots, |a|\}$, and a bijection $\nu : K \to \{1, \ldots, |\mathit{pr}|\}$, such that $D_b$ $\nu$-upward links to $\mathit{pr}$ and $D_c$ $\nu$-upward links to $\mathit{pr}'$.
From  $D_b$ $\nu$-upward links to $\mathit{pr}$, $P_b$ is a path in the ground graph corresponding to $D_b$, and \thref{theorem:upward:connections}, it follows that a positive literal $\mathit{pr}(\overline{v})$ in the body of one of the ground rules such that $b_{n_b}(\overline{b}_{n_b})(k) = \overline{v}(\nu(K))$ for any $k \in K$.
From this and the definition of $\mathit{ground}$, it follows that $\mathit{pr}(\overline{v}) \in \mathit{ground}(p)$.
From  $D_c$ $\nu$-upward links to $\mathit{pr}'$, $P_c$ is a path in the ground graph corresponding to $D_c$, and \thref{theorem:upward:connections}, it follows that a positive literal $\mathit{pr}'(\overline{v}')$ in the body of one of the ground rules such that $c_{n_c}(\overline{c}_{n_c})(k) = \overline{v}'(\nu(K))$ for any $k \in K$.
From this and the definition of $\mathit{ground}$, it follows that $\mathit{pr}'(\overline{v}') \in \mathit{ground}(p)$.
Finally, from the fact that $P_b$ and $P_c$ are tail-connected, it follows that $b_{n_b}(\overline{b}_{n_b}) = c_{n_c}(\overline{c}_{n_c})$.
From this, $b_{n_b}(\overline{b}_{n_b})(k) = \overline{v}(\nu(K))$ for any $k \in K$, and $c_{n_c}(\overline{c}_{n_c})(k) = \overline{v}'(\nu(K))$ for any $k \in K$, it follows that $\overline{v} = \overline{v}'$.
Therefore,  $\mathit{pr}(\overline{v}) \in \mathit{ground}(p)$ and $\mathit{pr}'(\overline{v}) \in \mathit{ground}(p)$.
This contradicts \thref{theorem:template:disjointness:semantics}, $\disjoint{\mathit{pr}}{\mathit{pr}'} \in {\cal D}$, and $p$ complies with ${\cal D}$.
\end{compactenum}
\end{compactenum}
This completes the proof of our claim.
\end{proof}

\subsubsection{Auxiliary Results}

Here we prove two auxiliary results that help in establishing that programs are acyclic. % by reducing common patterns to simpler patterns.
In particular, \thref{theorem:implication:between:type-1:structures} states pre-conditions that allows reducing the guardedness of a complex undirected structure to the guardedness of simpler structures.
Similarly, \thref{theorem:implication:between:type-2:structures} states pre-conditions that allows reducing  the guardedness of a complex directed structure to the guardedness of a sequence of simpler structures.
Note that \thref{theorem:implication:between:type-2:structures} can be easily extended to support (1) different forms of cycle combination, and (2) combinations of non-self-loop cycles.

\begin{proposition}\thlabel{theorem:implication:between:type-1:structures}
Let $p$ be a \problog{} program, ${\cal O}$ be an ordering template, ${\cal D}$ be a disjointness template, ${\cal U}$ be a uniqueness template, $C$ be an undirected cycle in $\graph(p)$, and $S = \langle D_1, D_2, D_3, U\rangle$ be an undirected unsafe structure.
If 
(1) $p$ complies with ${\cal O}$, ${\cal D}$, and ${\cal U}$,
(2) $C$ is equivalent to $D_1, U_1', U, U_2', D_3,D_2$,
(3) $S$ is $({\cal O}, {\cal D}, {\cal U})$-guarded,
then there is an undirected unsafe structure that covers $C$ and is  $({\cal O}, {\cal D}, {\cal U})$-guarded.
\end{proposition} 

\begin{proof}
Let $p$ be a \problog{} program, ${\cal O}$ be an ordering template, ${\cal D}$ be a disjointness template, ${\cal U}$ be a uniqueness template, $C$ be an undirected cycle in $\graph(p)$, and $S = \langle D_1, D_2, D_3, U\rangle$ be an undirected structure.
We assume that (1) $p$ complies with ${\cal O}$, ${\cal D}$, and ${\cal U}$,
(2) $C$ is equivalent to $D_1, U_1', U, U_2', D_3,D_2$,
(3) $S$ is $({\cal O}, {\cal D}, {\cal U})$-guarded.
We define the undirected unsafe structure $S' = \langle D_1, D_2, D_3, U_1', U, U_2'\rangle$.
From our assumption $C$ is equivalent to $D_1, U_1', U, U_2', D_3,D_2$.
Thus, $S'$ covers $C$.
Furthermore, since $S$ and $S'$ agree on all directed paths, which are the only ones that determine whether an undirected structure is guarded, it follows that $S'$ is  $({\cal O}, {\cal D}, {\cal U})$-guarded
\end{proof}

In the following, we say that a directed cycle is guarded by a set of predicates $O$ and a bijection $\nu$ iff the definition of guarded cycle is satisfied for the specific bijection $\nu$.

\begin{proposition}\thlabel{theorem:implication:between:type-2:structures}
Let $p$ be a \problog{} program, $C_1 = \mathit{pr}_1 \xrightarrow{r_1, i_1} \ldots \xrightarrow{r_{n-1}, i_{n-1}} \mathit{pr}_n \xrightarrow{r_n, i_n} \mathit{pr}_1$  be a directed cycle in $\graph(p)$, $O$ be a set of predicate symbols, $\mathit{pr}_1 \xrightarrow{r_{n+1}, i_{n+1}} pr_1$, and $o$ be a predicate symbol such that for all $o' \in O$, $|o| = |o'|$, $K$ be a  non-empty set $K \subseteq \{1, \ldots, |\mathit{pr}_1|\}$, and a bijection $\nu : K \to \{1, \ldots, \sfrac{|\mathit{o}|}{2}\}$.
If 
(1) $C_1$ is guarded for $O_1$ and $\nu$, and
(2) $\mathit{pr}_1 \xrightarrow{r_{n+1}, i_{n+1}} pr_1$ is guarded for $o$ and $\nu$, 
%(3) $\mathit{pr}_{n-1} \xrightarrow{r_{{y_{k+1}}-1}, i_{{y_{k+1}}-1}} pr_{n}$ $\nu'$-upward connects to $o_{k+1}(\overline{x}_{k+1})$
then $C_1, \mathit{pr}_1 \xrightarrow{r_{n+1}, i_{n+1}} pr_1$ is guarded for $O_1 \cup \{o\}$.
\end{proposition} 

\begin{proof}
Let $p$ be a \problog{} program, $C_1 = \mathit{pr}_1 \xrightarrow{r_1, i_1} \ldots \xrightarrow{r_{n-1}, i_{n-1}} \mathit{pr}_n \xrightarrow{r_n, i_n} \mathit{pr}_1$  be a directed cycle in $\graph(p)$, $O$ be a set of predicate symbols, $\mathit{pr}_1 \xrightarrow{r_{n+1}, i_{n+1}} pr_1$, and $o$ be a predicate symbol such that for all $o' \in O$, $|o| = |o'|$, $K$ be a  non-empty set $K \subseteq \{1, \ldots, |\mathit{pr}_1|\}$, $\nu$ be a bijection $\nu : K \to \{1, \ldots, \sfrac{|\mathit{o}|}{2}\}$ and $\nu'$ be the bijection $\nu'(i) = \nu(x) + \sfrac{|\mathit{o}|}{2}$ for all $1 \leq i \leq \sfrac{|\mathit{o}|}{2}$.
Furthermore, we assume that (1) $C_1$ is guarded for $O_1$ and $\nu$, and (2) $\mathit{pr}_1 \xrightarrow{r_{n+1}, i_{n+1}} pr_1$ is guarded for $o$ and $\nu$.

First, we rewrite $C_1, \mathit{pr}_1 \xrightarrow{r_{n+1}, i_{n+1}} pr_1$ as $\mathit{pr}_1 \xrightarrow{r_1, i_1} \ldots \xrightarrow{r_{n-1}, i_{n-1}} \mathit{pr}_n \xrightarrow{r_n, i_n} \mathit{pr}_{n+1} \xrightarrow{r_{n+1}, i_{n+1}} \mathit{pr}_1$.
From $C_1$ is guarded for $O_1$, it follows that there are integers $1 \leq y_1 < y_2 < \ldots < y_e \leq n$ such that $y_e = n$ and literals $o_1(\overline{x}_1), \ldots, o_e(\overline{x}_e)$ (where $o_j \in O$ and $|\overline{x}_j| = |o_j|$), such that 
for each $0 \leq k < e$, 
(1) $\mathit{pr}_{y_k} \xrightarrow{r_{y_k}, i_{y_k}} \ldots \xrightarrow{r_{{y_{k+1}}-1}, i_{{y_{k+1}}-1}} pr_{y_{k+1}}$ $\nu$-downward connects to $o_{k+1}(\overline{x}_{k+1})$, and 
(2) $\mathit{pr}_{y_{k+1}-1} \xrightarrow{r_{{y_{k+1}}-1}, i_{{y_{k+1}}-1}} pr_{y_{k+1}}$ $\nu'$-upward connects to $o_{k+1}(\overline{x}_{k+1})$, 
where $y_0 = 1$.
Furthermore, from $\mathit{pr}_1 \xrightarrow{r_{n+1}, i_{n+1}} pr_1$ is guarded for $o$ and $\nu$, it follows that there is a literal $o(\overline{x})$ such that
(3) $\mathit{pr}_{n+1} \xrightarrow{r_{n+1}, i_{n+1}} pr_1$ $\nu$-downward connects to $o(\overline{x})$, and
(4) $\mathit{pr}_{y_{n+1}} \xrightarrow{r_{n+1}, i_{n+1}} pr_{y_1}$ $\nu'$-upward connects to $o_{k+1}(\overline{x}_{k+1})$.
From (1)--(4), it therefore follows that $\mathit{pr}_1 \xrightarrow{r_1, i_1} \ldots \xrightarrow{r_{n-1}, i_{n-1}} \mathit{pr}_n \xrightarrow{r_n, i_n} \mathit{pr}_{n+1} \xrightarrow{r_{n+1}, i_{n+1}} \mathit{pr}_1$ is guarded for $\nu$ and $O \cup \{o\}$.
Indeed, there are integers $1 \leq y_1 < y_2 < \ldots < y_e < y_{e+1} \leq n+1$ such that $y_{e+1} = n+1$ and literals $o_1(\overline{x}_1), \ldots, o_e(\overline{x}_e),  o_{e+1}(\overline{x}_{e+1})$ (where $o_j \in O \cup \{o\}$ and $|\overline{x}_j| = |o_j|$) such that for each $0 \leq k < e+1$, 
(1) $\mathit{pr}_{y_k} \xrightarrow{r_{y_k}, i_{y_k}} \ldots \xrightarrow{r_{{y_{k+1}}-1}, i_{{y_{k+1}}-1}} pr_{y_{k+1}}$ $\nu$-downward connects to $o_{k+1}(\overline{x}_{k+1})$, and 
(2) $\mathit{pr}_{y_{k+1}-1} \xrightarrow{r_{{y_{k+1}}-1}, i_{{y_{k+1}}-1}} pr_{y_{k+1}}$ $\nu'$-upward connects to $o_{k+1}(\overline{x}_{k+1})$. 
\end{proof}

\subsection{Relaxed Acyclic Programs}

Acyclic \problog{} programs as defined above are very restrictive.
For instance, they cannot encode rules like those in Section~\ref{sect:language}. % and Appendix~\ref{app:examples}.
We now design \emph{relaxed acyclic programs}, a larger fragment of \problog{} where these examples can be encoded.

The key components of this new fragment are two syntactic transformations over \problog{} programs that, while not preserving the program's semantics, preserve certain key aspects of the program's structure.
Intuitively, a \problog{} program $p$ is a \emph{relaxed acyclic} program if the program obtained from $p$ by applying the transformations is an acyclic \problog{} program.

Finally, we develop a procedure for compiling any relaxed acyclic  \problog{} program into a poly-tree Bayesian Network.
Since any acyclic program is a relaxed acyclic program as well, this procedure can be applied also to acyclic programs.

\subsubsection{Rule Domination}

Rule $r_1$ is \emph{dominated} by rule $r_2$, written $r_1 \sqsubseteq r_2$, iff:
\begin{compactitem}
\item $\mathit{head}(r_1) = \mathit{head}(r_2)$,
\item $\mathit{cstr}(r_1) = \mathit{cstr}(r_2)$, 
\item for all $1 \leq i \leq |\mathit{body}(r_1)|$, then  $\mathit{pred}(\mathit{body}(r_1,i)) = \mathit{pred}(\mathit{body}(r_2,i))$ and $\mathit{args}(\mathit{body}(r_1,i)) = \mathit{args}(\mathit{body}(r_2,i))$.
\end{compactitem}
We extend the domination relation also to atoms and probabilistic atoms as follows:
$a_1(\overline{c}_1) \sqsubseteq a_2(\overline{c}_1)$ iff $a_1 = a_2$ and $\overline{c}_1 = \overline{c}_2$, and 
$\atom{v_1}{a_1(\overline{c}_1)} \sqsubseteq \atom{v_2}{a_2(\overline{c}_1)}$ iff $v_1 = v_2$, $a_1 = a_2$, and $\overline{c}_1 = \overline{c}_2$. 

Given a program $p$ and a rule $r \in p$, $[r]_{\sqsubseteq p}$ denotes the set $\{r' \in p \mid r' \sqsubseteq r\}$.
We say that a rule $r$  is \emph{maximal} in a program $p$ iff there is no rule $r' \in p$ such that $r \sqsubseteq r'$.
The \emph{kernel of $[r]_{\sqsubseteq p}$}, denoted $k([r]_{\sqsubseteq p})$, is the rule $r'$ defined as follows:
$\mathit{head}(r') = \mathit{head}(r)$, 
$|\mathit{body}(r')| = |\mathit{body}(r)|$, 
$\mathit{cstr}(r') = \mathit{cstr}(r)$, 
and for all $1 \leq i \leq |\mathit{body}(r')|$, then
 $\mathit{body}(r',i) = a_i(\overline{c}_i)$ if for all rules $r'' \in  [r]_{\sqsubseteq p}$ (1) $|\mathit{body}(r'')| \geq i$ and (2) $\mathit{body}(r'', i)$ is a positive literal and 
$\mathit{body}(r',i) = \neg a_i(\overline{c}_i)$ otherwise, where $a_i = \mathit{pred}(\mathit{body}(r,i))$ and $\overline{c}_i = \mathit{args}(\mathit{body}(r,i))$.

The \emph{maximal projection} of $p$, denoted $p \Downarrow_{\sqsubseteq}$, is $\{ k([r]_{\sqsubseteq p}) \mid \neg\exists r' \in p.\ r \sqsubseteq r' \}$.
A \emph{maximal partition} $\mathit{PP}$ of $p$ is a subset of $\mathbb{P}(p)$ such that:
(1) for each $R \in \mathit{PP}$, $R = [r]_{\sqsubseteq p}$ for some maximal rule $r \in p$, and
(2) for each $R_1, R_2 \in \mathit{PP}$, $R_1 \cap R_2 = \emptyset$.
A maximal partition $PP$ induces a \emph{maximal assignment} $\nu_{\mathit{PP}}: p \to \mathit{PP}$ such that $\nu_{\mathit{PP}}(r) = R$, where $R \in \mathit{PP}$ and $r \in R$ if $r$ is maximal, and $\nu_{\mathit{PP}}(r) = \emptyset$ otherwise.

The syntactic transformation $\alpha$, which removes all non-maximal rules, takes as input a \problog{} program $p$ and returns as its maximal projection $p \Downarrow_{\sqsubseteq}$.

\subsubsection{CPT-like predicates}

We now present CPT-like predicates, a special kind of probabilistic structure that can be formalized using annotated disjunctions.
While general annotated disjunctions cannot be used in our encoding, since they introduce, in general, cycles in the underlying Bayesian Network, CPT-like predicates can be still encoded as poly-trees.
Let $\Sigma$ be a first-order signature, $\mathbf{dom}$ be a finite domain, $\mathit{pr}$ be a predicate symbol in $\Sigma$ of arity $|\mathit{pr}|$, $p$ be a $(\Sigma, \mathbf{dom})$-\problog{} program, and $K$ be a set of distinct integer values such that for all $i \in K$, $1 \leq i \leq |\mathit{pr}|$.
Given a tuple $\overline{t}$, we denote by $\overline{t} \downarrow_{\{i_1, \ldots, i_n\}}$ the tuple $(\overline{t}(i_1), \ldots, \overline{t}(i_{n}))$.

We say that $\mathit{pr}$ is \textit{$K$-fixed domain with respect to $p$} iff for all rules $r \in p$, the following conditions hold:
\begin{compactitem}
\item if $\mathit{pred}(\mathit{head}(r)) = \mathit{pr}$, then $\mathit{args}(\mathit{head}(r))\downarrow_K$ is a tuple in $\mathbf{dom}^{|K|}$, and
\item if $\mathit{pred}(\mathit{body}(r, j)) = \mathit{pr}$, for some $1 \leq j \leq |\mathit{body}(r)|$, then  $\mathit{args}(\mathit{body}(r, j))\downarrow_K$ is a tuple in $\mathbf{dom}^{|K|}$.
\end{compactitem}
We denote by $\mathit{pdom}(\mathit{pr}, p, K)$ the set of $|K|$-tuples representing all constant values associated with the positions in $K$.
Namely, $\mathit{pdom}(\mathit{pr}, p, K) = \{ \mathit{args}(\mathit{head}(r))\downarrow_K \mid r \in p  \} \cup \{ \mathit{args}(\mathit{body}(r, j))\downarrow_K \mid r \in p \wedge 1 \leq i \leq |\mathit{body}(r)|\}$.
Finally, given an atom $a(\overline{x})$ and a $K$, we denote by $\beta(a(\overline{x}), K)$ the atom $a(\overline{x}\downarrow_K)$.
A set of rules $r_1, \ldots, r_n$ in $p$ is an \emph{annotated disjunction-set} iff they are the translation of an annotated disjunction to plain \problog{}.
Formally, a  set of rules $r_1, \ldots, r_n$ in $p$ is an \emph{annotated disjunction-set} iff there are a set of predicate symbols $\mathit{sw}_1, \ldots, \mathit{sw}_n$ and a sequence of literals $L$ such that 
(1) for $1 \leq i \leq n$, $\mathit{body}(r_i) = L, \neg \mathit{sw}_1(\overline{x}_1), \ldots, \neg \mathit{sw}_{i-1}(\overline{x}_{i-1}), \mathit{sw}_i(\overline{x}_i)$, where $\overline{x}_j = \mathit{args}(\mathit{head}(r_j))$ for $ 1 \leq j \leq n$, and
(2) for $1 \leq i \leq n$, $\mathit{sw}_i$ is used only to define probabilistic ground atoms such that if $\atom{v_1}{\mathit{sw}_i(\overline{c}_1}) \in p$ and $\atom{v_2}{\mathit{sw}_i(\overline{c}_2}) \in p$, then $v_1 = v_2$, and
(3) $\Sigma_{1 \leq i \leq n} p(i) \leq  1$, where $p(i) = v_i \cdot (1 - \Sigma_{1 \leq  j < i} p(j))$, where $v_i$ is the probability associated to the predicate symbol $\mathit{sw}_i$ in the program $p$.
We denote by (1) $\mathit{common}(R)$, where $R$ is an annotated disjunction-set, the list of literals that are common to all the rules in $R$, namely $\mathit{common}(R) = \bigcap_{r \in R} \mathit{body}(R)$ (with a slight abuse of notation, we use intersection for lists), by (2) $\mathit{heads}(R)$ the set containing the various heads of the rules in $R$, namely $\mathit{heads}(R) = \bigcup_{r \in R} \{\mathit{head}(r)\}$, and by (3) $\mathit{common}(R, \mathit{pr})$ the list of all literals in $\mathit{common}(R)$ whose predicate symbol is $\mathit{pr}$.

We say that an annotated disjunction set $R$ is an \emph{$\mathit{pr}$-annotated disjunction-set} iff for all rules $r \in R$, $\mathit{pred}(\mathit{head}(r)) = \mathit{pr}$.
Let $R_1$ and $R_2$ be  two $\mathit{pr}$-annotated disjunction-sets.
We say that $R_1$ and $R_2$ are $(\mathit{pr}_1, \mathit{pr}_2)$-disjoint iff  $\mathit{pr}_1(\overline{x}_1\downarrow_{\{1, \ldots, |\mathit{pr}|\} \setminus K}) \in \mathit{common}(R_1)$ and $\mathit{pr}_2(\overline{x}_2\downarrow_{\{1, \ldots, |\mathit{pr}|\} \setminus K}) \in \mathit{common}(R_2)$, $\mathit{pr}(\overline{x}_1) \in \mathit{heads}(R_1)$, and $\mathit{pr}(\overline{x}_2) \in \mathit{heads}(R_2)$.
We say that $R_1$ and $R_2$ are $(\mathit{pr},K)$-row-distinct iff 
(1) all literals in  $\mathit{common}(R_1, \mathit{pr})$ are positive,
(2) all literals in  $\mathit{common}(R_2, \mathit{pr})$ are positive,
(3) $|\mathit{common}(R_1, \mathit{pr})| = |\mathit{common}(R_2, \mathit{pr})|$, and
(4) $\{ \overline{x}_1\downarrow_{K} \mid \exists l \in \mathit{common}(R_1,\mathit{pr}).\, \overline{x}_1 = \mathit{args}(l)\} \neq \{ \overline{x}_2\downarrow_{K} \mid \exists l \in \mathit{common}(R_2,\mathit{pr}).\, \overline{x}_2 = \mathit{args}(l)\}$.

We say that $\mathit{pr}$ is \textit{$(K, {\cal D}, {\cal U})$-CPT-like} in a program $p$ iff 
\begin{compactitem}
\item for all rules $r \in p$ such that $\mathit{pred}(\mathit{head}(r)) = \mathit{pr}$, (1) $r$ is strongly connected for ${\cal U}$, and (2) $\mathit{body}(r) \neq \emptyset$,
\item $\mathit{pr}$ is $K$-fixed domain with respect to $p$,
\item for all maximal $\mathit{pr}$-annotated disjunction-sets $R \subseteq p$, for all $\mathit{pr}(\overline{x}) \in \mathit{heads}(R)$,  $\overline{x}\downarrow_{K} \in \mathit{pdom}(\mathit{pr}, p, K)$ for all $1 \leq i \leq n$,
\item for all maximal $\mathit{pr}$-annotated disjunction-sets $R \subseteq p$, $|\mathit{heads}(R)| = |\mathit{pdom}(\mathit{pr}, p, K)|$,
\item for all maximal $\mathit{pr}$-annotated disjunction-sets $R \subseteq p$, for all distinct $\mathit{pr}(\overline{x}_1), \mathit{pr}(\overline{x}_2) \in \mathit{heads}(R)$, $\overline{x}_1\downarrow_{K} \neq \overline{x}_2\downarrow_{K}$ and $\overline{x}_1\downarrow_{\{1, \ldots, |\mathit{pr}|\} \setminus K} = \overline{x}_2\downarrow_{\{1, \ldots, |\mathit{pr}|\} \setminus K}$,
\item for any two distinct and maximal $\mathit{pr}$-annotated disjunction-sets $R_1, R_2 \subseteq p$, one of the following conditions hold:
\begin{compactitem}
\item there is pair of predicate symbols $(\mathit{pr}_1, \mathit{pr}_2) \in {\cal D}$ such that  $R_1$ and $R_2$ are $(\mathit{pr}_1, \mathit{pr}_2)$-disjoint, or 
\item $R_1$ and $R_2$ are $(\mathit{pr},K)$-row-distinct. 
\end{compactitem}
\item There is no rule $r \in p$ such that (1) $\mathit{pred}(\mathit{head}(r)) = \mathit{pr}$, and (2) $r$ is not in a maximal $\mathit{pr}$-annotated disjunction set.
\end{compactitem}

%%%%%%%%%%%%% DEFINE SYNTACTIC TRANSFORMATION %%%%%%%%%%%%%%%
Let $\mathit{pr}$ be a predicate symbol and $K$ be a set of natural numbers.
The syntactic transformation $\Downarrow_{\mathit{pr},K}$ is defined as follows:
\begin{compactitem}
\item $(\mathit{pr}(\overline{x}))\Downarrow_{\mathit{pr},K}$ is $\mathit{pr}(\overline{x}\downarrow_{\{1, \ldots, |\overline{x}|\} \setminus K})$, 
\item $(\mathit{pr}'(\overline{x}))\Downarrow_{\mathit{pr},K} $ is $ \mathit{pr}'(\overline{x})$, where $\mathit{pr} \neq \mathit{pr}'$, 
\item $(\atom{v}{a(\overline{x})})\Downarrow_{\mathit{pr},K} $ is $ \atom{v}{(a(\overline{x})\Downarrow_{\mathit{pr},K})}$,
\item $(\neg a(\overline{x}))\Downarrow_{\mathit{pr},K} $ is $ \neg (a(\overline{x})\Downarrow_{\mathit{pr},K})$, and
\item $(h \leftarrow l_1, \ldots, l_n, c_1, \ldots, c_n) \Downarrow_{\mathit{pr},K} $ is $h\Downarrow_{\mathit{pr},K} \leftarrow l_1\Downarrow_{\mathit{pr},K}, \ldots, l_n\Downarrow_{\mathit{pr},K}, c_1, \ldots, c_n$, and
\item $p\Downarrow_{\mathit{pr},K}$, where $p$ is a program, is $\{r\Downarrow_{\mathit{pr},K} \mid r \in p\}$.
\end{compactitem}
Finally, we define the transformation $\Downarrow_{(\mathit{pr}_1,K_1), \ldots, (\mathit{pr}_n,K_n)}$ that is just $\Downarrow_{\mathit{pr}_1,K_1} \circ  \ldots \circ \Downarrow_{\mathit{pr}_n,K_n}$.

Given a program $p$, the function $\mu_{p,{\cal D}, {\cal U}}$ associates to each predicate symbol $\mathit{pr}$ in $\Sigma$ the maximal set $K$ such that $\mathit{pr}$ is $(K, {\cal D}, {\cal U})$-CPT-like with respect to $p$.
The syntactic transformation $\beta_{p,\cal D, U}$ takes as input a \problog{} program $p'$ and it returns as output the program $p'\Downarrow_{(\mathit{pr}_1, \mu_{p,{\cal D}, {\cal U}}(\mathit{pr}_1)), \ldots, (\mathit{pr}_n, \mu_{p,{\cal D}, {\cal U}}(\mathit{pr}_n))}$, where $\mathit{pr}_1, \ldots, \mathit{pr}_n$ are all predicate symbols in $\Sigma$.

Note that $p'$ is defined on a different vocabulary than $p$. 
The program $p'$ is defined over the ``reduced vocabulary'' obtained by applying $\beta$ to the atoms in $p$.

\subsubsection{Relaxed Acyclic \problog{} programs}

A \emph{relaxed acyclic $(\Sigma, \mathbf{dom})$-\problog{} program} is a $(\Sigma, \mathbf{dom})$-\problog{} program $p$ such that there are a $\Sigma$-ordering template ${\cal O}$, a $\Sigma$-disjointness template ${\cal D}$, and a $\Sigma$-uniqueness template ${\cal U}$ such that $\alpha(\beta_{p,\cal D,  U}(p))$ is $({\cal O}, {\cal D}, {\cal U})$-acyclic.
An \emph{aciclicity witness for a program $p$} is a tuple $({\cal O}, {\cal D}, {\cal U})$ such that $\alpha(\beta_{p,\cal D,  U}(p))$ is a $({\cal O}, {\cal D}, {\cal U})$-acyclic \problog{} program. 

Given a relaxed acyclic $(\Sigma, \mathbf{dom})$-\problog{} program $p$, we associate to each rule $r$ in $\alpha(\beta_{p,\cal D,  U}(p))$, the set $[r]_p$ of rules in $p$ defined as follows:
\[
[r]_p := \bigcup_{r' \in [r]_{\sqsubseteq \beta_{p,\cal D,  U}(p)}} \{r'' \in p \mid \beta_{p,\cal D,  U}(r'') = r' \} 
\]
The set $[r]_p$ contains all rules in $p$ that are represented by the rule $r$ in $\alpha(\beta_{p,\cal D,  U}(p))$.

\subsubsection{Compilation to Bayesian Networks}

Given a relaxed acyclic \problog{} program $p$, its encoding as a Bayesian Network $\mathit{BN} = (N,E,\mathit{cpt})$, denoted $\mathit{bn}(p)$, is defined by Algorithm~\ref{figure:acyclic:encoding:algorithm}.

\begin{algorithm*}[!hbtp]

    \begin{multicols}{2}
\DontPrintSemicolon
\KwIn{A relaxed acyclic $(\Sigma,\mathbf{dom})$-\problog{} program $p_0$ and a witness $({\cal O}, {\cal D}, {\cal U})$ for $p$.}
\KwOut{A Bayesian Network $(N,E,\mathit{cpt})$.}
\SetKwComment{Comment}{$\triangleright$}{}
%\Begin{
  \Comment{Transform the original program.}
  $p = \alpha(\beta_{p_0,\cal D,  U}(p_0))$\;
  
  \Comment{Initialize the domain of each predicate symbol.}
  $D = \emptyset$\;
  \For{$\mathit{pr} \in \Sigma$}
  {
  	\If{$|\mu_{p_0,\cal D, U}(\mathit{pr})|>0$}
  	{
  		$D(\mathit{pr}) = \mathit{pdom}(\mathit{pr}, p_0, \mu_{p_0,\cal D, U}(\mathit{pr})) \cup \{\bot\}$\;
  	}
  	\Else
  	{
  		$D(\mathit{pr}) = \{\top,\bot\}$\;
  	}
  }

  \Comment{Initialize the BN.}
  $N = \emptyset$\;
  $E = \emptyset$\;
  $\mathit{cpt} = \emptyset$\;

  \Comment{Creates the nodes.}
  \For{$r \in p$}  
  {
  	\For{$r' \in \mathit{ground}(p,r)$}
	{
			$N = N \cup \{X[\mathit{head}(r')]\}$\;
			\For{$a(\overline{c}) \in \mathit{body}^+(r')$}
			{
				$N = N \cup \{X[a(\overline{c})]\}$\;
			}
			\For{$\neg a(\overline{c}) \in \mathit{body}^-(r')$}
			{
				$N = N \cup \{X[a(\overline{c})]\}$\;
			}
	}
  }

  \For{$r \in p$}
  {
  	$\mu = \emptyset$\;
	\For{$s \in \mathit{ground}(p,r)$}
	{
			$\mu(\mathit{head}(s)) = \mu(\mathit{head}(s)) \cup \{ \mathit{pos}(s)\}$\;
	}

	\For{$a(\overline{c}) \in \mathit{domain}(\mu)$}
  	{
  		$N = N \cup \{X[r,a(\overline{c})]\}$\;
  		$E = E \cup \{X[r,a(\overline{c})] \to  X[a(\overline{c})] \}$\;
  		$K = \emptyset$\;
  			\For{$I \in \mu(a(\overline{c}))$}{  			
  				$K = K \cup \{X[r,I,a(\overline{c})]\}$\;
				\If{$I = \emptyset$}
		  		{
					\Comment{Here, $[r]_{p_0} = \{r\}$ and $D(a) = \{\top,\bot\}$.}	  			
		  			\If{$\exists v.\ \mathit{head}(r) = \atom{v}{a(\overline{c})}$}
		  			{
		  				$p = v$\;
		  			}
		  			\Else{
		  				$p = 1$\;
		  			}
		  			$\mathit{cpt}(X[r,\emptyset,a(\overline{c})]) = (\top \mapsto p,$\; 
		  			$\qquad \bot \mapsto (1-p))$\;
		  		}
		  		\Else
		  		{  				
		  			\For{$b(\overline{v}) \in I$}
		  			{
		  				$E = E \cup \{X[b(\overline{v})] \to X[r,I,a(\overline{c})] \}$\;
		  			}
		  			\For{$\neg b(\overline{v}) \in I$}
		  			{
		  				$E = E \cup \{X[b(\overline{v})] \to X[r,I,a(\overline{c})] \}$\;
		  			}
		  			$\mathit{cpt}(X[r, I, a(\overline{c})]) = \mathit{CPT}(I, r, D, \mu_{p_0,\cal D, U}, p_0)$\;
  				}  				
  			}
  			$(N',E',\mathit{cpt}') = \mathit{tree}(K, X[r,a(\overline{c})],D(a))$\;
  			$N = N \cup N'$\;
  			$E = E \cup E'$\;
  			$\mathit{cpt} = \mathit{cpt} \cup \mathit{cpt}'$\;
  	}
  }

  \Comment{Set the CPT for the variables associated to the atoms.}
  \For{$X[a(\overline{c})] \in N$}
  {
  	$\mathit{cpt}(X[a(\overline{c})]) = \mathit{CPT}_{\oplus}(|\{X[r,a(\overline{c})] \in N \}|+1, D(a))$\;
  }
  \Return{$(N,E,\mathit{cpt})$}
%}
\end{multicols}
\caption{Constructing the Bayesian Network. The sub-routines \emph{tree}, $\mathit{CPT}$, and $\mathit{CPT}_\oplus$ are shown in Algorithm~\ref{figure:acyclic:encoding:algorithm:auxiliary}.}
\label{figure:acyclic:encoding:algorithm}
\end{algorithm*}

{
\LinesNumberedHidden
\begin{algorithm*}[!hbtp]
\SetKwProg{Fn}{function}{}{}
\DontPrintSemicolon
\SetKwComment{Comment}{$\triangleright$}{}
\begin{multicols}{2}
	\Fn{$\mathit{CPT}((a_1(\overline{c}_1), \ldots, a_n(\overline{c}_n)),r, D, \mu, p_0)$}{
	$A = D(a_1) \times \ldots \times D(a_n) \times D(\mathit{pred}(\mathit{head}(r)))$\;
	$\mathit{cpt} = \emptyset$\;
	\For {$\overline{a} \in A$}{
		$\overline{a}' = \mathit{removeDuplicates}((a_1(\overline{c}_1), \ldots, a_n(\overline{c}_n)),\overline{a})$\;
		\If{$\mathit{satisfiable}(r,\overline{a},D, \mu, p_0) \wedge \forall 1 \leq i \leq n. \forall 1 \leq j \leq n.\ (i\neq j  \wedge a_i(\overline{c}_i) = a_j(\overline{c}_j) \Rightarrow \overline{a}(i) = \overline{a}(j))$}
		{
			$\mathit{cpt} = \mathit{cpt} \cup \{\overline{a}' \mapsto 1\}$\;
		}
		\Else{
			$\mathit{cpt} = \mathit{cpt} \cup \{\overline{a}' \mapsto 0\}$\;
		}
		
	}
	\Return{$\mathit{cpt}$}
	}
	\;
	\Fn{$\mathit{satisfiable}(r,(v_1, \ldots, v_n),D, \mu, p_0)$}{
		\If{$v_n \neq \bot$}
		{
			\Comment{At least one satisfiable assignment.}
			$\mathit{res} = \bot$\;
			\For{$r' \in [r]_{p_0}$}{
				\If{$\mathit{filter}(\mathit{head}(r'),D,\mu) = v_n$}{
					$\mathit{sat} = \top$\;
					\For{$1 \leq i \leq |\mathit{body}(r')|$}{
						\If{$v_i \not\in \mathit{filter}(\mathit{body}(r',i),D,\mu)$}{
							$\mathit{sat} = \bot$\;
						}
					}
					$\mathit{res} = \mathit{res} \vee \mathit{sat}$\;
				}
			}
			\Return{$\mathit{res}$}\;
		}	
		\If{$v_n = \bot$}
		{
			\Comment{All assignments must be unsatisfiable.}
			$\mathit{res} = \bot$\;
			\For{$r' \in [r]_{p_0}$}{
				%\If{$\mathit{filter}(\mathit{head}(r'),D,\mu) = v_n$}{
					$\mathit{sat} = \top$\;
					\For{$1 \leq i \leq |\mathit{body}(r')|$}{
						\If{$v_i \not\in \mathit{filter}(\mathit{body}(r',i),D,\mu)$}{
							$\mathit{sat} = \bot$\;
						}
					}
					$\mathit{res} = \mathit{res} \vee \mathit{sat}$\;
				%}
			}
			\Return{$\neg \mathit{res}$}\;
		}
	}
	\;
	\Fn{$\mathit{filter}(l,D,\mu)$}{	
		\If{$\exists a \in \Sigma, \overline{x} \in (\mathit{Var} \cup \mathbf{dom})^{|a|}.\, l = a(\overline{c})$}{
			\If{$\mu(a) \neq \emptyset$}{
				\Return{$\{x\downarrow_{\mu(a)}\}$}\;
			}
			\Else{
				\Return{$\{\top\}$}\;
			}
		}
		\If{$\exists a \in \Sigma, \overline{x} \in (\mathit{Var} \cup \mathbf{dom})^{|a|}.\, l = \neg a(\overline{c})$}{
			\If{$\mu(a) \neq \emptyset$}{
				\Return{$(D(a) \setminus \{\mu(a)\}) \cup \{\bot\}$}\;
			}
			\Else{
				\Return{$\{\bot\}$}\;
			}
		}
	}
	\;
	\Fn{$\mathit{CPT}_{\oplus}(n, D)$}{
		\Comment{If $D \neq \{\top, \bot\}$, then $n=2$ (since the program is relaxed-acyclic).}
		\If{$D \neq \{\top, \bot\} \wedge n = 2$}
		{
			$\mathit{cpt} = \emptyset$\;
			\For{$(v_1, v_2) \in D^2$}{
				\If{$v_1 = v_2$}{
					$\mathit{cpt} = \mathit{cpt} \cup \{(v_1, v_2) \mapsto 1\}$\;
				}
				\Else{
					$\mathit{cpt} = \mathit{cpt} \cup \{(v_1, v_2) \mapsto 0\}$\;				
				}
			}
			\Return{$\mathit{cpt}$}
		}	
		\If{$D = \{\top, \bot\}$}
		{
				$E = \{\top,\bot\}^{n}$\;
				$\mathit{cpt} = \emptyset$\;
				\For{$(v_1,\ldots,v_n) \in E$}{
					$K = \{ v_i \mid 1 \leq i \leq n-1 \wedge v_i = \top\}$\;
					\If{$(v_n = \top \wedge K \neq \emptyset) \vee (v_n = \bot \wedge K = \emptyset)$}{
						$\mathit{cpt} = \mathit{cpt} \cup \{\overline{v} \mapsto 1\}$\;
					}
					\Else{
						$\mathit{cpt} = \mathit{cpt} \cup \{\overline{v} \mapsto 0\}$\;
					}
					
				}
				\Return{$\mathit{cpt}$}\;
		}

	}	
	\;
	\Fn{$\mathit{tree}(\mathbb{X}, \mathit{root}, D)$}{
		$p = \mathit{nil}$\;
		$N = \mathbb{X}$\;
		$E = \emptyset$\;
		$\mathit{cpt} = \emptyset$\;
		$\mathbb{X}' = \emptyset$\;
		\If{$\mathbb{X} = \{x\}$}
		{
			$N = N \cup \{\mathit{root}\}$\;
			$E = E \cup \{x \to \mathit{root}\}$\;
			$\mathit{cpt}(\mathit{root}) = \mathit{CPT}_{\oplus}(2, D)$\;
		}
		\Else{
			\For{$v \in \mathbb{X}$}{
				\If{$p = \mathit{nil}$}
				{
					$p = v$\;
				}
				\Else
				{
					$\mathbb{X}' = \mathbb{X}' \cup X[p,v]$\;
					$N = N \cup X[p,v]$\;
					$E = E \cup \{p \to X[p,v], v \to X[p,v]\}$\;
					 $\mathit{cpt}(X[p,v]) = \mathit{CPT}_{\oplus}(3, D)$\;
					$p = \mathit{nil}$\;
				}
			}
			\If{$p \neq \mathit{nil}$}
			{
				$\mathbb{X}' = \mathbb{X}' \cup p$\;
			}
			$(N',E',\mathit{cpt}') = \mathit{tree}(\mathbb{X}',\mathit{root}, D)$\;
		}
		\Return{$(N\cup N', E \cup E', \mathit{cpt} \cup \mathit{cpt}')$}\;
	}\;
\end{multicols}
\caption{Auxiliary functions used in  Algorithm~\ref{figure:acyclic:encoding:algorithm}.}
\label{figure:acyclic:encoding:algorithm:auxiliary}
\end{algorithm*}
}

\subsubsection{Encoding's acyclicity}

Here, we prove that Algorithm~\ref{figure:acyclic:encoding:algorithm} produces a Bayesian Network that is a forest of poly-trees.
Note that Algorithm~\ref{figure:acyclic:encoding:algorithm} produces a forest of poly-trees and not just a single poly-tree because some predicates symbols may be independent.
For instance, the program consisting of the rules $a(x) \leftarrow b(x)$ and $c(x) \leftarrow d(x)$ corresponds to two poly-trees (one per rule).

\begin{proposition}\thlabel{theorem:bayesian:network:polytree}
Let $\Sigma$ be a first-order signature, $\mathbf{dom}$ be a finite domain, $p$ be a $(\Sigma,\mathbf{dom})$-relaxed acyclic \problog{} program, and $({\cal O}, {\cal D}, {\cal U})$ be a witness for $p$'s acyclicity.
The Bayesian Network $(N,E,\mathit{cpt})$ produced by Algorithm~\ref{figure:acyclic:encoding:algorithm} on input $p$ is a forest of poly-trees.
\end{proposition}

\begin{proof}
Let $\Sigma$ be a first-order signature, $\mathbf{dom}$ be a finite domain, $p$ be a $(\Sigma,\mathbf{dom})$-relaxed acyclic \problog{} program, and $({\cal O}, {\cal D}, {\cal U})$ be a witness for $p$'s acyclicity.
Furthermore, let $BN = (N,E,\mathit{cpt})$ be the Bayesian Network produced by Algorithm~\ref{figure:acyclic:encoding:algorithm} on input $p$.
There is a one-to-one mapping from paths in the ground graph $\mathit{gg}(\alpha(\beta_{p,\cal D,U}(p))))$ and paths in the Bayesian Network produced by Algorithm~\ref{figure:acyclic:encoding:algorithm}.
Namely, there is a path from $a(\overline{c})$ to $a'(\overline{c}')$ in $\mathit{gg}(\alpha(\beta_{p,\cal D,U}(p))))$ iff there is a path from $X[a(\overline{c})]$ to $X[a'(\overline{c}')]$ in $BN$.
Assume that there is a cycle in the undirected version of $BN$.
From this, it follows that there is an undirected cycle in $\mathit{gg}(\alpha(\beta_{p,\cal D,U}(p))))$.
This, however, contradicts~\thref{theorem:graph:acyclic} (since $p$ is relaxed acyclic, then $\alpha(\beta_{p,\cal D,U}(p)))$ is acyclic and there are no undirected cycles in $\mathit{gg}(\alpha(\beta_{p,\cal D,U}(p))))$).
\end{proof}

\subsubsection{Bayesian Networks}

A Bayesian Network $\mathit{BN}$ is a tuple $(N,E,\mathit{cpt})$, where $N$ is the set of nodes, $E$ is the set of edges, $\mathit{cpt}$ is a function associating to each node $n \in N$ its Conditional Probability Table.
For each node $n \in N$, we denote by $D(n)$ its domain, i.e., the possible values it can have.
Note that $D(n)$ can be immediately derived from $\mathit{cpt}(n)$.
We denote by $p(n)$ the set of $n$'s parents, i.e., $p(n) = \{ n' \mid n' \rightarrow n \in E\}$.
Furthermore, we denote by $\mathit{ancestors}(n)$ the set of $n$'s ancestors, i.e., $\mathit{ancestors}(n) = p(n) \cup \bigcup_{n' \in \mathit{p}(n)} \mathit{ancestors}(n')$.

A \textit{$\mathit{BN}$-total assignment} is a total function $\nu$ that associates to each $n \in N$ a value in $n \in \mathit{D}(n)$, whereas a \textit{$\mathit{BN}$-partial assignment} is a partial function $\nu$ that associates to each $n \in N'$, where $N' \subseteq N$, a value in $n \in \mathit{D}(n)$.
The probability defined by $\mathit{BN}$ given $\nu$, written $\llbracket \mathit{BN} \rrbracket(\nu)$ is
$\left[ \Sigma_{y_1} \ldots \Sigma_{y_m} \left(\Pi_{n \in N} \mathit{cpt}(n)\right) (Y_1 = y_1,\ldots,  Y_m = y_m) \right](X_1 \\ = v_1,  \ldots, X_n = v_n)$,
where $\mathit{dom}(\nu) = \{X_1,\ldots,X_n \}$, $N \setminus \mathit{dom}(\nu) = \{Y_1,\ldots,Y_m \}$, and $v_i = \nu(X_i)$ for $1 \leq i \leq n$.

\newcommand{\rest}[1]{#1\!\!\uparrow_p}
\newcommand{\transf}[1]{#1\!\!\downarrow_{p}}

\subsubsection{Correctness Proof}

Before presenting our correctness proof, we introduce some machinery.
Let $p$ be a relaxed acyclic \problog{} program,  $({\cal O, D, U})$ be a witness for $p$, and $\mathit{BN} = (N,E,\mathit{cpt})$ be the Bayesian Network produced by Algorithm~\ref{figure:acyclic:encoding:algorithm} having $p$ and  $({\cal O, D, U})$ as inputs.
For each rule atom $a(\overline{c}) \in \mathit{ground}(p)$, we denote by $\transf{a(\overline{c})}$ the atom obtained by applying the transformations $\alpha$ and $\beta$, whereas we denote by $\rest{a(\overline{c})}$ the constants not used in the transformed program, namely $\rest{a(\overline{c})}$ is $\overline{c}\downarrow_{\{1, \ldots, |\overline{c}|\}\setminus \mu_{p,\cal D, U}(a)}$.
Furthermore, given a rule $r \in p$ (respectively a ground rule $s \in \mathit{ground}(p,r)$), we denote by $\transf{r}$ (respectively $\transf{s}$) the rule obtained by applying the transformations $\alpha$ and $\beta$.
If $\mu_{p,\cal D, U}(a) = \emptyset$, then $\rest{a(\overline{c})} = \top$.
We say that the \textit{kernel of $\mathit{BN}$}, denoted $K(\mathit{BN})$, is the set of variables $\{ X[r,\emptyset,a(\overline{c})] \in N \mid  \exists v.\, r = \atom{v}{a(\overline{c})}\}$.
Note that all nodes $ n \in K(\mathit{BN})$ are boolean random variables by construction, namely $D(n) = \{\top,\bot\}$.
Given a $\mathit{BN}$-total assignment $\nu$, we say that $\nu$ is \textit{consistent} iff for all variables $n \in N\setminus K(\mathit{BN})$, $\mathit{cpt}(n)(\nu(P_1),\ldots, \nu(P_m),\nu(n)) = 1$, where $p(n) = \{P_1, \ldots, P_m\}$.
Furthermore, we say that $\nu$ is a \textit{model for a state $s$}, written $\nu \models s$, iff $a(\overline{c}) \in s \Leftrightarrow \nu(X[\transf{a(\overline{c})}]) = \rest{a(\overline{c})}$.

We first prove a simple results that connect the original program $p$ and its transformed version $\alpha(\beta(p))$.

\begin{proposition}\thlabel{theorem:original:to:transformed}
Let $\Sigma$ be a first-order signature, $\mathbf{dom}$ be a finite domain, $p$ be a $(\Sigma,\mathbf{dom})$-relaxed acyclic \problog{} program, $(\cal O,D,U)$ be a witness for $p$, $s$ be a $(\Sigma,\mathbf{dom})$-structure, and $p'$ be the program $\alpha(\beta_{p_0,\cal D,  U}(p_0))$.
Then, (1) $a \in \mathit{ground}(p)$ implies $\transf{a} \in \mathit{ground}(p')$, and
(2) $s \in \mathit{ground}(p,r)$ implies $\transf{s} \in \mathit{ground}(p,\transf{r})$.
\end{proposition}

\begin{proof}
Let $\Sigma$ be a first-order signature, $\mathbf{dom}$ be a finite domain, $p$ be a $(\Sigma,\mathbf{dom})$-relaxed acyclic \problog{} program, $(\cal O,D,U)$ be a witness for $p$, $s$ be a $(\Sigma,\mathbf{dom})$-structure, and $p'$ be the program $\alpha(\beta_{p_0,\cal D,  U}(p_0))$.

We first prove our first claim.
Let $a \in \mathit{ground}(p)$.
From this, it follows that there is an $i$ such that $a \in \mathit{ground}(p,i)$.
We now prove our claim by induction on $i$.
The base case is $a \in \mathit{ground}(p,0)$.
From this, it follows that $a$ is either a ground atom or a probabilistic ground atom.
From this and the definition of $\alpha$ and $\beta$, $\transf{a} = a$.
From this, it follows that $a \in p'$ and, therefore, $a \in \mathit{ground}(p')$.
For the induction step, assume that our claim holds for all $j < i$, we now show that it holds also for $i$.
The only interesting case is  $a \in \mathit{ground}(p,i) \setminus \mathit{ground}(p,i-1)$.
From this, it follows that there is a rule $r$ and a ground rule $s$ such that all $\mathit{body}^+(s) \subseteq \mathit{ground}(p,i-1)$.
From this, $\transf{r}\in p'$, the fact that the transformation does not introduce new positive literals, and the induction hypothesis, it follows that $\{\transf{b} \mid b \in \mathit{body}^+(s) \} \subseteq \mathit{ground}(p')$.
From this and $\mathit{ground}$'s definition, it follows that $\transf{a} \in \mathit{ground}(p')$.

The proof of our second claim is similar to the first one.
\end{proof}

Here are now ready to prove the two key lemmas for the encoding's correctness.

\begin{proposition}\thlabel{theorem:correctness:auxiliary:1}
Let $\Sigma$ be a first-order signature, $\mathbf{dom}$ be a finite domain, $p$ be a $(\Sigma,\mathbf{dom})$-relaxed acyclic \problog{} program, $(\cal O,D,U)$ be a witness for $p$, $s$ be a $(\Sigma,\mathbf{dom})$-structure, and $\mathit{BN} = (N,E,\mathit{cpt})$ be the Bayesian Network generated by Algorithm~\ref{figure:acyclic:encoding:algorithm} having $p$ and $(\cal O,D,U)$ as input.
There is a $p$-probabilistic assignment $f$ such that $\mathit{prob}(f) = k$ and $g_f(p) = s$ iff there is a $\mathit{BN}$-total assignment $\nu$ such that $\llbracket \mathit{BN} \rrbracket(\nu) = k$, $\nu$ is consistent, and $\nu \models s$.
\end{proposition}

\begin{proof}
Let $\Sigma$ be a first-order signature, $\mathbf{dom}$ be a finite domain, $p$ be a $(\Sigma,\mathbf{dom})$-relaxed acyclic \problog{} program, $(\cal O,D,U)$ be a witness for $p$, $s$ be a $(\Sigma,\mathbf{dom})$-structure, and $\mathit{BN} = (N,E,\mathit{cpt})$ be the Bayesian Network generated by Algorithm~\ref{figure:acyclic:encoding:algorithm} having $p$ and $(\cal O,D,U)$ as input.
We now prove both directions of our claim.

\para{\boldmath$(\Rightarrow)$} 
Let $f$ be a $p$-probabilistic assignment $f$ such that $\mathit{prob}(f) = k$ and $g_f(p) = s$.
Furthermore, let $\nu$ be the following $\mathit{BN}$-total assignment:
\[
\nu(n) = 
\begin{cases}
f(\atom{v}{a(\overline{c})}) & \text{if}\ n=X[\atom{v}{a(\overline{c})},\emptyset,a(\overline{c})]  \\
\top & \text{if}\ n=X[a(\overline{c}),\emptyset,a(\overline{c})] \\
k & \text{if}\ k \in D(n) \wedge \mathit{p}(n)=\{P_1, \ldots, P_m\} \\
& \quad \wedge \nu(P_1) = v_1 \wedge \ldots \wedge \nu(P_m) = v_m \\
& \quad \wedge \mathit{cpt}(n)(v_1, \ldots, v_m, k) = 1
\end{cases}
\]
The assignment $\nu$ is well-defined since (1) $\mathit{BN}$ is a forest of poly-trees (see \thref{theorem:bayesian:network:polytree}), and (2) for all nodes of the form $X[\atom{v}{a(\overline{c})},\emptyset,a(\overline{c})]$ and $X[{a(\overline{c})},\emptyset,a(\overline{c})]$, $D(n) = \{\top,\bot\}$ by construction (cf. Algorithm~\ref{figure:acyclic:encoding:algorithm}).
Furthermore, the assignment $\nu$ is consistent by construction.
Indeed, for all variables $n \in N \setminus K(\mathit{BN})$, the corresponding entry in the CPT is $1$.

We now prove that $\nu \models s$.
From \thref{theorem:correctness:auxiliary:2}, $\nu$'s consistency, and $f(\atom{v}{a(\overline{c})}) = \nu(X[\{\atom{v}{a(\overline{c})}\}, \emptyset,  a(\overline{c})])$ for all $\atom{v}{a(\overline{c})} \in p$, it follows that  $a(\overline{c}) \in  g_f(p)$ iff $\nu(X[\transf{a(\overline{c})}]) = \rest{a(\overline{c})}$.
From this and $g_f(p) = s$, it follows  $\nu \models s$.

Finally, we show that $\llbracket \mathit{BN} \rrbracket(\nu) = \mathit{prob}(f)$.
In more detail, $\llbracket \mathit{BN} \rrbracket(\nu) = \left[ \left(\Pi_{n \in N} \mathit{cpt}(n)\right) \right](\nu)$.
This can be equivalently rewritten as follows:
$\llbracket \mathit{BN} \rrbracket(\nu) = \left[ \left(\Pi_{n \in K(\mathit{BN})} \mathit{cpt}(n)\right) \cdot \\ \left(\Pi_{n \in N\setminus K(\mathit{BN})} \mathit{cpt}(n)\right)\right](\nu)$.
Furthermore, since $\nu$ is consistent and the CPTs associated with the variables in $N\setminus K(\mathit{BN})$ are deterministic, $\llbracket \mathit{BN} \rrbracket(\nu)$ can be simplified as $\llbracket \mathit{BN} \rrbracket(\nu) = \left[ \left(\Pi_{n \in K(\mathit{BN})} \mathit{cpt}(n)\right) \right](\nu)$.
From this and $\mathit{BN}$'s definition, it follows that $\llbracket \mathit{BN} \rrbracket(\nu) = \Pi_{\nu(X[\atom{v}{a(\overline{c})},\emptyset,a(\overline{c})]) = \top} v \cdot  \Pi_{\nu(X[\atom{v}{a(\overline{c})},\emptyset,a(\overline{c})]) = \bot} (1-v)$.
From this and $\nu$'s definition, it follows $\llbracket \mathit{BN} \rrbracket(\nu) = \Pi_{f(\atom{v}{a(\overline{c})}) = \top} v \cdot  \Pi_{f(\atom{v}{a(\overline{c})}) = \bot} (1-v)$, which is equivalent to $\mathit{prob}(f)$.

\para{\boldmath$(\Leftarrow)$} 
Let $\nu$ be a $\mathit{BN}$-total assignment such that $\llbracket \mathit{BN} \rrbracket(\nu) = k$, $\nu$ is consistent, and $\nu \models s$, and  $f$ be the following $p$-probabilistic assignment:
$f(\atom{v}{a(\overline{c})}) = \nu(X[\{\atom{v}{a(\overline{c})}\}, \emptyset,  a(\overline{c})])$.

We now show that $\llbracket \mathit{BN} \rrbracket(\nu) = \mathit{prob}(f)$.
In more detail, $\llbracket \mathit{BN} \rrbracket(\nu) = \left[ \left(\Pi_{n \in N} \mathit{cpt}(n)\right) \right](\nu)$.
Furthermore, since $\nu$ is consistent and the CPTs associated with the variables in $N\setminus K(\mathit{BN})$ are deterministic, $\llbracket \mathit{BN} \rrbracket(\nu)$ can be simplified as $\llbracket \mathit{BN} \rrbracket(\nu) = \left[ \left(\Pi_{n \in K(\mathit{BN})} \mathit{cpt}(n)\right) \right](\nu)$.
From this and $\mathit{BN}$'s definition, it follows that $\llbracket \mathit{BN} \rrbracket(\nu) = \Pi_{\nu(X[\atom{v}{a(\overline{c})},\emptyset,a(\overline{c})]) = \top} v \cdot  \Pi_{\nu(X[\atom{v}{a(\overline{c})},\emptyset,a(\overline{c})]) = \bot} (1-v)$.
From this and $\nu$'s definition, it follows that $\llbracket \mathit{BN} \rrbracket(\nu) = \Pi_{f(\atom{v}{a(\overline{c})}) = \top} v \cdot  \Pi_{f(\atom{v}{a(\overline{c})}) = \bot} (1-v)$, which is equivalent to $\mathit{prob}(f)$.

We still have to prove that $g_f(p) = s$.
From \thref{theorem:correctness:auxiliary:2}, $\nu$'s consistency, $f(\atom{v}{a(\overline{c})}) = \nu(X[\{\atom{v}{a(\overline{c})}\}, \emptyset,  a(\overline{c})])$, it follows that  $a(\overline{c}) \in  g_f(p)$ iff $\nu(X[\transf{a(\overline{c})}]) = \rest{a(\overline{c})}$.
From this and $\nu \models s$, it follows  $g_f(p) = s$. 
\end{proof}

\begin{proposition}\thlabel{theorem:correctness:auxiliary:2}
Let $\Sigma$ be a first-order signature, $\mathbf{dom}$ be a finite domain, $p$ be a $(\Sigma,\mathbf{dom})$-relaxed acyclic \problog{} program, $(\cal O,D,U)$ be a witness for $p$, $\mathit{BN} = (N,E,\mathit{cpt})$ be the Bayesian Network generated by Algorithm~\ref{figure:acyclic:encoding:algorithm} having $p$ and $(\cal O,D,U)$ as input, $f$ be a $p$-probabilistic assignment, and $\nu$ be a consistent $\mathit{BN}$-variable assignment such that $\nu(X[\atom{v}{a(\overline{c})}, \emptyset, a(\overline{c})]) = f(\atom{v}{a(\overline{c})})$ for all $\atom{v}{a(\overline{c})} \in p$. 
Then, $a(\overline{c}) \in  g_f(p,j,i)$, for some $j$ and $i$, iff $\nu(X[\transf{a(\overline{c})}]) = \rest{a(\overline{c})}$.
\end{proposition}

\begin{proof}
Let $\Sigma$ be a first-order signature, $\mathbf{dom}$ be a finite domain, $p$ be a $(\Sigma,\mathbf{dom})$-relaxed acyclic \problog{} program, $(\cal O,D,U)$ be a witness for $p$, $\mathit{BN} = (N,E,\mathit{cpt})$ be the Bayesian Network generated by Algorithm~\ref{figure:acyclic:encoding:algorithm} having $p$ and $(\cal O,D,U)$ as input, $f$ be a $p$-probabilistic assignment, and $\nu$ be a consistent $\mathit{BN}$-variable assignment such that $\nu(X[\atom{v}{a(\overline{c})}, \emptyset, a(\overline{c})]) = f(\atom{v}{a(\overline{c})})$ for all $\atom{v}{a(\overline{c})} \in p$. 
Furthermore, let $p'$ be the acyclic \problog{} program obtained after applying the $\alpha$ and $\beta$ transformations.
We claim that $a(\overline{c}) \in  g_f(p,\mu(a))$ iff $\nu(X[\transf{a(\overline{c})}]) = \rest{a(\overline{c})}$.
From this,  $a(\overline{c}) \in g_f(p)$ iff  $a(\overline{c}) \in g_f(p,j,i)$ for some $j$ and $i$, and $a(\overline{c}) \in g_f(p)$ iff $a(\overline{c}) \in g_f(p,\mu(a))$, it follows that $a(\overline{c}) \in  g_f(p,j,i)$, for some $j$ and $i$, iff $\nu(X[\transf{a(\overline{c})}]) = \rest{a(\overline{c})}$.
However, there may be some nodes in $\mathit{BN}$ that do not correspond to any ground rule or atom in $g_f(p,r)$ or $g_f(p)$.

We now prove, by induction on $\mu(a)$ (define in \S\ref{app:atklog:lite:extended:exact:grounding}), our claim that $a(\overline{c}) \in  g_f(p,\mu(a))$ iff $\nu(X[\transf{a(\overline{c})}]) = \rest{a(\overline{c})}$ (we denote the corresponding induction hypothesis as $(\star)$).

\para{Base Case}
For the base case, we assume that $\mu(a) = 0$.
We prove separately the two directions, namely (1) if $a(\overline{c}) \in  g_f(p,0,i)$ and $\mu(a) = 0$, then $\nu(X[\transf{a(\overline{c})}]) = \rest{a(\overline{c})}$, and (2) if  $\nu(X[\transf{a(\overline{c})}]) = \rest{a(\overline{c})}$ and $\mu(a) = 0$, then $a(\overline{c}) \in  g_f(p,0,i)$.
From this, it follows that if $\mu(a) = 0$, then $a(\overline{c}) \in  g_f(p,0)$ iff $\nu(X[\transf{a(\overline{c})}]) = \rest{a(\overline{c})}$.

\para{\boldmath$(\Rightarrow)$}
We prove, by induction on $i$, that if $a(\overline{c}) \in  g_f(p,0,i)$, then $\nu(X[\transf{a(\overline{c})}]) = \rest{a(\overline{c})}$ (we denote this induction hypothesis as $(\dagger)$).
From $a(\overline{c}) \in  g_f(p)$, $\mathit{g}_f(p) \subseteq \mathit{ground}(p)$, and \thref{theorem:original:to:transformed}, then $X[\transf{a(\overline{c})}] \in N$. 
The base case is as follows.
If $a(\overline{c}) \in  g_f(p,0,0)$, then there is a rule $r$ such that either $r = a(\overline{c})$ or $r = \atom{v}{a(\overline{c})}$ and $a(\overline{c}) \in g_f(p,r,0,0)$.
If $r = a(\overline{c})$, then $\nu(X[r,\emptyset,a(\overline{c})]) = \top$ due to $\nu$'s consistency,  $\mathit{g}_f(p) \subseteq \mathit{ground}(p)$, \thref{theorem:original:to:transformed}, $\transf{r} = r$, and $\transf{a(\overline{c})}  = a(\overline{c})$.
If $r = \atom{v}{a(\overline{c})}$, then  $f(\atom{v}{a(\overline{c})}) = \top$ follows from $a(\overline{c}) \in  g_f(p,0,0)$.
From this,  $\mathit{g}_f(p) \subseteq \mathit{ground}(p)$, \thref{theorem:original:to:transformed}, $\transf{r} = r$, $\transf{a(\overline{c})}  = a(\overline{c})$, and $\nu(X[\atom{v}{a(\overline{c})}, \emptyset, a(\overline{c})]) = f(\atom{v}{a(\overline{c})})$, it follows that there is a variable $X[r,\emptyset,a(\overline{c})] \in N$ such that $\nu(X[r,\emptyset,a(\overline{c})]) = \top$.
From $\nu(X[r,\emptyset, a(\overline{c})]) = \top$, $\mathit{BN}$'s definition, and $nu$'s definition, it follows that there are nodes $X[r,  a(\overline{c})], X[a(\overline{c})] \in N$ such that $\nu(X[r,  a(\overline{c})]) = \top$ and $\nu(X[a(\overline{c})]) = \top$ as well (note that $\transf{r} = r$, and $\transf{a(\overline{c})}  = a(\overline{c})$ in this case).
For the induction step, we assume that our claim holds for all $i' < i$.
The only interesting case is when $a(\overline{c}) \in  g_f(p,0,i) \setminus g_f(p,0,i-1)$.
From this, it follows that there is rule $r \in p$ such that $a(\overline{c}) \leftarrow b_1,\ldots,b_m \in g_f(p,r,0,i)$.
From this, it follows that $b_1,\ldots,b_m \in g_f(p,0,i-1)$ and $\mathit{body}^-(r) = \emptyset$.
From this, $g_f(p,0,i-1) \subseteq \mathit{ground}(p)$, \thref{theorem:original:to:transformed}, and the induction's hypothesis $(\dagger)$, it follows that there are nodes $X[\transf{b_1}], \ldots, X[\transf{b_m}] \in N$ and $\nu(X[\transf{b_1}]) = \rest{b_1},\ldots, \nu(X[\transf{b_m}]) = \rest{b_m}$.
From this, $\mathit{body}^-(r) = \emptyset$, $g_f(p,r,0,i) \subseteq g_f(p,r) \subseteq \mathit{ground}(p,r)$, \thref{theorem:original:to:transformed}, and $\mathit{BN}$'s construction, it follows that there is a variable $X[\transf{r}, \{\transf{b_1},\ldots,\transf{b_m}\}, \transf{a(\overline{c})}] \in N$ such that $\mathit{cpt}(X[\transf{r}, \{\transf{b_1}, \ldots,\transf{b_m}\}, \transf{a(\overline{c})}]) (\nu(X[\transf{b_1}]), \ldots, \nu(X[\transf{b_m}]), k)   = 1$ iff $k = \rest{a(\overline{c})}$. 
As a result, $\nu(X[\transf{r}, \{\transf{b_1},\ldots,\transf{b_m}\}, \transf{a(\overline{c})}]) = \rest{a(\overline{c})}$ since $\nu$ is consistent.
From this and $\mathit{BN}$'s construction, it follows that there are nodes $X[\transf{r},  \transf{a(\overline{c})}], X[\transf{a(\overline{c})}] \in N$ such that $\nu(X[\transf{r}, \transf{a(\overline{c})}]) = \top$ and $\nu(X[\transf{a(\overline{c})}]) = \rest{a(\overline{c})}$.
This completes the proof of the if direction.

\para{\boldmath$(\Leftarrow)$}
To prove that if $\nu(X[\transf{a(\overline{c})}]) = \rest{a(\overline{c})}$, then $a(\overline{c}) \in  g_f(p,0,i)$, for some $i$, we prove a stronger claim.
Namely, if $\nu(X[\transf{a(\overline{c})}]) = \rest{a(\overline{c})}$, then $a(\overline{c}) \in  g_f(p,0,\mathit{depth}(X[\transf{a(\overline{c})}]))$, where $\mathit{depth}(n) = 0$ if $\mathit{ancestors}(n)$ does not contain any variable of the form $X[b(\overline{v})]$ and $\mathit{depth}(n) = 1 + \mathit{max}_{n' \in \mathit{ancestors}(n)} depth(n')$ otherwise.
We prove our claim by induction on $\mathit{depth}(X[\transf{a(\overline{c})}])$ (we denote the induction hypothesis as $(\triangle)$).
The base case is as follows.
If $\mathit{depth}(X[\transf{a(\overline{c})}]) = 0$, then from  $\nu(X[\transf{a(\overline{c})}]) = \top$ and $\nu$'s consistency, it follows that there must be a rule $r$, of the form $r = \atom{v}{a(\overline{c})}$ or $r = a(\overline{c})$, such that $X[r,a(\overline{c})] \in N$ and $\nu(X[r,a(\overline{c})]) = \top$ (note that, in this case, $\transf{r} = r$ and $\transf{a(\overline{c})} = a(\overline{c})$).
From this and $\nu$'s consistency, there is a variable $X[r,\emptyset,a(\overline{c})] \in N$ such that $\nu(X[r,\emptyset,a(\overline{c})]) = \top$.
If $r = a(\overline{c})$, then $a(\overline{c}) \in g_f(p,0,0)$ by definition.
If $r = \atom{v}{a(\overline{c})}$, then from $\nu(X[r,\emptyset,a(\overline{c})]) = \top$ and $\nu(X[r,\emptyset,a(\overline{c})]) = f(\atom{v}{a(\overline{c})})$, it follows that $f(\atom{v}{a(\overline{c})}) = \top$.
From this, it follows that $a(\overline{c}) \in g_f(p,0,0)$.
For the induction step, we assume that our claim holds for all random variables of depth less than $k$.
From $\nu(X[\transf{a(\overline{c})}]) = \rest{a(\overline{c})}$ and $\nu$'s consistency, it follows that there is a rule $r \in p'$ such that $\nu(X[r,\transf{a(\overline{c})}]) = \rest{a(\overline{c})}$.
From this and $\nu$'s consistency, there is a set of positive ground literals $I = (b_1, \ldots, b_m)$ such that $X[r,(\transf{b_1}, \ldots, \transf{b_m}),\transf{a(\overline{c})}] \in N$ and $\nu(X[r,(\transf{b_1}, \ldots, \transf{b_m}),\transf{a(\overline{c})}]) = \rest{a(\overline{c})}$.
From this and $\mathit{BN}$'s construction (cf. the $\mathit{cpt}$ and $\mathit{satisfiable}$ procedures), it follows that there is a rule $r' \in p$ such that (1) $r \in [r']_p$, and (2) the values assigned by $\nu$ to the random variables associated to the atoms in $I$ produce a grounding $s$ of $r$ that satisfies all constraints and is consistent (namely, the body of $s$ does not contain both an atom and its negation and multiple copies of the same CPT-like atom in $r'$ are assigned to the same value in $r$).
From this, $\mathit{BN}$'s definition, and $\mu(a) = 0$, it follows that $\nu(X[\transf{b}]) = \rest{b}$  for all $b \in I$.
From this and the induction's hypothesis $(\triangle)$, it follows that $b \in g_f(p,0,\mathit{depth}(X[a(\overline{c})]) -1)$ for all  $b \in I$.
From this, $\nu(X[\transf{r},\{\transf{b_1}, \ldots, \transf{b_m}\},\transf{a(\overline{c})}]) = \rest{a(\overline{c})}$, and $\mathit{BN}$'s construction (cf. the $\mathit{cpt}$ and $\mathit{satisfiable}$ procedures), it follows that there is a rule $r' \in p$ such that (1) $\transf{r} = \transf{r'}$, and (2) $r'$ is satisfied by $\{{b_1}, \ldots, {b_m}\}$, i.e., there is a grounding $s'$ of $r'$ obtained by using a subset of the literals in  $\{{b_1}, \ldots, {b_m}\}$ and adding repeated occurrences of the literals if needed.
From this, it follows that $s$ is in $g_f(p,r,0,\mathit{depth}(X[\transf{a(\overline{c})}]))$.
From this, it follows that $a(\overline{c}) \in g_f(p,0,\mathit{depth}(\transf{X[a(\overline{c})}]))$.
This completes the proof the only if direction.

\para{Induction Step}
For the induction's step, we assume that $b(\overline{d}) \in g_f(p,\mu(b))$ iff $\nu(X[\transf{b(\overline{d})}]) = \rest{b(\overline{d})}$ holds for all $b$ such that $\mu(b) < \mu(a)$.
We now prove that $a(\overline{c}) \in g_f(p,\mu(a))$ iff $\nu(X[\transf{a(\overline{c})}]) = \rest{a(\overline{c})}$ as well.
In the following, let $k$ be $\mu(a)$.
We prove separately the two directions, namely (1) if $a(\overline{c}) \in  g_f(p,\mu(a),i)$, for some $i$, then $\nu(X[\transf{a(\overline{c})}]) = \rest{a(\overline{c})}$, and (2) if $\nu(X[\transf{a(\overline{c})}]) = \rest{a(\overline{c})}$, then $a(\overline{c}) \in  g_f(p,\mu(a),i)$, for some $i$. 
From this, it follows that $a(\overline{c}) \in  g_f(p,\mu(a))$ iff $\nu(X[\transf{a(\overline{c})}]) = \rest{a(\overline{c})}$.

\para{\boldmath$(\Rightarrow)$}
We prove, by induction on $i$, that if $a(\overline{c}) \in  g_f(p,k,i)$, for some $i$, then $\nu(X[\transf{a(\overline{c})}]) = \rest{a(\overline{c})}$ (we denote the induction hypothesis associated to this proof as $(\clubsuit)$).
From $a(\overline{c}) \in  g_f(p) \subseteq \mathit{ground}(p)$ and \thref{theorem:original:to:transformed}, then $X[\transf{a(\overline{c})}] \in N$. 
The base case is as follows.
Assume that $a(\overline{c}) \in  g_f(p,k,0)$.
From this and $a(\overline{c}) \not\in  g_f(p,k-1)$ (since $\mu(a) > k-1$), it follows that there is a rule $r$ such that either $r = a(\overline{c})$ or $r = \atom{v}{a(\overline{c})}$.
If $r = a(\overline{c})$, then there is a node $X[r,\emptyset,a(\overline{c})] \in N$ such that $\nu(X[r,\emptyset,a(\overline{c})]) = \top$ due to $\nu$'s consistency, $g_f(p,r) \subseteq \mathit{ground}(p,r)$, $\transf{r} = r$, $\transf{a(\overline{c})} = a(\overline{c})$, and \thref{theorem:original:to:transformed}.
If $r = \atom{v}{a(\overline{c})}$, then  $f(\atom{v}{a(\overline{c})}) = \top$ follows from $a(\overline{c}) \in  g_f(p,k,0)$.
From this, $g_f(p,r) \subseteq \mathit{ground}(p,r)$, \thref{theorem:original:to:transformed}, and $\nu(X[\{\atom{v}{a(\overline{c})}\}, \emptyset, a(\overline{c})]) = f(\atom{v}{a(\overline{c})})$, it follows that there is a node $X[r,\emptyset,a(\overline{c})] \in N$ such that $\nu(X[r,\emptyset,a(\overline{c})]) = \top$.
From $\nu(X[r,\emptyset, a(\overline{c})]) = \top$, $\mathit{BN}$'s definition, and $nu$'s definition, it follows that there are nodes $X[r, a(\overline{c})], X[a(\overline{c})] \in N$ such that $\nu(X[r, a(\overline{c})]) = \top$ and $\nu(X[a(\overline{c})]) = \top$ as well (note that  $\transf{r} = r$ and $\transf{a(\overline{c})} = a(\overline{c})$).
For the induction step, we assume that our claim holds for all $i' < i$.
Assume that $a(\overline{c}) \in  g_f(p,k,i)$.
The only interesting case is when $a(\overline{c}) \in  g_f(p,k,i) \setminus g_f(p,k,i-1)$.
From this, it follows that there is rule $r$ such that $r'=a(\overline{c}) \leftarrow b_1,\ldots,b_m$ and $r' \in g_f(p,r,k,i)$.
Furthermore, we denote by $b_1', \ldots, b_m'$ the atoms $\mathit{pos}(b_1), \ldots, \mathit{pos}(b_m)$.
From this, it follows that $b_e \in g_f(p,k,i-1)$ for all $b_e \in \mathit{body}^+(r')$ and $b_d' \not\in g_f(p,k -1)$ for all $b_d \in \mathit{body}^-(r')$.
From  $b_e \in g_f(p,k,i-1)$ for all $b_e \in \mathit{body}^+(r')$, $g_f(p,r) \subseteq \mathit{ground}(p,r)$, \thref{theorem:original:to:transformed},  and the induction's hypothesis $(\clubsuit)$, it follows that there is a node $X[\transf{b_e}] \in N$ such that $\nu(X[\transf{b_e}]) = \rest{b_e}$ for all $b_e \in \mathit{body}^+(r')$.
From $b_d' \not\in g_f(p,k -1)$ for all $b_d \in \mathit{body}^-(r')$, it follows that $b_d' \not\in g_f(p,\mu(\mathit{pred}(b_d)))$ for all $b_d \in \mathit{body}^-(r')$.
From this, $r' \in g_f(p,r,k,i)$, $g_f(p,r,k,i) \subseteq \mathit{ground}(p,r)$, \thref{theorem:original:to:transformed}, $\mu(\mathit{pred}(b_d)) < \mu(a)$ for all $b_d \in \mathit{body}^-(r')$, and the induction's hypothesis $(\star)$, it follows that there is a node $X[\transf{b_d'}] \in N$ such that $\nu(X[\transf{b_d'}]) \neq \rest{b_d'}$ for all $b_d \in \mathit{body}^-(r')$.
From $r' \in g_f(p,r,k,i)$, $g_f(p,r,k,i) \subseteq \mathit{ground}(p,r)$, \thref{theorem:original:to:transformed}, $\nu(X[\transf{b_e}]) = \rest{b_e}$ for all $b_e \in \mathit{body}^+(r')$, $\nu(X[\transf{b_d'}]) \neq \rest{b_d'}$ for all $b_d \in \mathit{body}^-(r')$, and $\mathit{BN}$'s construction, there is $X[\transf{r}, \{\transf{b_1'},\ldots,\transf{b_m'}\}, \transf{a(\overline{c})}] \in N$ such that $\mathit{cpt}(X[\transf{r}, \{\transf{b_1'},\ldots,\transf{b_m'}\}, \transf{a(\overline{c})}]) (\nu(X[\transf{b_1'}]), \ldots, \nu(X[\transf{b_m'}]),  k) = 1$ iff $k = \rest{a(\overline{c})}$. 
As a result, $\nu(X[\transf{r}, \{\transf{b_1'},\ldots,\transf{b_m'}\}, \transf{a(\overline{c})}]) = \top$ since $\nu$ is consistent.
From this and $\mathit{BN}$'s construction, it follows that there are variables $X[\transf{r}, \transf{a(\overline{c})}], X[\transf{a(\overline{c})}] \in N$ such that $\nu(X[\transf{r}, \transf{a(\overline{c})}]) = \rest{a(\overline{c})}$ and $\nu(X[\transf{a(\overline{c})}]) = \rest{a(\overline{c})}$.
This completes the proof of the if direction.

\para{\boldmath$(\Leftarrow)$}
To prove that if $\nu(X[\transf{a(\overline{c})}]) = \rest{a(\overline{c})}$, then $a(\overline{c}) \in  g_f(p,k,i)$, for some $i$, we prove a stronger claim.
Namely, if $\nu(X[\transf{a(\overline{c})}]) = \rest{a(\overline{c})}$, then $a(\overline{c}) \in  g_f(p,k,\mathit{depth}(X[a(\overline{c})]))$, where $\mathit{depth}(n)$ is as above.
We prove our claim by induction on $\mathit{depth}(X[\transf{a(\overline{c})}])$ (we denote the induction hypothesis associated to this proof as $(\spadesuit)$).
The base case is as follows.
If $\mathit{depth}(X[\transf{a(\overline{c})}])  = 0$, then from  $\nu(X[\transf{a(\overline{c})}]) = \top$ and $\nu$'s consistency, it follows that there must be a rule $r$, of the form $r = \atom{v}{a(\overline{c})}$ or $r = a(\overline{c})$, such that $X[r,a(\overline{c})] \in N$ and $\nu(X[r,a(\overline{c})]) = \top$ (note that $\transf{r} = r$ and $\transf{a(\overline{c})} = a(\overline{c})$.
From this and $\nu$'s consistency, there is a variable $X[r,\emptyset,a(\overline{c})] \in N$ such that $\nu(X[r,\emptyset,a(\overline{c})]) = \top$.
If $r = a(\overline{c})$, then $a(\overline{c}) \in g_f(p,k,0)$ by definition.
If $r = \atom{v}{a(\overline{c})}$, then from $\nu(X[r,\emptyset,a(\overline{c})]) = \top$ and $\nu(X[r,\emptyset,a(\overline{c})]) = f(\atom{v}{a(\overline{c})})$, it follows that $f(\atom{v}{a(\overline{c})}) = \top$.
From this, it follows that $a(\overline{c}) \in g_f(p,k,0)$.
For the induction step, we assume that our claim holds for all random variables of depth less than $k$.
We now show that it holds also for a variable $X[\transf{a(\overline{c})}]$ such that  $\mathit{depth}(X[\transf{a(\overline{c})}]) =i $.
From $\nu(X[\transf{a(\overline{c})}]) = \rest{a(\overline{c})}$ and $\nu$'s consistency, it follows that there is a rule $r' \in p'$ such that $X[r',\transf{a(\overline{c})}] \in N$  and $\nu(X[r',\transf{a(\overline{c})}]) = \rest{a(\overline{c})}$.
From this and $\nu$'s consistency, there is a set of ground atoms $I = (b_1, \ldots, b_m)$ such that $X[r',(\transf{b_1}, \ldots, \transf{b_m}),\transf{a(\overline{c})}] \in N$, $\nu(X[r',(\transf{b_1}, \ldots, \transf{b_m}),\transf{a(\overline{c})}]) = \rest{a(\overline{c})}$, and  $X[\transf{b}] \in N$ for all $b \in I$.
From this and $\mathit{BN}$'s construction (cf. the $\mathit{cpt}$ and $\mathit{satisfiable}$ procedures), it follows that there is a rule $r' \in p$ such that (1) $r \in [r']_p$, and (2) the values assigned by $\nu$ to the random variables associated to the atoms in $I$ produce a grounding $s$ of $r$ that satisfies all constraints and is consistent (namely, the body of $s$ does not contain both an atom and its negation and multiple copies of the same CPT-like atom in $r'$ are assigned to the same value in $r$).
Let $I^+$ be the atoms in $I$ that are assigned to positive literals in $s$ and $I^-$ be the atoms in $I$ that are assigned to negative literals (note that these two sets are disjoint).
From $\nu(X[r',(\transf{b_1}, \ldots, \transf{b_m}),\transf{a(\overline{c})}]) = \rest{a(\overline{c})}$, and $\nu$'s consistency, it follows that $\nu(X[\transf{b_e}]) = \rest{b_e}$ for all $b_e \in I^+$ and  $\nu(X[\transf{b_d}]) \neq \rest{b_d}$ for all $b_d \in I^-$.
From  $\nu(X[\transf{b_e}]) = \rest{b_e}$ and $\mathit{depth}(X[\transf{b_e}]) < \mathit{depth}(X[\transf{a(\overline{c})}])$ for all $b_e \in I^+$ and the induction's hypothesis $(\spadesuit)$, it follows that $b_e \in  g_f(p,k,\mathit{depth}(X[\transf{b_e}]))$ for all $b_e \in I^+$.
From this, the induction's hypothesis, and $\mathit{depth}(X[\transf{b_e}]) < \mathit{depth}(X[\transf{a(\overline{c})}])$, it follows that $b_e \in  g_f(p,k,i-1)$ for all $b_e \in \mathit{body}^+(r')$.
From  $\nu(X[\transf{b_d}]) \neq \rest{b_d}$ and $\mu(\mathit{pred}(b_d)) < \mu(a)$ for all $b_d \in I^-$ and the induction's hypothesis $(\star)$, it follows that $b_d \not\in  g_f(p,\mu(\mathit{pred}(b_d)))$ for all $b_d \in I^-$.
From this and $\mu(\mathit{pred}(b_d)) < \mu(a)$, it follows that $b_d \not\in  g_f(p,k-1)$ for all $b_d \in I^-$.
From $r'$ definition, $b_e \in  g_f(p,k,i-1)$ for all $b_e \in I^+$, and $b_d \not\in  g_f(p,k-1)$ for all $b_d \in I^-$, it follows that $r' \in g_f(p,r,k,i)$.
From this, it follows that $r' \in g_f(p,k,i)$ and $a(\overline{v}) \in g_f(p,k,i)$.
This completes the proof the only if direction.
\end{proof}

\thref{theorem:correctness} shows the correctness of our encoding.

\begin{proposition}\thlabel{theorem:correctness}
Let $\Sigma$ be a first-order signature, $\mathbf{dom}$ be a finite domain, $p$ be a $(\Sigma,\mathbf{dom})$-relaxed acyclic \problog{} program,
$(\cal O,D,U)$ be a witness for $p$, $\mathit{BN}$ be the Bayesian Network generated by Algorithm~\ref{figure:acyclic:encoding:algorithm} having $p$ and $(\cal O,D,U)$ as input, and $s$ be a $(\Sigma,\mathbf{dom})$-structure.
Furthermore, let $\mu$ be the $\mathit{BN}$-partial assignment such that 
(1) for any ground atom $a(\overline{c})$, $\mu(X[\transf{a(\overline{c}})]) = \rest{a(\overline{c})}$ iff $a(\overline{c}) \in s$, and
(2) $\mu(v)$ is undefined otherwise.
Then, $\llbracket p \rrbracket (s) = \llbracket \mathit{BN} \rrbracket(\mu)$.
\end{proposition}

\begin{proof}
Let $\Sigma$ be a first-order signature, $\mathbf{dom}$ be a finite domain, $p$ be a $(\Sigma,\mathbf{dom})$-relaxed acyclic \problog{} program,
$(\cal O,D,U)$ be a witness for $p$, $\mathit{BN}$ be the Bayesian Network generated by Algorithm~\ref{figure:acyclic:encoding:algorithm} having $p$ and $(\cal O,D,U)$ as input, and $s$ be a $(\Sigma,\mathbf{dom})$-structure.
Furthermore, let $\mu$ be the $\mathit{BN}$-partial assignment such that 
(1) for any ground atom $a(\overline{c})$, $\mu(X[\transf{a(\overline{c}})]) = \rest{a(\overline{c})}$ iff $a(\overline{c}) \in s$, and
(2) $\mu(v)$ is undefined otherwise.

The probability $\llbracket p \rrbracket (s)$ is $\Sigma_{f \in {\cal M}(p,s)} \mathit{prob}(f)$, where ${\cal M}(p,s)$ is the set of all  assignments $f$ such that $\mathit{WFM}(\mathit{instance}(p,f)) = s$.
Equivalently, ${\cal M}(p,s)$ is the set of all probabilistic assignments $f$ such that $g_f(p) =s$.
Let $K$ be the set of all total assignments that agree with $\mu$ for all variables of the form $X[\transf{a(\overline{c})}]$.
The probability $\llbracket \mathit{BN} \rrbracket(\mu)$ is $\Sigma_{\nu \in K} \llbracket \mathit{BN} \rrbracket(\nu)$.
Since any non-consistent assignment has probability $0$, $\llbracket \mathit{BN} \rrbracket(\mu) = \Sigma_{\nu \in K'} \llbracket \mathit{BN} \rrbracket(\nu)$, where $K'$ is the set of all consistent assignments in $K$.
From this, it follows that $\llbracket p \rrbracket (s) = \llbracket \mathit{BN} \rrbracket(\mu)$ iff $\Sigma_{f \in \{f \mid g_f(p) =s \}} \mathit{prob}(f) = \Sigma_{\nu \in K'} \llbracket \mathit{BN} \rrbracket(\nu)$.
The latter follows trivially from \thref{theorem:correctness:auxiliary:1}, which establishes a one-to-one mapping from $\{f \mid g_f(p) =s \}$ and $K'$ that preserves probabilities.
\end{proof}

\subsubsection{Size of the encoding}

Given a Bayesian Network $\mathit{BN} = (N,E,\mathit{cpt})$, the size of $\mathit{BN}$, denoted $|\mathit{BN}|$, is $|N| + |E| + \Sigma_{n \in N} |\mathit{cpt}(n)|$, where the size of a conditional probability table is just the number of rows in the table (i.e., the number of all assignments).
The size of an atom $a(\overline{c})$ is $|a|$, whereas the size of a rule $h \leftarrow l_1, \ldots, l_n$ is $|h| + \Sigma_{1\leq i\leq n} l_i$.
Finally, the size of a program $p$ is  $\Sigma_{r \in p} |r|$.
The \emph{ground version} of $p$, denoted $\mathit{gv}(p)$, is $\bigcup_{r \in p} \mathit{ground}(p,r)$, namely the relaxed grounding of all the rules in $p$.
Note that $\mathit{ground}(p) \subseteq \mathit{gv}(p)$.  
Given a \problog{} program $p$, we denote by $\mathit{edb}(p)$ the ground  (possibly probabilistic) atoms in $p$, whereas we denote by $\mathit{rules}(p)$ the rules, i.e., $\mathit{edb}(p) = \{ a(\overline{c}) \in p \} \cup \{ \atom{v}{a(\overline{c})} \in p \mid 0\leq v\leq 1\}$ and $\mathit{rules}(p) = \{ r \in p \mid \mathit{body}(r) \neq \emptyset \}$.

Let $p$ be a relaxed acyclic \problog{} program, $(\cal O,D,U)$ be a witness for $p$, $p' = \alpha(\beta_{p,\cal D, U}(p))$ be the transformed program,  $g = \mathit{gv}(p')$ be its ground version, and $\mathit{bn}(p, (\cal O,D,U)) = (N,E,\mathit{cpt})$ be the corresponding Bayesian Network derived by Algorithm~\ref{figure:angerona:algorithm}.
The number of nodes in $N$ is $O(|\mathit{rules}(p')|\cdot |g|)$. 
Indeed, there is a node $X[a(\overline{c})]$ for each ground atom in $g$. 
Moreover, for each rule $r \in p'$ and ground rule $r' \in g$, there are nodes $X[r,\mathit{body}(r'),\mathit{head}(r')]$ and $X[r,\mathit{head}(r')]$.
Finally, the number of intermediate nodes generated by the $\mathit{tree}$ procedure is twice the number of nodes of the form $X[r,\mathit{body}(r'),\mathit{head}(r')]$.

The number of edges in $E$ is  $O(|\mathit{rules}(p')|^2\cdot|g|)$.
Indeed, there is an edge $X[r,a(\overline{c})] \to X[a(\overline{c})]$ for each rule $r$ and ground atom $a(\overline{c})$.
Furthermore, there is an edge $X[b(\overline{v})] \to X[r, I,a(\overline{c})]$ for each rule $r$, ground rule $a(\overline{c}) \leftarrow I$, and atom $\overline{b}(\overline{v}) \in I$.
Finally, the number of edges introduced by the $\mathit{tree}$ procedure is $O(|g|)$.

The size of $\mathit{cpt}$ is $O(|\mathit{rules}(p')|\cdot |g| \cdot \mathit{max}(2,|\mathit{rules}(p)|)^{2+|\mathit{rules}(p)|})$.
Indeed, we have a CPT for each node $n \in N$.
The size of each CPT depends on (1) the number of parents for each node, and (2) the size of the domain associated to each variable.
In the worst case, the size of the domain of each node is $\mathit{max}(2,|\mathit{rules}(p)|)$ (since the size of the domain of a CPT-like predicate depends only on the rules --- there are no CPT-like predicates defined by ground atoms).
The number of parents for intermediate nodes is $2$, and the size of each CPT is $O(\mathit{max}(2,|\mathit{rules}(p)|)^3)$.
The maximum number of parents for nodes of the form $X[r,a(\overline{c})]$ is $1$, and the size of each CPT is $O(\mathit{max}(2,|\mathit{rules}(p)|)^2)$.
The maximum number of parents for nodes of the form $X[r,I,a(\overline{c})]$ is $|\mathit{rules}(p')|$, and the size of each CPT is $O(\mathit{max}(2,|\mathit{rules}(p)|)^{|\mathit{rules}(p')+1|})$.
Similarly, the maximum number of parents for nodes of the form $X[a(\overline{c})]$ is $1+|\mathit{rules}(p)|$, and the size of each CPT is $O(\mathit{max}(2,|\mathit{rules}(p)|)^{|\mathit{rules}(p')+2|})$.

As a result, $|\mathit{bn}(p, (\cal O,D,U))|$ is $O(|\mathit{rules}(p')|^2\cdot |g| \cdot \mathit{max}(2,|\mathit{rules}(p)|)^{2+|\mathit{rules}(p)|})$.
Furthermore, since $|g| \in O(|\mathit{edb}(p')|^{|\mathit{rules}(p')|})$, it follows that $|\mathit{bn}(p, (\cal O,D,U))|$ is $O(|\mathit{rules}(p')|^2\cdot |\mathit{edb}(p')|^{|\mathit{rules}(p')|} \cdot \mathit{max}(2,|\mathit{rules}(p)|)^{2+|\mathit{rules}(p)|})$.
Finally, since $|\mathit{edb}(p')| \in O(|\mathit{edb}(p)|)$ and $|\mathit{rules}(p') \in O(|\mathit{rules}(p)|)$, we can simplify the result as follows:
$|\mathit{bn}(p, (\cal O,D,U))|$ is $O(|\mathit{rules}(p)|^2\cdot |\mathit{edb}(p)|^{|\mathit{rules}(p)|} \cdot \mathit{max}(2,|\mathit{rules}(p)|)^{2+|\mathit{rules}(p)|})$.

\subsubsection{Complexity Proofs}
We first define the inference problem \textsc{Inf}.
Afterwards, we analyse its complexity.

\begin{problem}
\textsc{Inf} denotes the following decision problem:

\noindent
{\bf Input:} A first-order signature $\Sigma$, a finite domain $\mathbf{dom}$, a  $(\Sigma,\mathbf{dom})$-relaxed acyclic \problog{} program $p$, a set of ground literals $E$, and a ground atom $a(\overline{c})$.

\noindent
{\bf Output:} The probability of $a(\overline{c})$ given evidence in $E$.
\end{problem}

The data complexity of $\textsc{Inf}(\Sigma, \mathbf{dom}, p, E, a)$ for relaxed acyclic programs can be obtained by (1) fixing $\mathit{rules}(p)$ and varying only $\mathit{edb}(p)$ (and indirectly $\mathbf{dom}$), and (2) requiring the program to be relaxed acyclic.

The data complexity of the \textsc{Inf} problem for relaxed-acyclic programs $p$ is the complexity of the following decision problem:

\begin{problem}
Let $\Sigma$ be a first-order signature $\Sigma$, $R$ be a fixed set of \problog{} rules over $\Sigma$, $E$ be a set of ground literals,  $a(\overline{c})$ be a ground atom, and $(\cal O,D,U)$ be an ordering template, a disjointness template, and a uniqueness template.
$\textsc{Inf}_{\Sigma,R,E,a(\overline{c}),(\cal O,D,U)}^{ra}$ denotes the following problem:

\noindent
{\bf Input:} A set of probabilistic atoms $E'$ such that (1) atoms in $E$ and $a\overline{c}$ refer only to constant values in $E'$ and $R$, and (2) $R \cup E'$ is a relaxed acyclic \problog{} program and $(\cal O,D,U)$ is a witness for the acyclicity of $R \cup E'$.

\noindent
{\bf Output:} The probability of $a(\overline{c})$ given evidence in $E$.
\end{problem}

\begin{proposition}\thlabel{theorem:complexity}
$\textsc{Inf}_{\Sigma,R,E,a(\overline{c}),(\cal O,D,U)}^{ra}$ is in \textsc{Ptime}. % in terms of data complexity.
\end{proposition}

\begin{proof}
Let $\Sigma$ be a first-order signature $\Sigma$, $R$ be a fixed set of \problog{} rules over $\Sigma$, $E$ be a set of ground literals,  $a(\overline{c})$ be a ground atom, and  $(\cal O,D,U)$ be an ordering template, a disjointness template, and a uniqueness template.
We consider only inputs $E'$ such that $E' \cup R$ is a relaxed acyclic \problog{} program and  $(\cal O,D,U)$ is a witness for the acyclicity of $R \cup E'$.
The size of $E'$ is the sum of the sizes of all  atoms, where the size of an atom is its cardinality.

Let $e$ be the number of distinct constants occurring in $E' \cup R$ and $r$ be $|R|$.
Computing $\textsc{Inf}(\Sigma, E', p, E, a)$ can be done in the following steps:
\begin{compactenum}
\item Transform the original program $p$ into the program $p'$ (by applying the $\alpha$ and $\beta$ transformations).
\item Construct the ground version $g$ of the program $p'$.
\item Construct the Bayesian Network $\mathit{bn}(p,  (\cal O,D,U))$ from $g$.
\item Perform the inference on $\mathit{bn}(p, (\cal O,D,U))$.
\end{compactenum}
The first step can be performed in linear time in the size of $O(e + r)$.
The second step can be performed in $O(e^r)$ (because the grounding of $p'$ can be done in $O(e^{|\mathit{rules}(p')|})$ and $|\mathit{rules}(p')| \leq r$).
The third step can be performed in $O(r^4 \cdot e^{2 \cdot r} \cdot \mathit{max}(2,r)^{2+r})$ (constructing $N$ and $E$ can be done in $O(r^2 \cdot e^{2 \cdot r})$, whereas the $\mathit{cpt}$ can be constructed in $O(r^4 \cdot e^{r} \cdot \mathit{max}(2,r)^{2+r})$). 
The fourth step can be performed in $O(|\mathit{bn}(p)|)$ since $\mathit{bn}(p)$ is a forest of poly-trees (see \thref{theorem:bayesian:network:polytree}).
In particular, inference can be performed by (1) identifying the poly-tree that contains the atom $a(\overline{x})$ (in $O(|\mathit{bn}(p)|)$), (2) identifying the subset $E'' \subseteq E$ of all atoms in the poly-tree of $a(\overline{x})$ (again in $O(|\mathit{bn}(p)|)$), and (3) performing inference using belief-propagation (in $O(|\mathit{bn}(p)|)$~\cite{koller2009probabilistic}).
Therefore, the fourth step can be performed in $O(r^2 \cdot e^{r} \cdot \mathit{max}(2,r)^{2+r})$.
As a result, the answer to \textsc{Inf} can be computed in $O(r^2 \cdot e^{2 \cdot r} \cdot \mathit{max}(2,r)^{2+r})$.
Furthermore, since $r$ is fixed and the number of constants in $E'$ is $O(|E'|)$, there is a $k \in \mathbb{N}$ such that $e \in O(|E'| +k )$.
From this, it follows that there is a $k \in \mathbb{N}$ such that \textsc{Inf} can be computed in $O(|E'|^{k})$.
Hence, the complexity of $\textsc{Inf}_{\Sigma,R,E,a(\overline{c}),(\cal O,D,U)}^{ra}$ (and the data complexity of \textsc{Inf}) is \textsc{Ptime}.
\end{proof}

\textsc{Inf} for \problog{} is $\#P$-hard. 
This follows from the $\#P$-hardness of inference on arbitrary Bayesian networks (BNs)~\cite{koller2009probabilistic}, which can be encoded in \problog{}.
We now show that inference for \problog{} programs is $\#P$-hard even in terms of data complexity.
We do this using a reduction from the $\#\textsc{Sat}$ problem (counting the number of satisfying assignments of a propositional formula $\phi$), which is $\#P$-hard~\cite{pagourtzis2006complexity}.

\begin{proposition}
\textsc{Inf} is $\#P$-hard in terms of data complexity for \problog{} programs.
\end{proposition}

\begin{proof}
We show this by reducing the $\#\textsc{Sat}$ problem to inference for a \problog{} program $p$ where (1) the formula can be encoded in $\mathit{edb}(p)$, and (2) the rules are fixed.
Let $\phi$ be a propositional formula.

\para{First-order Signature} 
Let $\Sigma$ be the signature containing the following predicate symbols:
\begin{compactitem}
\item $\mathit{e}$ of arity $1$, which is used to store the identifiers associated to all sub-expressions of $\phi$,
\item $\mathit{state}$ of arity $1$, which is used to store the propositions in the model, 
\item $\mathit{sat}$ of arity $1$, which is used to denote whether an expression is satisfiable or not,
\item $\mathit{conj}$ of arity $3$ which is used to encode conjunctions,
\item $\mathit{disj}$ of arity $3$ which is used to encode disjunctions, 
\item $\mathit{neg}$ of arity $2$ which is used to encode negations, and
\item $\mathit{prop}$ of arity $2$ used to encode propositions.
\end{compactitem}

\para{Rules}
We now define a fixed set of rules encoding the semantics of propositional logic.
\begin{align*}
\mathit{sat}(x) &\leftarrow e(x), \mathit{prop}(x,y), \mathit{state}(y)\\
\mathit{sat}(x) &\leftarrow e(x), \mathit{neg}(x,y),\neg \mathit{sat}(y)\\
\mathit{sat}(x) &\leftarrow e(x), \mathit{conj}(x,y,z), \mathit{sat}(y),\mathit{sat}(z)\\
\mathit{sat}(x) &\leftarrow e(x), \mathit{disj}(x,y,z), \mathit{sat}(y)\\
\mathit{sat}(x) &\leftarrow e(x), \mathit{disj}(x,y,z), \mathit{sat}(z)
\end{align*}

\para{Encoding a formula $\phi$}
Given a formula $\phi$, we define the following encoding using probabilistic ground atoms.
We first associate to each sub-formula $\psi$ of $\phi$ a unique identifier $\mathit{id}_\psi$.
We also associate a unique identifier $\mathit{id}_p$ to each propositional symbol $p$.
For each sub-formula $\psi \wedge \gamma$, the ground atoms $\mathit{e}(\mathit{id}_{\psi \wedge \gamma})$ and $\mathit{conj}(\mathit{id}_{\psi \wedge \gamma}, \mathit{id}_\psi, \mathit{id}_\gamma)$ are in $E'$.
For each sub-formula $\psi \vee \gamma$, the ground atoms $\mathit{e}(\mathit{id}_{\psi \vee \gamma})$ and $\mathit{disj}(\mathit{id}_{\psi \vee \gamma}, \mathit{id}_\psi, \mathit{id}_\gamma)$ are in $E'$.
For each sub-formula $\neg \psi$, the ground atoms $\mathit{e}(\mathit{id}_{\neg \psi})$ and $\mathit{neg}(\mathit{id}_{\neg \psi}, \mathit{id}_\psi)$ are in $E'$.
For each sub-formula $\psi$ such that $\psi$ is a propositional symbol $p$, the ground atoms $\mathit{e}(\mathit{id}_{\psi})$ and $\mathit{prop}(\mathit{id}_{\psi}, \mathit{id}_p)$ are in $E'$. 
Finally, for each propositional symbol $p$, there is a propositional atom $\atom{\sfrac{1}{2}}{\mathit{state}(\mathit{id}_p)}$ in $E'$.

\para{Reduction}
There are $2^{n(\phi)}$ grounded instances of $ R \cup E'$, where $n(\phi)$ is the number of predicate symbols in $\phi$, each one with probability $\sfrac{1}{2^{n(\phi)}}$.
The grounded instances represent all possible assignments of $\top$ and $\bot$ to the proposition symbols in $\phi$.
It is easy to see that $\mathit{sat}(\mathit{id}_\phi)$ can be derived in a grounded instance iff it represents a model for $\phi$.
Therefore, the probability $\llbracket R \cup E'\rrbracket(\mathit{id}_\phi)$ is $\sfrac{k}{2^{n(\phi)}}$, where $k$ is the number of models that satisfy $\phi$.
\end{proof}

Note that the encoding shown in the previous proof can be tweaked to work for acyclic \problog{} programs but only for propositional formulae without repetitions of propositional symbols.
We remark that the $\#\textsc{Sat}$ problem restricted to formulae without repetitions is no longer $\#P$-hard.
Indeed, it can be solved in \textsc{Ptime} as follows:
given a formula $\phi$ without repetitions, we can construct a poly-tree boolean Bayesian Network $BN$ encoding $\phi$ in $O(|\phi|)$, where the nodes associated to sub-formulae of the from $p$ have a uniform probability distribution (i.e., $p$ is $\top$ with probability $\sfrac{1}{2}$ and $\bot$ with probability $\sfrac{1}{2}$).
Then, the probability associated to the root of $\phi$ is going to be $\sfrac{k}{2^{n(\phi)}}$, where $k$ is the number of satisfying assignments.
Since the inference on $\mathit{BN}$ can be done in linear time in $|\mathit{BN}|$ and $|BN| \in O(|\phi|)$, then the whole problem is in \textsc{Ptime}.

\subsubsection{Expressiveness}

Here we show that any BN that consists of a forest of poly-trees can be represented as a relaxed acyclic program.

\begin{proposition}
Any BN that is a forest of poly-trees can be represented as a relaxed acyclic \problog{} program.
\end{proposition}

\begin{proof}
\newcommand{\id}{\mathit{id}}
\newcommand{\cpt}{\mathit{cpt}}
\newcommand{\sw}{\mathit{sw}}

Let $\bn$ be a BN that is a forest of poly-trees.
We now construct the corresponding relaxed acyclic \problog{} program.
In particular, for each random variable $X$ in $\bn$, we show how to equivalently encode it as \problog{} rules.
Without loss of generality, we assume there is a unique mapping $\id$ from random variables in $\bn$ to predicate symbols identifiers.
With a slight abuse of notation, we use $X$ to refer both to the random variable and to the corresponding symbol $\id(X)$.

\para{Boolean Random Variables without parents}
For each boolean random variable $X$ such that (1) $p(X) = \emptyset$, and (2) $\cpt(X) = \{\top \mapsto v, \bot \mapsto (1-v)\}$, we introduce a probabilistic atom $\atom{v}{X}$.

\para{Non-Boolean Random Variables without parents}
For each non-boolean random variable $X$ with domain $\{ q_1, \ldots, q_n\}$ such that (1) $p(X) = \emptyset$, and (2) $\cpt(X) = \{q_1 \mapsto v_1, \ldots, q_n \mapsto v_n\}$, we introduce an annotated disjunction $\atom{v_1}{X(q_1)}; \ldots; \atom{v_n}{X(q_n)}$.

\para{Boolean Random Variables with parents}
Let $X$ be a boolean random variable $X$ with parents $p(X) = \{Y_1, \ldots, Y_n\}$, $\cpt(X)$ be the function associated with $X$, and $\overline{v}_1, \ldots, \overline{v}_m$ be all possible assignments to the variables in $p(X)$.
We introduce $m$ fresh predicate symbols $\sw_{X,\overline{v}_1}, \ldots, \sw_{X,\overline{v}_m}$.
For each $\overline{v}_i$, we introduce the probabilistic atom $\atom{c_i}{\sw_{X,\overline{v}_i}}$, where $\cpt(X)({\overline{v}_i},\top) = c_i$.
Finally, for each $\overline{v}_i$, we also introduce the following rules:
\begin{align*}
X &\leftarrow Y_1(\overline{v}_i(1)), \ldots, Y_n(\overline{v}_i(n)),\sw(\overline{l}_1), \sw_{X,\overline{v}_i}\\
&\vdots\\
X &\leftarrow Y_1(\overline{v}_i(1)), \ldots, Y_n(\overline{v}_i(n)),\sw(\overline{l}_{2^{i-1}}), \sw_{X,\overline{v}_i}
\end{align*}
where $Y_k(\top) = Y_k$, $Y_k(\bot) = \neg Y_k$, and $Y_1(v) = Y_1(v)$ if $v \not\in \{\top,\bot\}$, $\overline{l}_1, \ldots, \overline{l}_{2^{i-1}}$ are all possible values in $\{\top,\bot\}^{i-1}$, and $\sw(l_1, \ldots, l_{i-1})$ is the list of literals $\sw_{X,\overline{v}_1}(l_1), \ldots, \sw_{X,\overline{v}_{i-1}}(i-1)$.

\para{Non-Boolean Random Variables with parents}
Let $X$ be a non-boolean random variable $X$ with domain $\{q_1, \ldots, q_m\}$ and parents $p(X) = \{Y_1, \ldots, Y_n\}$, $\cpt(X)$ be the function associated with $X$, and $\overline{v}_1, \ldots, \overline{v}_m$ be all possible assignments to the variables in $p(X)$.
We introduce $m$ fresh predicate symbols $\sw_{X,\overline{v}_1}, \ldots, \sw_{X,\overline{v}_m}$.
For each $\overline{v}_i$, we introduce the probabilistic atom $\atom{c_i}{\sw_{X,\overline{v}_i}}$, where $\cpt(X)({\overline{v}_i},\top) = c_i$.
Finally, for each $\overline{v}_i$, we also introduce the following rules:
\begin{align*}
\atom{p_1^i}{X(q_1)}; \ldots; \atom{p_m^i}{X(q_m)} &\leftarrow Y_1(\overline{v}_i(1)), \ldots, Y_n(\overline{v}_i(n)),\\
& \qquad \sw(\overline{l}_1), \sw_{X,\overline{v}_i}\\
&\vdots\\
\atom{p_1^i}{X(q_1)}; \ldots; \atom{p_m^i}{X(q_m)} &\leftarrow Y_1(\overline{v}_i(1)), \ldots, Y_n(\overline{v}_i(n)),\\
& \qquad \sw(\overline{l}_{2^{i-1}}), \sw_{X,\overline{v}_i}
\end{align*}
where $Y_k(\top) = Y_k$, $Y_k(\bot) = \neg Y_k$, and $Y_1(v) = Y_1(v)$ if $v \not\in \{\top,\bot\}$, $\overline{l}_1, \ldots, \overline{l}_{2^{i-1}}$ are all possible values in $\{\top,\bot\}^{i-1}$, $\sw(l_1, \ldots, l_{i-1})$ is the list of literals $\sw_{X,\overline{v}_1}(l_1), \ldots, \sw_{X,\overline{v}_{i-1}}(i-1)$, and $\cpt(\overline{v}_i,q_j) = p_j^i$.

\para{Correctness of the encoding}
The encoding presented above encodes the CPTs for the random variables.
The probability that a random variable $X$ has value $v$ is exactly the same as the probability associated with the ground literal $X(v)$ (with the notation defined above).
The encoding for variables without parents directly encodes the corresponding CPTs.
For variables with parents, the only non-standard part is the use of $\sw(\overline{l}_{2^{i-1}}), \sw_{X,\overline{v}_i}$.
Note, however, that the literals in  $\sw(\overline{l}_{2^{i-1}})$ do not influence the derivation since there is a rule for any possible values for them.
We need them only for the encoding.
Therefore, also the encoding of random variables with parents directly encodes the corresponding CPTs.

\para{Relaxed Acyclicity}
Let $p$ be the program produced by the above construction.
It is easy to see that all predicates associated with non-boolean random variables are CPT-like (as we just encoded their CPTs).
Therefore, the program $\beta(p)$ is obtained by removing all constant values associated with the domains of non-boolean random values.
Finally, the program $\alpha(\beta(p))$ collapses the rules for random variables with parents into a single rule (this follows from our construction and the use of $\sw(\overline{l}_{2^{i-1}}), \sw_{X,\overline{v}_i}$ in the rules' bodies).
The acyclicity of  $p'= \alpha(\beta(p))$ follows from 
(1) $p'$ does not contain free-variables (i.e., both the literals and the ground atoms are only propositional facts),
(2) $\bn$ is a forest of poly-trees, and
(3) each predicate symbol $\sw_{X,\overline{v}_i}$ occurs only in the rule of $X$.
From (1), it follows that all rules are both strongly and weakly connected.
From (2) and (3), it follows that (a) there are no directed cycles in $\graph(p)$, and (b) for all undirected cycles $U$ in $\graph(p)$, there are $U',U''$ such that $U$ is equivalent to $U' \concat U''$ and $U'$ is $P' \xleftarrow{r,i} P \xrightarrow{r,i} P'$ for some $P, P', r, i$.
Since all rules are both strongly and weakly connected, it directly follows that the undirected unsafe structure $\langle P \xrightarrow{r,i} P', P \xrightarrow{r,i} P', P', U''\rangle$ is guarded.
Therefore, $p'$ is acyclic and $p$ is a relaxed acyclic \problog{} program.
\end{proof}

\clearpage

\section{ANGERONA}\label{app:enforcement:mechanism:proofs}
Here, we present additional details about \tool{}.
We also prove its security, complexity, and completeness.

\subsection{Checking Query Security}

In the following, let $D = \langle \Sigma, \mathbf{dom}\rangle$ be a database schema. 
Without loss of generality, we focus only on relational calculus formulae $\phi$ where  no distinct pair of quantifiers binds the same variable.

\para{Normal Form}
We say that a formula $\psi \wedge \neg\gamma$ is \emph{guarded} iff $\mathit{free}(\gamma) \subseteq \mathit{free}(\psi)$.
We say that a relational calculus formula $\phi$ is in \emph{Normal Form} (NF) iff 
(1) $\phi$ uses only existential quantifiers,
(2) negation is used only in sub-formulae of the form $\psi \wedge \neg\gamma$ and it is always guarded, 
(3) for any sub-formula $\psi \vee \gamma$, $\mathit{free}(\psi) = \mathit{free}(\gamma)$, 
(4) no distinct pair of quantifiers binds the same variable, and 
(5) there are no equality and inequality constraints.
Most of the time, domain-independent relational calculus formulae can be easily written in NF by just re-arranging sub-formulae.
We remark that any domain-independent relational calculus formula can be written in NF by 
(1) extending the database schema with two relations $\mathit{eq}$ and $\mathit{neq}$ encoding $=$ and $\neq$ among constants in $\mathbf{dom}$ (this is always possible because $\mathbf{dom}$ is finite),
(2) renaming the quantified variables in a unique way,
(3) replacing all universally quantified sub-formula $\forall x. \phi_e \rightarrow \psi$ with the equivalent existentially quantified version $\neg \exists x.\, \phi_e \wedge \neg \psi$,
(4) replacing each negated sub-formula $\neg \psi$ with the equivalent sub-formula $(\bigwedge_{x \in \mathit{free}(\psi)} \mathit{adom}(x)) \wedge \neg \psi$, where $\mathit{adom}(x)$ is $\bigvee_{R \in D} \bigvee_{1 \leq i \leq |R|} \exists x_1, \ldots, x_{i-1}, x_{i+1}, \ldots, x_{|R|}.\ R(x_1, \ldots, x_{i-1}, x, x_{i+1}, \ldots, x_{|R|}) $, and
(5) replacing sub-formulae of the form $\psi \vee \gamma$ with the equivalent formula $(\bigwedge_{x \in \mathit{free}(\gamma) \setminus \mathit{free}(\psi)} \mathit{adom}(x) \wedge \psi) \vee (\bigwedge_{x \in \mathit{free}(\psi) \setminus \mathit{free}(\gamma)} \mathit{adom}(x) \wedge \gamma)$.
Note that the resulting NF formula is equivalent to the original one (because the rewriting does not modify the formula's semantics).
Therefore, in the following we consider only NF relational calculus formulae.

\para{From Relational Calculus to Logic Programming}
Let $\phi$ be a NF relational calculus sentence.
We denote by $\mathit{sub}(\phi)$ the sequence $\phi_1, \ldots, \phi_n$ of all $\phi$'s sub-formulae ordered from the smallest to the largest, i.e., $\phi_n = \phi$, and by $\mathit{DIsub}(\phi)$ the sequence obtained by removing from the sequence $\mathit{sub}(\phi)$ (1)  all formulae of the form $\neg \psi$ , and (2) removing all sub-formulae consisting just of equality and inequality terms.
Since $\phi$ is in NF, all negated sub-formulae in $\phi$ appear only in NF in the sequence $\mathit{DIsub}(\phi)$.

The function $\mathit{PL}(\phi)$ encodes any NF relational calculus sentence as a set of equivalent logic programming rules.
We associate a unique predicate symbol $H_i$ and a set of rules $r(H_i)$  to each $\psi_i$ in $\mathit{DIsub}(\phi)$ as follows:
\begin{compactitem}
\item If $\psi_i := R(\overline{x}_i)$, for some $R \in \Sigma$, then $r(H_i)$ is $\{H_i(\overline{x}_i) \leftarrow R(\overline{x}_i)\}$, where $\overline{x}_i$ are $\psi_i$'s free variables.
\item If $\psi_i := \psi_j \wedge \neg \psi_j$, then $r(H_i)$ is $\{H_i(\overline{x}_i) \leftarrow H_j(\overline{x}_j), \neg H_k(\overline{x}_k)\}$, where $\overline{x}_i$ are $\psi_i$'s free variables, $\overline{x}_j$ are $\psi_j$'s free variables, and $\overline{x}_k$ are $\psi_k$'s free variables (note that $\overline{x}_i = \overline{x}_j$ and $\overline{x}_k \subseteq \overline{x}_j$).
\item If $\psi_i = \psi_j \wedge \psi_k$, then $r(H_i)$ is $\{H_i(\overline{x}_i) \leftarrow H_j(\overline{x}_j), H_k(\overline{x}_k)\}$, where $\overline{x}_i$ are $\psi_i$'s free variables, $\overline{x}_j$ are $\psi_j$'s free variables, and $\overline{x}_k$ are $\psi_k$'s free variables.
\item If $\psi_i = \psi_j \vee \psi_k$, then $r(H_i)$ contains the rules $H_i(\overline{x}_i) \leftarrow H_j(\overline{x}_i)$ and $H_i(\overline{x}_i) \leftarrow  H_k(\overline{x}_i)$ (note that $\mathit{free}(\psi_i) = \mathit{free}(\psi_j) = \mathit{free}(\psi_k)$).
\item If $\psi_i = \exists x.\, \psi_j$, then $r(H_i)$ is $\{H_i(\overline{x}_i) \leftarrow H_j(\overline{x}_j), c_1, \ldots, c_n\}$, where $\overline{x}_i$ are $\psi_i$'s free variables and $\overline{x}_j$ are $\psi_j$'s free variables (i.e., those in $\overline{x}_i$ and $x$).
\end{compactitem}
Furthermore, we denote $\mathit{head}(\phi)$ the predicate symbol associated to $\psi_m$.

\para{Checking Query Security}
\tool{}'s  enforcement algorithm is shown in Algorithm~\ref{figure:angerona:algorithm} (repeated from \S\ref{sect:enforcement}).

\begin{algorithm}[tp]
\DecMargin{10em}
\DontPrintSemicolon
\KwIn{A system configuration $C$, a system state $s = \langle \mathit{db}, U, P \rangle$, a history $h$, a $C$-\atklog{} model $\mathit{ATK}$, and an action $\langle u,q\rangle$.}
\KwOut{The security decision in $\{\top,\bot\}$.}

\SetKwProg{Fn}{function}{}{}
\SetKwIF{If}{ElseIf}{Else}{if}{}{else if}{else}{endif}
\Begin{
  \For{$\langle u, \psi, l\rangle \in \mathit{secrets}(P,u)$}
  	{
    	  \If{$\mathit{secure}(C, \mathit{ATK}, h, \langle u, \psi, l\rangle)$}
  	  {
		\If{$\mathit{pox}(C, \mathit{ATK}, h, \langle u, q \rangle, \top)$}{
			$h' :=  h \cdot \langle \langle u,q\rangle, \top, \top\rangle $\;
			\If{$\neg \mathit{secure}(C, \mathit{ATK}, h', \langle u, \psi, l\rangle)$}
  	  		{
  	  			\Return{$\bot$}\;
  	  		}
		}  	  
  	  	\If{$\mathit{pox}(C, \mathit{ATK}, h, \langle u,  q \rangle, \bot)$}{
			$h' :=  h \cdot \langle \langle u,q\rangle, \top, \bot \rangle$\;
			\If{$\neg \mathit{secure}(C, \mathit{ATK}, h', \langle u, \psi, l\rangle)$}
  	  		{
  	  			\Return{$\bot$}\;
  	  		}
		}
  	  }
  	}
  \Return{$\top$}
}
\;
%\;
\Fn{$\mathit{secure}(\langle D, \Gamma \rangle, \mathit{ATK},   h, \langle u, \psi, l\rangle)$}{
	%${\cal K} :=	 \llbracket K \rrbracket_{u'}(r)$\;
	%${\cal K} := \mathit{knowledge}(h,u)$\;
		$p := \mathit{ATK}(u)$\;
		%\For{$\phi \in {\cal K}$}
		\For{$\phi \in \mathit{knowledge}(h,u)$}
		{
			$p := p \cup \mathit{PL}(\phi) \cup \{\mathit{evidence}(\mathit{head}(\phi), \mathtt{true})\}$\;
		}
		$p := p \cup \mathit{PL}(\psi)$\;
		\Return{$\llbracket p\rrbracket_D(\mathit{head}(\psi)) < l$}	
%	}
} 
\;
\Fn{$\mathit{pox}(\langle D, \Gamma \rangle, \mathit{ATK},   h, \langle u,\psi \rangle, v)$}{
	%${\cal K} := \mathit{knowledge}(h,u)$\;
		$p := \mathit{ATK}(u)$\;
	%\For{$\phi \in {\cal K}$}
	\For{$\phi \in \mathit{knowledge}(h,u)$}
	{
		$p := p \cup \mathit{PL}(\phi) \cup \{\mathit{evidence}(\mathit{head}(\phi), \mathtt{true})\}$\;
	}
	\If{$v = \top$}{
		$p := p \cup \mathit{PL}(\psi)$\;
		\Return{$\llbracket p\rrbracket_D(\mathit{head}(\psi)) > 0$}
	}
	\Else
	{
		$p := p \cup \mathit{PL}(\neg \psi)$\;
		\Return{$\llbracket p\rrbracket_D(\mathit{head}(\neg \psi)) > 0$}	
	}	
} 
\caption{\tool{} Enforcement Algorithm (repeated from \S\ref{sect:enforcement}).}
\label{figure:angerona:algorithm}
\end{algorithm}

\subsection{Security Proof}

Let $C = \langle D, \Gamma\rangle$ be a system configuration and $h$ be a $C$-history. %, $\langle \mathit{db}, U, P\rangle$ be a $C$-state, and $p$ be a $D$-knowledge program.
The set $\mathit{knowledge}(h,u)$ is $\{ \phi \mid \exists i.\, h(i) = \langle \langle u, \phi \rangle, \top, \top \rangle \} \cup \{ \neg \phi \mid \exists i.\, h(i) = \langle \langle u, \phi \rangle, \top, \bot \rangle \}$.

We first prove that \tool{}'s result depends just on the queries in the history, and not on the actual database's state. 

\begin{proposition}\thlabel{theorem:angerona:same:result:knowledge}
Let $C = \langle D, \Gamma\rangle$ be a system configuration,  $\mathit{ATK}$ be a $C$-\atklog{} model, and $\langle u,q\rangle$ be a $C$-query.   
$\tool{}(C,s,h,\mathit{ATK},\langle u,q\rangle) = \tool{}(C,s',h',\mathit{ATK},\langle u,q\rangle)$ whenever $\mathit{knowledge}(h,u) = \mathit{knowledge}(h',u)$.
\end{proposition}

\begin{proof}
Let $C = \langle D, \Gamma\rangle$ be a system configuration, $\mathit{ATK}$ be an \atklog{} model, and $\langle u,q\rangle$ be a $C$-query.   
Furthermore, let $s,s'$ be $C$-states and $h,h'$ be $C$-histories such that  $\mathit{knowledge}(h,u) = \mathit{knowledge}(h',u)$.
Assume, for contradiction's sake, that $\tool{}(C,s,h,\mathit{ATK},\langle u,q\rangle) \neq \tool{}(C,s',h',\mathit{ATK},\langle u,q\rangle)$.
Without loss of generality, assume that $\tool{}(C,s,h,\mathit{ATK},\langle u,q\rangle) = \top$ and $\tool{}(C,s',h',\mathit{ATK},\langle u,q\rangle) = \bot$.
This happens iff either the result of $\mathit{pox}$ or $\mathit{secure}$ is different in the two cases.
The results of $\mathit{pox}$ and $\mathit{secure}$, however, depend just on the attacker's initial beliefs, the query $q$, and the set of formulae derive $\mathit{knowledge}(h,u)$.
From this and $\mathit{knowledge}(h,u) = \mathit{knowledge}(h',u)$, it follows that $\tool{}(C,s,h,\mathit{ATK},\langle u,q\rangle) = \tool{}(C,s',h',\mathit{ATK},\langle u,q\rangle)$.
\end{proof}

We now prove a key result for our security proof, namely that considering only the sentences in $\mathit{knowledge}(h,u)$ is enough to determine secrecy-preservation.

\begin{proposition}\thlabel{theorem:knowledge:is:enough}
Let $C = \langle D, \Gamma\rangle$ be a system configuration,  $\mathit{ATK}$ be a $C$-\atklog{} model, $f$ be the $C$-\acf{} obtained by parametrizing Algorithm~\ref{figure:angerona:algorithm} with $C$ and $\mathit{ATK}$, $u \in {\cal U}$ be a user, and $r = \langle s,h \rangle$ be a $(C,f)$-run.   
The following fact holds: $\llbracket r \rrbracket_{\sim_u} = \{\mathit{db} \in \Omega_D^\Gamma \mid \bigwedge_{\phi \in \mathit{knowledge}(h,u)} [\phi]^{\mathit{db}} = \top\}$. 
\end{proposition}

\begin{proof}

\para{$(\Rightarrow)$} 
Let $C = \langle D, \Gamma\rangle$ be a system configuration,  $\mathit{ATK}$ be an \atklog{} model, $f$ be the $C$-\acf{} obtained by parametrizing Algorithm~\ref{figure:angerona:algorithm} with $C$ and $\mathit{ATK}$, $u \in {\cal U}$ be a user,  and $r = \langle \langle \mathit{db},U,P \rangle,h \rangle$ be a $(C,f)$-run. 
Furthermore, let $r'  = \langle \langle \mathit{db}',U',P' \rangle,h' \rangle$ be a run in $[ r ]_{\sim_u}$.
From this, it follows that $h|_{u} = h'|_{u}$.
From this,  it follows that all queries in $\mathit{knowledge}(h,u)$ have the same result in $\mathit{db}$ and $\mathit{db}'$.
Therefore, $\mathit{db}' \in \{\mathit{db} \in \Omega_D^\Gamma \mid \bigwedge_{\phi \in \mathit{knowledge}(h,u)} [\phi]^{\mathit{db}} = \top\}$.

\para{$(\Leftarrow)$}
Let $C = \langle D, \Gamma\rangle$ be a system configuration, $\mathit{ATK}$ be an \atklog{} model, $f$ be the $C$-\acf{} obtained by parametrizing Algorithm~\ref{figure:angerona:algorithm} with $C$ and $\mathit{ATK}$, $u \in {\cal U}$ be a user,  and $r = \langle \langle \mathit{db},U,P \rangle,h \rangle$ be a $(C,f)$-run. 
Furthermore, let $\mathit{db}'$ be a database state in $\{\mathit{db}' \in \Omega_D^\Gamma \mid \bigwedge_{\phi \in \mathit{knowledge}(h,u)}[\phi]^{\mathit{db}'}\}$.
We now construct a $C$-state  $s$ and an history $h'$ such that (1) $\langle s, h' \rangle$ is a run, (2) $s.\mathit{db} = \mathit{db}'$, and (3) $\langle s, h' \rangle \in [ r ]_{\sim_u}$.
The $C$-state $s$ is defined as $\langle \mathit{db}', U, P\rangle$, whereas the history $h' = h|_{u}$.
We claim that $\langle s, h' \rangle$ is a run.
From this and  $h' = h|_{u}$, it follows that $\langle s, h' \rangle \sim_{u} r$.
From this, it follows that $\mathit{db}' \in \llbracket r \rrbracket_{\sim_u}$.

To prove our claim that $\langle \langle \mathit{db}', U, P\rangle, h|_u \rangle$ is a run, we prove the stronger fact that for all $0 \leq i \leq |h|$, $\langle \langle \mathit{db}', U, P\rangle, h^i|_u \rangle$ is a run.
We prove this by induction on $i$.

\para{Base case}
The base case $\langle \langle \mathit{db}', U, P\rangle, h^0|_u \rangle$ is trivial since $\mathit{db}' \in \Omega_D^\Gamma$ and $h^0|_u = \epsilon$, therefore $\langle \langle \mathit{db}', U, P\rangle, h^0|_u \rangle$ is a run.

\para{Induction Step}
Assume that $\langle \langle \mathit{db}', U, P\rangle, h^{i-1}|_u \rangle$ is a run.
We now show that $\langle \langle \mathit{db}', U, P\rangle, h^{i}|_u \rangle$ is a run as well.
Let $\langle \langle u',q\rangle, a, \mathit{res} \rangle$ be the last $C$-event in $h^i$.
There are two cases:
\begin{compactitem}
\item $u' \neq u$.
From this, $h^i|_u = h^{i-1}|_u$.
Therefore, $\langle \langle \mathit{db}', U, P\rangle, h^{i-1}|_u \rangle = \langle \langle \mathit{db}', U, P\rangle, h^{i}|_u \rangle$ and our claim directly follows from the induction hypothesis.

\item $u' = u$.
From this, $h^i|_u = h^{i-1}|_u \concat \langle \langle u,q\rangle, a, \mathit{res} \rangle$.
Assume, for contradiction's sake, that $\langle \langle \mathit{db}', U, P\rangle, h^{i}|_u \rangle$ is not a run.
From $\langle \langle \mathit{db}', U, P\rangle, h^{i-1}|_u \rangle$ is a run (which directly follows from the induction hypothesis), it follows that there are only three cases:
\begin{compactenum}
\item $f(\langle \mathit{db}', U,P\rangle, \langle u,q\rangle, h^{i-1}|_u ) \neq a$.
Since $a$ is the security decision associated with the last event in $h^i$, it follows that $a = f(\langle \mathit{db}, U,P\rangle, \langle u,q\rangle, h^{i-1} )$.
From this, it follows that $f(\langle \mathit{db}, U,P\rangle, \langle u,q\rangle, h^{i-1} ) \neq f(\langle \mathit{db}', U,P\rangle, \langle u,q\rangle, h^{i-1}|_u )$.
From $\mathit{knowledge}$'s definition, it follows that $\mathit{knowledge}(h^{i-1}, u) = \mathit{knowledge}(h^{i-1}|_{u}, u)$.
From this and \thref{theorem:angerona:same:result:knowledge}, it follows that $f(\langle \mathit{db}, U,P\rangle, \langle u,q\rangle, h^{i-1} ) = f(\langle \mathit{db}', U,P\rangle, \langle u,q\rangle, h^{i-1}|_u )$, leading to a contradiction.

\item $a = \bot$ but $\mathit{res} \neq \dagger$. This contradicts the fact that the history is derived from $h$, which comes from a proper run.

\item $a = \top$ but $\mathit{res} \neq [q]^{\mathit{db}'}$.
From $\mathit{res} \neq [q]^{\mathit{db}'}$ and $\mathit{res} = [q]^{\mathit{db}}$ (since $\mathit{res}$ comes from the run $r$), it follows that $[q]^\mathit{db} \neq  [q]^{\mathit{db}'}$.
There are two cases: 
\begin{compactitem}
\item $[q]^\mathit{db} = \top$. From this and the definition of $\mathit{knowledge}$, it follows that $q \in \mathit{knowledge}(h,u)$.
Therefore, $[q]^{\mathit{db}'} = \top$ follows from $\mathit{db}' \in \{\mathit{db}' \in \Omega_D^\Gamma \mid \bigwedge_{\phi \in \mathit{knowledge}(h,u)}[\phi]^{\mathit{db}'}\}$, leading to a contradiction.

\item $[q]^\mathit{db} = \bot$. From this and the definition of $\mathit{knowledge}$, it follows that $\neg q \in \mathit{knowledge}(h,u)$.
From this and $\mathit{db}' \in \{\mathit{db}' \in \Omega_D^\Gamma \mid \bigwedge_{\phi \in \mathit{knowledge}(h,u)}[\phi]^{\mathit{db}'}\}$, it follows that $[\neg q]^{\mathit{db}'} = \top$.
From this, $[q]^{\mathit{db}'} = \bot = [q]^\mathit{db}$, leading to a contradiction.
\end{compactitem} 

\end{compactenum}
\end{compactitem}
This completes the proof of our claim.
\end{proof}

We now prove that \tool{} provides the desired security guarantees.

\begin{proposition}
Let $C = \langle D, \Gamma\rangle$ be a system configuration, $\mathit{ATK}$ be a $C$-\atklog{} model, $f$ be the $C$-\acf{} obtained by parametrizing Algorithm~\ref{figure:angerona:algorithm} with $C$ and $\mathit{ATK}$, and 
$\mathit{ATK}' =  \lambda u \in {\cal U}. \llbracket \mathit{ATK}(u) \rrbracket_D$ be the $(C,f)$-attacker model associated to $\mathit{ATK}$.    
The \acf{} shown in Algorithm~\ref{figure:angerona:algorithm}, parametrized with $\mathit{ATK}$,  provides \confidentiality{} with respect to $C$ and $\mathit{ATK}'$.
\end{proposition}

\begin{proof}
Let $C = \langle D, \Gamma\rangle$ be a system configuration, $\mathit{ATK}$ be a $C$-\atklog{} model, $f$ be the $C$-\acf{} obtained by parametrizing Algorithm~\ref{figure:angerona:algorithm} with $C$ and $\mathit{ATK}$, and   
$\mathit{ATK}' =  \lambda u \in {\cal U}. \llbracket \mathit{ATK}(u) \rrbracket_D$ be the $(C,f)$-attacker model associated to $\mathit{ATK}$.   
Furthermore, let $f$ be the \acf{} shown in Algorithm~\ref{figure:angerona:algorithm}.
Assume, for contradiction's sake, that $f$ does not provide confidentiality with respect to $C$ and $\mathit{ATK}'$.
From this, it follows that there is a run $r =\langle \langle \mathit{db}, U, P\rangle, h \rangle$, a user $u \in U$, a secret $\langle u, \phi, l \rangle \in P$, and a $0 \leq i \leq |h|-1$ such that  $\llbracket \mathit{ATK}' \rrbracket(u,r^i) ( \llbracket \phi \rrbracket) < l$ and $\llbracket \mathit{ATK}' \rrbracket(u,r^{i+1}) ( \llbracket \phi \rrbracket) \geq l$.
As a result, the $i$-th $C$-query is the one that leaked information.
There are two cases:
\begin{compactitem}
\item The $C$-query is $\langle u,q\rangle$, for some $q$.
There are three cases:
\begin{compactitem}
\item The result of the PDP is $\top$ and the query $q$ holds in the current state.
From this, it follows that there is a database state (namely, the one in $r$) where the query $q$ holds and the database is consistent with $\mathit{knowledge}(h^i,u)$.
From this, it follows that $\llbracket \mathit{ATK}' \rrbracket (u,r^{i})(\llbracket q\rrbracket) > 0$.
From this and \thref{theorem:pox:sound:complete}, it follows that $\mathit{pox}(C, \mathit{ATK} ,h^i,\langle u,q \rangle, \top)$ returns true.
From this and the fact that the query has been authorized, it follows that $\mathit{secure}(C, \mathit{ATK} ,h^{i+1},\langle u, \phi, l \rangle)$ returned true.
From this, the fact that $r$ is actually a run (since $q$ is satisfied in $r.\mathit{db}$), and \thref{theorem:secure:sound:complete}, it follows that  $\llbracket \mathit{ATK}' \rrbracket(u,r^{i+1}) ( \llbracket \phi \rrbracket) < l$, leading to a contradiction.

\item The result of the PDP is $\top$ and the query $q$ does not hold in the current state.
From this, it follows that there is a database state (namely, the one in $r$) where the query $\neg q$ holds and the database is consistent with $\mathit{knowledge}(h^i,u)$.
From this, it follows that $\llbracket \mathit{ATK}' \rrbracket (u,r^{i})(\llbracket \neg q\rrbracket) > 0$.
From this and \thref{theorem:pox:sound:complete}, it follows that $\mathit{pox}(C, \mathit{ATK} ,h^i,\langle u,q \rangle, \bot)$ returns true.
From this and the fact that the query has been authorized, it follows that $\mathit{secure}(C, \mathit{ATK} ,h^{i+1},\langle u, \phi, l \rangle)$ returned true.
From this, the fact that $r$ is actually a run (since $q$ is satisfied in $r.\mathit{db}$), and \thref{theorem:secure:sound:complete}, it follows that  $\llbracket \mathit{ATK}' \rrbracket(u,r^{i+1}) ( \llbracket \phi \rrbracket) < l$, leading to a contradiction.

\item The result of the PDP is $\bot$ (namely the query is not authorized).
From this and \thref{theorem:angerona:same:result:knowledge}, it follows that for any run $r' \sim_u r^i$ the PDP result is the same.
Therefore, $\llbracket  r^{i} \rrbracket_{\sim_u} = \llbracket  r^{i+1} \rrbracket_{\sim_u}$.
From this, $\llbracket \mathit{ATK}' \rrbracket(u,r^i) ( \llbracket \phi \rrbracket) < l$, and \atklog{} semantics, it follows that $\llbracket \mathit{ATK}' \rrbracket(u,r^{i+1}) ( \llbracket \phi \rrbracket) < l$, leading to a contradiction.
\end{compactitem}

\item The $C$-query is $\langle u',q\rangle$, where $u' \neq u$ and $q$ is a query.
From this, it follows that $\llbracket \mathit{ATK}' \rrbracket(u,r^i) ( \llbracket \phi \rrbracket) = \llbracket \mathit{ATK}' \rrbracket(u,r^{i+1}) ( \llbracket \phi \rrbracket)$ since $u$'s belief does not change in response to a query from another user (since $r^i|_u = r^{i+1}|_u$).
From this and $\llbracket \mathit{ATK}' \rrbracket(u,r^{i}) ( \llbracket \phi \rrbracket) < l$, it follows that both $\llbracket \mathit{ATK}' \rrbracket(u,r^{i+1}) ( \llbracket \phi \rrbracket) < l$ and $\llbracket \mathit{ATK}' \rrbracket(u,r^{i+1}) ( \llbracket \phi \rrbracket) \geq l$, leading to a contradiction.
\end{compactitem}
Since all cases ended in contradiction, this completes the proof of our claim.
\end{proof}

\begin{proposition}\thlabel{theorem:secure:sound:complete}
Let $C = \langle D, \Gamma\rangle$ be a system configuration, $\mathit{ATK}$ be a $C$-\atklog{} model, $f$ be the $C$-\acf{} obtained by parametrizing Algorithm~\ref{figure:angerona:algorithm} with $C$ and $\mathit{ATK}$, 
$\mathit{ATK}' = \lambda u \in {\cal U}. \llbracket \mathit{ATK}(u) \rrbracket_D$ be the $(C,f)$-attacker model associated to $\mathit{ATK}$,
 $r$ be a run in $\mathit{runs}(C,f)$, and $\langle u,\psi,l \rangle$ be a secret in $r.S$.
Then, $\mathit{secure}(C, \mathit{ATK} , h,\langle u, \psi, l\rangle)$ returns $\mathit{true}$ iff $\llbracket \mathit{ATK}' \rrbracket (u,\langle s, h\rangle) (\llbracket \psi \rrbracket) < l$, for any state $s$ that is compatible with $h$.
\end{proposition}

\begin{proof}
Let $r = \langle s,h\rangle$ be a run such that $s$ is compatible with $h$.
The $\mathit{secure}$ procedure returns $\llbracket \mathit{ATK}(u) \rrbracket_D(\llbracket \psi \rrbracket | \bigcap_{\phi \in \mathit{knowledge}(h,u)} \llbracket \phi \rrbracket) < l$.
From \thref{theorem:knowledge:is:enough}, it follows that $\llbracket r \rrbracket_{\sim_u} = \{\mathit{db} \in \Omega_D^\Gamma \mid \bigwedge_{\phi \in \mathit{knowledge}(h,u)} [\phi]^{\mathit{db}} = \top\} = \bigcap_{\phi \in \mathit{knowledge}(h,u)} \llbracket \phi \rrbracket$.
We can therefore rewrite $\mathit{secure}$'s result as $\llbracket \mathit{ATK}(u) \rrbracket_D(\llbracket \psi \rrbracket | \llbracket \langle s, h\rangle \rrbracket_{\sim_u}) < l$.
This is exactly the definition of $\llbracket \mathit{ATK}' \rrbracket (u,\langle s, h\rangle) (\llbracket \psi \rrbracket) < l$.
\end{proof}

\begin{proposition}\thlabel{theorem:pox:sound:complete}
Let $C = \langle D, \Gamma\rangle$ be a system configuration, $\mathit{ATK}$ be a $C$-\atklog{} model, $f$ be the $C$-\acf{} obtained by parametrizing Algorithm~\ref{figure:angerona:algorithm} with $C$ and $\mathit{ATK}$, 
$\mathit{ATK}' = \lambda u \in {\cal U}. \llbracket \mathit{ATK}(u) \rrbracket_D$ be the $(C,f)$-attacker model associated to $\mathit{ATK}$,
 $r$ be a run in $\mathit{runs}(C,f)$, and $\langle u,q \rangle$ be a query.
Then, $\mathit{pox}(C, \mathit{ATK} , h,\langle u,q \rangle, v)$ returns $\mathit{true}$ iff $\llbracket \mathit{ATK}' \rrbracket (u,\langle s, h\rangle) (\llbracket \mathit{r}(q,v) \rrbracket) > 0$ for any state $s$ that is compatible with $h$, where $r(q,\top) = q$ and $r(q,\bot) = \neg q$.
\end{proposition}

\begin{proof}
Let $r = \langle s,h\rangle$ be a run such that $s$ is compatible with $h$.
The $\mathit{pox}$ procedure returns $\llbracket \mathit{ATK}(u) \rrbracket_D(\llbracket \mathit{r}(q,v) \rrbracket | \bigcap_{\phi \in \mathit{knowledge}(h,u)} \llbracket \phi \rrbracket) > 0$.
From \thref{theorem:knowledge:is:enough}, it follows that $\llbracket r \rrbracket_{\sim_u} = \{\mathit{db} \in \Omega_D^\Gamma \mid \bigwedge_{\phi \in \mathit{knowledge}(h,u)} [\phi]^{\mathit{db}} = \top\} = \bigcap_{\phi \in \mathit{knowledge}(h,u)} \llbracket \phi \rrbracket$.
We can therefore rewrite $\mathit{pox}$'s result as $\llbracket \mathit{ATK}(u) \rrbracket_D(\llbracket \mathit{r}(q,v) \rrbracket | \llbracket \langle s, h\rangle \rrbracket_{\sim_u}) < l$.
This is exactly the definition of $\llbracket \mathit{ATK}' \rrbracket (u,\langle s, h\rangle) (\llbracket \mathit{r}(q,v) \rrbracket) > 0$.
\end{proof}

\subsection{Complexity Proof}

First, we formalize literal queries, a fragment of relational calculus that can be composed with relaxed acyclic programs without modifying acyclicity.
Afterwards, we prove that for acyclic \atklog{} models and literal queries, \tool{} has \textsc{Ptime} data complexity.

A literal query is a quantifier-free relational calculus $(\Sigma, \mathbf{dom})$-formula either of the form $R(\overline{c})$ or $\neg R(\overline{c})$, where $R \in \Sigma$ and $\overline{c} \in \mathbf{dom}^{|R|}$.
We say that a literal query $R(\overline{c})$ (or $\neg R(\overline{c})$) is \emph{compatible with a relaxed acyclic program $p$} (with witness $(\cal O,D,U)$) iff $\mu_{p,\cal D,U}(R) \neq \emptyset$ implies  $\overline{c}\downarrow_{K} \in \mathit{pdom}(R,p,\mu_{p,\cal D,U}(R))$.
We now show that literal queries can be composed with acyclic  \problog{} programs without introducing cycles.

\begin{proposition}\thlabel{theorem:composition:literal:queries}
Let $D = \langle \Sigma, \mathbf{dom}\rangle$ be a database schema, $p$ be a relaxed acyclic program, $(\cal O,D,U)$ be a witness for $p$, $R$ be a predicate symbol in $\Sigma$, $\overline{c} \in \mathbf{dom}^{|R|}$ be a tuple, and $\phi$ be a boolean $D$-query.
If $\phi$ is a literal query compatible with $p$, then $p \cup \mathit{PL}(\phi)$ is a relaxed acyclic \problog{} program as well.
\end{proposition}

\begin{proof}
Let $D = \langle \Sigma, \Gamma\rangle$ be a database schema, $p$ be a relaxed acyclic program, and $\phi$ be a boolean $D$-query.
Furthermore, let $\phi$ be a literal query.
Assume, for contradiction's sake, that $p \cup \mathit{PL}(\phi)$ is not a relaxed acyclic program.
This happens iff the program $\alpha(\beta(p \cup \mathit{PL}(\phi)))$ is not an acyclic program.
This happens iff the program $\alpha(\beta(p \cup \mathit{PL}(\phi)))$ contains an unguarded unsafe structure $S$ in the dependency graph of $\alpha(\beta(p \cup \mathit{PL}(\phi)))$.
The only interesting case is when $S$ contains rules from  $\mathit{PL}(\phi)$.
If $\phi = R(\overline{c})$, then $\mathit{PL}(\phi)$ contains a single rule $\mathit{fresh} \leftarrow R(\overline{c})$, where $\mathit{fresh}$ is a fresh predicate symbol.
Similarly, if $\phi = \neg R(\overline{c})$, then $\mathit{PL}(\phi)$ contains two rules $\mathit{fresh}_1 \leftarrow \mathit{fresh}_2$ and $\mathit{fresh}_2 \leftarrow \neg R(\overline{c})$, where $\mathit{fresh}_1$ and $\mathit{fresh}_2$ are fresh predicate symbols.
From this and the fact that all rules in $\mathit{PL}(\phi)$ are both strongly and weakly connected for ${\cal U}$, it follows that there is always unguarded unsafe structure $S'$ that can be obtained from $S$ by removing the rules in $\mathit{PL}(\phi)$.
This contradicts the fact that $p$ is a relaxed acyclic program.
\end{proof}

Here, we show that for acyclic \atklog{} models, \tool{} has \textsc{Ptime} data complexity.

\begin{proposition}
Let $C$ be a system configuration and $f$ be the \acf{} shown in Figure~\ref{figure:angerona:algorithm}.
For any acyclic $C$-\atklog{} model $\mathit{ATK}$, any action $(u,q)$ such that $q$ is a literal query compatible with $p$, and any run $r \in \mathit{runs}(C,f)$ such that (1) all sentences in $\mathit{knowledge}(r.h,u)$, for any $u \in r.U$, are  literal queries compatible with $p$, and (2) all secrets in $r.S$ are literal queries compatible with $p$, the algorithm shown in Figure~\ref{figure:angerona:algorithm} has \textsc{Ptime} data complexity.
\end{proposition}

\begin{proof}
Let $B$ be the \problog{} program obtained from $\mathit{ATK}$.
The data complexity of $f$ is its complexity when only the database $r.\mathit{db}$ and $\mathit{edb}(B)$ change.
The algorithm in Figure~\ref{figure:angerona:algorithm} calls three times the $\mathit{secure}$ procedure and twice the $\mathit{pox}$ procedure.
The set $\mathit{knowledge}(h,u)$ can be constructed in $O(|h|)$.
From this, it follows that the program $p$ has size $O(|\mathit{edb}(B)|+|\mathit{rules}(B)|+|h|)$.
Furthermore, since all queries are literal queries, $B$ is a relaxed acyclic \problog{} program, and \thref{theorem:composition:literal:queries}, it follows that $p$ is a relaxed acyclic \problog{} program as well. 
From this and \thref{theorem:complexity}, the inference can be performed in $O(r^2\cdot e^{r} \cdot \mathit{max}(2,r)^{2+r})$, where $r \in O(|\mathit{rules}(B)|+(|r|+|\mathit{edb}(K)|)^{|\mathit{rules}(K)|}) $ and $e \in O(|\mathit{edb}(B)|)$.
Since $\mathit{rules}(B)$, $r$, and $K$ are fixed, then the data complexity of the $\mathit{secure}$ procedure is $O(|\mathit{edb}(B)|^k)$, for some $k \in \mathbb{N}$.
From this and the fact that Algorithm~\ref{figure:angerona:algorithm} three times the $\mathit{secure}$ procedure and twice the $\mathit{pox}$ procedure per secret in $r.S$, it follows that the data complexity of  Algorithm~\ref{figure:angerona:algorithm} is \textsc{Ptime}.
\end{proof}

\subsection{Completeness Proof}

We first introduce the notion of unconditionally secrecy-preserving query. Informally, a query $\langle u',q\rangle$ is unconditionally secrecy-preserving given a run $r$ and a secret $\langle u, \psi, l \rangle$ iff disclosing the result of $\langle u',q\rangle$ in any run $r' \sim_u r$ does not violate the secret.

\begin{definition}
Let $C = \langle D, \Gamma\rangle$ be a configuration, $f$ be a $C$-\acf{}, and $\mathit{ATK}$ be a $(C,f)$-attacker model.    
A query~$q$~is \emph{unconditionally secrecy-preserving} for a $(C,f)$-run $r$, a user $u$,  a secret $\langle u, \psi, l \rangle$,  and $\mathit{ATK}$ iff
 $\llbracket \mathit{ATK} \rrbracket (u,r)( \psi ) < l$ implies that
 \begin{inparaenum}[(a)]
 \item if $\llbracket \mathit{ATK} \rrbracket (u,r)( q ) > 0$, then $\llbracket \mathit{ATK} \rrbracket (u,r)( \psi \mid \{ \mathit{db} \in \llbracket r \rrbracket_{\sim_u} \mid [q]^\mathit{db} = \top \}) < l$, and 
 \item if $\llbracket \mathit{ATK} \rrbracket (u,r)( \neg q ) > 0$, then %$\llbracket \mathit{ATK} \rrbracket (u,\mathit{extend}(r, \langle u',q\rangle, \bot))( \psi) < l$,
 $\llbracket \mathit{ATK} \rrbracket (u,r)( \psi \mid \{ \mathit{db} \in \llbracket r \rrbracket_{\sim_u} \mid [q]^\mathit{db} = \bot \}) < l$.
 \end{inparaenum}
\end{definition}

We say that a \acf{} is \emph{complete} if it authorizes all unconditionally secrecy-preserving queries.
\tool{} is complete.
This directly follows from (1) \thref{theorem:knowledge:is:enough}, (2) the use of exact inference procedures for \problog{} programs, and (3) the fact that \tool{} directly checks whether queries are unconditionally secrecy-preserving or not.

\begin{proposition}
Let $C = \langle D, \Gamma\rangle$ be a system configuration, $\mathit{ATK}$ be a $C$-\atklog{} model, $f$ be the $C$-\acf{} obtained by parametrizing Algorithm~\ref{figure:angerona:algorithm} with $C$ and $\mathit{ATK}$,
% $\mathit{ATK}' = \langle \lambda u \in {\cal U}. \llbracket \mathit{ATK}(u) \rrbracket_D, \lambda u \in {\cal U}.\sim_u \rangle $ be the $(C,f)$-attacker model associated to $\mathit{ATK}$,
 $\mathit{ATK}' = \lambda u \in {\cal U}. \llbracket \mathit{ATK}(u) \rrbracket_D$ be the $(C,f)$-attacker model associated to $\mathit{ATK}$,
 $r = \langle \langle \mathit{db}, U,P \rangle, h \rangle$ be a run, and $\langle u,q\rangle$ be a $C$-query.
If $q$ is unconditionally secrecy-preserving for $u$, for all secrets $\langle u, \psi, l \rangle \in \mathit{secrets}(P,u)$, $r$, and $f$ authorizes $\langle u,q\rangle$. 
\end{proposition} 

\begin{proof}
Let $C = \langle D, \Gamma\rangle$ be a system configuration, $\mathit{ATK}$ be a $C$-\atklog{} model, $f$ be the $C$-\acf{} obtained by parametrizing Algorithm~\ref{figure:angerona:algorithm} with $C$ and $\mathit{ATK}$, 
$\mathit{ATK}' = \lambda u \in {\cal U}. \llbracket \mathit{ATK}(u) \rrbracket_D$ be the $(C,f)$-attacker model associated to $\mathit{ATK}$, 
$r = \langle \langle \mathit{db}, U,P \rangle, h \rangle$ be a run, and $\langle u,q\rangle$ be a $C$-query.
Assume, for contradiction's sake, that our claim does not hold.
Namely, there is query $\langle u,q\rangle$ such that (1) the query is unconditionally secrecy-preserving $u$, for all secrets $\langle u, \psi, l \rangle \in \mathit{secrets}(P,u)$, $r$, and $\mathit{ATK}'$, and (2) \tool{} (parametrized with $\mathit{ATK}$) does not authorize $\langle u,q\rangle$. 
Since \tool{} does not authorize $\langle u,q\rangle$, it means that $\tool{}(C,\langle \mathit{db}, U,P \rangle, h, \mathit{ATK}, \langle u,q \rangle) = \bot$.
From this, it follows that there is a secret $\langle u, \psi, l \rangle \in \mathit{secrets}(P,u)$ such that:
\begin{compactitem}
\item  $\mathit{secure}(C,\mathit{ATK}, h , \langle u, \psi, l \rangle) = \top$, and
\item one of the two cases hold:
\begin{compactitem}
\item $\mathit{pox}(C,\mathit{ATK},  h, \langle u, q \rangle, \top) = \top$ and $\mathit{secure}(C,\mathit{ATK},  h', \langle u, \psi, l \rangle) = \bot$, where $h' = h \cdot \langle \langle u,q \rangle, \top, top \rangle$, or
\item $\mathit{pox}(C,\mathit{ATK},  h, \langle u, q \rangle, \bot) = \top$ and $\mathit{secure}(C,\mathit{ATK}, \langle \mathit{db}, U,P \rangle, h'', \langle u, \psi, l \rangle) = \bot$, where $h'' = h \cdot \langle \langle u,q \rangle, \top, \bot \rangle$.
\end{compactitem}
\end{compactitem}
Without loss of generality, we assume that $\mathit{pox}(C,\mathit{ATK},  h, \langle u, q \rangle, \top) = \top$ and $\mathit{secure}(C,\mathit{ATK},  h' \langle u, \psi, l \rangle) = \bot$, where $h' = h \cdot \langle \langle u,q \rangle, \top, top \rangle$ (the proof for the other case is identical).
From $\mathit{secure}(C,\mathit{ATK}, h,  \langle u, \psi, l \rangle) = \top$, \thref{theorem:secure:sound:complete}, and the fact that $s$ is compatible with $h$, it follows that $\llbracket \mathit{ATK}' \rrbracket(u,r)(\psi) < l$.
From  $\mathit{pox}(C,\mathit{ATK}, h, \langle u, q \rangle, \top) = \top$, \thref{theorem:secure:sound:complete},  and the fact that $s$ is compatible with $h$, it follows that $\llbracket \mathit{ATK}' \rrbracket(u,r)(q) > 0$.
From this and \thref{theorem:angerona:same:result:knowledge}, it follows that there is a system state $s' = \langle \mathit{db}', U,P\rangle$ such that (1) $q$ holds in $\mathit{db}'$, (2) $s'$ is compatible with $h|_u$, and (3) $r' = \langle s', h|_u\rangle$ is a run indistinguishable from $r$.
From this, $\mathit{secure}(C,\mathit{ATK},  h', \langle u, \psi, l \rangle) = \bot$, $\mathit{secure}(C,\mathit{ATK},  h', \langle u, \psi, l \rangle) = \mathit{secure}(C,\mathit{ATK}, \langle \mathit{db}, U,P \rangle, h'|_u, \langle u, \psi, l \rangle)$, and \thref{theorem:secure:sound:complete}, it follows that  $\llbracket \mathit{ATK}' \rrbracket(u, \langle s', h'|_u\rangle)(\psi) \geq l$.
From this and the definition of $h'|_u$, it follows that  $\llbracket \mathit{ATK}' \rrbracket(u,r')(\psi \mid \{ \mathit{db} \in \llbracket r \rrbracket_{\sim_u} \mid [q]^\mathit{db} = \top \} ) \geq l$.
From this, $r \sim_u r'$, and $\llbracket \mathit{ATK}' \rrbracket(u,r) = \llbracket \mathit{ATK}' \rrbracket(u,r')$ if $r \sim_u r'$, it follows that $\llbracket \mathit{ATK}' \rrbracket(u,r)(\psi \mid \{ \mathit{db} \in \llbracket r \rrbracket_{\sim_u} \mid [q]^\mathit{db} = \top \} ) \geq l$.
Therefore, we have that $\llbracket \mathit{ATK}' \rrbracket(u,r)(\psi) < l$, $\llbracket \mathit{ATK}' \rrbracket(u,r)(q) > 0$, and $\llbracket \mathit{ATK}' \rrbracket(u,r)(\psi \mid \{ \mathit{db} \in \llbracket r \rrbracket_{\sim_u} \mid [q]^\mathit{db} = \top \} ) \geq l$.
This contradicts the fact that $q$ is an unconditionally secrecy-preserving query and completes the proof of our claim.
\end{proof}

}

\end{document}